\documentclass[journal]{IEEEtran}
\usepackage{amsmath,amsfonts,amssymb,amsthm,mathtools}

\usepackage{algorithm}
\usepackage{algpseudocode}
\usepackage{array}
\usepackage[caption=false,font=normalsize,labelfont=sf,textfont=sf]{subfig}
\usepackage{textcomp}
\usepackage{stfloats}
\usepackage{url}
\usepackage{multirow}
\usepackage{verbatim}
\usepackage{graphicx}
\usepackage{cite}
\usepackage{geometry}
\usepackage{float}
\usepackage{enumitem}
\usepackage{tabularx, booktabs, ragged2e}
\newcolumntype{L}{>{\RaggedRight\arraybackslash}X}
\usepackage{placeins}
\usepackage{makecell}
\usepackage{etoolbox} 
\newcommand{\Dd}{\text{Dd}}
\newcommand{\Atom}{\text{Atom}}

\newcommand{\Cn}{\text{Cn}}
\newcommand{\enc}{\text{enc}}

\newcommand{\E}{\mathbb{E}}

\newcommand{\timeindex}{\mathsf{time}}

\newcommand{\syn}{\equiv_{\mathcal L_{\mathrm{sem}}}}

\DeclareMathOperator{\supp}{supp}

\newtheorem{axiom}{Axiom}[section]
\newtheorem{assumption}{Assumption}[section]
\newtheorem{theorem}{Theorem}[section]
\newtheorem{lemma}{Lemma}[section]
\newtheorem{proposition}{Proposition}[section]
\newtheorem{corollary}{Corollary}[section]
\newtheorem{definition}{Definition}[section]
\newtheorem{remark}{Remark}[section]
\newtheorem{example}{Example}[section]

\numberwithin{equation}{section}

\makeatletter
\renewenvironment{proof}[1][\proofname]{\par
  \pushQED{\qed}%
  \normalfont \topsep6\p@\@plus6\p@\relax
  \trivlist
  \item[\hskip\labelsep
        \itshape
    #1\@addpunct{.}]\ignorespaces
}{%
  \popQED\endtrivlist\@endpefalse
}
\makeatother

\begin{document}

\title{Semantic Channel Theory: Deductive Compression and Structural Fidelity for Multi-Agent Communication}

\author{Jianfeng~Xu$^{1}$%
\thanks{$^{1}$Koguan School of Law, China Institute for Smart Justice, School of Computer Science, Shanghai Jiao Tong University, Shanghai 200030, China. Email: xujf@sjtu.edu.cn}%
}

\maketitle

\begin{abstract}
Shannon's information theory deliberately excludes message semantics.
This paper develops a rigorous framework for semantic communication that
integrates formal proof systems with Shannon-theoretic tools.
We introduce an axiomatic information model comprising
$\mathcal L_{\mathrm{sem}}$-definable state sets linked by computable
enabling maps, and define the semantic channel as a composition of Markov
kernels whose supports respect the enabling structure.
A fixed proof system induces an irredundant semantic core and a
derivation-depth stratification, enabling four distortion measures of
increasing semantic depth: Hamming, closure, depth, and a parameterized
composite.
Six families of computable semantic channel invariants are defined and
their inter-relationships established, including a data processing bound,
a semantic Fano bound, and an ideal-channel collapse theorem.
The central quantitative result is a deductive compression gain: under
closure-based fidelity, the minimum blocklength is determined by the
irredundant core size rather than the full knowledge-base size.
We instantiate the framework for heterogeneous multi-agent communication,
introducing an overlap decomposition that yields necessary and sufficient
conditions for closure-reliable communication.
A semantic bottleneck phenomenon is identified in broadcast settings:
vocabulary mismatch imposes irreducible fidelity limitations even over
noiseless carriers.
All results are verified on an explicit Datalog instance.
\end{abstract}


\section{Introduction}
\label{sec:introduction}

\IEEEPARstart{C}{ommunication} is, in Shannon's celebrated
formulation, the problem of reproducing at one point a message
selected at another~\cite{shannon1948mathematical}.
Shannon's mathematical theory deliberately sets aside the
\emph{meaning} of messages: ``the semantic aspects of
communication are irrelevant to the engineering problem''
\cite{shannon1948mathematical}.
This separation has been extraordinarily productive, yielding
the fundamental limits of data compression, channel coding, and
network information
theory~\cite{cover2006elements,csiszar2011information}.
Yet already in the companion essay by
Weaver~\cite{weaver2017recent}, a three-level hierarchy was
envisaged: the \emph{technical} problem of accurate symbol
transmission (Level~A), the \emph{semantic} problem of conveying
intended meaning (Level~B), and the \emph{effectiveness} problem
of achieving the desired effect at the receiver (Level~C).
Shannon's theory addresses Level~A with unmatched precision; the
present paper develops a rigorous framework for Level~B that
is mathematically compatible with---and strictly
generalizes---Level~A.

The need for such a framework has intensified with the rise of
knowledge-intensive communication systems.
In multi-agent coordination, retrieval-augmented generation (RAG)
pipelines, and federated knowledge-base synchronization, the
communicating parties exchange not raw symbols but structured
\emph{semantic states}---ground facts, rules, or queries drawn
from a shared or partially overlapping logical vocabulary.
A symbol-level error that leaves the deductive content unchanged
is harmless, while a symbol-level ``success'' that corrupts a
single irredundant axiom may destroy an entire branch of
derivable knowledge.
Classical information theory, which treats source and channel
alphabets as unstructured label sets, cannot make this
distinction.
A principled theory of semantic communication must therefore
integrate \emph{logical structure}---deductive closure,
irredundant bases, derivation depth---with
\emph{information-theoretic structure}---entropy, mutual
information, channel capacity, and coding theorems.

\subsection*{Related Work and Gap Identification}

Semantic information has a long intellectual history, from the
possible-worlds measure of Carnap and
Bar-Hillel~\cite{carnap1952outline} and Floridi's theory of
strongly semantic information~\cite{floridi2004outline},
through Kolchinsky and Wolpert's viability-based
formulation~\cite{kolchinsky2018semantic}, to the recent
deep-learning-driven systems exemplified by
DeepSC~\cite{xie2021deep} and surveyed
in~\cite{qin2021semantic,gunduz2022beyond,luo2022semantic}.
A parallel line of work has established a mathematical theory
of semantic communication based on synonymous
mappings~\cite{niu2024mathematical,zhang2024modern}, deriving
semantic entropy, capacity, and rate-distortion functions, with
companion coding algorithms for many-to-one
sources~\cite{ma2025theory}, semantic
rate-distortion computation~\cite{han2025extended}, and
semantic arithmetic coding~\cite{liang2025semantic}.
On the multi-agent side, recent contributions address
heterogeneous semantics via Bayesian contextual
reasoning~\cite{seo2023bayesian}, logic-driven resilience
frameworks~\cite{alshammari2026logic}, and goal-oriented
metric unification~\cite{li2024toward,wu2024toward}.
The logical substrate for structured communication draws on
descriptive complexity~\cite{immerman1999descriptive},
Datalog fixpoint
theory~\cite{ceri1989you,abiteboul1995foundations}, and
knowledge compilation~\cite{dantsin2001complexity}.

These diverse lines of research each illuminate important
facets of semantic communication, and together they suggest a
natural next step: \emph{integrating formal proof systems
(deductive closure, irredundant cores, derivation depth) with
Shannon-theoretic tools (entropy, mutual information, channel
capacity, coding theorems) to characterize the fundamental
limits of semantic communication under knowledge-base
heterogeneity}.
The synonymous-mapping
approach~\cite{niu2024mathematical} elegantly extends Shannon
theory through probabilistic equivalence classes and opens the
door to incorporating logical inference structure; the
deep-learning
approaches~\cite{xie2021deep,qin2021semantic} demonstrate
impressive practical performance and motivate the development
of formal guarantees on deductive fidelity; the multi-agent
frameworks~\cite{seo2023bayesian,alshammari2026logic} advance
coordination and resilience analysis and invite companion
coding-theoretic results; and the logical-complexity
literature~\cite{immerman1999descriptive,marx2013tractable}
provides the mathematical substrate on which communication-rate
and channel-capacity analyses can be built.
A detailed discussion of each line of work and its relationship
to the present framework is provided in
Section~\ref{sec:related-works}.

\subsection*{Motivation and Approach}

The gap identified above motivates the central question of this
paper:

\begin{quote}
\itshape
Can one develop a unified framework that combines formal
logic (proof systems, deductive closure, irredundant bases) with
Shannon-theoretic tools (entropy, mutual information, channel
capacity, coding theorems) to characterize the fundamental limits
of semantic communication, and can this framework handle the
vocabulary heterogeneity that arises naturally in multi-agent
settings?
\end{quote}

Our approach to this question originates from an
\emph{ontological} perspective on information~\cite{xu2014objective,xu2024research,xu2025general,qiu2025research}.
Rather than treating information as a probability distribution on
an abstract alphabet (the Shannon view) or as a set of excluded
possible worlds (the Carnap--Bar-Hillel view), we model
information as a collection of \emph{states} of objects in
time---formalized as \(\mathcal L_{\mathrm{sem}}\)-definable
binary relations in a finite many-sorted logical
structure---together with a \emph{computable enabling mechanism}
that links semantic states to physical carrier states.
This ontological information model, developed axiomatically in
Section~\ref{sec:model}, provides the deterministic and logical
substrate on which probabilistic structure is erected.
The key insight is that the proof system shared by communicating
agents induces a natural stratification of the semantic state
space into an \emph{irredundant core}
\(\Atom(S_O)\)---the smallest set from which all stored knowledge
can be re-derived---and a collection of \emph{stored shortcuts}
\(J=S_O\setminus\Atom(S_O)\) that are deductively redundant.
This stratification has direct coding-theoretic consequences:
under a \emph{closure-based} fidelity criterion (which deems a
reconstruction acceptable if it preserves the deductive closure
of the original knowledge base), only the \(|\Atom(S_O)|\)
irredundant core elements need to be transmitted reliably; the
remaining \(|S_O|-|\Atom(S_O)|\) states can be recovered by the
receiver's inference engine at zero additional channel cost.

Building on this foundation, we introduce the \emph{semantic
channel} as a composition of three enabling kernels---encoding,
carrier transmission, and decoding---each constrained by the
enabling structure of an underlying information model.
We then develop a suite of \emph{semantic channel
invariants}---computable scalar quantities that characterize the
channel's information throughput, structural fidelity, and
noise-pair structure---and derive \emph{semantic coding theorems}
that quantify the achievable rates and minimum blocklengths under
semantic reliability criteria.
Finally, we instantiate the framework in a
\emph{heterogeneous multi-agent} setting where the sender and
each receiver maintain different knowledge bases, and derive
results---including a ``semantic bottleneck'' phenomenon in
broadcast channels and a vocabulary-independent deductive
compression ratio---that have no counterpart in classical
Shannon theory.

\subsection*{Main Contributions}

The main contributions of this paper are as follows.

\begin{enumerate}[label=\textup{C\arabic*.},leftmargin=*]
  \item \emph{Axiomatic information model
        \textup{(Section~\ref{sec:model})}:}
        We introduce a many-sorted logical framework comprising
        \(\mathcal L_{\mathrm{sem}}\)-definable state sets, a
        computable enabling mechanism linking semantic and carrier
        states, composable information models, and an ideal-information
        notion based on order-preserving semantic synonymy~\(\syn\).
        A fixed proof system~\(\mathsf{PS}\) induces an
        irredundant semantic core~\(\Atom(S_O)\) and a
        derivation-depth stratification~\(\Dd(\cdot\mid B)\),
        both of which are shown to be computable invariants.
        A noisy-information extension via set perturbations
        \((S_O^-,S_O^+)\) enables the analysis of non-ideal
        carrier realizations.

  \item \emph{Semantic channel and enabling kernels
        \textup{(Section~\ref{sec:channel})}:}
        We define the semantic channel as a composition of Markov
        kernels whose supports are constrained by the enabling
        maps of three composable information models (encoding,
        carrier, decoding).
        The end-to-end kernel is shown to be an enabling kernel
        for the composite model, ensuring that structural
        constraints propagate through the probabilistic layer.

  \item \emph{Semantic distortion measures
        \textup{(Section~\ref{subsec:distortion})}:}
        We introduce four distortion functions of increasing
        semantic depth: Hamming (symbol-level), closure (deductive
        content), depth (inferential complexity), and a
        parameterized composite.
        Under closure distortion, errors on redundant states incur
        zero penalty---a property absent from classical distortion
        theory.

  \item \emph{Semantic channel invariants
        \textup{(Section~\ref{subsec:invariants})}:}
        We define six families of computable invariants---source-side
        structural, set-level fidelity, noise-pair probabilistic,
        structural quality, receiver-side comparison, and
        information-theoretic---and establish their inter-relationships,
        including a data processing bound
        \(C_{\mathrm{sem}}(W)\le C(W)\), a semantic Fano bound
        connecting the probabilistic core preservation index to
        mutual information, and an ideal-channel collapse theorem.

  \item \emph{Deductive compression gain
        \textup{(Section~\ref{subsec:coding})}:}
        We derive semantic converse and achievability bounds
        showing that the minimum blocklength for closure-reliable
        communication is determined by \(\log|\Atom(S_O)|\) rather
        than \(\log|S_O|\), yielding a deductive compression ratio
        \(\log|\Atom(S_O)|/\log|S_O|\) that measures the
        channel-use savings afforded by receiver-side inference.
        A semantic rate--distortion function
        \(R_{\mathrm{sem}}(D;d_{\Cn})\) is shown to satisfy
        \(R_{\mathrm{sem}}(0)\le\log|\Atom(S_O)|\).

  \item \emph{Heterogeneous multi-agent communication
        \textup{(Section~\ref{sec:application})}:}
        We formalize knowledge-base heterogeneity via a pairwise
        \emph{overlap decomposition} and derive
        necessary and sufficient conditions for closure-reliable
        communication in terms of core coverage and surplus
        derivability.
        The deductive compression ratio is shown to be
        invariant under vocabulary heterogeneity.
        In the broadcast setting, we prove that the minimum
        blocklength is independent of the number of receivers
        and identify a \emph{semantic bottleneck}
        phenomenon---an irreducible fidelity limitation due to
        vocabulary mismatch that persists even over noiseless
        carriers---that has no counterpart in classical broadcast
        theory.
        All results are verified on an explicit Datalog
        knowledge-base instance with full numerical computation of
        every invariant.
\end{enumerate}

\subsection*{Paper Organization}

The remainder of this paper is organized as follows.
Section~\ref{sec:model} establishes the axiomatic foundations:
the logical language, information models, synonymous state sets,
irredundant cores, derivation depth, and noisy information.
Section~\ref{sec:channel} builds the probabilistic layer:
enabling kernels, the semantic channel, distortion measures,
channel invariants, and preliminary coding theorems.
Section~\ref{sec:application} instantiates the framework for
heterogeneous multi-agent communication, derives the main
analytical results, and presents the numerical example.
Section~\ref{sec:conclusion} discusses the significance of the
results, identifies limitations, and outlines directions for
future research.

\textit{Notation.}
Sets are denoted by uppercase letters, elements by lowercase.
\(\mathbb{S}_O\) is the ambient semantic universe;
\(S_O,S_C\) are semantic and carrier state sets;
\(\Cn(\cdot)\) is the deductive closure operator;
\(\Atom(S_O)\) is the irredundant core;
\(\Dd(s\mid B)\) is the derivation depth;
\(\syn\) denotes semantic synonymy;
\(\kappa:X\rightsquigarrow Y\) denotes a Markov kernel;
\(H(\cdot)\), \(I(\cdot;\cdot)\) are Shannon entropy and mutual
information (base~2, in bits);
\(d_H\), \(d_{\Cn}\), \(d_{\Dd}\), \(d_{\mathrm{sem}}\) are
the four distortion functions.
Detailed notation is introduced where needed; a summary for the
multi-agent application appears in Table~\ref{tab:notation-iv}.


\section{Related Works}
\label{sec:related-works}

\noindent\textbf{Semantic Information Theory.}
The formal study of semantic information dates back to
Carnap and Bar-Hillel~\cite{carnap1952outline}, who proposed
measuring the information content of a proposition by the set
of possible worlds it excludes.
Floridi~\cite{floridi2004outline} further developed a theory
of strongly semantic information, defining information as
well-formed, meaningful, and truthful data.
These pioneering contributions established the philosophical
and logical foundations for semantic information and opened up
the question of how to connect such measures to operational
coding-theoretic frameworks in the tradition of
Shannon~\cite{shannon1948mathematical}.
Recently, Niu and
Zhang~\cite{niu2024mathematical,zhang2024modern} established a
systematic mathematical theory of semantic communication,
introducing \emph{synonymous mapping} as a core concept and
deriving semantic entropy, semantic channel capacity, and
rate-distortion functions.
A key result of their framework is that the semantic channel
capacity can exceed the Shannon capacity, \(C_s\ge C\), by
exploiting an average synonymous length~\(S\) that allows
multiple distinguishable source sequences to be treated as
equivalent when they express the same meaning.
Their framework elegantly extends Shannon theory by partitioning
the source alphabet into synonymous equivalence classes, and
leaves room for further exploration of how \emph{logical
inference structure}---proof systems, deductive closure, and
irredundant cores---can be integrated into the
information-theoretic framework.
Our work pursues this complementary direction: the compression
gain in our framework arises from the receiver's ability to
\emph{re-derive} redundant states via logical inference, which
complements the source-side synonymous-collapsing mechanism
of~\cite{niu2024mathematical,zhang2024modern}.

\noindent\textbf{Semantic Coding Theorems and Algorithms.}
Building on semantic information theory, several recent works
have derived coding theorems and practical algorithms.
Ma et~al.~\cite{ma2025theory} proved a semantic channel coding
theorem for many-to-one sources, where multiple source
sequences express the same meaning, using a generalized Fano's
inequality.
Han et~al.~\cite{han2025extended} proposed an extended
Blahut--Arimoto (EBA) algorithm to compute the semantic
rate-distortion function under synonymous mappings, and further
integrated it with simulated annealing to find optimal
synonymous mappings.
Liang et~al.~\cite{liang2025semantic} introduced semantic
arithmetic coding (SAC), achieving higher compression efficiency
than classical arithmetic coding by performing encoding over
synonymous sets.
These works provide valuable coding-theoretic results and
computational methods for achieving semantic compression and
transmission limits within the synonymous-mapping paradigm, and
invite further investigation into alternative compression
mechanisms grounded in logical structure.
Our work complements these contributions by approaching the
fundamental limits from a \emph{logical-structural} perspective:
the minimum blocklength for closure-reliable communication is
determined by \(\log|\Atom(S_O)|\) rather than \(\log|S_O|\)
(Theorem~\ref{thm:achievability} and
Corollary~\ref{cor:min-blocklength}), where the compression
mechanism is receiver-side deductive inference rather than
source-side synonymous collapsing.

\noindent\textbf{Heterogeneous and Multi-Agent Semantic
Communication.}
The challenge of heterogeneous knowledge bases in multi-agent
systems has gained increasing attention.
Seo et~al.~\cite{seo2023bayesian} addressed the inverse
contextual reasoning problem, where an agent must infer the
communication context of others from noisy observations, using
Bayesian inference and Markov Chain Monte Carlo methods.
Alshammari and Bennis~\cite{alshammari2026logic} proposed a
logic-driven semantic communication framework for resilient
multi-agent systems, formalizing epistemic and action resilience
using temporal epistemic logic and Kripke structures.
These works advance multi-agent semantic communication from
algorithmic and logical-resilience perspectives, respectively,
and open the door to further study of the coding-theoretic
limits of communication under vocabulary mismatch.
Our work addresses this aspect by introducing an \emph{overlap
decomposition}
(Definition~\ref{def:overlap-decomposition}) that characterizes
vocabulary mismatch via core loss~\((A_-^{ij})\) and
non-derivable surplus~\((S_{+,n}^{ij})\), and by deriving
minimum-blocklength bounds and identifying a \emph{semantic
bottleneck} phenomenon
(Proposition~\ref{prop:broadcast-bottleneck}) that complement
the algorithmic and resilience-oriented analyses
of~\cite{seo2023bayesian,alshammari2026logic}.

\noindent\textbf{Goal-Oriented and Task-Oriented
Communications.}
Goal-oriented communication has emerged as a prominent
paradigm, where the effectiveness of communication is measured
by task performance rather than bit-level
accuracy~\cite{strinati20216g,kountouris2021semantics}.
Li et~al.~\cite{li2024toward} proposed a goal-oriented tensor
framework that unifies various significance metrics such as Age
of Information, Value of Information, and Age of Incorrect
Information.
Wu et~al.~\cite{wu2024toward} introduced information-theoretic
metrics including semantic entropy, semantic distortion, and
semantic communication rate, and provided guidelines for
designing interpretable semantic communication systems.
These works establish a rich landscape of task-specific fidelity
metrics and motivate further exploration of how such metrics
relate to the logical structure of the communicated content.
Our distortion measures
(Definitions~\ref{def:hamming-distortion}--\ref{def:depth-distortion})
offer one such connection: the composite semantic distortion
\(d_{\mathrm{sem}}=\alpha\,d_H+\beta\,d_{\Cn}+\gamma\,d_{\Dd}\)
subsumes Hamming distortion (classical), closure distortion
(deductive content), and depth distortion (inferential
complexity) as special cases, and the semantic Fano bound
(Theorem~\ref{thm:semantic-fano}) links these proof-system-aware
distortion measures to mutual information, offering a
perspective that complements the metric-design viewpoint
of~\cite{li2024toward,wu2024toward}.

\noindent\textbf{Logical and Database Foundations.}
The logical substrate of our framework draws on descriptive
complexity and database theory.
Immerman~\cite{immerman1999descriptive} established the
correspondence between \(\mathrm{FO(LFP)}\) and polynomial
time on ordered finite structures, which we use to ensure
computational tractability of all semantic operations
(Assumption~\ref{assump:ordered-structures}).
Marx~\cite{marx2013tractable} introduced submodular width as a
complexity measure for conjunctive query evaluation, and
Abo~Khamis and Chen~\cite{abokhamis2025jaguar} recently
proposed the Jaguar algorithm achieving
\(O(N^{\mathrm{subw}(Q)+\epsilon})\)-time query evaluation.
Mu~\cite{mu2024identifying} studied the identification of
formula roles in inconsistency, introducing causal and
informational responsibility measures that parallel, at a
methodological level, our distinction between irredundant core
elements and redundant stored shortcuts.
These results from logic and database theory provide the
technical foundations on which our framework builds, and they
suggest further opportunities for cross-fertilization with
communication theory.
In our framework, the irredundant core~\(\Atom(S_O)\) serves as
the minimal generating set for a deductive closure
(Definition~\ref{def:atom-so},
Proposition~\ref{prop:atom-core-correct}), and its cardinality
directly determines the minimum number of channel uses under
closure reliability
(Corollary~\ref{cor:min-blocklength})---a connection between
proof-system structure and communication fundamental limits that,
to the best of our knowledge, has not been previously
established.

\noindent\textbf{Summary of Positioning.}
The works reviewed above have made substantial contributions to
semantic information measures, coding algorithms, multi-agent
coordination, goal-oriented metrics, and logical complexity,
collectively advancing the field toward a mature theory of
semantic communication.
Our work aims to build upon these contributions by integrating
\emph{formal proof systems} with \emph{Shannon-theoretic tools},
with the goal of deriving \emph{computable semantic channel
invariants} and \emph{deductive compression gains} applicable
to heterogeneous multi-agent settings.
We view our framework as complementary to the existing body of
work: the synonymous-mapping
approach~\cite{niu2024mathematical} captures source-side
semantic equivalence, the deep-learning
approaches~\cite{xie2021deep,qin2021semantic} provide practical
system implementations, the multi-agent
frameworks~\cite{seo2023bayesian,alshammari2026logic} address
coordination and resilience, and the
logical-complexity
literature~\cite{immerman1999descriptive,marx2013tractable}
supplies the mathematical substrate.
Our contribution is to connect these threads through a unified
coding-theoretic framework in which the proof-system structure
of the communicated knowledge base enters directly into the
channel capacity, distortion, and blocklength expressions.

\section{System Model and Axiomatic Foundations}
\label{sec:model}

This section lays the formal foundations on which the semantic channel theory 
and the associated information metrics are built.
The framework is organized in five layers:
(i)~a fixed many-sorted logical language~\(\mathcal L\), a designated semantic
sublanguage~\(\mathcal L_{\mathrm{sem}}\), and an effective proof system;
(ii)~a two-domain \emph{information model} with time-indexed semantic and carrier state
sets \((S_O,S_C)\) linked by a computable enabling mechanism;
(iii)~an \emph{ideal-information} notion characterized by order-preserving semantic synonymy
\(S_O\syn S_C\);
(iv)~an \emph{irredundant semantic core}~\(\Atom(S_O)\) and a base-relative
\emph{derivation-depth} stratification \(\Dd(\cdot\mid B)\), both induced by the immediate
consequence operator~\(T_{\mathsf{PS}}\); and
(v)~a \emph{noisy-information} framework that extends the ideal theory to non-ideal carrier
realizations via set perturbations of~\(S_O\).
Together these components supply the deterministic and logical substrate to which
probabilistic structure---Markov kernels, source distributions, semantic distortion
measures---is adjoined in Section~\ref{sec:channel} to define the \emph{semantic channel}
and its invariants.

\subsection{Underlying Logical System \(\mathcal{L}\) and Expressible State Sets}
\label{subsec:logic}

Throughout we fix a many-sorted logical language \(\mathcal{L}\), taken to be
first-order logic with least fixed-point operators \(\mathrm{FO(LFP)}\)
\cite{immerman1999descriptive}, extended with multiple sorts including at least
\(\mathsf{Obj}\) (objects/entities), \(\mathsf{Time}\) (time points), and
\(\mathsf{Carrier}\) (physical carriers).

\begin{assumption}[Finite ordered structures (standing convention)]
\label{assump:ordered-structures}
We restrict attention to \emph{finite} \(\mathcal{L}\)-structures.
Moreover, whenever descriptive-complexity claims are invoked, we work over \emph{ordered} finite
structures, i.e., structures equipped with a built-in linear order (e.g., a distinguished binary
relation symbol \(<\) interpreting a total order on the active domain(s)).
This convention is used solely to justify the standard correspondence that \(\mathrm{FO(LFP)}\)
captures polynomial time on ordered finite structures
\cite{immerman1982relational,vardi1982complexity}.
\end{assumption}

\begin{remark}[Ordered representations vs.\ semantic content]
\label{rem:ordered-vs-semantic-language}
The built-in linear order on finite structures is used only as a \emph{representation artifact}
to support descriptive-complexity statements (e.g., that \(\mathrm{FO(LFP)}\) captures PTIME on ordered
finite structures). It is treated as an element of \(\Sigma_{\mathrm{rep}}\) and is therefore excluded from \(\mathcal L_{\mathrm{sem}}\)
(Assumption~\ref{assump:semantic-sublanguage}).
\end{remark}

\begin{assumption}[Designated semantic sublanguage via a semantic signature]
\label{assump:semantic-sublanguage}
Fix the non-logical signature \(\Sigma\) of the ambient language \(\mathcal L\).
We assume a fixed partition
\[
\Sigma \;=\; \Sigma_{\mathrm{sem}} \,\dot\cup\, \Sigma_{\mathrm{rep}},
\]
where \(\Sigma_{\mathrm{rep}}\) contains auxiliary representation/bookkeeping symbols
(in particular, the built-in order symbol \(<\) whenever present).

\smallskip
Define the designated semantic sublanguage \(\mathcal L_{\mathrm{sem}}\) to be the restriction of
\(\mathcal L\) to the semantic signature:
\[
\mathcal L_{\mathrm{sem}} \;:=\; \mathcal L\!\upharpoonright_{\Sigma_{\mathrm{sem}}},
\]

\noindent For example, if \(\mathcal L=\mathrm{FO(LFP)}[\Sigma]\), then
\(\mathcal L_{\mathrm{sem}}=\mathrm{FO(LFP)}[\Sigma_{\mathrm{sem}}]\).

\smallskip
\noindent\textbf{Disclosure convention.}
Every non-logical symbol used in formal definitions is understood to belong to exactly one of
\(\Sigma_{\mathrm{sem}}\) or \(\Sigma_{\mathrm{rep}}\); symbols are semantic by default unless explicitly
declared representation-only.
All notions of semantic indistinguishability and definable recoding in
Section~\ref{subsec:ideal} are taken relative to \(\mathcal L_{\mathrm{sem}}\).
\end{assumption}

Since we work with finite \(\mathcal L\)-structures, model checking
\((\mathfrak R\models\varphi)\) is decidable:
given a finite encoding of \(\mathfrak R\) and a sentence \(\varphi\in\mathcal L\), one can
decide whether \(\mathfrak R\models\varphi\).
(For \(\mathcal L=\mathrm{FO(LFP)}\) on ordered finite structures, the \emph{data complexity}---i.e.,
the complexity measured in the size of \(\mathfrak R\) for a fixed formula---is polynomial
\cite{immerman1982relational,vardi1982complexity};
the \emph{combined complexity} in both \(|\mathfrak R|\) and \(|\varphi|\) is higher.
Only effective decidability, not a specific complexity bound, is used in the sequel.)
Assumption~\ref{assump:ordered-structures} is used only when invoking descriptive-complexity correspondences
(e.g., \(\mathrm{FO(LFP)}\) capturing PTIME on \emph{ordered} finite structures).
This semantic decidability is used only for the evaluation of \(\mathcal L\)-descriptions on finite structures,
and should be kept conceptually separate from syntactic derivability in the proof system fixed below.

\begin{remark}[Definability and effective enumerability on finite structures]
\label{rem:definable-enumerable}
Since the ambient structure \(\mathfrak R\) is finite
(Assumption~\ref{assump:ordered-structures}), the extension of any
\(\mathcal L\)-definable (or \(\mathcal L_{\mathrm{sem}}\)-definable) relation in
\(\mathfrak R\) can be computed by exhaustive model checking over the finite domain.
In particular, for any \(\mathcal L_{\mathrm{sem}}\)-definable state set \(S\) in
\(\mathfrak R\), the set of pairs \((x,t)\) satisfying \(S(x,t)\) is effectively enumerable.
This bridge between definability and computability is invoked without further comment
whenever a definable relation is treated as computable data.
\end{remark}

\begin{definition}[Expressible object--time state sets]
\label{def:expressible-states}
Let \(\mathfrak R\) be a finite \(\mathcal L\)-structure and let \(\mathcal L'\subseteq \mathcal L\) be a fixed sublanguage.

Assume that \(x\) (resp.\ \(t\)) ranges over a fixed object sort (resp.\ time sort) of \(\mathcal L\).
We allow \(X\) and \(T\) to be \emph{definable subdomains} of these sorts: namely, \(X\) and \(T\) are specified by unary
\(\mathcal L\)-formulas \(\delta_X(x)\) and \(\delta_T(t)\).
(When \(X\) or \(T\) is the full carrier of its sort, the corresponding \(\delta\)-formula can be taken as \(\top\).)

\smallskip
\noindent\textbf{Language convention for subdomains.}
The subdomain formulas \(\delta_X\) and \(\delta_T\) are permitted to belong to the full ambient language~\(\mathcal L\),
not only to the sublanguage~\(\mathcal L'\). Thus the domains \(X\) and \(T\) are fixed as part of the ambient structure,
while the \(\mathcal L'\)-definability requirement below applies only to the state relation~\(S\) itself.

\smallskip
An \emph{(object--time) state domain} over \((X,T)\) is an \emph{\(\mathcal L'\)-definable} binary relation \(S(x,t)\) such that
\[
\mathfrak R \models \forall x\,\forall t\ \bigl(S(x,t)\rightarrow (\delta_X(x)\wedge \delta_T(t))\bigr).
\]
We read \(S(x,t)\) as ``\(x\) is in a (well-formed) state at time \(t\)''.

\smallskip
A set of states \(A\subseteq X\times T\) is \(\mathcal L'\)-\emph{expressible} (over \(S\))
if there exists an \(\mathcal L'\)-formula \(\varphi(x,t)\) such that for all \(x,t\),
\[
(x,t)\in A \quad\Longleftrightarrow\quad \mathfrak R \models S(x,t)\wedge \varphi(x,t).
\]

\smallskip
\noindent\textbf{Convention (state notation as definitional abbreviation).}
Whenever convenient, we write a state variable \(s\in S\) as a meta-level abbreviation for a pair \((x,t)\)
with \(S(x,t)\). Under this convention,
\[
\timeindex(s)=t,
\qquad
\text{and (when needed) } \mathsf{obj}(s)=x,
\]
and quantification over states abbreviates quantification over \((x,t)\) guarded by \(S(x,t)\).

\smallskip
When \(\mathcal L'=\mathcal L_{\mathrm{sem}}\), we also say ``\(\mathcal L_{\mathrm{sem}}\)-definable/expressible''
in this sense.
\end{definition}

\begin{assumption}[Fixed effective proof system for the inference fragment]
\label{assump:proof-system}
We fix a syntactic \emph{inference fragment} \(\mathcal L_{\mathrm{kb}}\subseteq \mathcal L_{\mathrm{sem}}\)
(typically a decidable rule/KB language such as Datalog or a Horn fragment)
and an effective proof system \(\mathsf{PS}\) over the syntax of \(\mathcal L_{\mathrm{kb}}\) such that
proof checking is decidable.

For any finite \(\Gamma\subseteq \mathcal L_{\mathrm{kb}}\) and any \(\varphi\in \mathcal L_{\mathrm{kb}}\), write
\(\Gamma \vdash_{\mathrm{kb}} \varphi\) to denote derivability in \(\mathsf{PS}\), and define the deductive closure operator
\[
\Cn(\Gamma)\;:=\;\{\varphi\in \mathcal L_{\mathrm{kb}}:\ \Gamma\vdash_{\mathrm{kb}}\varphi\}.
\]

\noindent\textbf{Basic closure properties.}
The operator \(\Cn\) satisfies three standing properties used throughout without further comment:
\begin{enumerate}[label=\textup{(Cn\arabic*)}]
  \item \emph{Reflexivity (inclusion of premises):}
        \(\Gamma\subseteq\Cn(\Gamma)\) for every \(\Gamma\subseteq\mathcal L_{\mathrm{kb}}\).
  \item \emph{Monotonicity (weakening):}
        if \(\Gamma\subseteq\Gamma'\), then \(\Cn(\Gamma)\subseteq\Cn(\Gamma')\).
  \item \emph{Idempotence (cut):}
        \(\Cn\bigl(\Cn(\Gamma)\bigr)=\Cn(\Gamma)\) for every \(\Gamma\subseteq\mathcal L_{\mathrm{kb}}\).
        Equivalently: if \(\Delta\subseteq\Cn(\Gamma)\), then \(\Cn(\Delta)\subseteq\Cn(\Gamma)\).
\end{enumerate}
Properties \textup{(Cn1)} and \textup{(Cn2)} follow from the standard weakening rule; \textup{(Cn3)}
follows from the transitivity (cut) property of derivability: if every premise in a derivation of
\(\varphi\) from \(\Delta\) is itself derivable from \(\Gamma\), then \(\varphi\) is derivable
from~\(\Gamma\).

\noindent\textbf{State--formula convention.}
When elements of a state set \(S_O\) are used as premises for \(\Cn(\cdot)\)
(Section~\ref{subsec:atomic-derivation}), each state \((o,\tau)\in S_O\) is identified with the
corresponding ground atom \(S_O(o,\tau)\in\mathcal L_{\mathrm{kb}}\).
\end{assumption}

\subsection{Computable Information Axioms and Information Model}
\label{subsec:info-instance}

\textit{Basic objects (preview).}
We fix semantic/carrier domains with time-indexed state sets; the precise assumptions are stated in
Axioms~\ref{ax:time-domains}--\ref{ax:enabling-mapping}.

\begin{definition}[State spaces as object--time relations (notation)]
\label{def:binary-attribute}
Information will be modeled using domains \(O,C\) with time domains \(T_O,T_C\) and state sets
\(S_O\subseteq O\times T_O\) and \(S_C\subseteq C\times T_C\).
We write \(s_o=(o,\tau)\in S_O\) and \(s_c=(c,\theta)\in S_C\), where the second coordinate is the time index.
\end{definition}

\begin{axiom}[Time domains, indexing, and precedence]
\label{ax:time-domains}
Fix a finite many-sorted \(\mathcal L\)-structure \(\mathfrak R\).
Let \(O,C,T_O,T_C\) be \(\mathcal L\)-definable domains in \(\mathfrak R\), and let \(S_O\subseteq O\times T_O\) and \(S_C\subseteq C\times T_C\) be
\(\mathcal L_{\mathrm{sem}}\)-definable state sets (i.e., definable by \(\mathcal L_{\mathrm{sem}}\)-formulas)
over the \(\mathcal L\)-definable domains \(O,T_O\) and \(C,T_C\) in the sense of
Definition~\ref{def:expressible-states} (with \(\mathcal L'=\mathcal L_{\mathrm{sem}}\)).

\smallskip
Time indexing is the second coordinate. For \(s_o=(o,\tau)\in S_O\) and \(s_c=(c,\theta)\in S_C\), define
\[
\timeindex_O(s_o):=\tau,
\qquad
\timeindex_C(s_c):=\theta,
\]
as definitional abbreviations.

\smallskip
There exist \(\mathcal L_{\mathrm{sem}}\)-definable relations \(\prec_O\subseteq T_O\times T_O\),
\(\prec_C\subseteq T_C\times T_C\), and \(\prec_{OC}\subseteq T_O\times T_C\) such that:
\begin{enumerate}[label=\textup{(T\arabic*)}]
  \item \((T_O,\prec_O)\) and \((T_C,\prec_C)\) are linearly ordered time domains.
  \item \textbf{Cross-domain compatibility:} \(\prec_{OC}\) is monotone with respect to \(\prec_O\) and \(\prec_C\), i.e.,
  for all \(\tau,\tau'\in T_O\) and \(\theta,\theta'\in T_C\),
  \[
    (\tau'\prec_O \tau \ \wedge\ \tau\prec_{OC}\theta)\ \Rightarrow\ \tau'\prec_{OC}\theta,
   \]
   and
   \[
    (\tau\prec_{OC}\theta \ \wedge\ \theta\prec_C\theta')\ \Rightarrow\ \tau\prec_{OC}\theta'.
  \]
  \item \textbf{Nontriviality (mild totality):} for every \(\tau\in T_O\) there exists \(\theta\in T_C\) such that
  \(\tau\prec_{OC}\theta\).
\end{enumerate}
\end{axiom}

\begin{axiom}[State representation]
\label{ax:state-representation}
There exist injective encodings
\[
\enc_O:S_O\to\{0,1\}^*,
\qquad
\enc_C:S_C\to\{0,1\}^*
\]
such that the following predicates are decidable from the codes:
\begin{enumerate}[label=\textup{(R\arabic*)}]
  \item \emph{Semantic-time precedence:} given \(\enc_O(s_o)\) and \(\enc_O(s_o')\), decide whether
        \(\timeindex_O(s_o)\prec_O \timeindex_O(s_o')\).
  \item \emph{Carrier-time precedence:} given \(\enc_C(s_c)\) and \(\enc_C(s_c')\), decide whether
        \(\timeindex_C(s_c)\prec_C \timeindex_C(s_c')\).
  \item \emph{Cross-domain precedence:} given \(\enc_O(s_o)\) and \(\enc_C(s_c)\), decide whether
        \(\timeindex_O(s_o)\prec_{OC}\timeindex_C(s_c)\).
  \item \emph{Membership (optional):} given \(w\in\{0,1\}^*\), decide whether \(w\in \enc_O(S_O)\) (resp.\ \(w\in \enc_C(S_C)\)).
\end{enumerate}
\end{axiom}

\begin{axiom}[Enabling mapping (realization mechanism)]
\label{ax:enabling-mapping}
There exists a relation
\[
R_{\mathcal E}\ \subseteq\ O\times T_O\times C\times T_C,
\]
and the induced set-valued map \(\mathcal E:S_O\Rightarrow S_C\), defined for \(s_o=(o,\tau)\in S_O\) by
\[
\mathcal E(s_o)
\;:=\;
\{\, (c,\theta)\in S_C:\ (o,\tau,c,\theta)\in R_{\mathcal E}\,\},
\]
such that:
\begin{enumerate}[label=\textup{(E\arabic*)}]
  \item \emph{Totality on semantic states:}
        for every \(s_o\in S_O\), \(\mathcal E(s_o)\neq\varnothing\).
  \item \emph{Coverage of the carrier state space:}
        for every \(s_c\in S_C\), there exists \(s_o\in S_O\) with \(s_c\in \mathcal E(s_o)\).
        Equivalently, \(\bigcup_{s_o\in S_O}\mathcal E(s_o)=S_C\).
  \item \emph{Computable enabling selector:}
        there exists a (partial) computable function \(e:\{0,1\}^*\to\{0,1\}^*\) that is total on \(\enc_O(S_O)\),
        such that for every \(s_o\in S_O\), \(e(\enc_O(s_o))\) halts and equals \(\enc_C(s_c)\) for some \(s_c\in\mathcal E(s_o)\).
  \item \emph{Precedence compatibility:}
        if \((o,\tau)\in S_O\), \((c,\theta)\in S_C\), and \((o,\tau,c,\theta)\in R_{\mathcal E}\),
        then \(\tau \prec_{OC} \theta\).
\end{enumerate}
\end{axiom}

\begin{remark}[Notes on the enabling mapping]
\label{rem:enabling-notes}
\textup{(i)}~Unless stated otherwise, the relation \(R_{\mathcal E}\) is treated as meta-level instance data and is
\emph{not} assumed to be \(\mathcal L\)-definable or \(\mathcal L_{\mathrm{sem}}\)-definable.
When a definable channel law is needed, one may add the extra assumption that \(R_{\mathcal E}\) is
\(\mathcal L_{\mathrm{sem}}\)-definable.
\textup{(ii)}~Clause \textup{(E2)} does not claim that \(S_C\) contains all physically possible carrier states;
\(S_C\) comprises the admissible semantic-bearing carrier states, and pure noise states can be modeled
outside~\(S_C\).
\textup{(iii)}~Clause \textup{(E3)} requires only a computable \emph{selector} that produces one admissible realization
\(s_c\in\mathcal E(s_o)\) per~\(s_o\); it does not imply that the full relation \(R_{\mathcal E}\) or the full
fiber \(\mathcal E(s_o)\) is decidable or effectively enumerable.
\end{remark}

\begin{definition}[Information model]
\label{def:info-instance}
An \emph{information model} is the tuple
\[
  \mathcal I
  \;=\;
  \langle
    O,\; T_O,\; S_O,\; C,\; T_C,\; S_C,\; R_{\mathcal E}
  \rangle,
\]
together with the temporal structure (orders and cross-domain precedence) postulated in
Axiom~\ref{ax:time-domains}.
\end{definition}

\begin{remark}[Background data suppressed from the instance notation]
\label{rem:instance-temporal-data}
For brevity, we suppress from the tuple notation of \(\mathcal I\) both
(i) the temporal relations \(\prec_O,\prec_C,\prec_{OC}\) specified by Axiom~\ref{ax:time-domains}, and
(ii) the representation/encoding maps \(\enc_O,\enc_C\) postulated in Axiom~\ref{ax:state-representation}.
All of these are treated as background instance data fixed throughout the paper.
\end{remark}

\begin{definition}[Composition of information models]
\label{def:model-composition}
Let \(\mathcal I_1\) and \(\mathcal I_2\) be information models
(Definition~\ref{def:info-instance}) with induced enabling maps
\(\mathcal E_1:S_O^{(1)}\Rightarrow S_C^{(1)}\) and
\(\mathcal E_2:S_O^{(2)}\Rightarrow S_C^{(2)}\), respectively.
The pair \((\mathcal I_1,\mathcal I_2)\) is \emph{composable} if
\(S_C^{(1)}=S_O^{(2)}\) as state sets (with compatible encodings).

\smallskip
For a composable pair, the \emph{composite information model}
\(\mathcal I_2\circ\mathcal I_1\) is the information model whose semantic state set is
\(S_O^{(1)}\), whose carrier state set is \(S_C^{(2)}\), and whose
\emph{composite enabling map}
\(\mathcal E_{2\circ 1}:S_O^{(1)}\Rightarrow S_C^{(2)}\) is defined by
\begin{equation}\label{eq:composite-enabling}
  \mathcal E_{2\circ 1}(s_o)
  \;:=\;
  \bigcup_{s_c\,\in\,\mathcal E_1(s_o)}\mathcal E_2(s_c),
  \qquad s_o\in S_O^{(1)}.
\end{equation}
The remaining instance data (domains, time indices, cross-domain precedence) of the
composite model are inherited from the constituents in the natural way and are suppressed
from notation following the convention of Remark~\ref{rem:instance-temporal-data}.
\end{definition}

\begin{proposition}[Composition preserves enabling axioms]
\label{prop:composition-axioms}
Let \((\mathcal I_1,\mathcal I_2)\) be a composable pair of information models, each
satisfying Axiom~\ref{ax:enabling-mapping}.
Then the composite enabling map
\(\mathcal E_{2\circ 1}\) of \(\mathcal I_2\circ\mathcal I_1\) satisfies:
\begin{enumerate}[label=\textup{(\roman*)}]
  \item \emph{Totality \textup{(E1)}:}
        \(\mathcal E_{2\circ 1}(s_o)\neq\varnothing\) for every
        \(s_o\in S_O^{(1)}\).
  \item \emph{Coverage \textup{(E2)}:}
        \(\displaystyle\bigcup_{s_o\in S_O^{(1)}}\mathcal E_{2\circ 1}(s_o)=S_C^{(2)}\).
  \item \emph{Computable selector \textup{(E3)}:}
        if \(e_1\) and \(e_2\) are computable enabling selectors for \(\mathcal I_1\)
        and \(\mathcal I_2\) respectively, then \(e_{2\circ 1}:=e_2\circ e_1\) is a
        computable enabling selector for \(\mathcal I_2\circ\mathcal I_1\).
\end{enumerate}
\end{proposition}

\begin{proof}
\textup{(i)}\
For any \(s_o\in S_O^{(1)}\), totality of \(\mathcal E_1\) gives some
\(s_c\in\mathcal E_1(s_o)\).
Since \(s_c\in S_C^{(1)}=S_O^{(2)}\), totality of
\(\mathcal E_2\) gives \(\mathcal E_2(s_c)\neq\varnothing\), so
\(\mathcal E_{2\circ 1}(s_o)\supseteq\mathcal E_2(s_c)\neq\varnothing\).

\smallskip\noindent
\textup{(ii)}\
Take any \(s'\in S_C^{(2)}\).
Coverage of \(\mathcal E_2\) yields
\(s_c\in S_O^{(2)}=S_C^{(1)}\) with \(s'\in\mathcal E_2(s_c)\).
Coverage of \(\mathcal E_1\) then yields \(s_o\in S_O^{(1)}\) with
\(s_c\in\mathcal E_1(s_o)\), whence
\(s'\in\mathcal E_{2\circ 1}(s_o)\).

\smallskip\noindent
\textup{(iii)}\
The selector \(e_1\) is total on codes of \(S_O^{(1)}\) and outputs codes of
elements in \(\mathcal E_1(s_o)\subseteq S_C^{(1)}=S_O^{(2)}\);
\(e_2\) is total on codes of \(S_O^{(2)}\).
Hence \(e_{2\circ 1}=e_2\circ e_1\) is total on codes of \(S_O^{(1)}\),
computable as a composition of computable functions, and satisfies
\(e_{2\circ 1}(s_o)\in\mathcal E_2(e_1(s_o))\subseteq\mathcal E_{2\circ 1}(s_o)\).
\end{proof}

\begin{remark}[Associativity and iterated composition]
\label{rem:composition-assoc}
Composition of enabling maps is associative: for a composable triple
\((\mathcal I_1,\mathcal I_2,\mathcal I_3)\),
\begin{align*}
  \mathcal E_{3\circ(2\circ 1)}(s_o) 
  \;=\;
  &\bigcup_{s\,\in\,\mathcal E_{2\circ 1}(s_o)}\mathcal E_3(s) \\
  \;=\;
  &\bigcup_{s_c\,\in\,\mathcal E_1(s_o)}\; 
  \bigcup_{s\,\in\,\mathcal E_2(s_c)}\mathcal E_3(s) \\
  \;=\;
  &\mathcal E_{(3\circ 2)\circ 1}(s_o),
\end{align*}
so \((\mathcal I_3\circ\mathcal I_2)\circ\mathcal I_1
     =\mathcal I_3\circ(\mathcal I_2\circ\mathcal I_1)\).
Consequently, an iterated composition
\(\mathcal I_n\circ\cdots\circ\mathcal I_1\) is unambiguously defined whenever
consecutive pairs are composable.
This allows any end-to-end information pipeline to be modeled as a single
composite information model at the coarsest granularity, or decomposed into
sub-models at any desired level of refinement; the two viewpoints yield the same
enabling structure.
\end{remark}

\subsection{Synonymous State Sets and Ideal Information}
\label{subsec:ideal}

This subsection formalizes when two state sets carry the same \emph{semantic content}
under an \emph{ideal} (lossless) notion of synonymy. In the ideal setting adopted here,
temporal precedence is part of the semantic structure; hence any synonymy must preserve it.

\textit{Semantic sublanguage.}
All definability requirements below are taken relative to the designated semantic sublanguage
\(\mathcal L_{\mathrm{sem}}\subseteq\mathcal L\) from Assumption~\ref{assump:semantic-sublanguage}.

\textit{Induced precedence on states.}
Let \(S\subseteq X\times T\) be a state set and let \(\prec\) be a linear order on \(T\).
For \(s=(x,t)\) and \(s'=(x',t')\) in \(S\), define the induced strict precedence relation
\[
s\prec_S s' \quad:\Longleftrightarrow\quad t\prec t'.
\]

\begin{definition}[\(\mathcal L_{\mathrm{sem}}\)-definable coding isomorphism between state sets]
\label{def:definable-coding}
Let \(S_1\subseteq X_1\times T_1\) and \(S_2\subseteq X_2\times T_2\) be \(\mathcal L_{\mathrm{sem}}\)-definable
state sets in the ambient structure \(\mathfrak R\).

A relation \(G_{12}\subseteq S_1\times S_2\) is an \emph{\(\mathcal L_{\mathrm{sem}}\)-definable coding isomorphism graph}
from \(S_1\) to \(S_2\) if there exists an \(\mathcal L_{\mathrm{sem}}\)-formula \(\psi(x_1,t_1,x_2,t_2)\) such that
for all \(x_1,t_1,x_2,t_2\),
\begin{align*}
G_{12}\bigl((x_1,t_1),(x_2,t_2)\bigr)
&\Longleftrightarrow S_1(x_1,t_1)\wedge S_2(x_2,t_2) \nonumber \\
&\qquad \wedge \psi(x_1,t_1,x_2,t_2),
\end{align*}
and \(G_{12}\) is a bijection between \(S_1\) and \(S_2\) in the sense of \textup{(B1)}--\textup{(B2)} below:
\begin{enumerate}[label=\textup{(B\arabic*)}]
  \item for every \(s_1\in S_1\) there exists a unique \(s_2\in S_2\) with \(G_{12}(s_1,s_2)\);
  \item for every \(s_2\in S_2\) there exists a unique \(s_1\in S_1\) with \(G_{12}(s_1,s_2)\).
\end{enumerate}
The uniqueness in \textup{(B1)}--\textup{(B2)} induces mutually inverse maps
\(\tau_{12}:S_1\to S_2\) and \(\tau_{21}:S_2\to S_1\).
\end{definition}

\begin{definition}[Synonymous state sets (ideal, order-preserving)]
\label{def:synonymous-states}
Let \(S_1\subseteq X_1\times T_1\) and \(S_2\subseteq X_2\times T_2\) be
\(\mathcal L_{\mathrm{sem}}\)-definable state sets in \(\mathfrak R\), equipped with
linear orders \(\prec_1\) on \(T_1\) and \(\prec_2\) on \(T_2\).
We say \(S_1\) and \(S_2\) are \emph{synonymous} (in the ideal, order-sensitive sense), written
\[
S_1 \syn S_2,
\]
if there exists an \(\mathcal L_{\mathrm{sem}}\)-definable coding isomorphism graph \(G_{12}\subseteq S_1\times S_2\)
(Definition~\ref{def:definable-coding}) inducing a bijection \(\tau_{12}:S_1\to S_2\) such that \(\tau_{12}\) is a \emph{time-order isomorphism} for the induced precedence on states, i.e.,
\[
s\prec_{S_1} s' \ \Longleftrightarrow\ \tau_{12}(s)\prec_{S_2}\tau_{12}(s')\quad (\forall s,s'\in S_1),
\]
where \(\prec_{S_i}\) is induced from \(\prec_i\) via \( (x,t)\prec_{S_i}(x',t')\Longleftrightarrow t\prec_i t'\).
\end{definition}

\begin{remark}[Nonlinearity, time-slice constraints, and role of \(\mathcal L_{\mathrm{sem}}\)]
\label{rem:syn-notes}
The induced relation \(\prec_S\) compares states only by their time indices and is generally not a linear order
when multiple states share the same time point.
The preservation/reflection condition in Definition~\ref{def:synonymous-states} implies that
for each time point \(t\in T_1\), every state \((x,t)\in S_1\) is mapped by \(\tau_{12}\)
to a state whose time index is the same value \(t'\in T_2\)
(since states at the same time are pairwise \(\prec_{S_1}\)-incomparable, and \(\tau_{12}\)
preserves and reflects the induced precedence).
Hence \(\tau_{12}\) restricts to a bijection between the time-slices
\(\{x:(x,t)\in S_1\}\) and \(\{y:(y,t')\in S_2\}\), and corresponding time-slices must have
equal cardinality.
Restricting coding/decoding definability to \(\mathcal L_{\mathrm{sem}}\) prevents auxiliary representation predicates
(e.g., a built-in order used only for encoding) from trivializing equivalence or enabling arbitrary matchings
between equipotent finite sets.
\end{remark}

\begin{proposition}[\(\syn\) is an equivalence relation]
\label{prop:syn-equiv}
The relation \(\syn\) on \(\mathcal L_{\mathrm{sem}}\)-definable state sets is reflexive, symmetric, and transitive.
\end{proposition}

\begin{proof}
Reflexivity is witnessed by the \(\mathcal L_{\mathrm{sem}}\)-definable identity graph.
Symmetry follows by swapping coordinates in the witnessing graph.
For transitivity, let \(G_{12}\subseteq S_1\times S_2\) and \(G_{23}\subseteq S_2\times S_3\) be the two witnessing
graphs. Define the composed graph by
\(G_{13}(s_1,s_3):\Leftrightarrow \exists\, s_2\,[G_{12}(s_1,s_2)\wedge G_{23}(s_2,s_3)]\),
where, following the state notation convention of Definition~\ref{def:expressible-states},
the existential quantifier \(\exists\,s_2\) abbreviates \(\exists\,x_2\,\exists\,t_2\,[S_2(x_2,t_2)\wedge\cdots]\).
This is \(\mathcal L_{\mathrm{sem}}\)-definable since \(\mathcal L_{\mathrm{sem}}\) is closed under
first-order connectives and quantification over its own sorts.
The uniqueness clauses \textup{(B1)}--\textup{(B2)} of each component graph ensure that \(G_{13}\) again
satisfies \textup{(B1)}--\textup{(B2)}, so the induced map is a bijection that preserves and reflects
the induced time precedence by composition of the two order isomorphisms.
\end{proof}

\begin{definition}[Ideal information]
\label{def:ideal-info}
Let \(\mathcal I=\langle O,T_O,S_O,C,T_C,S_C,R_{\mathcal E}\rangle\) be an information model.
We call \(\mathcal I\) \emph{ideal (with respect to \(\mathcal L_{\mathrm{sem}}\))} if there exists an
\(\mathcal L_{\mathrm{sem}}\)-definable coding isomorphism graph \(G_{OC}\subseteq S_O\times S_C\)
witnessing \(S_O \syn S_C\), with induced bijection \(\tau_{OC}:S_O\to S_C\), such that the enabling map is
singleton-valued and agrees with \(\tau_{OC}\): for every \(s_o\in S_O\),
\[
\mathcal E(s_o) = \{\tau_{OC}(s_o)\}.
\]
\end{definition}

\begin{remark}[Cross-sort definability, witness dependence, and canonical representative]
\label{rem:ideal-notes}
The requirement \(S_O \syn S_C\) is nontrivial in a many-sorted setting: the witnessing
\(\mathcal L_{\mathrm{sem}}\)-formula must relate variables of the object sort and the carrier sort.
If \(\Sigma_{\mathrm{sem}}\) contains no relations or functions that can link these sorts
(directly or indirectly), then no definable coding isomorphism exists and the instance cannot be ideal.
Moreover, the witnessing isomorphism \(\tau_{OC}\) need not be unique; hence the pointwise identification
between \(S_O\) and \(S_C\) depends on the chosen witness.

\smallskip
\noindent\textbf{Convention (fixing a witness and canonical representative).}
Whenever we restrict attention to an ideal instance, we fix one witnessing graph \(G_{OC}\) and its induced bijection
\(\tau_{OC}:S_O\to S_C\) (with inverse \(\tau_{CO}\)) once and for all, identify \(S_O\) and \(S_C\) via this choice,
and take \(S_O\) as the canonical semantic representative of~\(\mathcal I\).
With a mild abuse of notation, we may refer to the ``semantics of \(\mathcal I\)'' simply as~\(S_O\).

\smallskip
When \(\mathcal I\) is not ideal, one can still analyze its semantic structure
through a carrier-equivalent semantic component; this is the subject of the noisy information
framework developed in Section~\ref{subsec:noisy}.
\end{remark}

\subsection{Irredundant Cores and Derivation Depth}
\label{subsec:atomic-derivation}

This subsection introduces two concepts unified by the immediate consequence operator
\(T_{\mathsf{PS}}\) of the fixed proof system:
(i) an \emph{irredundant semantic core} \(\Atom(S_O)\), defined via the deductive closure
\(\Cn=T_{\mathsf{PS}}^\omega\), that separates intrinsic vs.\ operational inference;
(ii) a base-relative \emph{derivation depth} \(\Dd(\cdot\mid B)\), defined as the stratum index
of the \(T_{\mathsf{PS}}\)-iteration from~\(B\), that measures the minimum number of
single-step inferences needed to reach a given formula.

\smallskip
\noindent\textbf{Ambient semantic state space.}
We fix a countable set \(\mathbb{S}_O\supseteq S_O\) of potential semantic states.
Under the State--formula convention of Assumption~\ref{assump:proof-system}, each element of
\(\mathbb{S}_O\) is identified with a ground atom of \(\mathcal L_{\mathrm{kb}}\)
(extending the identification of \(S_O\)-elements with atoms of the form \(S_O(o,\tau)\)).
In particular, the derivability relation \(\vdash_{\mathrm{kb}}\) and the operator
\(T_{\mathsf{PS}}\) introduced below act on (subsets of) \(\mathbb{S}_O\) via this
identification.

The set \(\mathbb{S}_O\) is equipped with an injective encoding
\(\enc_O:\mathbb S_O\to\{0,1\}^*\) extending that of \(S_O\)
(Axiom~\ref{ax:state-representation}) and an extension of \(\timeindex_O\) to \(\mathbb{S}_O\),
such that membership in \(\mathbb{S}_O\) is decidable from the codes
(i.e., given \(w\in\{0,1\}^*\), one can decide whether \(w\in\enc_O(\mathbb S_O)\)).
Further standing properties of \(\mathbb{S}_O\)
(closure under definable recodings) are specified in Assumption~\ref{assump:semantic-universe}.

\smallskip
\textit{Available premise bases.}
A finite set \(B\subseteq \mathbb S_O\) is called an \emph{available premise base} if its elements are treated as
depth-\(0\) premises when measuring derivation depth from \(B\).

\begin{assumption}[Finite and effectively listable knowledge bases]
\label{assump:finite-so}
The knowledge bases \(S_O\) considered in this paper are finite and effectively listable under a fixed canonical order.
\end{assumption}

\begin{assumption}[Effective redundancy test (core extractability)]
\label{assump:core-extractable}
For the knowledge bases \(S_O\) considered, the predicate
\(
s\in \Cn(\Gamma)
\)
is decidable whenever \(\Gamma\subseteq S_O\) is finite and \(s\in S_O\).
(Equivalently: redundancy of a stored formula relative to the remaining stored formulas is decidable.)
\end{assumption}

\begin{remark}[When Assumption~\ref{assump:core-extractable} is reasonable]
\label{rem:core-extractable-context}
Assumption~\ref{assump:core-extractable} is strong for unrestricted logics, but it is satisfied in many
practically relevant decidable fragments used for knowledge bases and rule
systems~\cite{ceri1989you,abiteboul1995foundations,
dantsin2001complexity} (e.g., Datalog/Horn-style
settings and bounded-domain theories), where entailment and redundancy checking are decidable and often tractable;
see, e.g.,~\cite{abiteboul1995foundations,dantsin2001complexity}.
\end{remark}

\begin{definition}[Semantic atomic core (canonical irredundant generating set)]
\label{def:atom-so}
Let \(S_O\) be a finite knowledge base.
Define \(\Atom(S_O)\) to be the output of the following fixed deterministic \emph{irredundantization} procedure:

\begin{enumerate}[label=\textup{(\roman*)}]
  \item Initialize \(A\gets S_O\).
  \item Scan elements of \(S_O\) in the fixed canonical order; for each \(s\in S_O\), if
  \(s\in \Cn(A\setminus\{s\})\), then set \(A\gets A\setminus\{s\}\).
  \item Output \(A\) and denote it by \(\Atom(S_O)\).
\end{enumerate}
\end{definition}

\begin{proposition}[Core correctness: equivalence, irredundancy, and derivability]
\label{prop:atom-core-correct}
Under Assumptions~\ref{assump:finite-so} and~\ref{assump:core-extractable}, for every finite knowledge base \(S_O\),
the set \(A:=\Atom(S_O)\) satisfies:
\begin{enumerate}[label=\textup{(\roman*)}]
  \item \emph{Closure equivalence:} \(\Cn(A)=\Cn(S_O)\).
  \item \emph{Irredundancy:} for every \(a\in A\), \(a\notin \Cn(A\setminus\{a\})\).
  \item \emph{Canonicality:} \(\Atom(S_O)\) is uniquely determined by \(S_O\) and the fixed canonical order.
  \item \emph{Stored derivability:} \(S_O\subseteq \Cn(\Atom(S_O))\).
\end{enumerate}
\end{proposition}

\begin{proof}
We first verify that each removal step preserves the deductive closure.
Suppose at some stage the current set is~\(A\) and some \(s\in A\) satisfies
\(s\in\Cn(A\setminus\{s\})\); we claim \(\Cn(A)=\Cn(A\setminus\{s\})\).
Since \(A\setminus\{s\}\subseteq A\), monotonicity~\textup{(Cn2)} gives
\(\Cn(A\setminus\{s\})\subseteq\Cn(A)\).
For the reverse inclusion: reflexivity~\textup{(Cn1)} gives
\(A\setminus\{s\}\subseteq\Cn(A\setminus\{s\})\), and by hypothesis
\(s\in\Cn(A\setminus\{s\})\), so \(A\subseteq\Cn(A\setminus\{s\})\).
By monotonicity, \(\Cn(A)\subseteq\Cn\bigl(\Cn(A\setminus\{s\})\bigr)\), and by
idempotence~\textup{(Cn3)},
\(\Cn\bigl(\Cn(A\setminus\{s\})\bigr)=\Cn(A\setminus\{s\})\).
Hence \(\Cn(A)\subseteq\Cn(A\setminus\{s\})\).

By induction over the finitely many removal steps in the scan of
Definition~\ref{def:atom-so}, the final set \(A\) satisfies
\(\Cn(A)=\Cn(S_O)\), proving~\textup{(i)}.

For irredundancy~\textup{(ii)}: any \(a\) retained in~\(A\) was \emph{not} removable at its
scan stage, i.e., \(a\notin\Cn(A'\setminus\{a\})\) for the set~\(A'\) at that stage.
Since subsequent removals only shrink~\(A'\) to the final~\(A\subseteq A'\),
monotonicity~\textup{(Cn2)} gives
\(\Cn(A\setminus\{a\})\subseteq\Cn(A'\setminus\{a\})\), so
\(a\notin\Cn(A\setminus\{a\})\).

Canonicality~\textup{(iii)} follows from determinism of the procedure and the fixed
canonical order.

For~\textup{(iv)}, reflexivity~\textup{(Cn1)} gives
\(S_O\subseteq\Cn(S_O)=\Cn(A)\) by~\textup{(i)}.
\end{proof}

\begin{definition}[Intrinsic vs.\ operational premise bases]
\label{def:intrinsic-operational-bases}
For a knowledge base \(S_O\), define:
\begin{align*}
 & A := \Atom(S_O) \quad\text{(core premises)},\\
 & J := S_O\setminus A \quad\text{(stored shortcuts)}.
\end{align*}
We call \(A\) the \emph{intrinsic} premise base and \(A\cup J=S_O\) the \emph{operational} premise base.
\end{definition}

\textit{Derivation depth from the immediate consequence operator.}
We derive the derivation depth directly from
the proof system \(\mathsf{PS}\) (Assumption~\ref{assump:proof-system}) via its
\emph{immediate consequence operator}, which simultaneously characterizes \(\Cn\) and stratifies
the derivation process.

\begin{definition}[Immediate consequence operator]
\label{def:T-operator}
Write \(\Gamma\vdash_{\mathrm{kb}}^{1}s\) to denote that \(s\) is derivable from premises in
\(\Gamma\) by \emph{exactly one inference step} (i.e., one rule application) of~\(\mathsf{PS}\).
For any set \(\Gamma\subseteq \mathbb{S}_O\), define
\[
  T_{\mathsf{PS}}(\Gamma)
  \;:=\;
  \Gamma \;\cup\;
  \bigl\{\, s\in \mathbb{S}_O :\,
     \Gamma\vdash_{\mathrm{kb}}^{1} s \,\bigr\}.
\]
Define the iteration \(T^0_{\mathsf{PS}}(\Gamma):=\Gamma\) and
\(T^{n+1}_{\mathsf{PS}}(\Gamma):=T_{\mathsf{PS}}\bigl(T^n_{\mathsf{PS}}(\Gamma)\bigr)\) for \(n\ge 0\).
\end{definition}

\begin{axiom}[Properties of the immediate consequence operator]
\label{ax:T-operator}
The operator \(T_{\mathsf{PS}}\) satisfies:
\begin{enumerate}[label=\textup{(IC\arabic*)}]
  \item \emph{Monotonicity:}
        \(\Gamma\subseteq\Gamma'\) implies \(T_{\mathsf{PS}}(\Gamma)\subseteq T_{\mathsf{PS}}(\Gamma')\).
  \item \emph{Computability:}
        for any finite \(\Gamma\subseteq\mathbb{S}_O\) (given by the codes of its elements),
        the set \(T_{\mathsf{PS}}(\Gamma)\) is finite and effectively computable.
  \item \emph{Closure characterization:}
        \(\Cn(\Gamma)=\bigcup_{n\ge 0}T^n_{\mathsf{PS}}(\Gamma)\) for every finite
        \(\Gamma\subseteq\mathbb{S}_O\).
  \item \emph{Finite stabilization:}
        for every finite \(\Gamma\subseteq\mathbb{S}_O\), the ascending chain
        \(\bigl(T^n_{\mathsf{PS}}(\Gamma)\bigr)_{n\ge 0}\) stabilizes in finitely many steps;
        i.e., there exists \(N<\infty\) such that
        \(T^{N+1}_{\mathsf{PS}}(\Gamma)=T^N_{\mathsf{PS}}(\Gamma)=\Cn(\Gamma)\).
\end{enumerate}
\end{axiom}

\begin{remark}[When Axiom~\ref{ax:T-operator} holds; derived properties]
\label{rem:T-operator-context}
In Datalog and Horn-clause settings over finite domains, the immediate consequence operator
\(T_P\) satisfies all four properties:
monotonicity is a standard property of \(T_P\);
computability follows from the decidability of single-step rule application;
the closure characterization is the classical van Emden--Kowalski fixpoint
theorem~\cite{abiteboul1995foundations}; and
finite stabilization follows from finiteness of the Herbrand base.
More broadly, Axiom~\ref{ax:T-operator} holds for any effective proof system whose
single-step relation is decidable and whose ground term universe is finite.

\smallskip
Two useful consequences:
\textup{(a)}~\textup{(IC2)}+\textup{(IC4)} imply that \(\Cn(\Gamma)\) is \emph{finite}
for every finite~\(\Gamma\);
\textup{(b)}~\textup{(IC2)}+\textup{(IC3)}+\textup{(IC4)} together imply
Assumption~\textup{\ref{assump:core-extractable}} as a special case (deciding
\(s\in\Cn(\Gamma)\) by iterating \(T_{\mathsf{PS}}\) from~\(\Gamma\) until stabilization).
We state Assumption~\textup{\ref{assump:core-extractable}} separately to clarify the minimal
requirements for core extraction independently of the \(T_{\mathsf{PS}}\)-axiomatics.
\end{remark}

\begin{definition}[Base-relative derivation depth (stratum)]
\label{def:derivation-depth}
Let \(B\subseteq\mathbb{S}_O\) be a finite set of \emph{available premises}.
The \emph{derivation depth} (or \emph{stratum index}) of \(s\in\mathbb{S}_O\) from \(B\) is
\begin{equation}\label{eq:Dd-stratum}
  \Dd(s\mid B)
  \;:=\;
  \min\bigl\{\,n\ge 0 : s\in T^n_{\mathsf{PS}}(B)\,\bigr\},
\end{equation}
with the convention \(\Dd(s\mid B):=\infty\) if \(s\notin\Cn(B)\).
\end{definition}

\begin{remark}[Recursive characterization]
\label{rem:Dd-recursive}
The stratum definition~\eqref{eq:Dd-stratum} admits an equivalent recursive form:
\(\Dd(s\mid B)=0\) if \(s\in B\), and otherwise
\begin{align*}
\Dd(s \mid B) = 1 + \min \Bigl\{ & \max_{s'\in P} \Dd(s'\mid B) : \nonumber \\
                                 & P \subseteq \Cn(B) \text{ finite}, \; P \vdash_{\mathrm{kb}}^{1} s \Bigr\}
\end{align*}
where the minimum is over all finite premise sets \(P\) from which \(s\) is derivable
in one step of~\(\mathsf{PS}\).
This makes explicit that \(\Dd(s\mid B)\) equals the depth of the \emph{shallowest} proof tree
of \(s\) from~\(B\).
\end{remark}

\begin{lemma}[Properties of derivation depth]
\label{lem:depth-properties}
Under Axiom~\ref{ax:T-operator}, for any finite premise base \(B\subseteq\mathbb{S}_O\):
\begin{enumerate}[label=\textup{(\roman*)}]
  \item For every \(s\in\Cn(B)\), \(\Dd(s\mid B)\) is a unique, finite, computable non-negative integer.
  \item \(\Dd(s\mid B)=0\) if and only if \(s\in B\).
  \item \emph{Monotonicity in the base:} if \(B\subseteq B'\subseteq\mathbb{S}_O\)
        with \(B'\) finite and \(s\in\Cn(B)\), then \(\Dd(s\mid B')\le\Dd(s\mid B)\).
\end{enumerate}
\end{lemma}

\begin{proof}
\textup{(i)}\
By \textup{(IC4)}, the iteration \(\bigl(T^n_{\mathsf{PS}}(B)\bigr)_{n\ge 0}\) stabilizes at
\(\Cn(B)\) in finitely many steps.
Since \(s\in\Cn(B)\), the set \(\{n:s\in T^n_{\mathsf{PS}}(B)\}\) is nonempty;
its minimum is therefore a unique finite non-negative integer.
Computability: iterate \(T_{\mathsf{PS}}\) from \(B\) (computable by \textup{(IC2)})
and check membership at each stage until \(s\) appears.

\smallskip\noindent
\textup{(ii)}\
\(\Dd(s\mid B)=0\) iff \(s\in T^0_{\mathsf{PS}}(B)=B\).

\smallskip\noindent
\textup{(iii)}\
Monotonicity of \(T_{\mathsf{PS}}\) \textup{(IC1)} yields, by induction on \(n\),
\(T^n_{\mathsf{PS}}(B)\subseteq T^n_{\mathsf{PS}}(B')\) for all \(n\ge 0\).
Hence
\(\min\{n:s\in T^n_{\mathsf{PS}}(B')\}\le\min\{n:s\in T^n_{\mathsf{PS}}(B)\}\),
i.e., \(\Dd(s\mid B')\le\Dd(s\mid B)\).
\end{proof}

\begin{definition}[Intrinsic and operational derivation depths]
\label{def:int-op-depth}
Let \(S_O\) be a knowledge base with core \(A:=\Atom(S_O)\).
For any \(q\in\Cn(A)=\Cn(S_O)\), define
\[
  n_{\mathrm{int}}(q) := \Dd(q \mid A),
  \qquad
  n_{\mathrm{op}}(q) := \Dd(q \mid S_O).
\]
\end{definition}

\begin{definition}[Semantic atomicity measure]
\label{def:atomicity-measure}
The \emph{semantic atomicity} of an information model \(\mathcal I\) with semantic space \(S_O\) is
\[
  \mathsf{A}(\mathcal I)\;:=\;\lvert \Atom(S_O)\rvert.
\]
\end{definition}

\begin{definition}[Maximum intrinsic derivation depth]
\label{def:max-depth}
The \emph{maximum intrinsic derivation depth} of an information model \(\mathcal I\) with semantic space \(S_O\) is
\[
  \mathsf{D_d}(\mathcal I)
  \;:=\;
  \max_{q\,\in\, S_O}\;\Dd\bigl(q \mid \Atom(S_O)\bigr),
\]
with the convention \(\max\varnothing := 0\) when \(S_O=\varnothing\).
\end{definition}

\begin{theorem}[Computable semantic invariants of an ideal information model]
\label{thm:semantic-invariants}
Let \(\mathcal I\) be an ideal information model
\textup{(Definition~\ref{def:ideal-info})} with canonical semantic representative \(S_O\).
Under Assumptions~\textup{\ref{assump:finite-so}}--\textup{\ref{assump:core-extractable}}
and Axiom~\textup{\ref{ax:T-operator}},
the following hold.
\begin{enumerate}[label=\textup{(\roman*)}]
  \item \emph{Well-definedness and uniqueness.}
        The semantic atomicity \(\mathsf{A}(\mathcal I)=\lvert\Atom(S_O)\rvert\) and the maximum intrinsic derivation depth
        \(\mathsf{D_d}(\mathcal I)=\max_{q\in S_O}\Dd(q\mid\Atom(S_O))\)
        are uniquely determined by \(S_O\) and the fixed canonical order.
  \item \emph{Finiteness and computability.}
        Both \(\mathsf{A}(\mathcal I)\) and \(\mathsf{D_d}(\mathcal I)\) are finite non-negative integers,
        and are computable from \(S_O\).
  \item \emph{Monotonicity of derivation depth.}
        For every \(q\in S_O\),
        \[
          n_{\mathrm{op}}(q)\;\le\; n_{\mathrm{int}}(q)\;\le\;\mathsf{D_d}(\mathcal I).
        \]
  \item \emph{Characterization of zero depth.}
        \(\mathsf{D_d}(\mathcal I)=0\) if and only if \(\Atom(S_O)=S_O\)
        \textup{(}i.e., \(S_O\) is already irredundant\textup{)}.
\end{enumerate}
\end{theorem}

\begin{proof}
\textup{(i)}\
\(\Atom(S_O)\) is uniquely determined by Proposition~\ref{prop:atom-core-correct}\textup{(iii)};
hence \(\mathsf{A}(\mathcal I)\) is unique.
By Proposition~\ref{prop:atom-core-correct}\textup{(iv)},
every \(q\in S_O\) lies in \(\Cn(\Atom(S_O))\),
so the stratum \(\Dd(q\mid\Atom(S_O))\) is unique by Lemma~\ref{lem:depth-properties}\textup{(i)},
and \(\mathsf{D_d}(\mathcal I)\) is unique.

\smallskip\noindent
\textup{(ii)}\
\(\Atom(S_O)\) is computable by the deterministic procedure of Definition~\ref{def:atom-so}
under Assumptions~\ref{assump:finite-so}--\ref{assump:core-extractable},
so \(\mathsf{A}(\mathcal I)=|\Atom(S_O)|\) is computable.
By Lemma~\ref{lem:depth-properties}\textup{(i)},
each \(\Dd(q\mid\Atom(S_O))\) is finite and computable.
Since \(S_O\) is finite and effectively listable (Assumption~\ref{assump:finite-so}),
the maximum \(\mathsf{D_d}(\mathcal I)\) is computable by enumeration.

\smallskip\noindent
\textup{(iii)}\
Since \(\Atom(S_O)\subseteq S_O\), Lemma~\ref{lem:depth-properties}\textup{(iii)} gives
\(\Dd(q\mid S_O)\le\Dd(q\mid\Atom(S_O))\) for every
\(q\in\Cn(\Atom(S_O))\supseteq S_O\)
(the inclusion by Proposition~\ref{prop:atom-core-correct}\textup{(iv)}).
Restricting to \(q\in S_O\) yields \(n_{\mathrm{op}}(q)\le n_{\mathrm{int}}(q)\).
The bound \(n_{\mathrm{int}}(q)\le\mathsf{D_d}(\mathcal I)\) is immediate from the definition
of the maximum.

\smallskip\noindent
\textup{(iv)}\
If \(\Atom(S_O)=S_O\), then every \(q\in S_O\) satisfies \(q\in\Atom(S_O)=T^0_{\mathsf{PS}}(\Atom(S_O))\),
so \(\Dd(q\mid\Atom(S_O))=0\) and \(\mathsf{D_d}(\mathcal I)=0\).
Conversely, if \(\mathsf{D_d}(\mathcal I)=0\), then \(\Dd(q\mid\Atom(S_O))=0\) for all
\(q\in S_O\), which by Lemma~\ref{lem:depth-properties}\textup{(ii)} implies \(q\in\Atom(S_O)\);
thus \(S_O\subseteq\Atom(S_O)\), and since \(\Atom(S_O)\subseteq S_O\) by construction,
\(\Atom(S_O)=S_O\).
\end{proof}

\begin{remark}[Interpretation and connection to logical depth]
\label{rem:invariants-interpretation}
The intrinsic depth \(n_{\mathrm{int}}\) captures inferential content relative to the irredundant core,
whereas the operational depth \(n_{\mathrm{op}}\) captures actual online cost under stored shortcuts.
Since \(\Dd(q\mid B)\) is the depth of the shallowest proof tree from \(B\),
it is closely related to Bennett's logical depth~\cite{bennett1988logical};
later sections connect intrinsic/operational derivation depth to conditional logical depth
under explicit simulation assumptions.

\smallskip
The irredundant core \(\Atom(S_O)\) is also related to the notion of a
\emph{prime implicate basis} or \emph{irredundant base} in knowledge
compilation~\cite{dantsin2001complexity} and to the concept of a
\emph{minimal generating set} for a deductive closure in logic
programming~\cite{abiteboul1995foundations,ceri1989you}.
The key difference in our setting is that \(\Atom(S_O)\) is defined
relative to a fixed proof system~\(\mathsf{PS}\) and a fixed canonical
order, yielding a unique and computable representative; moreover, the
derivation-depth stratification~\(\Dd(\cdot\mid\Atom(S_O))\) provides
a \emph{quantitative} measure of inferential complexity that goes
beyond the binary distinction between base and derived facts.
\end{remark}

\begin{remark}[Application to time-free knowledge bases]
\label{rem:time-free-kb}
In many knowledge-base applications (including the Datalog examples of
Section~\ref{sec:application}), the time coordinate plays no active role:
all states reside at a single (implicit) time point, and the time domain
\(T_O\) degenerates to a singleton.
In this case the object--time pair \((o,\tau)\) reduces to the object~\(o\)
alone, the induced precedence~\(\prec_S\) is trivial (no distinct time points
to compare), and the synonymy condition of
Definition~\ref{def:synonymous-states} reduces to the existence of an
\(\mathcal L_{\mathrm{sem}}\)-definable bijection between state sets.
All definitions and results of Sections~\ref{sec:model}--\ref{sec:channel}
remain valid in this degenerate setting; the time-indexed formulation is
retained for generality.
\end{remark}

\subsection{Noisy Information}
\label{subsec:noisy}

Sections~\ref{subsec:info-instance}--\ref{subsec:atomic-derivation} developed the theory of
information models and their semantic invariants under the assumption that the model is
\emph{ideal}, i.e., \(S_O\syn S_C\).
In practice, a given information model need not be ideal: the carrier may realize states that
were not semantically intended, or fail to realize some intended states.
This subsection introduces a framework for analyzing such non-ideal models by constructing,
for each information model whose carrier admits a semantic representation, a noisy information
object whose semantic component is synonymous with the carrier space.
The semantic component is a set perturbation of the original \(S_O\);
the perturbation pair \((S_O^-,S_O^+)\) quantifies the gap between intended and
carrier-equivalent semantics, and the semantic invariants of the perturbed component serve as
computable proxies for the effective semantic structure of the original model as seen through
its carrier realization.

\begin{assumption}[Common semantic universe for set operations]
\label{assump:semantic-universe}
The ambient set \(\mathbb{S}_O\supseteq S_O\) introduced in
Section~\ref{subsec:atomic-derivation} is assumed to satisfy the following additional
closure property: if \(S\subseteq\mathbb{S}_O\) is an
\(\mathcal L_{\mathrm{sem}}\)-definable state set in \(\mathfrak R\) and
\(S'\) is an \(\mathcal L_{\mathrm{sem}}\)-definable state set over subdomains of the
object and time sorts with \(S'\syn S\) witnessed by an \(\mathcal L_{\mathrm{sem}}\)-definable
coding isomorphism graph (Definition~\ref{def:definable-coding}),
then \(S'\subseteq\mathbb{S}_O\).
(Since elements of \(\mathbb{S}_O\) are object-sort/time-sort pairs, the closure condition
is relevant only when \(S'\) is over the same sorts; for state sets over other sorts
the condition is vacuously satisfied.)

Moreover, \(\mathbb{S}_O\) is effectively representable: the injective encoding \(\enc_O\) and
time-index map \(\timeindex_O\) from Section~\ref{subsec:atomic-derivation} extend to all of
\(\mathbb{S}_O\), and membership in \(\mathbb{S}_O\) remains decidable from the codes.
\end{assumption}

\begin{definition}[Noisy semantic base (set perturbation)]
\label{def:noisy}
Let \(S_O\subseteq \mathbb S_O\) be the intended semantic state set.
A \emph{noisy semantic base} associated with \(S_O\) is any set of the form
\[
  \tilde S_O := (S_O \setminus S_O^-) \cup S_O^+,
\]
where \(S_O^- \subseteq S_O\) (lost states) and
\(S_O^+ \subseteq \mathbb S_O\setminus S_O\) (spurious states).
We call \(S_O^-\) and \(S_O^+\) the \emph{noise pair} and say the perturbation is
\emph{trivial} when \(S_O^-=S_O^+=\varnothing\), i.e., \(\tilde S_O=S_O\).
\end{definition}

\begin{assumption}[Carrier representability]
\label{assump:carrier-rep}
For the information model \(\mathcal I\) under consideration, there exists an
\(\mathcal L_{\mathrm{sem}}\)-definable state set
\(S'_O\subseteq\mathbb{S}_O\) such that \(S'_O\syn S_C\)
(Definition~\ref{def:synonymous-states}).
\end{assumption}

\begin{remark}[When carrier representability holds]
\label{rem:carrier-rep-context}
Assumption~\ref{assump:carrier-rep} requires a cross-sort
\(\mathcal L_{\mathrm{sem}}\)-definable bijection between some
\(\mathcal L_{\mathrm{sem}}\)-definable subset of \(\mathbb{S}_O\)
and \(S_C\); by Remark~\ref{rem:ideal-notes}, this is nontrivial in a many-sorted setting.
The assumption is naturally satisfied whenever the semantic signature
\(\Sigma_{\mathrm{sem}}\) contains at least one relation or function symbol that links the
object sort and the carrier sort (directly or through intermediate sorts), and the ambient
structure \(\mathfrak R\) interprets this symbol so as to induce a bijection between the
time-slices of a suitable subset of \(\mathbb{S}_O\) and those of~\(S_C\).
In typical semantic communication scenarios, such a link is built into the model
(e.g., an encoding/decoding relation between messages and signals).
When \(\mathcal I\) is already ideal, the assumption holds with \(S'_O=S_O\) and the
perturbation is trivial.
\end{remark}

\begin{definition}[Noisy information (carrier-equivalent representation)]
\label{def:noisy-info}
Let \(\mathcal I=\langle O,T_O,S_O,C,T_C,S_C,R_{\mathcal E}\rangle\) be an information model
satisfying Assumption~\ref{assump:carrier-rep}, and let
\(\tilde S_O\subseteq\mathbb{S}_O\) be an \(\mathcal L_{\mathrm{sem}}\)-definable state set
with \(\tilde S_O\syn S_C\),
witnessed by an \(\mathcal L_{\mathrm{sem}}\)-definable coding isomorphism graph
\(G\subseteq \tilde S_O\times S_C\) inducing a bijection
\(\tau_{OC}:\tilde S_O\to S_C\) (with inverse \(\tau_{CO}\)).

The \emph{noisy information} associated with \(\mathcal I\) via \((\tilde S_O,G)\) is the
tuple
\[
  \tilde{\mathcal I}
  \;=\;
  \bigl\langle\,
    \tilde S_O,\;
    S_C,\;
    G,\;
    \tau_{OC},\;
    \tau_{CO}
  \,\bigr\rangle,
\]
together with the noise pair \((S_O^-,\,S_O^+)\) defined by
\[
  S_O^-:=S_O\setminus\tilde S_O
  \qquad\text{(intended but unrealized)},
\]
and
\[
  S_O^+:=\tilde S_O\setminus S_O
  \qquad\text{(unintended but realized)},
\]
so that \(\tilde S_O=(S_O\setminus S_O^-)\cup S_O^+\) is a noisy semantic base of \(S_O\)
in the sense of Definition~\ref{def:noisy}.
\end{definition}

\begin{remark}[Connection to ideal information and computability of semantic invariants]
\label{rem:noisy-ideal-connection}
By construction, \(\tilde S_O\syn S_C\) and the singleton map
\(\tilde{\mathcal E}(\tilde s_o):=\{\tau_{OC}(\tilde s_o)\}\) is a bijective enabling
assignment that agrees with the witnessing isomorphism, reproducing the algebraic
content of Definition~\ref{def:ideal-info} (synonymy plus singleton enabling map).
Whether the induced enabling relation additionally satisfies the precedence compatibility
clause~\textup{(E4)} of Axiom~\textup{\ref{ax:enabling-mapping}} depends on the interaction
between the witnessing graph~\(G\) and the cross-domain precedence~\(\prec_{OC}\);
this may be imposed as an extra condition on the construction when the full information model
axiomatics are needed.

Independently of the full model structure, the semantic invariants
\(\mathsf{A}\) and \(\mathsf{D_d}\)
(Definitions~\ref{def:atomicity-measure}--\ref{def:max-depth}) depend only on the semantic
space and the proof system.
The hypotheses used in the proof of Theorem~\ref{thm:semantic-invariants} hold
for~\(\tilde S_O\):
\textup{(a)}~\(\tilde S_O\) is finite, being in bijection with the finite set~\(S_C\);
\textup{(b)}~\(\tilde S_O\) is effectively listable under the canonical order inherited from
\(\enc_O\), since \(S_C\) is finite and \(\tau_{CO}\) is computable
(induced by an \(\mathcal L_{\mathrm{sem}}\)-definable graph in the finite
structure~\(\mathfrak R\));
\textup{(c)}~the effective redundancy test holds, since
\(\tilde S_O\subseteq\mathbb{S}_O\) is finite and
Axiom~\textup{\ref{ax:T-operator}} allows deciding \(s\in\Cn(\Gamma)\) by iterating
\(T_{\mathsf{PS}}\) until stabilization
\textup{(Remark~\ref{rem:T-operator-context}(b))}.

We therefore define
\(\mathsf{A}(\tilde{\mathcal I}):=|\Atom(\tilde S_O)|\) and
\(\mathsf{D_d}(\tilde{\mathcal I}):=\max_{q\in\tilde S_O}\Dd(q\mid\Atom(\tilde S_O))\);
by the same reasoning as in Theorem~\ref{thm:semantic-invariants}
(applied to \(\tilde S_O\) in place of~\(S_O\)), both are well-defined, finite,
and computable.
\end{remark}

\begin{proposition}[Existence of carrier-equivalent noisy information]
\label{prop:noisy-exist}
Under Assumption~\ref{assump:carrier-rep}, for every information model
\(\mathcal I=\langle O,T_O,S_O,C,T_C,S_C,R_{\mathcal E}\rangle\),
there exists a noisy information \(\tilde{\mathcal I}\)
\textup{(Definition~\ref{def:noisy-info})} such that:
\begin{enumerate}[label=\textup{(\roman*)}]
  \item \(\tilde S_O\syn S_C\), i.e., the semantic component of \(\tilde{\mathcal I}\)
        is synonymous with the carrier space.
  \item \(\tilde{\mathcal I}\) shares the carrier space \(S_C\) with \(\mathcal I\).
  \item If \(\mathcal I\) is already ideal, the perturbation can be chosen to be trivial:
        \(\tilde S_O=S_O\), \(S_O^-=S_O^+=\varnothing\).
\end{enumerate}
\end{proposition}

\begin{proof}
By Assumption~\ref{assump:carrier-rep}, there exists an
\(\mathcal L_{\mathrm{sem}}\)-definable
\(S'_O\subseteq\mathbb{S}_O\) with \(S'_O\syn S_C\).
Set \(\tilde S_O:=S'_O\) and let \(G\) be a witnessing
\(\mathcal L_{\mathrm{sem}}\)-definable coding isomorphism graph.
Then \((S_O^-,S_O^+)=(S_O\setminus S'_O,\;S'_O\setminus S_O)\), and
\(\tilde S_O=(S_O\setminus S_O^-)\cup S_O^+=S'_O\), confirming
\textup{(i)} and~\textup{(ii)}.
For \textup{(iii)}: if \(\mathcal I\) is ideal, then \(S_O\syn S_C\) by
Definition~\ref{def:ideal-info}; choosing \(S'_O:=S_O\) yields
\(S_O^-=S_O^+=\varnothing\).
\end{proof}

\begin{remark}[Non-uniqueness and the role of the noise pair]
\label{rem:noise-interpretation}
The noisy information \(\tilde{\mathcal I}\) depends on the choice of \(\tilde S_O\)
(and its witnessing graph \(G\)); it is not in general unique.
Different choices yield different noise pairs, reflecting different ways of ``reading''
the carrier semantically.
Regardless of the choice, the noise pair \((S_O^-,S_O^+)\) provides a structured
decomposition of the discrepancy between the intended semantics \(S_O\) and the
carrier-equivalent semantics \(\tilde S_O\):
\(S_O^-\) collects semantic states that the carrier fails to realize, and
\(S_O^+\) collects carrier-realized states absent from the original semantic intent.
The semantic invariants \(\mathsf{A}(\tilde{\mathcal I})\) and
\(\mathsf{D_d}(\tilde{\mathcal I})\) then serve as computable measures of the
\emph{effective} semantic structure of \(\mathcal I\) as mediated by its carrier.
\end{remark}

\begin{definition}[Closure fidelity]
\label{def:closure-fidelity}
For finite sets \(S,\hat S\subseteq\mathbb{S}_O\), the \emph{closure fidelity} is
the Jaccard index of their deductive closures~\cite{lipkus1999proof}:
\begin{equation}\label{eq:closure-fidelity}
  \mathsf{F}_{\Cn}(S,\hat S)
  \;:=\;
  \frac{\lvert\Cn(S)\cap\Cn(\hat S)\rvert}
       {\lvert\Cn(S)\cup\Cn(\hat S)\rvert},
\end{equation}
with \(0/0:=1\).
We have \(\mathsf{F}_{\Cn}\in[0,1]\) with
\(\mathsf{F}_{\Cn}(S,\hat S)=1\) if and only if
\(\Cn(S)=\Cn(\hat S)\).
\end{definition}

\begin{definition}[Atom-core preservation ratio]
\label{def:atom-preservation}
For a knowledge base \(S_O\) with core \(A=\Atom(S_O)\) and any
\(\hat S\subseteq\mathbb{S}_O\), define
\begin{equation}\label{eq:rho-atom}
  \rho_{\Atom}(S_O,\hat S)
  \;:=\;
  \frac{\lvert A\cap\hat S\rvert}{\lvert A\rvert},
\end{equation}
with the convention \(0/0:=1\) when \(A=\varnothing\).
This measures the fraction of irredundant core elements of \(S_O\) that are
present in~\(\hat S\).
\end{definition}

\begin{proposition}[Noise pair, core preservation, and closure fidelity]
\label{prop:noise-fidelity}
Let \(\tilde S_O=(S_O\setminus S_O^-)\cup S_O^+\) be a noisy semantic base of
\(S_O\) \textup{(Definition~\ref{def:noisy})} and let \(A=\Atom(S_O)\).
\begin{enumerate}[label=\textup{(\roman*)}]
  \item \emph{Core preservation:}
        \(\rho_{\Atom}(S_O,\tilde S_O)=1\) if and only if
        \(A\cap S_O^-=\varnothing\), i.e., no core element is lost.
  \item \emph{Closure inclusion under full core preservation:}
        if \(A\cap S_O^-=\varnothing\), then
        \(\Cn(S_O)\subseteq\Cn(\tilde S_O)\).
  \item \emph{Closure equality:}
        if \(A\cap S_O^-=\varnothing\) and
        \(S_O^+\subseteq\Cn(S_O)\)
        \textup{(}all spurious states are derivable from the original\textup{)},
        then \(\Cn(S_O)=\Cn(\tilde S_O)\) and
        \(\mathsf{F}_{\Cn}(S_O,\tilde S_O)=1\).
  \item \emph{Trivial noise:}
        \(S_O^-=S_O^+=\varnothing\) implies
        \(\rho_{\Atom}=1\) and \(\mathsf{F}_{\Cn}=1\).
\end{enumerate}
\end{proposition}

\begin{proof}
\textup{(i)}\
\(A\cap\tilde S_O=A\cap\bigl[(S_O\setminus S_O^-)\cup S_O^+\bigr]
=(A\setminus S_O^-)\cup(A\cap S_O^+)\).
Since \(A\subseteq S_O\) and \(S_O^+\subseteq\mathbb S_O\setminus S_O\),
\(A\cap S_O^+=\varnothing\).
Hence \(A\cap\tilde S_O=A\setminus S_O^-\), so
\(\rho_{\Atom}=|A\setminus S_O^-|/|A|=1\) iff \(A\cap S_O^-=\varnothing\).

\smallskip\noindent
\textup{(ii)}\
If \(A\cap S_O^-=\varnothing\), then \(A\subseteq S_O\setminus S_O^-\subseteq\tilde S_O\).
By monotonicity~\textup{(Cn2)},
\(\Cn(A)\subseteq\Cn(\tilde S_O)\).
Since \(\Cn(A)=\Cn(S_O)\)
\textup{(Proposition~\ref{prop:atom-core-correct}(i))},
\(\Cn(S_O)\subseteq\Cn(\tilde S_O)\).

\smallskip\noindent
\textup{(iii)}\
By~\textup{(ii)}, \(\Cn(S_O)\subseteq\Cn(\tilde S_O)\).
For the reverse:
\(S_O\setminus S_O^-\subseteq S_O\subseteq\Cn(S_O)\) by
reflexivity~\textup{(Cn1)}, and
\(S_O^+\subseteq\Cn(S_O)\) by hypothesis.
Hence \(\tilde S_O\subseteq\Cn(S_O)\).
Monotonicity gives \(\Cn(\tilde S_O)\subseteq\Cn(\Cn(S_O))=\Cn(S_O)\) by
idempotence~\textup{(Cn3)}.
Therefore \(\Cn(\tilde S_O)=\Cn(S_O)\) and
\(\mathsf{F}_{\Cn}(S_O,\tilde S_O)=1\).

\smallskip\noindent
\textup{(iv)}\
Immediate from \(\tilde S_O=S_O\).
\end{proof}

\section{Semantic Channel}
\label{sec:channel}

Section~\ref{sec:model} established the \emph{structural framework} of information
models---state spaces, enabling maps, deductive closure, irredundant cores, derivation
depth, and noisy perturbations---without invoking any probabilistic concepts.
The enabling map \(\mathcal E\) captures which carrier realizations are
\emph{admissible} for each semantic state, but assigns no likelihood to them; the noise
pair \((S_O^-,S_O^+)\) measures set-theoretic discrepancy between intended and realized
semantics without modeling its stochastic origin.
Importantly, Section~\ref{sec:model} does not \emph{exclude} randomness: the state
spaces and enabling maps are defined in full generality as (possibly random) mathematical
objects.
What Section~\ref{sec:model} does not provide is the \emph{probabilistic apparatus}
needed to reason about the likelihood of specific carrier realizations, the statistical
behavior of end-to-end transmission, or the fundamental limits of semantic
communication~\cite{cover2006elements,polyanskiy2025information}.

This section furnishes precisely that apparatus.
The central construction is the \emph{semantic channel}: a stochastic map---realized as
a composition of Markov kernels---that transforms semantic source states into
reconstructed semantic states through a carrier medium subject to noise.
The construction is governed by the following
\emph{design principle}:

\medskip
\begin{quote}
\itshape
Every Markov kernel in the semantic channel framework is an
\emph{enabling kernel} for some information model---the probabilistic refinement of
a set-valued enabling map whose support is constrained by the enabling structure of
that model.
\end{quote}

\medskip\noindent
Under this principle, the three stages of semantic communication---encoding,
carrier transmission, and decoding---are modeled as three composable information models
(Definition~\ref{def:model-composition}), each equipped with an enabling kernel
(Definition~\ref{def:enabling-kernel}).
The end-to-end semantic channel kernel arises as the composition of these three
enabling kernels and is itself an enabling kernel for the composite information model
(Proposition~\ref{prop:enabling-compose}).
This ensures that the structural constraints established in
Section~\ref{sec:model}---in particular, the enabling support inclusion and the
noise framework of Section~\ref{subsec:noisy}---are inherited by and preserved through
the probabilistic layer.

The section proceeds in five stages.
Section~\ref{subsec:prob-structure} endows the finite state spaces of
Section~\ref{sec:model} with probability distributions, introduces the Markov kernel
as the fundamental probabilistic building block, and defines the
\emph{enabling kernel}---the universal interface between structural enabling maps and
stochastic transitions.
Section~\ref{subsec:channel-def} applies this machinery to define the encoding kernel,
the carrier channel kernel, the decoding kernel, and their composition into an
end-to-end \emph{semantic channel kernel}, each arising from a dedicated information
model.
Section~\ref{subsec:distortion} introduces \emph{semantic distortion} measures that
go beyond symbol-level error rates by exploiting the deductive closure and
derivation-depth structure inherited from Section~\ref{sec:model}.
Section~\ref{subsec:invariants} defines the principal invariants of a semantic
channel---semantic mutual information, semantic channel capacity, and structural
fidelity indices---and relates them to the classical Shannon capacity of the underlying
carrier channel.
Section~\ref{subsec:coding} formulates a semantic channel coding problem and states
preliminary achievability and converse bounds.

Throughout, we fix the information model
\(\mathcal I=\langle O,T_O,S_O,C,T_C,S_C,R_{\mathcal E}\rangle\)
from Section~\ref{sec:model} and adopt all standing conventions and axioms
established there.

\subsection{Probabilistic Structure on State Spaces}
\label{subsec:prob-structure}

Because all state spaces in our framework are finite
(Assumptions~\ref{assump:ordered-structures} and~\ref{assump:finite-so}), the required
measure-theoretic apparatus reduces to discrete probability distributions and
stochastic matrices.
This subsection introduces the key primitives: probability distributions, Markov
kernels, and---most importantly---the \emph{enabling kernel}, which provides the
universal probabilistic interface to the structural enabling maps of
Section~\ref{sec:model}.

\smallskip
\noindent\textbf{Probabilistic notation convention.}
Throughout Section~\ref{sec:channel}, random elements drawn from a distribution~\(P\)
on a state space~\(S\) are denoted by the corresponding sans-serif letter:
uppercase for random variables (e.g., \(\mathsf{S}_o\) for a random semantic state
drawn from~\(P_O\)) and lowercase for their realizations
(e.g., \(s_o\in S_O\)).
Shannon entropy, conditional entropy, and mutual information are denoted
\(H(\cdot)\), \(H(\cdot\mid\cdot)\), and \(I(\cdot\,;\cdot)\) respectively, with
subscripts indicating the underlying distribution when
ambiguous~\cite{cover2006elements}.
Expectations are denoted~\(\E[\cdot]\).
These standard information-theoretic quantities are instantiated on the semantic and
carrier state spaces; their interaction with the proof-system structure of
Section~\ref{sec:model} is the subject of
Sections~\ref{subsec:distortion}--\ref{subsec:invariants}.

\begin{definition}[Probability distribution on a finite set]
\label{def:prob-dist}
Let \(S\) be a nonempty finite set.
A \emph{probability distribution} (or \emph{probability mass function}) on~\(S\) is a
function \(P:S\to[0,1]\) satisfying
\[
  \sum_{s\in S}P(s)=1.
\]
The \emph{support} of~\(P\) is \(\supp(P):=\{s\in S:P(s)>0\}\).
We write \(\Delta(S)\) for the set of all probability distributions on~\(S\)
(the \emph{probability simplex} over~\(S\)), and set
\(\Delta(\varnothing):=\varnothing\) by convention.
\end{definition}

\begin{definition}[Semantic source]
\label{def:semantic-source}
A \emph{semantic source} for an information model~\(\mathcal I\) is a pair
\((\mathcal I,\,P_O)\) where \(P_O\in\Delta(S_O)\).
We refer to~\(P_O\) as the \emph{source distribution} and, when \(\mathcal I\) is
clear from context, abbreviate the semantic source to~\((S_O,P_O)\).

A semantic source is \emph{full-support} if \(\supp(P_O)=S_O\), and \emph{uniform} if
\(P_O(s_o)=1/|S_O|\) for every \(s_o\in S_O\).
\end{definition}

\begin{remark}[A semantic source is not merely a random variable]
\label{rem:source-not-rv}
In Shannon's framework~\cite{shannon1948mathematical,cover2006elements},
a discrete memoryless source is fully specified by a probability distribution on a
finite alphabet.
A semantic source carries additional structure: the underlying set~\(S_O\) is endowed
with a deductive closure operator~\(\Cn\), an irredundant core~\(\Atom(S_O)\), and a
derivation-depth stratification~\(\Dd\) (Section~\ref{subsec:atomic-derivation}).
These inferential structures are invisible to classical information-theoretic
measures (entropy, mutual information) but will enter the semantic distortion and
capacity definitions of Sections~\ref{subsec:distortion}--\ref{subsec:invariants}.
\end{remark}

\begin{definition}[Markov kernel between finite state spaces]
\label{def:markov-kernel}
Let \(X\) and \(Y\) be nonempty finite sets.
A \emph{Markov kernel} (or \emph{transition kernel}) from~\(X\) to~\(Y\) is a function
\[
  \kappa\;:\;X\times Y\;\longrightarrow\;[0,1]
\]
such that for every \(x\in X\),
\[
  \sum_{y\in Y}\kappa(y\mid x)=1,
\]
i.e., \(\kappa(\cdot\mid x)\in\Delta(Y)\) for each~\(x\).
We write \(\kappa:X\rightsquigarrow Y\) and, equivalently, view~\(\kappa\) as the
\(|X|\times|Y|\) row-stochastic matrix
\([\kappa(y\mid x)]_{x\in X,\,y\in Y}\).
A kernel is \emph{deterministic} if for every \(x\in X\) there exists a unique
\(y_x\in Y\) with \(\kappa(y_x\mid x)=1\); in this case \(\kappa\) is identified
with the function \(f:X\to Y\) given by \(f(x):=y_x\).
\end{definition}

\begin{definition}[Kernel composition and identity kernel]
\label{def:kernel-composition}
Let \(\kappa_1:X\rightsquigarrow Y\) and \(\kappa_2:Y\rightsquigarrow Z\) be Markov
kernels.
Their \emph{composition} \(\kappa_2\circ\kappa_1:X\rightsquigarrow Z\) is defined by
\begin{equation}\label{eq:kernel-composition}
  (\kappa_2\circ\kappa_1)(z\mid x)
  \;:=\;
  \sum_{y\in Y}\kappa_1(y\mid x)\;\kappa_2(z\mid y),
  \qquad x\in X,\; z\in Z.
\end{equation}
The \emph{identity kernel} on~\(X\) is the deterministic kernel
\(\mathrm{id}_X:X\rightsquigarrow X\) with
\(\mathrm{id}_X(x'\mid x):=\mathbf{1}[x'=x]\).
\end{definition}

\begin{proposition}[Closure under composition]
\label{prop:kernel-composition}
The composition \(\kappa_2\circ\kappa_1\) is a Markov kernel from~\(X\) to~\(Z\).
Kernel composition is associative, and the identity kernel is a two-sided identity:
\(\mathrm{id}_Y\circ\kappa_1=\kappa_1\) and \(\kappa_2\circ\mathrm{id}_Y=\kappa_2\).
\end{proposition}

\begin{proof}
For each \(x\in X\),
\(\sum_{z}(\kappa_2\circ\kappa_1)(z\mid x)
=\sum_{z}\sum_{y}\kappa_1(y\mid x)\,\kappa_2(z\mid y)
=\sum_{y}\kappa_1(y\mid x)\bigl(\sum_{z}\kappa_2(z\mid y)\bigr)
=\sum_{y}\kappa_1(y\mid x)=1\),
and each term is non-negative, so \(\kappa_2\circ\kappa_1\) is a Markov kernel.
Associativity and the identity laws follow from rearranging finite sums
(equivalently, from associativity of matrix multiplication and the property of the
identity matrix).
\end{proof}

\begin{definition}[Enabling kernel for an information model]
\label{def:enabling-kernel}
Let \(\mathcal I\) be an information model with semantic state set \(S_O\),
carrier state set \(S_C\), and enabling map
\(\mathcal E:S_O\Rightarrow S_C\)
(Axiom~\ref{ax:enabling-mapping}).
An \emph{enabling kernel} for~\(\mathcal I\) is a Markov kernel
\(\kappa:S_O\rightsquigarrow S_C\)
satisfying the \emph{enabling support constraint}: for every \(s_o\in S_O\),
\begin{equation}\label{eq:enabling-support}
  \supp\bigl(\kappa(\cdot\mid s_o)\bigr)
  \;\subseteq\;
  \mathcal E(s_o).
\end{equation}
Denote the set of all enabling kernels for \(\mathcal I\) by
\(\mathcal K(\mathcal I)\).
\end{definition}

\begin{proposition}[Every information model admits an enabling kernel]
\label{prop:enabling-existence}
For any information model \(\mathcal I\) satisfying
Axiom~\textup{\ref{ax:enabling-mapping}},
\(\mathcal K(\mathcal I)\neq\varnothing\).
In particular, the computable enabling selector~\(e\) of
Axiom~\textup{\ref{ax:enabling-mapping}(E3)} induces a deterministic enabling kernel
\(\kappa_e\in\mathcal K(\mathcal I)\) via
\(\kappa_e(s_c\mid s_o):=\mathbf{1}[s_c=e(s_o)]\).
\end{proposition}

\begin{proof}
By \textup{(E3)}, \(e(s_o)\in\mathcal E(s_o)\) for every \(s_o\in S_O\).
Hence \(\supp(\kappa_e(\cdot\mid s_o))=\{e(s_o)\}\subseteq\mathcal E(s_o)\),
and \(\kappa_e\) is a Markov kernel since
\(\sum_{s_c}\kappa_e(s_c\mid s_o)=1\) for each~\(s_o\).
\end{proof}

\begin{proposition}[Enabling kernels compose under model composition]
\label{prop:enabling-compose}
Let \((\mathcal I_1,\mathcal I_2)\) be a composable pair of information models
\textup{(Definition~\ref{def:model-composition})}.
If \(\kappa_1\in\mathcal K(\mathcal I_1)\) and
\(\kappa_2\in\mathcal K(\mathcal I_2)\), then
\[
  \kappa_2\circ\kappa_1
  \;\in\;
  \mathcal K(\mathcal I_2\circ\mathcal I_1).
\]
\end{proposition}

\begin{proof}
By Proposition~\ref{prop:kernel-composition},
\(\kappa_2\circ\kappa_1:S_O^{(1)}\rightsquigarrow S_C^{(2)}\) is a Markov kernel.
It remains to verify the enabling support constraint.
Let \(s_o\in S_O^{(1)}\) and suppose
\((\kappa_2\circ\kappa_1)(s'\mid s_o)>0\) for some \(s'\in S_C^{(2)}\).
By~\eqref{eq:kernel-composition},
\[
  \sum_{s_c\in S_C^{(1)}}\kappa_1(s_c\mid s_o)\,\kappa_2(s'\mid s_c)>0,
\]
so there exists \(s_c\in S_C^{(1)}\) with \(\kappa_1(s_c\mid s_o)>0\) and
\(\kappa_2(s'\mid s_c)>0\).
By the enabling support constraint of~\(\kappa_1\),
\(s_c\in\mathcal E_1(s_o)\);
by that of~\(\kappa_2\), \(s'\in\mathcal E_2(s_c)\).
Hence
\(
  s'\in\mathcal E_2(s_c)
  \subseteq
  \bigcup_{u\in\mathcal E_1(s_o)}\mathcal E_2(u)
  =\mathcal E_{2\circ 1}(s_o),
\)
confirming
\(\supp\bigl((\kappa_2\circ\kappa_1)(\cdot\mid s_o)\bigr)
  \subseteq\mathcal E_{2\circ 1}(s_o)\).
\end{proof}

\begin{remark}[Design principle realized]
\label{rem:design-principle}
Proposition~\ref{prop:enabling-compose} is the probabilistic counterpart of
Proposition~\ref{prop:composition-axioms}: just as composite enabling maps
preserve the structural axioms of Section~\ref{sec:model}, composite enabling kernels
preserve the enabling support constraint.
This guarantees that the end-to-end semantic channel kernel constructed in
Section~\ref{subsec:channel-def} is itself an enabling kernel for the composite
information model, fulfilling the design principle stated in the preamble of this
section.
\end{remark}

\subsection{Semantic Channel: Definition via Markov Kernels}
\label{subsec:channel-def}

This subsection assembles the end-to-end \emph{semantic channel} from three
information models---encoding, carrier transmission, and decoding---each equipped with
an enabling kernel.
The received carrier space \(\hat S_C\) and the reconstructed semantic space
\(\hat S_O\) are \emph{not} postulated independently; they arise as carrier state sets
of the channel and decoding information models, respectively.

\subsubsection*{Stage 1: Encoding}

\begin{definition}[Encoding kernel]
\label{def:encoding-kernel}
Let \(\mathcal I=\langle O,T_O,S_O,C,T_C,S_C,R_{\mathcal E}\rangle\) be the
information model fixed in Section~\ref{sec:model}, with enabling map
\(\mathcal E:S_O\Rightarrow S_C\).
An \emph{encoding kernel} is an enabling kernel for~\(\mathcal I\)
\textup{(Definition~\ref{def:enabling-kernel})}, i.e., a Markov kernel
\(\kappa_{\enc}:S_O\rightsquigarrow S_C\) satisfying
\(\kappa_{\enc}\in\mathcal K(\mathcal I)\).
\end{definition}

\begin{remark}[Deterministic encoding as a degenerate kernel]
\label{rem:det-encoding}
The computable enabling selector~\(e\) of
Axiom~\ref{ax:enabling-mapping}\textup{(E3)} induces a deterministic encoding kernel
\(\kappa_e(s_c\mid s_o):=\mathbf{1}[s_c=e(s_o)]\),
which belongs to \(\mathcal K(\mathcal I)\) by
Proposition~\ref{prop:enabling-existence}.
Thus every information model admits at least one encoding kernel.
Conversely, any \(\kappa_{\enc}\in\mathcal K(\mathcal I)\)
can be interpreted as a \emph{randomized selection} from the set of admissible carrier
realizations~\(\mathcal E(s_o)\); randomized encoding may strictly outperform
deterministic encoding in the presence of channel noise, paralleling the classical
situation in Shannon theory~\cite{cover2006elements}.
\end{remark}

\begin{proposition}[Induced joint and marginal distributions]
\label{prop:induced-dist}
Let \((S_O,P_O)\) be a semantic source and
\(\kappa_{\enc}\in\mathcal K(\mathcal I)\) an encoding kernel.
Define:
\begin{enumerate}[label=\textup{(\roman*)}]
  \item the \emph{joint distribution} on \(S_O\times S_C\) by
        \(P_{OC}(s_o,s_c):=P_O(s_o)\,\kappa_{\enc}(s_c\mid s_o)\);
  \item the \emph{induced carrier distribution} by
        \(P_C(s_c):=\sum_{s_o\in S_O}P_O(s_o)\,\kappa_{\enc}(s_c\mid s_o)\).
\end{enumerate}
Then \(P_{OC}\in\Delta(S_O\times S_C)\), \(P_C\in\Delta(S_C)\),
\(P_C\) is the \(S_C\)-marginal of~\(P_{OC}\), and
\begin{equation}\label{eq:support-joint}
  \supp(P_{OC})
  \;\subseteq\;
  \bigl\{(s_o,s_c)\in S_O\times S_C : s_c\in\mathcal E(s_o)\bigr\}.
\end{equation}
Moreover, if \(P_O\) is full-support and \(\kappa_{\enc}\) is deterministic
with selector~\(e\), then \(\supp(P_C)=e(S_O)\).
\end{proposition}

\begin{proof}
Non-negativity is immediate.
Normalization:
\(\sum_{s_o,s_c}P_{OC}(s_o,s_c)
=\sum_{s_o}P_O(s_o)\sum_{s_c}\kappa_{\enc}(s_c\mid s_o)
=\sum_{s_o}P_O(s_o)=1\).
Hence \(P_{OC}\in\Delta(S_O\times S_C)\) and \(P_C\) is its \(S_C\)-marginal, so
\(P_C\in\Delta(S_C)\).
The support inclusion~\eqref{eq:support-joint} follows from the enabling support
constraint~\eqref{eq:enabling-support}:
if \(P_{OC}(s_o,s_c)>0\), then \(P_O(s_o)>0\) and
\(\kappa_{\enc}(s_c\mid s_o)>0\), hence
\(s_c\in\supp(\kappa_{\enc}(\cdot\mid s_o))\subseteq\mathcal E(s_o)\).
The final claim follows because under deterministic encoding,
\(P_C(s_c)=\sum_{s_o:e(s_o)=s_c}P_O(s_o)\), which is positive if and only if
\(s_c\in e(S_O)\) (since \(P_O\) is full-support).
\end{proof}

\begin{remark}[Pushforward under synonymy]
\label{rem:pushforward-syn}
If \(S_1\syn S_2\) via \(\tau_{12}\)
(Definition~\ref{def:synonymous-states}) and \(P_1\in\Delta(S_1)\), the
\emph{pushforward distribution} \((\tau_{12})_\# P_1\in\Delta(S_2)\) is defined by
\(
  \bigl[(\tau_{12})_\# P_1\bigr](s_2)
  :=
  P_1\bigl(\tau_{12}^{-1}(s_2)\bigr)
\)
for \(s_2\in S_2\).
Since \(\tau_{12}\) is a bijection that preserves and reflects the induced time
precedence, the pushforward preserves any temporal correlation structure of the
source.
When \(\mathcal I\) is ideal and \(S_O\syn S_C\) via \(\tau_{OC}\)
(Definition~\ref{def:ideal-info}), the deterministic encoding kernel
\(\kappa_e\) with \(e=\tau_{OC}\) yields
\(P_C=(\tau_{OC})_\# P_O\).
\end{remark}

\subsubsection*{Stage 2: Carrier Channel}

\begin{definition}[Carrier channel information model]
\label{def:carrier-channel-model}
A \emph{carrier channel information model} (composable with \(\mathcal I\)) is an
information model
\[
  \mathcal I_{\mathrm{ch}}
  \;=\;
  \bigl\langle\,
    C,\; T_C,\; S_C,\;
    \hat C,\; T_{\hat C},\; \hat S_C,\;
    R_{\mathcal E}^{\mathrm{ch}}
  \,\bigr\rangle,
\]
satisfying Axiom~\textup{\ref{ax:enabling-mapping}}, whose semantic state set is
\(S_C\) (the carrier state set of~\(\mathcal I\)) and whose carrier state set is a
nonempty finite set~\(\hat S_C\) called the \emph{received carrier state space},
equipped with an injective encoding
\(\enc_{\hat C}:\hat S_C\to\{0,1\}^*\).

\smallskip
The induced enabling map
\(\mathcal E_{\mathrm{ch}}:S_C\Rightarrow\hat S_C\) specifies, for each transmitted
carrier state, the set of physically possible received states.
The \emph{carrier channel kernel} is any enabling kernel for~\(\mathcal I_{\mathrm{ch}}\):
\[
  W\;\in\;\mathcal K(\mathcal I_{\mathrm{ch}}),
  \qquad
  W:S_C\rightsquigarrow\hat S_C.
\]
\end{definition}

\begin{remark}[Common special cases of the carrier channel]
\label{rem:carrier-channel-cases}
\textup{(i)}~In many applications the physical channel output alphabet coincides with
its input alphabet; setting \(\hat S_C:=S_C\) and
\(\mathcal E_{\mathrm{ch}}(s_c):=S_C\) for all \(s_c\) recovers the standard
unconstrained discrete memoryless channel (DMC) of
Shannon theory~\cite{shannon1948mathematical}.

\smallskip\noindent
\textup{(ii)}~When the carrier channel is \emph{noiseless},
\(\hat S_C=S_C\) and \(\mathcal E_{\mathrm{ch}}(s_c)=\{s_c\}\) for every
\(s_c\in S_C\); the unique enabling kernel is the identity kernel
\(W=\mathrm{id}_{S_C}\).

\smallskip\noindent
\textup{(iii)}~In the terminology of Section~\ref{sec:model}, the roles of
``semantic space'' and ``carrier space'' are relative to each information model.
In \(\mathcal I_{\mathrm{ch}}\), the transmitted carrier states \(S_C\) play the
structural role of \(S_O\), and the received carrier states \(\hat S_C\) play the
role of~\(S_C\).
This relabeling is purely notational; it reflects the fact that each link in a
communication chain is itself an instance of the general information model framework.
\end{remark}

\subsubsection*{Stage 3: Decoding}

\begin{definition}[Decoding information model]
\label{def:decoding-model}
A \emph{decoding information model} (composable with \(\mathcal I_{\mathrm{ch}}\)) is
an information model
\[
  \mathcal I_{\mathrm{dec}}
  \;=\;
  \bigl\langle\,
    \hat C,\; T_{\hat C},\; \hat S_C,\;
    \hat O,\; T_{\hat O},\; \hat S_O,\;
    R_{\mathcal E}^{\mathrm{dec}}
  \,\bigr\rangle,
\]
satisfying Axiom~\textup{\ref{ax:enabling-mapping}}, whose semantic state set is
\(\hat S_C\) (the carrier state set of~\(\mathcal I_{\mathrm{ch}}\)) and whose carrier
state set
\(\hat S_O\subseteq\mathbb{S}_O\) is called the
\emph{reconstructed semantic state space}, inheriting the encoding~\(\enc_O\), the
time-index map~\(\timeindex_O\), and the proof system
\((\mathsf{PS},\,T_{\mathsf{PS}},\,\Cn)\) from
Sections~\textup{\ref{subsec:logic}}--\textup{\ref{subsec:atomic-derivation}}.

\smallskip
The induced enabling map
\(\mathcal E_{\mathrm{dec}}:\hat S_C\Rightarrow\hat S_O\) specifies, for each received
carrier state, the set of admissible decoded semantic states.
The \emph{decoding kernel} is any enabling kernel for~\(\mathcal I_{\mathrm{dec}}\):
\[
  D\;\in\;\mathcal K(\mathcal I_{\mathrm{dec}}),
  \qquad
  D:\hat S_C\rightsquigarrow\hat S_O.
\]
\end{definition}

\begin{remark}[Inherited proof-system structure and heterogeneous receivers]
\label{rem:decoded-space-structure}
\textup{(i)}~Since \(\hat S_O\subseteq\mathbb{S}_O\) and the proof system
\(\mathsf{PS}\) acts on all of~\(\mathbb{S}_O\)
(Section~\ref{subsec:atomic-derivation}), the deductive closure
\(\Cn(\hat S_O)\), the irredundant core \(\Atom(\hat S_O)\), and the derivation
depth \(\Dd(q\mid\Atom(\hat S_O))\) are all well-defined for~\(\hat S_O\).
In particular, the semantic invariants \(\mathsf{A}\) and \(\mathsf{D_d}\)
(Definitions~\ref{def:atomicity-measure}--\ref{def:max-depth}) are defined for
\(\hat S_O\) by the same constructions as for~\(S_O\).
This structural inheritance is what distinguishes the semantic channel framework from a
purely symbol-level treatment.

\smallskip\noindent
\textup{(ii)}~When the receiver seeks to reconstruct states in the original semantic
vocabulary, one takes \(\hat S_O:=S_O\).
When the receiver operates with a coarser or finer semantic vocabulary
than the sender, \(\hat S_O\) may differ from~\(S_O\);
this generality is relevant to heterogeneous agent scenarios
(e.g., in multi-agent systems where different agents maintain different knowledge
bases~\cite{lamb2020graph,garcez2023neurosymbolic}).
More generally, the same sender communicating with different receivers gives rise
to different information models (sharing the same \(S_O\) but differing in \(S_C\),
\(\hat S_C\), \(\hat S_O\), and in the associated enabling maps and noise structures).
\end{remark}

\subsubsection*{End-to-End Composition}

The three information models
\(\mathcal I\), \(\mathcal I_{\mathrm{ch}}\), and \(\mathcal I_{\mathrm{dec}}\)
form a composable chain:
\(S_C(\mathcal I)=S_O(\mathcal I_{\mathrm{ch}})=S_C\) and
\(S_C(\mathcal I_{\mathrm{ch}})=S_O(\mathcal I_{\mathrm{dec}})=\hat S_C\).
The composite information model
(Definition~\ref{def:model-composition}, Remark~\ref{rem:composition-assoc})
\begin{equation}\label{eq:composite-model}
  \mathcal I_{\mathrm{sem}}
  \;:=\;
  \mathcal I_{\mathrm{dec}}\circ\mathcal I_{\mathrm{ch}}\circ\mathcal I
\end{equation}
has semantic state set \(S_O\), carrier state set \(\hat S_O\), and composite enabling
map
\begin{equation}\label{eq:composite-enabling-e2e}
  \mathcal E_{\mathrm{sem}}(s_o)
  \;=\;
  \bigcup_{s_c\,\in\,\mathcal E(s_o)}\;
  \bigcup_{\hat s_c\,\in\,\mathcal E_{\mathrm{ch}}(s_c)}
  \mathcal E_{\mathrm{dec}}(\hat s_c),
  \; s_o\in S_O.
\end{equation}

\subsubsection*{End-to-End Noise Structure}

The composite information model \(\mathcal I_{\mathrm{sem}}\) has semantic
state set~\(S_O\) and carrier state set
\(\hat S_O\subseteq\mathbb S_O\)
(Definition~\ref{def:decoding-model}).
Since both \(S_O\) and \(\hat S_O\) reside in the common semantic universe
\(\mathbb S_O\), the noisy information framework of
Section~\ref{subsec:noisy} applies directly to the end-to-end model,
yielding a structured decomposition of the semantic gap between sender and
receiver.

\begin{proposition}[End-to-end carrier representability and noisy semantic base]
\label{prop:e2e-carrier-rep}
The composite information model
\(\mathcal I_{\mathrm{sem}}
  =\mathcal I_{\mathrm{dec}}\circ\mathcal I_{\mathrm{ch}}\circ\mathcal I\)
satisfies Assumption~\textup{\ref{assump:carrier-rep}}.
Consequently, setting \(\tilde S_O:=\hat S_O\), the reconstructed
semantic space is a noisy semantic base of~\(S_O\)
\textup{(Definition~\ref{def:noisy})} with \emph{end-to-end noise pair}
\begin{equation}\label{eq:e2e-noise-minus}
  S_O^{-}:=S_O\setminus\tilde S_O
  \qquad\textup{(intended but unreconstructable)},
\end{equation}
\begin{equation}\label{eq:e2e-noise-plus}
  S_O^{+}:=\tilde S_O\setminus S_O
  \qquad\textup{(unintended but reachable)},
\end{equation}
so that
\begin{equation}\label{eq:tildeSO-perturbation}
  \tilde S_O
  \;=\;(S_O\setminus S_O^{-})\cup S_O^{+}.
\end{equation}
Moreover, the noisy information
\(\tilde{\mathcal I}_{\mathrm{sem}}\)
\textup{(Definition~\ref{def:noisy-info})} associated with
\(\mathcal I_{\mathrm{sem}}\) is well-defined, and the semantic
invariants
\(\mathsf A(\tilde{\mathcal I}_{\mathrm{sem}})
  =\lvert\Atom(\tilde S_O)\rvert\) and
\(\mathsf{D_d}(\tilde{\mathcal I}_{\mathrm{sem}})
  =\max_{q\in\tilde S_O}\Dd\bigl(q\mid\Atom(\tilde S_O)\bigr)\)
are finite and computable
\textup{(Remark~\ref{rem:noisy-ideal-connection})}.
\end{proposition}

\begin{proof}
In the composite model \(\mathcal I_{\mathrm{sem}}\), the carrier state set
is~\(\hat S_O\).
By Definition~\ref{def:decoding-model},
\(\hat S_O\subseteq\mathbb S_O\) and \(\hat S_O\) is
\(\mathcal L_{\mathrm{sem}}\)-definable
(Axiom~\ref{ax:time-domains} applied to~\(\mathcal I_{\mathrm{dec}}\)).
Setting \(S'_O:=\hat S_O\), the reflexivity of~\(\syn\)
(Proposition~\ref{prop:syn-equiv}) gives
\(S'_O\syn\hat S_O\), witnessed by the identity graph
\(G_{\mathrm{id}}(s,s'):=\mathbf 1[s=s']\).
Hence Assumption~\ref{assump:carrier-rep} is satisfied with
\(\tilde S_O=\hat S_O\).
The noise pair follows from Definition~\ref{def:noisy-info}, and the
computability statement from Remark~\ref{rem:noisy-ideal-connection}
(whose hypotheses are verified by the same argument given there, with
\(\tilde S_O\) in place of the general carrier-equivalent state set).
\end{proof}

\begin{corollary}[End-to-end core preservation and closure fidelity]
\label{cor:e2e-noise-fidelity}
Let \(A=\Atom(S_O)\) and let \((S_O^{-},S_O^{+})\) be the end-to-end
noise pair of Proposition~\textup{\ref{prop:e2e-carrier-rep}}.
Then all four conclusions of
Proposition~\textup{\ref{prop:noise-fidelity}} hold with
\(\tilde S_O=\hat S_O\):
\begin{enumerate}[label=\textup{(\roman*)}]
  \item \(\rho_{\Atom}(S_O,\tilde S_O)=1\) if and only if
        \(A\cap S_O^{-}=\varnothing\)
        \textup{(}no core element is lost\textup{)}.
  \item If \(A\cap S_O^{-}=\varnothing\), then
        \(\Cn(S_O)\subseteq\Cn(\tilde S_O)\).
  \item If \(A\cap S_O^{-}=\varnothing\) and
        \(S_O^{+}\subseteq\Cn(S_O)\), then
        \(\Cn(S_O)=\Cn(\tilde S_O)\) and
        \(\mathsf{F}_{\Cn}(S_O,\tilde S_O)=1\).
  \item \(S_O^{-}=S_O^{+}=\varnothing\) \textup{(}i.e.,
        \(\tilde S_O=S_O\)\textup{)} implies
        \(\rho_{\Atom}=1\) and \(\mathsf{F}_{\Cn}=1\).
\end{enumerate}
\end{corollary}

\begin{proof}
Immediate from Propositions~\ref{prop:e2e-carrier-rep}
and~\ref{prop:noise-fidelity}, since \(\tilde S_O\) is a noisy
semantic base of~\(S_O\) with noise pair \((S_O^{-},S_O^{+})\).
\end{proof}

\begin{remark}[Stage-wise noise decomposition]
\label{rem:stage-noise}
The noisy information framework of Section~\ref{subsec:noisy} applies
not only to the composite model~\(\mathcal I_{\mathrm{sem}}\) but also
to each constituent model individually, provided
Assumption~\ref{assump:carrier-rep} holds at that stage.
In particular:
(a)~the encoding model~\(\mathcal I\) has semantic space~\(S_O\) and
carrier space~\(S_C\); if Assumption~\ref{assump:carrier-rep} holds
for~\(\mathcal I\), there exists a noisy semantic base
\(\tilde S_O^{(\mathrm{enc})}\subseteq\mathbb S_O\) with
\(\tilde S_O^{(\mathrm{enc})}\syn S_C\), yielding an encoding noise pair
\((S_O^{-,\mathrm{enc}},\,S_O^{+,\mathrm{enc}})\);
(b)~the carrier channel model~\(\mathcal I_{\mathrm{ch}}\) has its own
noise pair capturing physical-layer distortion;
(c)~the decoding model~\(\mathcal I_{\mathrm{dec}}\) has a noise pair
capturing decoder-induced mismatch.
The end-to-end noise pair
\((S_O^{-},S_O^{+})\) subsumes the combined effect of all three stages.
This stage-wise view can guide the identification of the dominant noise
source; a detailed treatment of the composition of stage-wise noise
pairs is deferred to future work.
\end{remark}

\begin{remark}[Notation convention: \(\tilde S_O\) and~\(\hat S_O\)]
\label{rem:tildeSO-convention}
By Proposition~\ref{prop:e2e-carrier-rep}, \(\tilde S_O=\hat S_O\)
as subsets of~\(\mathbb S_O\).
We adopt the following convention for the remainder of this section:
\(\hat S_O\) is used when the reconstructed space is viewed as the
\emph{output alphabet} of the end-to-end channel (the coding-theoretic
perspective), while \(\tilde S_O\) is used when it is viewed as a
\emph{noisy semantic base} of~\(S_O\) with noise pair
\((S_O^{-},S_O^{+})\) (the noise-structure perspective).
The semantic invariants \(\Atom(\tilde S_O)\), \(\Cn(\tilde S_O)\), and
\(\Dd(\cdot\mid\Atom(\tilde S_O))\) are always well-defined since
\(\tilde S_O\subseteq\mathbb S_O\).
\end{remark}


\subsubsection*{Probabilistic Characterization of the End-to-End Noise Pair}

The noise pair \((S_O^{-},S_O^{+})\) is a \emph{set-level}
(deterministic) decomposition of the gap between \(S_O\)
and~\(\tilde S_O\).
The semantic channel kernel
\(\kappa_{\mathrm{sem}}:S_O\rightsquigarrow\tilde S_O\) endows this
decomposition with \emph{probabilistic} content by specifying, for each
sent state~\(s_o\), the likelihood that the reconstruction lands in
each region of~\(\tilde S_O\).

\begin{definition}[Noise-region partition and per-input probabilities]
\label{def:noise-regions}
The noise pair \((S_O^{-},S_O^{+})\) of
Proposition~\ref{prop:e2e-carrier-rep} induces a partition of
\(\tilde S_O\) into a \emph{preserved region} and a
\emph{spurious region}:
\begin{align}
  \tilde S_O^{\cap}
  &\;:=\; S_O\cap\tilde S_O
  \;=\; S_O\setminus S_O^{-}
  &&\text{(preserved)},
  \label{eq:region-preserved}\\[3pt]
  \tilde S_O^{+}
  &\;:=\; S_O^{+}
  \;=\; \tilde S_O\setminus S_O
  &&\text{(spurious)}.
  \label{eq:region-spurious}
\end{align}
The two regions are disjoint:
\(\tilde S_O=\tilde S_O^{\cap}\;\dot\cup\;\tilde S_O^{+}\).
For each \(s_o\in S_O\), define the \emph{preservation probability} and
the \emph{spurious probability}:
\begin{align}
  p_{\cap}(s_o)
  &\;:=\;
  \sum_{\hat s_o\in\tilde S_O^{\cap}}
  \kappa_{\mathrm{sem}}(\hat s_o\mid s_o),
  \label{eq:p-preserved}\\[3pt]
  p_{+}(s_o)
  &\;:=\;
  \sum_{\hat s_o\in\tilde S_O^{+}}
  \kappa_{\mathrm{sem}}(\hat s_o\mid s_o),
  \label{eq:p-spurious}
\end{align}
with \(p_{\cap}(s_o)+p_{+}(s_o)=1\).
For an element \(a\in\Atom(S_O)\cap\tilde S_O^{\cap}\), the
\emph{per-core-element self-preservation probability} is
\begin{equation}\label{eq:pi-self}
  \pi(a)\;:=\;\kappa_{\mathrm{sem}}(a\mid a).
\end{equation}
\end{definition}

\begin{proposition}[Noise-region decomposition of channel behavior]
\label{prop:prob-noise}
Let \((S_O,P_O)\) be a semantic source and let
\(\mathfrak C\) be a semantic channel with end-to-end noise pair
\((S_O^{-},S_O^{+})\) and kernel
\(\kappa_{\mathrm{sem}}:S_O\rightsquigarrow\tilde S_O\).
Let \(A=\Atom(S_O)\).
\begin{enumerate}[label=\textup{(\roman*)}]
  \item \emph{Distortion decomposition:}\;
        The expected Hamming distortion decomposes as
        \begin{align}\label{eq:dH-noise-decomp}
          \bar d_H(\mathfrak C,P_O)
          \;=\;&
          \underbrace{\sum_{s_o\in S_O}P_O(s_o)
            \!\!\sum_{\substack{\hat s_o\in\tilde S_O^{\cap}\\
                               \hat s_o\neq s_o}}
            \!\!\kappa_{\mathrm{sem}}(\hat s_o\mid s_o)
          }_{\displaystyle\bar d_H^{\cap}
             \;\text{(within-vocabulary confusion)}} \nonumber\\
          \;+\;&
          \underbrace{\sum_{s_o\in S_O}P_O(s_o)\,p_{+}(s_o)
          }_{\displaystyle\bar d_H^{+}
             \;\text{(spurious substitution)}}.
        \end{align}
  \item \emph{Necessary condition for perfect symbol fidelity:}\;
        \(\bar d_H(\mathfrak C,P_O)=0\) for a full-support~\(P_O\)
        requires \(S_O^{-}=\varnothing\)
        \textup{(}i.e., \(S_O\subseteq\tilde S_O\)\textup{)} and
        \(\kappa_{\mathrm{sem}}(s_o\mid s_o)=1\) for every
        \(s_o\in S_O\).
  \item \emph{Probabilistic core preservation:}\;
        If \(A\cap S_O^{-}=\varnothing\), then \(\pi(a)\) is
        well-defined for every \(a\in A\) and
        \begin{equation}\label{eq:avg-core-preservation}
          \bar\pi(\mathfrak C,P_O)
          \;:=\;
          \frac{
            \sum_{a\in A}P_O(a)\;\pi(a)
          }{P_O(A)}
        \end{equation}
        measures the average self-preservation probability among core
        elements
        \textup{(}where \(P_O(A):=\sum_{a\in A}P_O(a)\)\textup{)}.
  \item \emph{Sufficient condition for closure preservation
        with high probability:}\;
        If \(A\cap S_O^{-}=\varnothing\) and
        \(S_O^{+}\subseteq\Cn(S_O)\), then by
        Corollary~\textup{\ref{cor:e2e-noise-fidelity}(iii)},
        \(\Cn(S_O)=\Cn(\tilde S_O)\).
        Moreover, substituting any redundant state
        \(s_o\in S_O\setminus A\) by any
        \(\hat s_o\in\tilde S_O\) preserves the deductive
        closure:
        \(\Cn\bigl((S_O\setminus\{s_o\})\cup\{\hat s_o\}\bigr)
          =\Cn(S_O)\).
        \textup{(}This closure-preservation property is
        formalized as zero closure distortion in
        Section~\textup{\ref{subsec:distortion}};
        see Definition~\textup{\ref{def:closure-distortion}}
        and
        Remark~\textup{\ref{rem:closure-distortion-content}(a)}.\textup{)}
\end{enumerate}
\end{proposition}

\begin{proof}
\textup{(i)}\
Since \(\tilde S_O=\tilde S_O^{\cap}\;\dot\cup\;\tilde S_O^{+}\),
\begin{align*}
  \bar d_H(\mathfrak C,P_O)
  =\;&\sum_{s_o}P_O(s_o)
    \sum_{\hat s_o\neq s_o}\kappa_{\mathrm{sem}}(\hat s_o\mid s_o)\\
  =\;&\sum_{s_o}P_O(s_o)\Bigl[
    \sum_{\substack{\hat s_o\in\tilde S_O^{\cap}\\\hat s_o\neq s_o}}
    \kappa_{\mathrm{sem}}(\hat s_o\mid s_o) \\
    +\;&\sum_{\hat s_o\in\tilde S_O^{+}}
    \kappa_{\mathrm{sem}}(\hat s_o\mid s_o)\Bigr].
\end{align*}
For \(s_o\in S_O\) and \(\hat s_o\in\tilde S_O^{+}\subseteq
\mathbb S_O\setminus S_O\), one has \(\hat s_o\neq s_o\), so the second
sum equals~\(p_{+}(s_o)\).

\smallskip\noindent
\textup{(ii)}\
Since \(d_H\ge 0\) and \(\kappa_{\mathrm{sem}}(\hat s_o\mid s_o)\ge 0\),
\(\bar d_H=0\) under full-support~\(P_O\) requires
\(d_H(s_o,\hat s_o)=0\) whenever
\(P_O(s_o)>0\) and \(\kappa_{\mathrm{sem}}(\hat s_o\mid s_o)>0\).
Hence \(\hat s_o=s_o\) for every reachable pair, i.e.,
\(\kappa_{\mathrm{sem}}(s_o\mid s_o)=1\) for all \(s_o\in S_O\).
This requires \(s_o\in\tilde S_O\) for every \(s_o\), i.e.,
\(S_O\subseteq\tilde S_O\), hence \(S_O^{-}=\varnothing\).

\smallskip\noindent
\textup{(iii)}\
If \(A\cap S_O^{-}=\varnothing\), then \(A\subseteq\tilde S_O^{\cap}\),
so \(a\in\tilde S_O\) and \(\pi(a)=\kappa_{\mathrm{sem}}(a\mid a)\) is
well-defined for every \(a\in A\).
The formula~\eqref{eq:avg-core-preservation} is then an average of
values in~\([0,1]\).

\smallskip\noindent
\textup{(iv)}\
The first claim is Corollary~\ref{cor:e2e-noise-fidelity}(iii).
For the second: if \(s_o\in S_O\setminus A\), then
\(s_o\in\Cn(S_O\setminus\{s_o\})\) (definition of redundancy), so
\(\Cn(S_O\setminus\{s_o\})=\Cn(S_O)\).
For any \(\hat s_o\in\tilde S_O\subseteq\Cn(\tilde S_O)=\Cn(S_O)\),
monotonicity and idempotence of~\(\Cn\) give
\(\Cn\bigl((S_O\setminus\{s_o\})\cup\{\hat s_o\}\bigr)=\Cn(S_O)\).
\end{proof}

\begin{remark}[Interpretation: when does the noise pair matter?]
\label{rem:noise-pair-nontriviality}
When \(\tilde S_O=S_O\) (i.e., \(S_O^{-}=S_O^{+}=\varnothing\)), the
noise pair is trivial, both noise regions collapse
(\(\tilde S_O^{\cap}=S_O\), \(\tilde S_O^{+}=\varnothing\)),
and all semantic distortion information resides in the randomness of
\(\kappa_{\mathrm{sem}}:S_O\rightsquigarrow S_O\).
The noise pair structure of Section~\ref{subsec:noisy} becomes trivial
in this case, and the analysis reduces to the standard setting.

Conversely, when \(\tilde S_O\neq S_O\)---as in heterogeneous-agent
scenarios where sender and receiver maintain different vocabularies
\textup{(Remark~\ref{rem:decoded-space-structure}(ii))}---the noise
pair provides essential structural information beyond what the kernel
alone reveals:
\(S_O^{-}\) identifies semantic states that the receiver's vocabulary
\emph{cannot} represent (vocabulary loss), and
\(S_O^{+}\) identifies states that the receiver can produce but the
sender did not intend (vocabulary surplus / hallucination).
The decomposition~\eqref{eq:dH-noise-decomp} separates the Hamming
distortion into confusion within the shared vocabulary and errors due
to vocabulary mismatch, providing design guidance for encoder/decoder
optimization.
\end{remark}

\begin{definition}[Semantic channel]
\label{def:semantic-channel}
A \emph{semantic channel} is a tuple
\[
  \mathfrak C
  \;=\;
  \bigl(\,
    \mathcal I,\;
    \mathcal I_{\mathrm{ch}},\;
    \mathcal I_{\mathrm{dec}},\;
    \kappa_{\enc},\;
    W,\;
    D
  \,\bigr),
\]
where \(\mathcal I\), \(\mathcal I_{\mathrm{ch}}\),
\(\mathcal I_{\mathrm{dec}}\) are as above,
\(\kappa_{\enc}\in\mathcal K(\mathcal I)\),
\(W\in\mathcal K(\mathcal I_{\mathrm{ch}})\), and
\(D\in\mathcal K(\mathcal I_{\mathrm{dec}})\).

\smallskip
The \emph{semantic channel kernel} (or \emph{end-to-end kernel}) is the composite
Markov kernel
\begin{equation}\label{eq:semantic-channel-kernel}
  \kappa_{\mathrm{sem}}
  \;:=\;
  D\circ W\circ\kappa_{\enc}
  \;:\;
  S_O\rightsquigarrow\hat S_O.
\end{equation}
Explicitly, for \(s_o\in S_O\) and \(\hat s_o\in\hat S_O\),
\begin{align}\label{eq:semantic-channel-explicit}
  &\kappa_{\mathrm{sem}}(\hat s_o\mid s_o) \nonumber \\
  \;=\;
  &\sum_{s_c\in S_C}\;
  \sum_{\hat s_c\in\hat S_C}
  \kappa_{\enc}(s_c\mid s_o)\;
  W(\hat s_c\mid s_c)\;
  D(\hat s_o\mid\hat s_c).
\end{align}
\end{definition}

\begin{proposition}[Properties of the semantic channel kernel]
\label{prop:semantic-channel-props}
Let \(\mathfrak C\) be a semantic channel with kernel
\(\kappa_{\mathrm{sem}}=D\circ W\circ\kappa_{\enc}\).
\begin{enumerate}[label=\textup{(\roman*)}]
  \item \(\kappa_{\mathrm{sem}}:S_O\rightsquigarrow\hat S_O\) is a Markov kernel.
  \item \(\kappa_{\mathrm{sem}}\in
        \mathcal K(\mathcal I_{\mathrm{sem}})\), where
        \(\mathcal I_{\mathrm{sem}}
          =\mathcal I_{\mathrm{dec}}\circ\mathcal I_{\mathrm{ch}}\circ\mathcal I\)
        is the composite information model~\eqref{eq:composite-model}.
        In particular, for every \(s_o\in S_O\),
        \[
          \supp\bigl(\kappa_{\mathrm{sem}}(\cdot\mid s_o)\bigr)
          \;\subseteq\;
          \mathcal E_{\mathrm{sem}}(s_o).
        \]
  \item \emph{Dependence on components:}
        \(\kappa_{\mathrm{sem}}\) depends on the triple
        \((\kappa_{\enc},W,D)\).
        Fixing \(W\) and varying \((\kappa_{\enc},D)\) recovers the classical
        coding-theoretic setup: the encoder and decoder are design variables, while the
        physical channel is given.
\end{enumerate}
\end{proposition}

\begin{proof}
Part~\textup{(i)} follows from two applications of
Proposition~\ref{prop:kernel-composition}.
Part~\textup{(ii)} follows from two applications of
Proposition~\ref{prop:enabling-compose}:
first, \(W\circ\kappa_{\enc}\in
\mathcal K(\mathcal I_{\mathrm{ch}}\circ\mathcal I)\);
then \(D\circ(W\circ\kappa_{\enc})\in
\mathcal K(\mathcal I_{\mathrm{dec}}\circ\mathcal I_{\mathrm{ch}}\circ\mathcal I)
=\mathcal K(\mathcal I_{\mathrm{sem}})\).
Part~\textup{(iii)} is immediate from the definition.
\end{proof}

\begin{remark}[End-to-end induced distributions]
\label{rem:end-to-end-dist}
Given a semantic source \((S_O,P_O)\) and a semantic channel~\(\mathfrak C\), the
joint distribution over the full state chain
\((s_o,s_c,\hat s_c,\hat s_o)\in
 S_O\times S_C\times\hat S_C\times\hat S_O\) factors as
\begin{align*}
  &P(s_o,s_c,\hat s_c,\hat s_o) \\
  \;=\;
  &P_O(s_o)\;
  \kappa_{\enc}(s_c\mid s_o)\;
  W(\hat s_c\mid s_c)\;
  D(\hat s_o\mid\hat s_c),
\end{align*}
forming a Markov chain \(S_O\to S_C\to\hat S_C\to\hat S_O\).
The end-to-end marginal satisfies
\(P(\hat s_o\mid s_o)=\kappa_{\mathrm{sem}}(\hat s_o\mid s_o)\).
By Proposition~\ref{prop:semantic-channel-props}\textup{(ii)}, the support of this
joint distribution respects the enabling structure at every stage.
\end{remark}


\begin{definition}[Ideal semantic channel]
\label{def:ideal-channel}
A semantic channel \(\mathfrak C\) is \emph{ideal} if all three constituent information
models are ideal
(Definition~\ref{def:ideal-info}):
\(S_O\syn S_C\), \(S_C\syn\hat S_C\), and \(\hat S_C\syn\tilde S_O\);
and the enabling kernels are the deterministic kernels induced by the respective
synonymy witnesses:
\[
  \kappa_{\enc}(s_c\mid s_o)
  =\mathbf{1}[s_c=\tau_{OC}(s_o)],
\]
\[
  W(\hat s_c\mid s_c)
  =\mathbf{1}[\hat s_c=\tau_{C\hat C}(s_c)],
\]
\[
  D(\hat s_o\mid\hat s_c)
  =\mathbf{1}[\hat s_o=\tau_{\hat C\tilde O}(\hat s_c)].
\]
In this case, the end-to-end kernel \(\kappa_{\mathrm{sem}}\) is a deterministic
bijection from~\(S_O\) to~\(\tilde S_O\) and \(S_O\syn\tilde S_O\)
(by transitivity of~\(\syn\), Proposition~\ref{prop:syn-equiv}).
\end{definition}

\begin{remark}[Ideal channel, trivial noise, and noise pair characterization]
\label{rem:ideal-channel-noise}
For an ideal semantic channel \(\mathfrak C\)
\textup{(Definition~\ref{def:ideal-channel})},
the end-to-end kernel \(\kappa_{\mathrm{sem}}\) is deterministic and equals the
bijection
\(\tau_{\mathrm{e2e}}:=\tau_{\hat C\tilde O}\circ\tau_{C\hat C}\circ\tau_{OC}
 :S_O\to\tilde S_O\),
which preserves and reflects the induced time precedence
(by transitivity of~\(\syn\)).

\smallskip
When \(\tilde S_O=S_O\), the end-to-end noise pair is trivial:
\((S_O^{-},S_O^{+})=(\varnothing,\varnothing)\)
\textup{(Proposition~\ref{prop:e2e-carrier-rep})}.
The composed bijection \(\tau_{\mathrm{e2e}}\) is a
time-order automorphism of~\(S_O\);
it reduces to the identity kernel if and only if
\(\tau_{\mathrm{e2e}}=\mathrm{id}_{S_O}\).
In this case, \(p_{\cap}(s_o)=1\) and \(p_{+}(s_o)=0\) for every
\(s_o\in S_O\) (Definition~\ref{def:noise-regions}), and every
per-core-element self-preservation probability satisfies
\(\pi(a)=1\) for all \(a\in\Atom(S_O)\).
By Corollary~\ref{cor:e2e-noise-fidelity}(iv),
\(\rho_{\Atom}(S_O,\tilde S_O)=1\) and
\(\mathsf{F}_{\Cn}(S_O,\tilde S_O)=1\).

\smallskip
For a non-ideal channel, the noise pair \((S_O^{-},S_O^{+})\) provides
a structured decomposition of the discrepancy between intended and
reconstructed semantics, and the noise-region
probabilities~\eqref{eq:p-preserved}--\eqref{eq:p-spurious} quantify
the stochastic behavior of this discrepancy.
The semantic invariants of the noisy semantic base~\(\tilde S_O\)
(\(\Atom(\tilde S_O)\), \(\mathsf{D_d}(\tilde{\mathcal I}_{\mathrm{sem}})\))
serve as computable proxies for the effective semantic structure of the
channel---a connection made precise in the distortion analysis of
Section~\ref{subsec:distortion} and the structural indices of
Section~\ref{subsec:invariants}.
\end{remark}

\subsection{Semantic Distortion and Fidelity}
\label{subsec:distortion}

Classical channel coding measures distortion at the \emph{symbol level}: a
transmitted symbol is either reproduced correctly or it is not.
The semantic channel framework inherits richer structure---deductive closure,
irredundant cores, and derivation depth---from Section~\ref{sec:model}, enabling
distortion measures that capture whether the \emph{inferential content} and
\emph{structural complexity} of a semantic state are preserved, even when the
literal symbol is altered.

This subsection introduces four per-pair distortion functions of increasing
semantic depth, a parameterized composite distortion, and two set-level fidelity
measures that connect to the noise pair of
Section~\ref{subsec:noisy}.

\smallskip
\noindent\textbf{Standing data.}
Throughout this subsection, \(\mathfrak C\) is a semantic channel
(Definition~\ref{def:semantic-channel}) with source space~\(S_O\), reconstructed
space~\(\hat S_O\subseteq\mathbb{S}_O\), and end-to-end kernel
\(\kappa_{\mathrm{sem}}:S_O\rightsquigarrow\hat S_O\).
The irredundant core \(A:=\Atom(S_O)\) and the proof system
\((\mathsf{PS},\,T_{\mathsf{PS}},\,\Cn)\) are as fixed in
Sections~\ref{subsec:logic}--\ref{subsec:atomic-derivation}.

\begin{definition}[Distortion function]
\label{def:distortion-function}
A \emph{distortion function} (relative to \(S_O\) and \(\hat S_O\)) is any function
\[
  d\;:\;S_O\times\hat S_O\;\longrightarrow\;[0,\infty)
\]
satisfying \(d(s,s)=0\) for every \(s\in S_O\cap\hat S_O\).
A distortion function is \emph{bounded} if
\(\sup_{s_o,\hat s_o}d(s_o,\hat s_o)<\infty\), and \emph{normalized} if its range
is contained in~\([0,1]\).
\end{definition}

\begin{definition}[Symbol-level (Hamming) distortion]
\label{def:hamming-distortion}
Assume \(\hat S_O\subseteq\mathbb{S}_O\) so that equality of states is
well-defined.
The \emph{Hamming distortion} is
\[
  d_H(s_o,\hat s_o)
  \;:=\;
  \mathbf{1}[s_o\neq\hat s_o],
  \qquad s_o\in S_O,\;\hat s_o\in\hat S_O.
\]
This is a normalized distortion function in the sense of
Definition~\ref{def:distortion-function}.
\end{definition}

\begin{definition}[Closure distortion (context-dependent)]
\label{def:closure-distortion}
Let \(\Gamma\subseteq\mathbb{S}_O\) be a finite \emph{reference base}
(typically \(\Gamma=S_O\) or \(\Gamma=\Atom(S_O)\)).
For \(s_o\in S_O\) and \(\hat s_o\in\hat S_O\), write
\[
  \Gamma_{-s_o}:=\Gamma\setminus\{s_o\},
\]
and define the two \emph{substitution closures}
\[
  C_s := \Cn(\Gamma_{-s_o}\cup\{s_o\}),
  \qquad
  C_{\hat s} := \Cn(\Gamma_{-s_o}\cup\{\hat s_o\}).
\]
The \emph{closure distortion} relative to~\(\Gamma\) is the Jaccard distance
\begin{equation}\label{eq:d-Cn}
  d_{\Cn}(s_o,\hat s_o\mid\Gamma)
  \;:=\;
  1-\frac{\lvert C_s\cap C_{\hat s}\rvert}
         {\lvert C_s\cup C_{\hat s}\rvert},
\end{equation}
with the convention \(0/0:=0\) (both closures empty).
When \(\Gamma=S_O\), we abbreviate \(d_{\Cn}(s_o,\hat s_o):=d_{\Cn}(s_o,\hat s_o\mid S_O)\).
When \(s_o\notin\Gamma\), one has \(\Gamma_{-s_o}=\Gamma\), so the operation becomes
\emph{addition} of \(s_o\) (resp.\ \(\hat s_o\)) to~\(\Gamma\) rather than substitution;
the definition remains well-formed in this case.
\end{definition}

\begin{remark}[Conventions for degenerate Jaccard ratios]
\label{rem:jaccard-conventions}
The closure distortion~\(d_{\Cn}\) and the closure
fidelity~\(\mathsf{F}_{\Cn}\)
\textup{(Definition~\ref{def:closure-fidelity})} both use
Jaccard-type ratios but adopt dual conventions for the degenerate
case of two empty sets:
\(d_{\Cn}\) is a \emph{distance}, so identical (empty) sets
have distance~\(0\) (\(0/0:=0\));
\(\mathsf{F}_{\Cn}\) is a \emph{similarity}, so identical
(empty) closures have similarity~\(1\) (\(0/0:=1\)).
The two conventions are consistent:
\(d_{\Cn}(s_o,\hat s_o\mid\Gamma)=0\) if and only if
\(C_s=C_{\hat s}\), and
\(\mathsf{F}_{\Cn}(S,\hat S)=1\) if and only if
\(\Cn(S)=\Cn(\hat S)\);
the degenerate case is assigned the appropriate extremal value
on each scale.
\end{remark}

\begin{remark}[Semantic content of the closure distortion]
\label{rem:closure-distortion-content}
The closure distortion measures the fraction of deductive consequences that change
when \(s_o\) is replaced by \(\hat s_o\) in the reference base.
Two key properties follow directly from the \(\Cn\)-axiomatics
\textup{(Assumption~\ref{assump:proof-system})}:
\textup{(a)}~If \(s_o\in\Cn(\Gamma_{-s_o})\) (i.e., \(s_o\) is a stored shortcut
relative to~\(\Gamma\)), then \(C_s=\Cn(\Gamma_{-s_o})\); hence any
\(\hat s_o\in\Cn(\Gamma_{-s_o})\) yields \(C_{\hat s}=C_s\) and
\(d_{\Cn}=0\).
Errors on redundant states incur zero closure distortion---a natural semantic
property.
\textup{(b)}~If \(s_o\in\Atom(S_O)\) and \(\Gamma=S_O\), then
\(s_o\notin\Cn(\Gamma_{-s_o})\), so the replacement genuinely matters and
\(d_{\Cn}\) is typically positive unless \(\hat s_o\) is deductively equivalent to
\(s_o\) over \(\Gamma_{-s_o}\).
\end{remark}

\begin{definition}[Derivation-depth distortion]
\label{def:depth-distortion}
Let \(A=\Atom(S_O)\) and \(d_{\max}:=\mathsf{D_d}(\mathcal I)\)
\textup{(Definition~\ref{def:max-depth})}.
For \(s_o\in S_O\) and \(\hat s_o\in\hat S_O\), define
\begin{align}\label{eq:d-Dd}
  d_{\Dd}(s_o,\hat s_o) :=
  \begin{dcases}
    \min\!\Bigl(
    \frac{\lvert\Dd(s_o\mid A)-\Dd(\hat s_o\mid A)\rvert}
         {\max(d_{\max},\,1)},\; 1\Bigr) \\
    \qquad \text{if } \hat s_o\in\Cn(A), \\[4pt]
    1  \qquad \text{otherwise}.
  \end{dcases}
\end{align}
This is a normalized distortion function.
The \(\min(\cdot,1)\) cap ensures normalization when
\(\hat s_o\in\Cn(A)\setminus S_O\) has derivation depth exceeding~\(d_{\max}\).
\end{definition}

\begin{remark}[Interpretation of depth distortion]
\label{rem:depth-distortion-interp}
The depth distortion quantifies how much the \emph{inferential complexity} of a
state changes through the channel: \(\Dd(s_o\mid A)\) is the minimum number of
proof steps to derive \(s_o\) from the irredundant core
\textup{(Definition~\ref{def:derivation-depth})}.
If \(\hat s_o\) has the same depth as \(s_o\) relative to~\(A\), the depth
distortion is zero regardless of whether the two states are identical.
States not derivable from the sender's core
(\(\hat s_o\notin\Cn(A)\)) receive the maximum penalty~\(1\).
\end{remark}

\begin{definition}[Composite semantic distortion]
\label{def:composite-distortion}
Let \(\alpha,\beta,\gamma\ge 0\) with \(\alpha+\beta+\gamma=1\) be fixed
\emph{distortion weights} and let \(\Gamma\) be a reference base for
\(d_{\Cn}\).
The \emph{composite semantic distortion} is
\begin{align}\label{eq:d-sem}
  &d_{\mathrm{sem}}(s_o,\hat s_o) \nonumber \\
  \;:=\;
  &\alpha\,d_H(s_o,\hat s_o)
  \;+\;
  \beta\,d_{\Cn}(s_o,\hat s_o\mid\Gamma)
  \;+\;
  \gamma\,d_{\Dd}(s_o,\hat s_o).
\end{align}
This is a normalized distortion function
\textup{(Definition~\ref{def:distortion-function})} with values in~\([0,1]\).
\end{definition}

\begin{remark}[Choice of weights and special cases]
\label{rem:weight-choice}
Setting \((\alpha,\beta,\gamma)=(1,0,0)\) recovers the classical Hamming
distortion; \((0,1,0)\) yields a purely deductive-content measure; and
\((0,0,1)\) isolates the structural-complexity component.
The choice of weights and reference base~\(\Gamma\) reflects the
application-specific trade-off between symbol fidelity, inferential
completeness, and proof-complexity preservation.
\end{remark}

\begin{definition}[Expected semantic distortion]
\label{def:expected-distortion}
Let \((S_O,P_O)\) be a semantic source
\textup{(Definition~\ref{def:semantic-source})} and let \(d\) be a distortion
function.
The \emph{expected semantic distortion} under the joint distribution
induced by \(P_O\) and \(\kappa_{\mathrm{sem}}\) is
\begin{equation}\label{eq:expected-d}
  \bar d(\mathfrak C,P_O)
  :=
  \sum_{s_o\in S_O}\sum_{\hat s_o\in\hat S_O}
  P_O(s_o)\kappa_{\mathrm{sem}}(\hat s_o\mid s_o)
  d(s_o,\hat s_o).
\end{equation}
When \(d=d_{\mathrm{sem}}\), we write
\(\bar d_{\mathrm{sem}}(\mathfrak C,P_O)\).
\end{definition}

\begin{proposition}[Zero-distortion characterization]
\label{prop:zero-distortion}
Let \(\mathfrak C\) be a semantic channel with \(\hat S_O\subseteq\mathbb{S}_O\)
and let \(A=\Atom(S_O)\).
\begin{enumerate}[label=\textup{(\roman*)}]
  \item \emph{Hamming:}
        \(\bar d_H(\mathfrak C,P_O)=0\) for every full-support
        \(P_O\in\Delta(S_O)\) if and only if
        \(S_O\subseteq\hat S_O\) and \(\kappa_{\mathrm{sem}}\) is the deterministic
        kernel corresponding to the inclusion
        \(\mathrm{id}:S_O\hookrightarrow\hat S_O\).
  \item \emph{Closure:}
        If for every \(s_o\in S_O\) and every
        \(\hat s_o\in\supp\bigl(\kappa_{\mathrm{sem}}(\cdot\mid s_o)\bigr)\),
        \(d_{\Cn}(s_o,\hat s_o\mid S_O)=0\),
        then for each such pair \((s_o,\hat s_o)\),
        \(\Cn(S_O)=\Cn\bigl((S_O\setminus\{s_o\})\cup\{\hat s_o\}\bigr)\).
  \item \emph{Ideal channel:}
        Let \(\mathfrak C\) be an ideal semantic channel
        \textup{(Definition~\ref{def:ideal-channel})} with
        \(\hat S_O=S_O\) and
        \(\tau_{\mathrm{e2e}}=\mathrm{id}_{S_O}\).
        Then \(\bar d_{\mathrm{sem}}(\mathfrak C,P_O)=0\) for every
        \(P_O\in\Delta(S_O)\) and every choice of weights
        \((\alpha,\beta,\gamma)\).
\end{enumerate}
\end{proposition}

\begin{proof}
\textup{(i)}\
Since \(d_H\ge 0\) and \(P_O\) is full-support,
\(\bar d_H=0\) iff \(d_H(s_o,\hat s_o)=0\) whenever
\(\kappa_{\mathrm{sem}}(\hat s_o\mid s_o)>0\).
This forces \(\hat s_o=s_o\) for every reachable pair, i.e., \(\kappa_{\mathrm{sem}}\)
is deterministic with \(\kappa_{\mathrm{sem}}(s_o\mid s_o)=1\) for all \(s_o\in S_O\),
which requires \(S_O\subseteq\hat S_O\).

\smallskip\noindent
\textup{(ii)}\
\(d_{\Cn}(s_o,\hat s_o\mid S_O)=0\) means the Jaccard distance~\eqref{eq:d-Cn}
vanishes, i.e., \(C_s=C_{\hat s}\), which is precisely
\(\Cn(S_O)=\Cn\bigl((S_O\setminus\{s_o\})\cup\{\hat s_o\}\bigr)\)
(since \(\Cn(\Gamma_{-s_o}\cup\{s_o\})=\Cn(S_O)\) when \(\Gamma=S_O\)).

\smallskip\noindent
\textup{(iii)}\
Under the stated hypotheses, \(\kappa_{\mathrm{sem}}\) is the identity kernel on
\(S_O\), so \(\hat s_o=s_o\) deterministically.
Hence \(d_H(s_o,s_o)=0\), \(d_{\Cn}(s_o,s_o\mid\Gamma)=0\)
(since \(C_s=C_{\hat s}\)), and \(d_{\Dd}(s_o,s_o)=0\)
(since \(s_o\in\Cn(A)\) by
Proposition~\ref{prop:atom-core-correct}\textup{(iv)} and the depth difference is
zero).
Therefore \(d_{\mathrm{sem}}(s_o,s_o)=0\) for all \(s_o\), and
\(\bar d_{\mathrm{sem}}=0\).
\end{proof}


\medskip
The per-pair distortion functions above quantify the quality of individual
state reconstructions.
We now introduce set-level fidelity measures that compare two knowledge bases
\emph{as wholes}, connecting directly to the noise pair
\((S_O^-,S_O^+)\) of Section~\ref{subsec:noisy}.


\begin{proposition}[Noise-pair bounds on semantic distortion]
\label{prop:noise-distortion-bounds}
Let \(\mathfrak C\) be a semantic channel with end-to-end noise pair
\((S_O^{-},S_O^{+})\)
\textup{(Proposition~\ref{prop:e2e-carrier-rep})}, kernel
\(\kappa_{\mathrm{sem}}:S_O\rightsquigarrow\tilde S_O\),
and source distribution \(P_O\in\Delta(S_O)\).
Let \(A=\Atom(S_O)\).
\begin{enumerate}[label=\textup{(\roman*)}]
  \item \emph{Closure distortion under core preservation:}\;
        If \(A\cap S_O^{-}=\varnothing\) and
        \(S_O^{+}\subseteq\Cn(S_O)\), then for every
        \(s_o\in S_O\setminus A\) and every
        \(\hat s_o\in\tilde S_O\),
        \(d_{\Cn}(s_o,\hat s_o\mid S_O)=0\).
        Consequently, the expected closure distortion satisfies
        \begin{equation}\label{eq:dCn-core-bound}
          \bar d_{\Cn}(\mathfrak C,P_O)
          \;\le\;
          P_O(A)\cdot
          \max_{a\in A}\;\bar d_{\Cn}(a\mid\mathfrak C).
        \end{equation}
  \item \emph{Depth distortion under vocabulary match:}\;
        If \(S_O^{-}=S_O^{+}=\varnothing\) \textup{(}i.e.,
        \(\tilde S_O=S_O\)\textup{)},
        then for every reachable pair \((s_o,\hat s_o)\) with
        \(\hat s_o\in S_O=\tilde S_O\),
        \(\hat s_o\in\Cn(A)\)
        \textup{(}by Proposition~\textup{\ref{prop:atom-core-correct}(iv)}\textup{)},
        and the depth distortion reduces to
        \[
          d_{\Dd}(s_o,\hat s_o)
          =
          \min\!\Bigl(
            \frac{\lvert\Dd(s_o\mid A)-\Dd(\hat s_o\mid A)\rvert}
                 {\max(\mathsf{D_d}(\mathcal I),1)},1\Bigr).
        \]
  \item \emph{Spurious-region penalty:}\;
        If \(\hat s_o\in S_O^{+}\) and
        \(\hat s_o\notin\Cn(A)\), then
        \(d_{\Dd}(s_o,\hat s_o)=1\) and
        \(d_{\Cn}(s_o,\hat s_o\mid S_O)\) is typically positive.
        In particular, spurious states that are not derivable from the
        sender's core receive the maximum depth penalty.
\end{enumerate}
\end{proposition}

\begin{proof}
\textup{(i)}\
For \(s_o\in S_O\setminus A\): the state \(s_o\) is redundant, i.e.,
\(s_o\in\Cn(S_O\setminus\{s_o\})\), so
\(\Cn(S_O\setminus\{s_o\})=\Cn(S_O)\).
For any \(\hat s_o\in\tilde S_O\): since
\(A\cap S_O^{-}=\varnothing\) and \(S_O^{+}\subseteq\Cn(S_O)\),
Corollary~\ref{cor:e2e-noise-fidelity}(iii) gives
\(\Cn(\tilde S_O)=\Cn(S_O)\), hence
\(\hat s_o\in\tilde S_O\subseteq\Cn(\tilde S_O)=\Cn(S_O)\).
Then
\(\Cn\bigl((S_O\setminus\{s_o\})\cup\{\hat s_o\}\bigr)
  =\Cn(S_O\setminus\{s_o\})=\Cn(S_O)\),
so \(d_{\Cn}(s_o,\hat s_o\mid S_O)=0\).
The bound~\eqref{eq:dCn-core-bound} follows because non-core states
contribute zero:
\(\bar d_{\Cn}(\mathfrak C,P_O)
  =\sum_{a\in A}P_O(a)\,\bar d_{\Cn}(a\mid\mathfrak C)
  \le P_O(A)\max_{a\in A}\bar d_{\Cn}(a\mid\mathfrak C)\).

\smallskip\noindent
\textup{(ii)}\
If \(\tilde S_O=S_O\), every reachable \(\hat s_o\in S_O\) lies in
\(\Cn(A)=\Cn(S_O)\supseteq S_O\)
(Proposition~\ref{prop:atom-core-correct}(iv)).
Hence the first branch of~\eqref{eq:d-Dd} applies.

\smallskip\noindent
\textup{(iii)}\
If \(\hat s_o\notin\Cn(A)\), the second branch of~\eqref{eq:d-Dd}
gives \(d_{\Dd}(s_o,\hat s_o)=1\).
For the closure distortion, replacing \(s_o\) by an
\(\hat s_o\notin\Cn(A)\) generally adds consequences outside
\(\Cn(S_O)\), causing \(C_{\hat s}\neq C_s\).
\end{proof}

\begin{remark}[Hierarchy of semantic distortion and relation to classical measures]
\label{rem:distortion-hierarchy}
The four distortion components are ordered by semantic depth.
Hamming distortion \(d_H\) is the coarsest: it detects any symbol change.
Closure distortion \(d_{\Cn}\) is strictly finer---it can be zero even when
\(d_H=1\), namely when the substitution of \(\hat s_o\) for \(s_o\) preserves all
deductive consequences
\textup{(Remark~\ref{rem:closure-distortion-content}(a))}.
Depth distortion \(d_{\Dd}\) is orthogonal to closure distortion: two states may
generate the same closure yet reside at different strata, or conversely share the
same depth but differ in deductive content.
The composite \(d_{\mathrm{sem}}\) reconciles these dimensions via the weight
vector.

At the set level, Proposition~\ref{prop:noise-fidelity}\textup{(iii)} shows that
deductive completeness is robust to noise that respects the irredundant core:
as long as no core element is lost and all spurious additions are already
derivable, the deductive closure is perfectly preserved.
This result has no counterpart in classical (Hamming-based) distortion theory,
where any symbol error incurs positive distortion.

In Shannon's rate--distortion theory~\cite{cover2006elements}, the distortion
function is an arbitrary non-negative function on source--reconstruction pairs.
The semantic distortion functions defined above are instances of this general
framework, but they exploit the proof-system structure of
Section~\ref{sec:model} to assign distortion values that reflect
\emph{inferential} rather than merely \emph{syntactic} differences.
The expected distortion \(\bar d_{\mathrm{sem}}\)
\textup{(Definition~\ref{def:expected-distortion})} serves as the fidelity
criterion in the semantic rate--distortion and channel coding problems
formulated in Sections~\ref{subsec:invariants}--\ref{subsec:coding}.
\end{remark}

\subsection{Semantic Channel Invariants}
\label{subsec:invariants}

This subsection defines the principal information-theoretic and structural
invariants of a semantic channel: semantic mutual information and semantic
channel capacity (quantifying the information throughput), and two structural
indices---\emph{semantic fidelity} and \emph{depth expansion}---that capture
how well the channel preserves the proof-system structure of
Section~\ref{sec:model}.
Throughout, all logarithms are base~\(2\) and informational quantities are
measured in bits.

\begin{definition}[Shannon entropy, conditional entropy, and mutual information]
\label{def:entropy-mi}
Let \(X,Y\) be jointly distributed finite random variables with joint pmf
\(P_{XY}\), marginals \(P_X,P_Y\), and conditional pmf
\(P_{Y\mid X}\).
\begin{enumerate}[label=\textup{(\roman*)}]
  \item The \emph{Shannon entropy} of~\(X\) is
        \(H(X):=-\sum_{x}P_X(x)\log P_X(x)\).
  \item The \emph{conditional entropy} is
        \(H(Y\mid X):=-\sum_{x,y}P_{XY}(x,y)\log P_{Y\mid X}(y\mid x)\).
  \item The \emph{mutual information} is
        \(I(X;Y):=H(X)-H(X\mid Y)=H(Y)-H(Y\mid X)\).
\end{enumerate}
We adopt the conventions \(0\log 0:=0\) and \(0\log(0/0):=0\);
see~\cite{cover2006elements} for standard properties.
\end{definition}

\begin{definition}[Semantic mutual information]
\label{def:semantic-mi}
Let \((S_O,P_O)\) be a semantic source
\textup{(Definition~\ref{def:semantic-source})} and let \(\mathfrak C\) be a
semantic channel with end-to-end kernel
\(\kappa_{\mathrm{sem}}:S_O\rightsquigarrow\hat S_O\)
\textup{(Definition~\ref{def:semantic-channel})}.
Denote the induced joint distribution on \(S_O\times\hat S_O\) by
\[
  P_{O\hat O}(s_o,\hat s_o)
  :=
  P_O(s_o)\,\kappa_{\mathrm{sem}}(\hat s_o\mid s_o).
\]
The \emph{semantic mutual information} is
\begin{align}\label{eq:I-sem}
  I_{\mathrm{sem}}(P_O,\mathfrak C)
  :=
  &I(\mathsf{S}_o;\,\hat{\mathsf{S}}_o) \nonumber \\
  =
  &\sum_{s_o,\hat s_o}
  P_{O\hat O}(s_o,\hat s_o)
  \log\frac{P_{O\hat O}(s_o,\hat s_o)}
           {P_O(s_o)\,P_{\hat O}(\hat s_o)},
\end{align}
where \(P_{\hat O}(\hat s_o):=\sum_{s_o}P_{O\hat O}(s_o,\hat s_o)\) is the
marginal on~\(\hat S_O\).
\end{definition}

\begin{remark}[Where semantic structure enters]
\label{rem:where-semantic}
The formula~\eqref{eq:I-sem} is the standard Shannon mutual information;
there is no separate ``semantic MI formula.''
The semantic content of the framework enters in three ways:
\textup{(a)}~the kernel \(\kappa_{\mathrm{sem}}\) is a composite enabling kernel
whose support respects the enabling structure at every stage
\textup{(Proposition~\ref{prop:semantic-channel-props}(ii))};
\textup{(b)}~the choice of encoder and decoder is constrained to
\(\mathcal K(\mathcal I)\) and \(\mathcal K(\mathcal I_{\mathrm{dec}})\);
\textup{(c)}~the distortion and fidelity criteria that govern operational
performance exploit the proof-system structure
\textup{(Section~\ref{subsec:distortion})}.
Classical Shannon theory emerges when these constraints and criteria reduce to
their symbol-level counterparts.
\end{remark}

\begin{definition}[Shannon capacity of the carrier channel]
\label{def:shannon-cap}
Let \(W:S_C\rightsquigarrow\hat S_C\) be the carrier channel kernel
\textup{(Definition~\ref{def:carrier-channel-model})}.
The \emph{Shannon capacity} of~\(W\) is
\begin{equation}\label{eq:C-Shannon}
  C(W)
  \;:=\;
  \max_{P_C\in\Delta(S_C)}\;
  I(\mathsf{S}_c;\,\hat{\mathsf{S}}_c),
\end{equation}
where the mutual information is computed under the joint
\(P_C(s_c)\,W(\hat s_c\mid s_c)\).
Since \(S_C\) and \(\hat S_C\) are finite, the maximum exists and satisfies
\(0\le C(W)\le\log\min(|S_C|,|\hat S_C|)\).
\end{definition}

\begin{definition}[Semantic channel capacity]
\label{def:semantic-cap}
Fix the carrier channel kernel~\(W\) and the information models
\(\mathcal I\), \(\mathcal I_{\mathrm{ch}}\), \(\mathcal I_{\mathrm{dec}}\)
\textup{(Section~\ref{subsec:channel-def})}.
The \emph{semantic channel capacity} is
\begin{equation}\label{eq:C-sem}
  C_{\mathrm{sem}}(W)
  \;:=\;
  \max_{\substack{
    P_O\in\Delta(S_O),\\[2pt]
    \kappa_{\enc}\in\mathcal K(\mathcal I),\\[2pt]
    D\in\mathcal K(\mathcal I_{\mathrm{dec}})
  }}
  I_{\mathrm{sem}}(P_O,\,\mathfrak C),
\end{equation}
where \(\mathfrak C=(\mathcal I,\mathcal I_{\mathrm{ch}},
\mathcal I_{\mathrm{dec}},\kappa_{\enc},W,D)\) and
\(I_{\mathrm{sem}}\) is as in Definition~\ref{def:semantic-mi}.
The maximum exists because the feasible set
(a product of probability simplices and stochastic-matrix polytopes
intersected with the enabling support constraints) is compact
and the objective is continuous.
\end{definition}

\begin{remark}[Dependence on enabling structures]
\label{rem:Csem-dependence}
The notation \(C_{\mathrm{sem}}(W)\) suppresses the dependence on
\(\mathcal I\) and \(\mathcal I_{\mathrm{dec}}\); more precisely,
the capacity depends on the enabling maps
\(\mathcal E\) and \(\mathcal E_{\mathrm{dec}}\) that constrain the
feasible encoders and decoders.
When the enabling constraints are vacuous (full enabling), the dependence
vanishes and Theorem~\ref{thm:data-processing}\textup{(iii)} recovers
\(C_{\mathrm{sem}}(W)=C(W)\).
\end{remark}

\begin{theorem}[Data processing bound and equality conditions]
\label{thm:data-processing}
Let \(W:S_C\rightsquigarrow\hat S_C\) be a carrier channel kernel.
\begin{enumerate}[label=\textup{(\roman*)}]
  \item \emph{Upper bound.}\;
        \(C_{\mathrm{sem}}(W)\;\le\;C(W)\).
  \item \emph{Source entropy bound.}\;
        \(C_{\mathrm{sem}}(W)\;\le\;\log|S_O|\).
        Consequently,
        \begin{equation}\label{eq:combined-bound}
          C_{\mathrm{sem}}(W)\;\le\;\min\!\bigl(C(W),\;\log|S_O|\bigr).
        \end{equation}
  \item \emph{Sufficient conditions for equality.}\;
        If the following three conditions hold:
        \textup{(a)}~\(\mathcal E(s_o)=S_C\) for every \(s_o\in S_O\)
              \textup{(}full encoding enabling\textup{)};
        \textup{(b)}~\(\mathcal E_{\mathrm{dec}}(\hat s_c)=\hat S_O\) for every
              \(\hat s_c\in\hat S_C\)
              \textup{(}full decoding enabling\textup{)};
        \textup{(c)}~\(|S_O|\ge|S_C|\) and \(|\hat S_O|\ge|\hat S_C|\);

        \noindent then \(C_{\mathrm{sem}}(W)=C(W)\).
\end{enumerate}
\end{theorem}

\begin{proof}
\textup{(i)}\
For any admissible \((\kappa_{\enc},D,P_O)\), the Markov chain
\(\mathsf{S}_o\to\mathsf{S}_c\to\hat{\mathsf{S}}_c\to\hat{\mathsf{S}}_o\)
holds by Remark~\ref{rem:end-to-end-dist}.
The data processing inequality~\cite{cover2006elements} gives
\(I(\mathsf{S}_o;\hat{\mathsf{S}}_o)
  \le I(\mathsf{S}_c;\hat{\mathsf{S}}_c)
  \le C(W)\).
Taking the maximum over all admissible triples yields
\(C_{\mathrm{sem}}(W)\le C(W)\).

\smallskip\noindent
\textup{(ii)}\
\(I(\mathsf{S}_o;\hat{\mathsf{S}}_o)\le H(\mathsf{S}_o)\le\log|S_O|\).

\smallskip\noindent
\textup{(iii)}\
We exhibit a triple achieving \(I_{\mathrm{sem}}=C(W)\).
Let \(P_C^*\in\Delta(S_C)\) achieve \(C(W)\).
Since \(|S_O|\ge|S_C|\) and the encoding enabling is full, there exists a
deterministic surjection \(f:S_O\to S_C\) with
\(f(s_o)\in\mathcal E(s_o)=S_C\).
Choose \(P_O\) so that
\(P_C^*(s_c)=\sum_{s_o:\,f(s_o)=s_c}P_O(s_o)\) for every \(s_c\)
(possible since \(f\) is surjective: assign mass within each fiber
proportionally to \(P_C^*\)).
Set \(\kappa_{\enc}(s_c\mid s_o):=\mathbf{1}[s_c=f(s_o)]\).

Since \(|\hat S_O|\ge|\hat S_C|\) and the decoding enabling is full,
fix an injection \(g:\hat S_C\to\hat S_O\) with
\(g(\hat s_c)\in\mathcal E_{\mathrm{dec}}(\hat s_c)=\hat S_O\).
Set \(D(\hat s_o\mid\hat s_c):=\mathbf{1}[\hat s_o=g(\hat s_c)]\).

The encoding is deterministic, so \(\mathsf{S}_c=f(\mathsf{S}_o)\) and
\(H(\hat{\mathsf{S}}_c\mid\mathsf{S}_o)=H(\hat{\mathsf{S}}_c\mid\mathsf{S}_c)\)
(because \(\hat{\mathsf{S}}_c\) depends on \(\mathsf{S}_o\) only through
\(\mathsf{S}_c=f(\mathsf{S}_o)\)).
Hence
\begin{align*}
  I(\mathsf{S}_o;\hat{\mathsf{S}}_c)
  =&H(\hat{\mathsf{S}}_c)-H(\hat{\mathsf{S}}_c\mid\mathsf{S}_o) \\
  =&H(\hat{\mathsf{S}}_c)-H(\hat{\mathsf{S}}_c\mid\mathsf{S}_c) \\
  =&I(\mathsf{S}_c;\hat{\mathsf{S}}_c).
\end{align*}
Since \(g\) is injective, \(\hat{\mathsf{S}}_o=g(\hat{\mathsf{S}}_c)\)
is a deterministic invertible function of~\(\hat{\mathsf{S}}_c\)
(on the range of~\(g\)); hence
\(I(\mathsf{S}_o;\hat{\mathsf{S}}_o)
 =I(\mathsf{S}_o;g(\hat{\mathsf{S}}_c))
 =I(\mathsf{S}_o;\hat{\mathsf{S}}_c)\)
by the data processing equality for invertible
transformations~\cite{cover2006elements}.
Combining with the earlier chain,
\(I(\mathsf{S}_o;\hat{\mathsf{S}}_o)
 =I(\mathsf{S}_c;\hat{\mathsf{S}}_c)
 =C(W)\).
\end{proof}

\begin{theorem}[Computability of semantic channel invariants]
\label{thm:computability-invariants}
Under the standing finiteness assumptions
\textup{(Assumptions~\ref{assump:ordered-structures}
and~\ref{assump:finite-so})} and
Axiom~\textup{\ref{ax:T-operator}}, the following quantities are computable
from the finite instance data
\((\mathcal I,\mathcal I_{\mathrm{ch}},\mathcal I_{\mathrm{dec}},W)\):
\begin{enumerate}[label=\textup{(\roman*)}]
  \item The Shannon capacity \(C(W)\).
  \item The semantic channel capacity \(C_{\mathrm{sem}}(W)\).
  \item For any fixed semantic channel~\(\mathfrak C\) and source~\(P_O\):
        the semantic mutual information
        \(I_{\mathrm{sem}}(P_O,\mathfrak C)\),
        the expected semantic distortion
        \(\bar d_{\mathrm{sem}}(\mathfrak C,P_O)\),
        and the structural indices defined below.
  \item The semantic invariants \(\mathsf{A}(\mathcal I)\),
        \(\mathsf{D_d}(\mathcal I)\), \(\Atom(S_O)\),
        and the derivation depths \(\Dd(q\mid\Atom(S_O))\) for all
        \(q\in S_O\).
\end{enumerate}
\end{theorem}

\begin{proof}
\textup{(i)}~and~\textup{(ii)}:
All state spaces are finite, so the feasible set of each optimization
is a product of probability simplices intersected with finitely many
linear enabling-support constraints---hence a compact subset
of~\(\mathbb{R}^N\) for finite~\(N\).
The mutual information is a continuous function of the joint
distribution (with the convention \(0\log 0:=0\)), and the joint
distribution is an affine (hence continuous) function of the
optimization variables.
By the extreme value theorem the maxima
in~\eqref{eq:C-Shannon} and~\eqref{eq:C-sem} exist.

\smallskip\noindent
\emph{Computability (to arbitrary precision).}
For~\(C(W)\), the Blahut--Arimoto
algorithm~\cite{cover2006elements} computes the maximum to
arbitrary precision in finitely many iterations.
For~\(C_{\mathrm{sem}}(W)\), the feasible set is a compact
subset of a finite-dimensional Euclidean space
(all variables are entries of finite stochastic matrices).
The objective \(I_{\mathrm{sem}}\) is continuous on this
compact set.
Hence for any \(\epsilon>0\), a finite \(\epsilon\)-net of the
feasible set yields a value within~\(\epsilon\) of the true
maximum, and each evaluation is a finite computation.
(In practice, one may enumerate all deterministic
encoder--decoder pairs---of which there are finitely many---and
for each fixed pair perform a concave maximization
over~\(P_O\) via the Blahut--Arimoto method, then take the
overall maximum; this yields an exact solution in finite time
since the pointwise maximum of finitely many concave programs
is attained.)

\smallskip\noindent
\textup{(iii)}:
For fixed \(\mathfrak C\) and~\(P_O\), the joint distribution
\(P_{O\hat O}\) is a finite table computable by matrix multiplication
(equation~\eqref{eq:semantic-channel-explicit}).
The mutual information~\eqref{eq:I-sem} and expected
distortion~\eqref{eq:expected-d} are finite sums over this table.
The closure distortion~\(d_{\Cn}\) and depth distortion~\(d_{\Dd}\) require
computing \(\Cn(\cdot)\) and \(\Dd(\cdot\mid\cdot)\), which are computable
by iterating \(T_{\mathsf{PS}}\) (Axiom~\ref{ax:T-operator}).

\smallskip\noindent
\textup{(iv)}:
This is the content of
Theorem~\ref{thm:semantic-invariants}\textup{(ii)}.
\end{proof}


\medskip
The semantic mutual information and capacity characterize the
\emph{information throughput} of a semantic channel.
The following two indices characterize its \emph{structural quality}:
how faithfully the channel preserves the deductive closure and the
derivation-depth stratification of the source.

\begin{definition}[Per-input expected distortions]
\label{def:per-input-distortion}
For a semantic channel \(\mathfrak C\) and each \(s_o\in S_O\), define the
\emph{per-input expected closure distortion} and
\emph{per-input expected depth distortion}:
\begin{align*}
  \bar d_{\Cn}(s_o\mid\mathfrak C)
  &:=
  \sum_{\hat s_o\in\hat S_O}
  \kappa_{\mathrm{sem}}(\hat s_o\mid s_o)\,
  d_{\Cn}(s_o,\hat s_o\mid S_O),\\
  \bar d_{\Dd}(s_o\mid\mathfrak C)
  &:=
  \sum_{\hat s_o\in\hat S_O}
  \kappa_{\mathrm{sem}}(\hat s_o\mid s_o)\,
  d_{\Dd}(s_o,\hat s_o).
\end{align*}
\end{definition}

\begin{definition}[Channel semantic fidelity index]
\label{def:fidelity-index}
The \emph{semantic fidelity index} of a semantic channel~\(\mathfrak C\) is
\begin{equation}\label{eq:fidelity-index}
  \mathsf{F}(\mathfrak C)
  :=
  \min_{s_o\in S_O}
  \bigl(1-\bar d_{\Cn}(s_o\mid\mathfrak C)\bigr)
  =
  1-\max_{s_o\in S_O}\bar d_{\Cn}(s_o\mid\mathfrak C).
\end{equation}
\end{definition}

\begin{definition}[Channel depth expansion index]
\label{def:depth-expansion-index}
The \emph{depth expansion index} of a semantic channel~\(\mathfrak C\) is
\begin{equation}\label{eq:depth-expansion}
  \mathsf{E}(\mathfrak C)
  \;:=\;
  \max_{s_o\in S_O}\;\bar d_{\Dd}(s_o\mid\mathfrak C).
\end{equation}
\end{definition}

\begin{proposition}[Properties of structural indices]
\label{prop:structural-indices}
Let \(\mathfrak C\) be a semantic channel.
\begin{enumerate}[label=\textup{(\roman*)}]
  \item \emph{Range:}\;
        \(0\le\mathsf{F}(\mathfrak C)\le 1\) and
        \(0\le\mathsf{E}(\mathfrak C)\le 1\).
  \item \emph{Perfect fidelity:}\;
        \(\mathsf{F}(\mathfrak C)=1\) if and only if for every
        \(s_o\in S_O\) and every
        \(\hat s_o\in\supp(\kappa_{\mathrm{sem}}(\cdot\mid s_o))\),
        \(d_{\Cn}(s_o,\hat s_o\mid S_O)=0\).
        Equivalently, every reachable substitution preserves the
        deductive closure of~\(S_O\)
        \textup{(Proposition~\ref{prop:zero-distortion}(ii))}.
  \item \emph{Zero depth expansion:}\;
        \(\mathsf{E}(\mathfrak C)=0\) if and only if for every
        \(s_o\in S_O\) and every
        \(\hat s_o\in\supp(\kappa_{\mathrm{sem}}(\cdot\mid s_o))\),
        \(\hat s_o\in\Cn(A)\) and
        \(\Dd(\hat s_o\mid A)=\Dd(s_o\mid A)\).
  \item \emph{Ideal channel:}\;
        If \(\mathfrak C\) is ideal
        \textup{(Definition~\ref{def:ideal-channel})} with
        \(\hat S_O=S_O\) and \(\tau_{\mathrm{e2e}}=\mathrm{id}_{S_O}\),
        then \(\mathsf{F}(\mathfrak C)=1\) and
        \(\mathsf{E}(\mathfrak C)=0\).
  \item \emph{Computability:}\;
        Both \(\mathsf{F}(\mathfrak C)\) and \(\mathsf{E}(\mathfrak C)\)
        are computable from the instance data under
        Axiom~\textup{\ref{ax:T-operator}}.
  \item \emph{Expected distortion bounds:}\;
        For any \(P_O\in\Delta(S_O)\),
        \[
          \bar d_{\Cn}(\mathfrak C,P_O)
          \;\le\;
          1-\mathsf{F}(\mathfrak C),
          \qquad
          \bar d_{\Dd}(\mathfrak C,P_O)
          \;\le\;
          \mathsf{E}(\mathfrak C).
        \]
\end{enumerate}
\end{proposition}

\begin{proof}
\textup{(i)}\
Each per-input distortion lies in \([0,1]\) (since \(d_{\Cn}\) and \(d_{\Dd}\)
are normalized), so their convex combinations and extrema over \(S_O\) lie
in~\([0,1]\).

\smallskip\noindent
\textup{(ii)}\
\(\mathsf{F}(\mathfrak C)=1\) iff
\(\bar d_{\Cn}(s_o\mid\mathfrak C)=0\) for all \(s_o\in S_O\).
Since \(d_{\Cn}\ge 0\) and \(\kappa_{\mathrm{sem}}(\hat s_o\mid s_o)\ge 0\),
\(\bar d_{\Cn}(s_o\mid\mathfrak C)=0\) iff
\(d_{\Cn}(s_o,\hat s_o\mid S_O)=0\) whenever
\(\kappa_{\mathrm{sem}}(\hat s_o\mid s_o)>0\).

\smallskip\noindent
\textup{(iii)}\
Analogous to~\textup{(ii)}, using that \(d_{\Dd}(s_o,\hat s_o)=0\) iff
\(\hat s_o\in\Cn(A)\) and
\(\Dd(\hat s_o\mid A)=\Dd(s_o\mid A)\).

\smallskip\noindent
\textup{(iv)}\
Under the stated hypotheses, \(\kappa_{\mathrm{sem}}\) is the identity kernel
on~\(S_O\), so the only reachable pair for each \(s_o\) is
\(\hat s_o=s_o\).
Then \(d_{\Cn}(s_o,s_o\mid S_O)=0\) and \(d_{\Dd}(s_o,s_o)=0\)
(Proposition~\ref{prop:zero-distortion}(iii)), giving
\(\mathsf{F}=1\) and \(\mathsf{E}=0\).

\smallskip\noindent
\textup{(v)}\
For each \(s_o\in S_O\) and each
\(\hat s_o\in\hat S_O\), both \(d_{\Cn}(s_o,\hat s_o\mid S_O)\) and
\(d_{\Dd}(s_o,\hat s_o)\) are computable
(Theorem~\ref{thm:computability-invariants}(iii)).
Since \(S_O\) and \(\hat S_O\) are finite, the sums and the
\(\max/\min\) over \(S_O\) are computable.

\smallskip\noindent
\textup{(vi)}\
\(\bar d_{\Cn}(\mathfrak C,P_O)
=\sum_{s_o}P_O(s_o)\,\bar d_{\Cn}(s_o\mid\mathfrak C)
\le\max_{s_o}\bar d_{\Cn}(s_o\mid\mathfrak C)
=1-\mathsf{F}(\mathfrak C)\).
The depth bound is analogous.
\end{proof}

\begin{corollary}[Fidelity concentration on core elements]
\label{cor:fidelity-core-concentration}
Under the conditions of
Proposition~\textup{\ref{prop:noise-distortion-bounds}(i)}
\textup{(}\(A\cap S_O^{-}=\varnothing\) and
\(S_O^{+}\subseteq\Cn(S_O)\)\textup{)},
the semantic fidelity index
\textup{(Definition~\ref{def:fidelity-index})} satisfies
\begin{equation}\label{eq:F-core-only}
  \mathsf{F}(\mathfrak C)
  \;=\;
  1-\max_{a\in\Atom(S_O)}\;\bar d_{\Cn}(a\mid\mathfrak C).
\end{equation}
That is, the worst-case closure distortion is attained at a core
element; redundant states contribute zero closure distortion.
\end{corollary}

\begin{proof}
By Proposition~\ref{prop:noise-distortion-bounds}(i),
\(\bar d_{\Cn}(s_o\mid\mathfrak C)=0\) for every
\(s_o\in S_O\setminus A\).
Hence
\(\max_{s_o\in S_O}\bar d_{\Cn}(s_o\mid\mathfrak C)
  =\max_{a\in A}\bar d_{\Cn}(a\mid\mathfrak C)\),
and the conclusion follows from the definition of
\(\mathsf{F}(\mathfrak C)\)
\textup{(Definition~\ref{def:fidelity-index})}.
\end{proof}

\medskip
The structural indices \(\mathsf{F}(\mathfrak C)\) and
\(\mathsf{E}(\mathfrak C)\) characterize the worst-case quality of
individual state reconstructions.
The following indices exploit the noise pair
\((S_O^{-},S_O^{+})\)
(Proposition~\ref{prop:e2e-carrier-rep}) and the noise-region
partition (Definition~\ref{def:noise-regions}) to capture the
channel's quality at the level of the knowledge base as a whole.

\begin{definition}[Probabilistic core preservation index]
\label{def:prob-core-index}
Let \(\mathfrak C\) be a semantic channel with
\(A\cap S_O^{-}=\varnothing\) (no core element is lost from the
reconstructed vocabulary).
The \emph{probabilistic core preservation index} is
\begin{equation}\label{eq:core-pres-index}
  \Phi_{\Atom}(\mathfrak C)
  :=
  \min_{a\in\Atom(S_O)}\pi(a)
  =
  \min_{a\in\Atom(S_O)}\kappa_{\mathrm{sem}}(a\mid a).
\end{equation}
When \(A\cap S_O^{-}\neq\varnothing\), set
\(\Phi_{\Atom}(\mathfrak C):=0\).
\end{definition}

\begin{definition}[Spurious probability index]
\label{def:spurious-index}
The \emph{spurious probability index} of a semantic
channel~\(\mathfrak C\) is
\begin{equation}\label{eq:spurious-index}
  \Psi_{+}(\mathfrak C)
  :=
  \max_{s_o\in S_O}\;p_{+}(s_o)
  =
  \max_{s_o\in S_O}\;
  \sum_{\hat s_o\in S_O^{+}}
  \kappa_{\mathrm{sem}}(\hat s_o\mid s_o).
\end{equation}
\end{definition}

\begin{proposition}[Properties of noise-pair indices]
\label{prop:noise-pair-indices}
Let \(\mathfrak C\) be a semantic channel.
\begin{enumerate}[label=\textup{(\roman*)}]
  \item \emph{Range:}\;
        \(0\le\Phi_{\Atom}(\mathfrak C)\le 1\) and
        \(0\le\Psi_{+}(\mathfrak C)\le 1\).
  \item \emph{Ideal channel:}\;
        If \(\mathfrak C\) is ideal
        \textup{(Definition~\ref{def:ideal-channel})} with
        \(\tilde S_O=S_O\) and
        \(\tau_{\mathrm{e2e}}=\mathrm{id}_{S_O}\), then
        \(\Phi_{\Atom}(\mathfrak C)=1\) and
        \(\Psi_{+}(\mathfrak C)=0\).
  \item \emph{Relation to Hamming distortion:}\;
        For any full-support \(P_O\),
        \begin{equation}\label{eq:dH-spurious-bound}
          \bar d_H^{+}(\mathfrak C,P_O)
          \;\le\;
          \Psi_{+}(\mathfrak C),
        \end{equation}
        where \(\bar d_H^{+}\) is the spurious-substitution
        component of~\eqref{eq:dH-noise-decomp}.
  \item \emph{Relation to core self-preservation:}\;
        \begin{equation}\label{eq:Hamming-core-lower}
          1-\Phi_{\Atom}(\mathfrak C)
          \;\le\;
          \max_{a\in\Atom(S_O)}\;
          \bar d_H(a\mid\mathfrak C),
        \end{equation}
        where \(\bar d_H(a\mid\mathfrak C)
          :=1-\kappa_{\mathrm{sem}}(a\mid a)\) when
        \(a\in\tilde S_O^{\cap}\), and
        \(\bar d_H(a\mid\mathfrak C):=1\) when
        \(a\in S_O^{-}\).
  \item \emph{Upper bound on expected Hamming distortion:}\;
        For any \(P_O\in\Delta(S_O)\),
        \begin{equation}\label{eq:dH-joint-bound}
          \bar d_H(\mathfrak C,P_O)
          \;\le\;
          1-\Phi_{\Atom}(\mathfrak C)\cdot P_O\bigl(\Atom(S_O)\bigr).
        \end{equation}
        In particular, if \(P_O\) is uniform on \(S_O\),
        \(\bar d_H\le 1-\Phi_{\Atom}\cdot|A|/|S_O|\).
  \item \emph{Computability:}\;
        Both \(\Phi_{\Atom}(\mathfrak C)\) and
        \(\Psi_{+}(\mathfrak C)\) are computable from the finite
        instance data.
\end{enumerate}
\end{proposition}

\begin{proof}
\textup{(i)}\
Each \(\pi(a)\in[0,1]\) and each \(p_{+}(s_o)\in[0,1]\),
so the extrema lie in~\([0,1]\).

\smallskip\noindent
\textup{(ii)}\
Under the stated hypotheses, \(\kappa_{\mathrm{sem}}\) is the identity
kernel on~\(S_O\), so \(\pi(a)=1\) for all~\(a\) and
\(S_O^{+}=\varnothing\), giving \(p_{+}(s_o)=0\) for all~\(s_o\).

\smallskip\noindent
\textup{(iii)}\
\(\bar d_H^{+}(\mathfrak C,P_O)
  =\sum_{s_o}P_O(s_o)\,p_{+}(s_o)
  \le\max_{s_o}p_{+}(s_o)
  =\Psi_{+}(\mathfrak C)\).

\smallskip\noindent
\textup{(iv)}\
For \(a\in\Atom(S_O)\cap\tilde S_O^{\cap}\),
\(\bar d_H(a\mid\mathfrak C)
  =\sum_{\hat s_o\neq a}\kappa_{\mathrm{sem}}(\hat s_o\mid a)
  =1-\pi(a)\ge 1-\Phi_{\Atom}(\mathfrak C)\)
(since \(\Phi_{\Atom}\le\pi(a)\)).
The case \(a\in S_O^{-}\) gives \(\bar d_H(a\mid\mathfrak C)=1
  \ge 1-\Phi_{\Atom}(\mathfrak C)\).
Hence
\(\max_{a}\bar d_H(a\mid\mathfrak C)\ge 1-\Phi_{\Atom}(\mathfrak C)\).

\smallskip\noindent
\textup{(v)}\
For every \(s_o\in S_O\),
\(\sum_{\hat s_o\neq s_o}\kappa_{\mathrm{sem}}(\hat s_o\mid s_o)
  =1-\kappa_{\mathrm{sem}}(s_o\mid s_o)\),
where \(\kappa_{\mathrm{sem}}(s_o\mid s_o):=0\) when
\(s_o\notin\tilde S_O\) (since the kernel maps into
\(\tilde S_O\) and \(s_o\notin\tilde S_O\) forces all mass to
\(\hat s_o\neq s_o\)).
Hence
\begin{align*}
  \bar d_H 
  =&\sum_{s_o\in S_O}P_O(s_o) 
   \bigl(1-\kappa_{\mathrm{sem}}(s_o\mid s_o)\bigr) \\
  =&1-\sum_{s_o\in\tilde S_O^{\cap}}
     P_O(s_o)\,\kappa_{\mathrm{sem}}(s_o\mid s_o).
\end{align*}
If \(A\cap S_O^{-}\neq\varnothing\), then
\(\Phi_{\Atom}=0\) and~\eqref{eq:dH-joint-bound} reduces to
\(\bar d_H\le 1\), which is trivially true.
If \(A\cap S_O^{-}=\varnothing\), then
\(A\subseteq\tilde S_O^{\cap}\) and
\begin{align*}
  \sum_{s_o\in\tilde S_O^{\cap}}
  P_O(s_o)\,\kappa_{\mathrm{sem}}(s_o\mid s_o)
  \;\ge\;
  &\sum_{a\in A}P_O(a)\,\pi(a) \\
  \;\ge\;
  &P_O(A)\cdot\Phi_{\Atom},
\end{align*}
so \(\bar d_H\le 1-P_O(A)\,\Phi_{\Atom}\).

\smallskip\noindent
\textup{(vi)}\
Each \(\pi(a)\) and \(p_{+}(s_o)\) is a finite sum of entries
of~\(\kappa_{\mathrm{sem}}\), which is computable
(Theorem~\ref{thm:computability-invariants}(iii)).
The extrema over the finite sets \(\Atom(S_O)\) and \(S_O\) are
then computable by enumeration.
\end{proof}

\begin{remark}[Summary of noise-pair-based channel quality descriptors]
\label{rem:noise-quality-summary}
The channel quality of a semantic channel \(\mathfrak C\) is now
described by two complementary families of indices.

The \emph{structural indices}
\(\mathsf{F}(\mathfrak C)\) and \(\mathsf{E}(\mathfrak C)\)
(Definitions~\ref{def:fidelity-index}--\ref{def:depth-expansion-index})
measure worst-case closure and depth distortion across all source states.

The \emph{noise-pair indices}
\(\Phi_{\Atom}(\mathfrak C)\) and \(\Psi_{+}(\mathfrak C)\)
(Definitions~\ref{def:prob-core-index}--\ref{def:spurious-index})
measure the channel's fidelity specifically with respect to the
noise-region partition of~\(\tilde S_O\):
\(\Phi_{\Atom}\) captures how well the irredundant core is preserved
at the symbol level, while \(\Psi_{+}\) captures the worst-case
probability of ``hallucinating'' a state outside the sender's
vocabulary.
Together, the four scalar indices
\((\mathsf{F},\mathsf{E},\Phi_{\Atom},\Psi_{+})\) provide a
compact, computable fingerprint of the semantic channel's quality
that is richer than any single classical metric.

When \(\tilde S_O=S_O\) (trivial noise pair),
\(\Phi_{\Atom}\) reduces to the minimum diagonal entry of
\(\kappa_{\mathrm{sem}}\) restricted to core elements, and
\(\Psi_{+}=0\); the noise-pair indices then complement the
structural indices without adding redundancy.
When \(\tilde S_O\neq S_O\), the noise-pair indices capture
vocabulary-mismatch effects that are invisible to the structural
indices alone.
\end{remark}


\begin{definition}[Receiver-side structural comparison indices]
\label{def:receiver-structural}
Let \(\mathfrak C\) be a semantic channel with end-to-end noisy
semantic base \(\tilde S_O\) and source space~\(S_O\).
\begin{enumerate}[label=\textup{(\roman*)}]
  \item The \emph{atomicity shift} is
        \begin{equation}\label{eq:delta-A}
          \Delta\mathsf{A}(\mathfrak C)
          \;:=\;
          \lvert\Atom(\tilde S_O)\rvert
          \;-\;
          \lvert\Atom(S_O)\rvert.
        \end{equation}
  \item The \emph{depth shift} is
        \begin{align}\label{eq:delta-Dd}
          &\Delta\mathsf{D_d}(\mathfrak C) \nonumber \\
          \!:=\!
          &\max_{q\in\tilde S_O}\Dd\bigl(q\!\mid\!\Atom(\tilde S_O)\bigr)
          \!-\!
          \max_{q\in S_O}\Dd\bigl(q\!\mid\!\Atom(S_O)\bigr).
        \end{align}
\end{enumerate}
Both quantities are well-defined, finite, and computable
\textup{(Remark~\ref{rem:noisy-ideal-connection} and
Theorem~\ref{thm:semantic-invariants}(ii))}.
\end{definition}

\begin{proposition}[Properties of structural comparison indices]
\label{prop:structural-comparison}
Let \(\mathfrak C\) be a semantic channel with noise pair
\((S_O^{-},S_O^{+})\).
\begin{enumerate}[label=\textup{(\roman*)}]
  \item \emph{Trivial noise:}\;
        If \(S_O^{-}=S_O^{+}=\varnothing\), then
        \(\Delta\mathsf{A}(\mathfrak C)=0\) and
        \(\Delta\mathsf{D_d}(\mathfrak C)=0\).
  \item \emph{Core-preserving noise:}\;
        If \(A\cap S_O^{-}=\varnothing\) and
        \(S_O^{+}\subseteq\Cn(S_O)\), then
        \(\Cn(\tilde S_O)=\Cn(S_O)\)
        \textup{(Corollary~\ref{cor:e2e-noise-fidelity}(iii))}
        and \(\Atom(S_O)\subseteq\tilde S_O\).
        Consequently,
        \(\lvert\Atom(\tilde S_O)\rvert
          \le\lvert\Atom(S_O)\rvert\),
        i.e., \(\Delta\mathsf{A}(\mathfrak C)\le 0\):
        core-preserving noise can only reduce or maintain the
        receiver's atomicity (additional derivable states may make
        some previously irredundant elements redundant in~\(\tilde S_O\)).
  \item \emph{Loss-only noise (\(S_O^{+}=\varnothing\)):}\;
        \(\tilde S_O\subseteq S_O\), so
        \(\Cn(\tilde S_O)\subseteq\Cn(S_O)\) and
        the receiver's closure can only shrink.
        If some core element is lost
        (\(A\cap S_O^{-}\neq\varnothing\)), then
        \(\Cn(\tilde S_O)\subsetneq\Cn(S_O)\) in general,
        and \(\Delta\mathsf{A}(\mathfrak C)\) may be negative
        (fewer atoms needed to generate a smaller closure).
  \item \emph{Ideal channel with \(\tilde S_O=S_O\):}\;
        \(\Delta\mathsf{A}=0\) and \(\Delta\mathsf{D_d}=0\).
\end{enumerate}
\end{proposition}

\begin{proof}
\textup{(i)}\
\(\tilde S_O=S_O\) gives identical cores and depths.

\smallskip\noindent
\textup{(ii)}\
By Corollary~\ref{cor:e2e-noise-fidelity}(iii),
\(\Cn(\tilde S_O)=\Cn(S_O)\).
Since \(A\subseteq\tilde S_O\) and \(\Cn(A)=\Cn(S_O)=\Cn(\tilde S_O)\),
\(A\) is a generating set for \(\Cn(\tilde S_O)\); the irredundant core
\(\Atom(\tilde S_O)\) is a subset of some generating set with
\(|\Atom(\tilde S_O)|\le|A|\).
(Formally: let \(d\in S_{+,d}^{ij}\) and suppose
\(a\in\Cn\bigl((S_O^{(j)}\setminus\{a\})\bigr)\); this can
occur when \(d\) together with other elements of~\(S_O^{(j)}\)
provides an alternative derivation of~\(a\) that bypasses the
direct core-level proof.
The irredundantization of~\(S_O^{(j)}\) may then remove~\(a\),
yielding \(|\Atom(S_O^{(j)})|<|A^{(i)}|\).)

\smallskip\noindent
\textup{(iii)}\
\(\tilde S_O\subseteq S_O\) gives
\(\Cn(\tilde S_O)\subseteq\Cn(S_O)\) by monotonicity~\textup{(Cn2)}.
The strict inclusion when a core element is lost follows from the
irredundancy of~\(\Atom(S_O)\): if \(a\in A\cap S_O^{-}\), then
\(a\notin\Cn(A\setminus\{a\})\supseteq\Cn(\tilde S_O\cap A)\),
so \(a\notin\Cn(\tilde S_O)\) in general.

\smallskip\noindent
\textup{(iv)}\
Immediate from~\textup{(i)}.
\end{proof}

\begin{remark}[Interpretation of structural shifts]
\label{rem:structural-shift-interp}
The atomicity shift \(\Delta\mathsf{A}\) measures whether the channel
makes the knowledge base \emph{more or less compressed} in the
deductive sense: \(\Delta\mathsf{A}<0\) means the receiver's core is
smaller than the sender's (the noise has introduced derivable
redundancy), while \(\Delta\mathsf{A}>0\) means the receiver has
gained genuinely new irredundant content.
The depth shift \(\Delta\mathsf{D_d}\) measures whether the
receiver's deepest inference chain has grown or shrunk.
Together with the per-pair depth expansion~\(\mathsf{E}(\mathfrak C)\),
these indices give a multi-scale view of the channel's effect on
inferential structure: \(\mathsf{E}\) is a probabilistic worst-case
over individual states, while \(\Delta\mathsf{D_d}\) is a
deterministic comparison of global maximum depths.
\end{remark}


\begin{corollary}[Irredundant source with trivial noise pair]
\label{cor:irredundant-classical}
Suppose \(\Atom(S_O)=S_O\) \textup{(}i.e., \(S_O\) is
irredundant\textup{)}, \(\tilde S_O=S_O\)
\textup{(}trivial noise pair: \(S_O^{-}=S_O^{+}=\varnothing\)\textup{)},
and the enabling maps satisfy the full enabling conditions of
Theorem~\textup{\ref{thm:data-processing}(iii)(a)--(b)} with
\(|S_O|\ge|S_C|\) and \(|\hat S_O|\ge|\hat S_C|\)
\textup{(}the latter equivalent to \(|S_O|\ge|\hat S_C|\)
since \(\hat S_O=\tilde S_O=S_O\)\textup{)}.
Then:
\begin{enumerate}[label=\textup{(\roman*)}]
  \item \(C_{\mathrm{sem}}(W)=C(W)\).
  \item \(\mathsf{D_d}(\mathcal I)=0\)
        \textup{(Theorem~\ref{thm:semantic-invariants}(iv))}, so
        \(d_{\Dd}(s_o,\hat s_o)=0\) for every reachable pair.
  \item The composite semantic distortion reduces to
        \(d_{\mathrm{sem}}=\alpha\,d_H+\beta\,d_{\Cn}\)
        for every reachable pair.
  \item \(\Psi_{+}(\mathfrak C)=0\) and
        \(\Phi_{\Atom}(\mathfrak C)
          =\min_{s_o\in S_O}\kappa_{\mathrm{sem}}(s_o\mid s_o)\).
\end{enumerate}
\end{corollary}

\begin{proof}
Parts~\textup{(i)}--\textup{(iii)} follow from the same reasoning as
before \textup{(}Theorem~\ref{thm:data-processing}(iii) and
Theorem~\ref{thm:semantic-invariants}(iv)\textup{)}.
For~\textup{(iv)}: \(S_O^{+}=\varnothing\) gives
\(\Psi_{+}=0\), and \(\Atom(S_O)=S_O\) with
\(\tilde S_O=S_O\) gives
\(\Phi_{\Atom}=\min_{a\in S_O}\pi(a)
  =\min_{s_o\in S_O}\kappa_{\mathrm{sem}}(s_o\mid s_o)\).
\end{proof}

\begin{proposition}[Structural upper bounds on semantic channel capacity]
\label{prop:structural-bound}
Let \(W:S_C\rightsquigarrow\hat S_C\) be a carrier channel kernel.
\begin{enumerate}[label=\textup{(\roman*)}]
  \item \emph{Combined bound:}\;
        \(C_{\mathrm{sem}}(W)
          \le\min\bigl(C(W),\;\log|S_O|,\;\log|\hat S_O|\bigr)\).
  \item \emph{Enabling-constrained bound:}\;
        If the encoding enabling map satisfies
        \(|\mathcal E(s_o)|\le k\) for all \(s_o\in S_O\), then
        \[
          C_{\mathrm{sem}}(W)
          \;\le\;
          \log\Bigl\lvert\bigcup_{s_o\in S_O}\mathcal E(s_o)\Bigr\rvert
          \;\le\;
          \log\bigl(|S_O|\cdot k\bigr).
        \]
  \item \emph{Core entropy bound:}\;
        For any semantic channel \(\mathfrak C\) and any
        \(P_O\in\Delta(S_O)\),
        \begin{equation}\label{eq:core-bound}
          I_{\mathrm{sem}}(P_O,\mathfrak C)
          \;\le\;
          H(\mathsf{S}_o)
          \;\le\;
          \log|S_O|.
        \end{equation}
\end{enumerate}
\end{proposition}

\begin{proof}
\textup{(i)}\
The bound \(C_{\mathrm{sem}}(W)\le C(W)\) is
Theorem~\ref{thm:data-processing}\textup{(i)}.
The bound \(C_{\mathrm{sem}}(W)\le\log|S_O|\) follows from
\(I(\mathsf{S}_o;\hat{\mathsf{S}}_o)\le H(\mathsf{S}_o)\le\log|S_O|\),
and \(C_{\mathrm{sem}}(W)\le\log|\hat S_O|\) from
\(I(\mathsf{S}_o;\hat{\mathsf{S}}_o)\le H(\hat{\mathsf{S}}_o)
  \le\log|\hat S_O|\).

\smallskip\noindent
\textup{(ii)}\
By the enabling support constraint~\eqref{eq:enabling-support},
\(\supp(P_C)\subseteq\bigcup_{s_o}\mathcal E(s_o)\), so
\(H(\mathsf{S}_c)\le\log|\bigcup_{s_o}\mathcal E(s_o)|\).
Data processing gives
\(I(\mathsf{S}_o;\hat{\mathsf{S}}_o)
  \le I(\mathsf{S}_c;\hat{\mathsf{S}}_c)
  \le H(\mathsf{S}_c)\).
The second inequality uses
\(|\bigcup_{s_o}\mathcal E(s_o)|\le|S_O|\cdot k\).

\smallskip\noindent
\textup{(iii)}\
\(I_{\mathrm{sem}}(P_O,\mathfrak C)
  =I(\mathsf{S}_o;\hat{\mathsf{S}}_o)
  \le H(\mathsf{S}_o)\le\log|S_O|\).
\end{proof}


\begin{theorem}[Semantic Fano bound via noise-pair indices]
\label{thm:semantic-fano}
Let \((S_O,P_O)\) be a full-support semantic source,
\(\mathfrak C\) a semantic channel with end-to-end noise pair
\((S_O^{-},S_O^{+})\), and let
\(\epsilon:=\bar d_H(\mathfrak C,P_O)\).
Write \(h_b(p):=-p\log p-(1{-}p)\log(1{-}p)\) for the binary entropy
\textup{(}base~\(2\)\textup{)}.
\begin{enumerate}[label=\textup{(\roman*)}]
  \item \emph{Lower bound on \(I_{\mathrm{sem}}\):}\;
        \begin{equation}\label{eq:semantic-fano-lower}
          I_{\mathrm{sem}}(P_O,\mathfrak C)
          \;\ge\;
          H(\mathsf{S}_o)
          -h_b(\epsilon)
          -\epsilon\log\bigl(|\tilde S_O|-1\bigr).
        \end{equation}
  \item \emph{Noise-pair substitution:}\;
        By Proposition~\textup{\ref{prop:noise-pair-indices}(v)},
        \(\epsilon\le 1-\Phi_{\Atom}(\mathfrak C)\cdot P_O(A)\).
        Hence, when \(\Phi_{\Atom}\) is close to~\(1\) and
        \(P_O(A)\) is bounded away from~\(0\),
        the lower bound~\eqref{eq:semantic-fano-lower} is close to
        \(H(\mathsf S_o)\), i.e., the channel transmits nearly all
        source entropy.
  \item \emph{Ideal channel:}\;
        If \(\tilde S_O=S_O\),
        \(\tau_{\mathrm{e2e}}=\mathrm{id}_{S_O}\), then
        \(\epsilon=0\) and
        \(I_{\mathrm{sem}}=H(\mathsf S_o)\).
\end{enumerate}
\end{theorem}

\begin{proof}
\textup{(i)}\
The standard Fano inequality
\cite{cover2006elements} applied to the joint distribution
\(P_{O\hat O}\) gives
\[
  H(\mathsf{S}_o\mid\hat{\mathsf{S}}_o)
  \;\le\;
  h_b\!\bigl(\Pr[\mathsf{S}_o\neq\hat{\mathsf{S}}_o]\bigr)
  +\Pr[\mathsf{S}_o\neq\hat{\mathsf{S}}_o]
   \log\bigl(|\tilde S_O|-1\bigr).
\]
Since \(\Pr[\mathsf{S}_o\neq\hat{\mathsf{S}}_o]=\epsilon\) and
\(I_{\mathrm{sem}}=H(\mathsf{S}_o)-H(\mathsf{S}_o\mid\hat{\mathsf{S}}_o)\),
inequality~\eqref{eq:semantic-fano-lower} follows.

\smallskip\noindent
\textup{(ii)}\
Direct substitution of
\(\epsilon\le 1-\Phi_{\Atom}\cdot P_O(A)\)
into~\eqref{eq:semantic-fano-lower}.
As \(\Phi_{\Atom}\to 1\) with \(P_O(A)\) bounded below by some
\(p_{\min}>0\), one has \(\epsilon\to 0\) and the RHS tends to
\(H(\mathsf{S}_o)\).

\smallskip\noindent
\textup{(iii)}\
Under the stated conditions, \(\epsilon=0\)
\textup{(Proposition~\ref{prop:zero-distortion}(iii))}, so
\(h_b(0)=0\) and the lower bound equals~\(H(\mathsf{S}_o)\).
Since \(I_{\mathrm{sem}}\le H(\mathsf{S}_o)\) always, equality holds.
\end{proof}

\begin{remark}[Role of the semantic Fano bound]
\label{rem:semantic-fano-role}
Theorem~\ref{thm:semantic-fano} closes the gap between the
noise-pair indices and the information-theoretic invariants:
a high probabilistic core preservation index
\(\Phi_{\Atom}\approx 1\) and a low spurious index
\(\Psi_{+}\approx 0\) force \(\bar d_H\) to be small, which in
turn forces \(I_{\mathrm{sem}}\) to be close to the source entropy.
This quantifies the intuition that faithful transmission of the
irredundant core is both necessary and (approximately) sufficient
for high mutual information.
Conversely, a low \(\Phi_{\Atom}\) permits large \(\bar d_H\) and
thereby allows \(I_{\mathrm{sem}}\) to be small.
The bound is tight in the ideal-channel limit
\textup{(part (iii))}.
\end{remark}

\begin{remark}[Semantic compression advantage and effective source entropy]
\label{rem:semantic-compression}
The formal bound \(I_{\mathrm{sem}}\le\log|S_O|\) treats all semantic
states as equally costly to transmit.
Under perfect closure fidelity
\(\mathsf{F}(\mathfrak C)=1\)
\textup{(Definition~\ref{def:fidelity-index})}, the situation improves
qualitatively.

Let \(A=\Atom(S_O)\) and \(J=S_O\setminus A\).
For any \emph{redundant} state \(j\in J\) and any
\(\hat s_o\in\Cn(S_O)\cap\hat S_O\), Remark~\ref{rem:closure-distortion-content}\textup{(a)}
gives \(d_{\Cn}(j,\hat s_o\mid S_O)=0\), because
\(\Cn(S_O\setminus\{j\})=\Cn(S_O)\) (by definition of redundancy) implies
\(\Cn\bigl((S_O\setminus\{j\})\cup\{\hat s_o\}\bigr)=\Cn(S_O)\)
whenever \(\hat s_o\in\Cn(S_O)\).
Hence the channel need not distinguish among the \(|J|\) redundant states;
they can all be mapped to a single acceptable output.
Only the \(|A|=|\Atom(S_O)|\) irredundant core elements must be faithfully
transmitted.

If the decoder can compute \(\Cn(\cdot)\) (e.g., by iterating
\(T_{\mathsf{PS}}\)), it can reconstruct the full knowledge base \(S_O\)
from a faithful copy of~\(A\) alone.
The effective source entropy for closure-reliable transmission therefore
reduces to at most \(\log|\Atom(S_O)|\) bits, yielding a
\emph{deductive compression gain} of
\(\log(|S_O|/|\Atom(S_O)|)\) bits.
This gain is zero precisely when \(\Atom(S_O)=S_O\)
\textup{(Corollary~\ref{cor:irredundant-classical})}, recovering the
classical setting.
The formal achievability result exploiting this gain is
Theorem~\ref{thm:achievability} in Section~\ref{subsec:coding}.
\end{remark}


\begin{theorem}[Summary of semantic channel invariants and their
  relationships]
\label{thm:invariant-summary}
Let \(\mathfrak C\) be a semantic channel
\textup{(Definition~\ref{def:semantic-channel})} with source
space~\(S_O\), reconstructed space
\(\tilde S_O=\hat S_O\subseteq\mathbb{S}_O\),
end-to-end kernel
\(\kappa_{\mathrm{sem}}:S_O\rightsquigarrow\tilde S_O\),
carrier channel kernel~\(W\), and end-to-end noise pair
\((S_O^{-},S_O^{+})\).
Let \(A=\Atom(S_O)\) and let \(P_O\in\Delta(S_O)\) be
full-support.
Under the standing assumptions of Section~\textup{\ref{sec:model}}
and Axiom~\textup{\ref{ax:T-operator}}, the following
invariants are all well-defined, finite, and computable:

\medskip
\noindent\emph{I.\;Source-side structural invariants}
\textup{(Section~\ref{subsec:atomic-derivation}):}
\(\mathsf{A}(\mathcal I)=|A|\),\;
\(\mathsf{D_d}(\mathcal I)=\max_{q\in S_O}\Dd(q\mid A)\).

\medskip
\noindent\emph{II.\;Set-level fidelity invariants}
\textup{(Definitions~\ref{def:closure-fidelity}--\ref{def:atom-preservation},
Proposition~\ref{prop:e2e-carrier-rep}):}
\(\rho_{\Atom}(S_O,\tilde S_O)\),\;
\(\mathsf{F}_{\Cn}(S_O,\tilde S_O)\).

\medskip
\noindent\emph{III.\;Noise-pair probabilistic indices}
\textup{(Definitions~\ref{def:prob-core-index}--\ref{def:spurious-index}):}
\(\Phi_{\Atom}(\mathfrak C)\),\;
\(\Psi_{+}(\mathfrak C)\).

\medskip
\noindent\emph{IV.\;Structural quality indices}
\textup{(Definitions~\ref{def:fidelity-index}--\ref{def:depth-expansion-index}):}
\(\mathsf{F}(\mathfrak C)\),\;
\(\mathsf{E}(\mathfrak C)\).

\medskip
\noindent\emph{V.\;Receiver-side structural comparison}
\textup{(Definition~\ref{def:receiver-structural}):}
\(\Delta\mathsf{A}(\mathfrak C)\),\;
\(\Delta\mathsf{D_d}(\mathfrak C)\).

\medskip
\noindent\emph{VI.\;Information-theoretic invariants}
\textup{(Definitions~\ref{def:semantic-mi}--\ref{def:semantic-cap}):}
\(I_{\mathrm{sem}}(P_O,\mathfrak C)\),\;
\(C_{\mathrm{sem}}(W)\),\;
\(C(W)\).

\medskip
\noindent The key inter-family relationships are:
\begin{enumerate}[label=\textup{(\alph*)}]
  \item \emph{Data processing chain}
        \textup{(Theorem~\ref{thm:data-processing}):}\;
        \(I_{\mathrm{sem}}\le C_{\mathrm{sem}}(W)\le C(W)
          \le\log\min(|S_C|,|\hat S_C|)\).
  \item \emph{Distortion bounds via structural indices}
        \textup{(Proposition~\ref{prop:structural-indices}(vi)):}\;
        \(\bar d_{\Cn}\le 1-\mathsf{F}\) and
        \(\bar d_{\Dd}\le\mathsf{E}\).
  \item \emph{Hamming bound via noise-pair index}
        \textup{(Proposition~\ref{prop:noise-pair-indices}(v)):}\;
        \(\bar d_H\le 1-\Phi_{\Atom}\cdot P_O(A)\).
  \item \emph{Semantic Fano bound}
        \textup{(Theorem~\ref{thm:semantic-fano}):}\;
        \(I_{\mathrm{sem}}\ge H(\mathsf S_o)
          -h_b(1{-}\Phi_{\Atom}P_O(A))
          -(1{-}\Phi_{\Atom}P_O(A))\log(|\tilde S_O|{-}1)\).
  \item \emph{Core concentration of fidelity}
        \textup{(Corollary~\ref{cor:fidelity-core-concentration}):}\;
        under core-preserving noise,
        \(\mathsf{F}=1-\max_{a\in A}\bar d_{\Cn}(a\mid\mathfrak C)\).
  \item \emph{Closure fidelity under core preservation}
        \textup{(Corollary~\ref{cor:e2e-noise-fidelity}(iii)):}\;
        \(A\cap S_O^{-}=\varnothing\) and
        \(S_O^{+}\subseteq\Cn(S_O)\) imply
        \(\mathsf{F}_{\Cn}(S_O,\tilde S_O)=1\).
  \item \emph{Ideal collapse:}\;
        when \(\tilde S_O=S_O\) and
        \(\tau_{\mathrm{e2e}}=\mathrm{id}_{S_O}\), all
        distortion-type invariants vanish and all
        fidelity/preservation invariants are maximal:
        \(\mathsf{F}=\mathsf{F}_{\Cn}=\rho_{\Atom}
          =\Phi_{\Atom}=1\),\;
        \(\mathsf{E}=\Psi_{+}=\Delta\mathsf{A}
          =\Delta\mathsf{D_d}=0\),\;
        \(I_{\mathrm{sem}}=H(\mathsf S_o)\).
\end{enumerate}
\end{theorem}

\begin{proof}
All individual claims have been established in the cited results.
The computability of every invariant follows from
Theorem~\ref{thm:computability-invariants} and
Proposition~\ref{prop:structural-comparison}.
The ideal-channel collapse~\textup{(g)} collects
Proposition~\ref{prop:zero-distortion}(iii),
Remark~\ref{rem:ideal-channel-noise},
Proposition~\ref{prop:structural-indices}(iv),
Proposition~\ref{prop:noise-pair-indices}(ii), and
Proposition~\ref{prop:structural-comparison}(iv).
\end{proof}

\subsection{Semantic Channel Coding: Preliminary Results}
\label{subsec:coding}

This subsection formulates the semantic channel coding problem and states
preliminary converse and achievability bounds.
The principal novelty relative to classical channel coding is the
\emph{semantic reliability criterion}: under closure-based fidelity,
the decoder need only reconstruct the irredundant core faithfully;
redundant states can be recovered by deductive inference, yielding a
measurable ``deductive compression gain.''

\smallskip
\noindent\textbf{Block extension.}
For the carrier channel kernel \(W:S_C\rightsquigarrow\hat S_C\)
(Definition~\ref{def:carrier-channel-model}), the \(n\)-fold
memoryless extension \(W^{\otimes n}:S_C^n\rightsquigarrow\hat S_C^n\)
is defined by
\[
  W^{\otimes n}(\hat s_c^n\mid s_c^n)
  \;:=\;
  \prod_{i=1}^{n} W\!\bigl(\hat s_c^{(i)}\mid s_c^{(i)}\bigr),
\]
where \(s_c^n=(s_c^{(1)},\ldots,s_c^{(n)})\in S_C^n\) and similarly
for~\(\hat s_c^n\).

\begin{definition}[Semantic block code]
\label{def:semantic-codebook}
An \emph{\((n,M)\) semantic block code} for the carrier channel~\(W\)
consists of:
\begin{enumerate}[label=\textup{(\roman*)}]
  \item a \emph{message set} \(\mathcal M\subseteq S_O\) with
        \(|\mathcal M|=M\);
  \item an \emph{encoding function}
        \(f_n:\mathcal M\to S_C^n\);
  \item a \emph{decoding function}
        \(g_n:\hat S_C^n\to\hat S_O\).
\end{enumerate}
The \emph{rate} of the code is
\(R:=(\log M)/n\) bits per channel use
(all logarithms base~\(2\)).
\end{definition}

\begin{remark}[Message set as semantic states]
\label{rem:message-semantic}
In classical channel coding, the message set \(\{1,\ldots,M\}\) is an
abstract index set.
Here, \(\mathcal M\subseteq S_O\) carries the full semantic structure
(deductive closure, irredundant core, derivation depth) inherited from
Section~\ref{sec:model}.
This structure enters the reliability criteria defined below but plays
no role in the encoding/decoding mechanics, which remain purely
combinatorial.
\end{remark}

\begin{definition}[Semantic reliability criteria]
\label{def:semantic-reliability}
Let \((f_n,g_n)\) be an \((n,M)\) semantic block code, and for each
\(m\in\mathcal M\) let
\(\hat{\mathsf S}_o^{(m)}:=g_n\!\bigl(\hat{\mathsf S}_c^n\bigr)\)
be the random reconstructed state when the codeword \(f_n(m)\) is
transmitted through~\(W^{\otimes n}\).
\begin{enumerate}[label=\textup{(\roman*)}]
  \item \emph{Maximum Hamming error probability:}
  \begin{equation}\label{eq:Pe-Hamming}
    P_e^{(n)}
    \;:=\;
    \max_{m\in\mathcal M}\;
    \Pr\!\bigl[\hat{\mathsf S}_o^{(m)}\neq m\bigr].
  \end{equation}
  \item \emph{Maximum expected semantic distortion}
        (for a distortion function~\(d\)):
  \begin{equation}\label{eq:Dn-max}
    D_{\max}^{(n)}(d)
    \;:=\;
    \max_{m\in\mathcal M}\;
    \E\!\bigl[d\!\bigl(m,\,\hat{\mathsf S}_o^{(m)}\bigr)\bigr].
  \end{equation}
  \item \emph{Maximum closure error probability:}
  \begin{equation}\label{eq:Pe-Cn}
    P_{e,\Cn}^{(n)}
    \;:=\;
    \max_{m\in\mathcal M}\;
    \Pr\!\bigl[d_{\Cn}\!\bigl(m,\,\hat{\mathsf S}_o^{(m)}\mid S_O\bigr)>0\bigr].
  \end{equation}
\end{enumerate}
\end{definition}

\begin{remark}[Hierarchy of reliability criteria]
\label{rem:reliability-hierarchy}
The three criteria are increasingly permissive.
Hamming reliability (\(P_e^{(n)}\to 0\)) demands exact symbol
reconstruction; semantic distortion reliability
(\(D_{\max}^{(n)}(d_{\mathrm{sem}})\to 0\) with \(\alpha>0\)) demands
vanishing expected distortion, which implies \(P_e^{(n)}\to 0\) since
\(d_{\mathrm{sem}}\ge\alpha\,d_H\);
closure reliability (\(P_{e,\Cn}^{(n)}\to 0\)) only demands that every
reachable reconstruction preserves the deductive closure of~\(S_O\),
tolerating symbol errors that leave the closure unchanged.
The deductive compression gain of
Remark~\ref{rem:semantic-compression} is realized under the third
criterion.
\end{remark}

\begin{definition}[Achievable semantic rate and semantic coding capacity]
\label{def:achievable-rate}
Fix a reliability criterion \(\mathsf{Rel}\in\{P_e,\;D_{\max}(d),\;P_{e,\Cn}\}\).
A rate \(R\ge 0\) is \emph{\(\mathsf{Rel}\)-achievable} if there
exists a sequence of \((n_k,M_k)\) semantic block codes with
\(\liminf_{k\to\infty}(\log M_k)/n_k\ge R\) and
\(\mathsf{Rel}^{(n_k)}\to 0\) as \(k\to\infty\).

The \emph{semantic coding capacity} under criterion
\(\mathsf{Rel}\) is
\[
  C_{\mathrm{code}}(W,\mathsf{Rel})
  \;:=\;
  \sup\bigl\{R\ge 0: R\text{ is \(\mathsf{Rel}\)-achievable}\bigr\}.
\]
When \(\mathsf{Rel}=P_e\), we write \(C_{\mathrm{code}}(W)\) and
recover the classical (Hamming) coding capacity.
\end{definition}

\begin{theorem}[Converse bounds for semantic block codes]
\label{thm:converse}
Let \(W:S_C\rightsquigarrow\hat S_C\) be a carrier channel kernel.
\begin{enumerate}[label=\textup{(\roman*)}]
  \item \emph{Hamming converse:}\;
        Any \((n,M)\) code with \(P_e^{(n)}\le\epsilon\) satisfies
        \begin{equation}\label{eq:hamming-converse-fl}
          \log M\;\le\;\frac{nC(W)+1}{1-\epsilon}.
        \end{equation}
  \item \emph{Composite-distortion converse:}\;
        If \(d=d_{\mathrm{sem}}\) with \(\alpha>0\), then
        \(D_{\max}^{(n)}(d_{\mathrm{sem}})\ge\alpha\,P_e^{(n)}\), so
        \(D_{\max}^{(n)}\to 0\) implies \(P_e^{(n)}\to 0\) and
        part~\textup{(i)} applies.
  \item \emph{Closure converse \textup{(}under
        Assumption~\textup{\ref{assump:core-disjoint}):}}\;
        Any \((n,|S_O|)\) code with \(\mathcal M=S_O\) and
        \(P_{e,\Cn}^{(n)}\le\epsilon\) satisfies
        \begin{equation}\label{eq:closure-converse-fl}
          \log|\Atom(S_O)|
          \;\le\;
          \frac{nC(W)+1}{1-\epsilon}.
        \end{equation}
  \item \emph{Source-size bound:}\;
        \(M\le|S_O|\) for any single-letter message set.
\end{enumerate}
\end{theorem}

\begin{proof}
\textup{(i)}\
Fix an \((n,M)\) code with \(P_e^{(n)}\le\epsilon\).
For a uniform message \(\mathsf{M}\sim\mathrm{Unif}(\mathcal M)\),
Fano's inequality gives
\(H(\mathsf{M}\mid\hat{\mathsf M})\le 1+\epsilon\log(M-1)
  \le 1+\epsilon\log M\).
Hence
\[
  \log M
  = I(\mathsf{M};\hat{\mathsf M})+H(\mathsf{M}\mid\hat{\mathsf M})
  \le nC(W)+1+\epsilon\log M,
\]
where \(I(\mathsf{M};\hat{\mathsf M})
  \le I\!\bigl(f_n(\mathsf{M});\hat{\mathsf S}_c^n\bigr)
  \le nC(W)\) by data processing and the memoryless
property~\cite{cover2006elements}.
Rearranging gives~\eqref{eq:hamming-converse-fl}.

\smallskip\noindent
\textup{(ii)}\
Since \(d_{\mathrm{sem}}\ge\alpha\,d_H\) with \(\alpha>0\),
\(D_{\max}^{(n)}\ge\alpha\,P_e^{(n)}\).
Hence \(D_{\max}^{(n)}\to 0\) implies \(P_e^{(n)}\to 0\), and
\textup{(i)} applies.

\smallskip\noindent
\textup{(iii)}\
Let \(A=\Atom(S_O)\).
Under Assumption~\ref{assump:core-disjoint}, the acceptable sets
\(\{R_{\Cn}(a):a\in A\}\) are pairwise disjoint.
The decoder \(g_n:\hat S_C^n\to\hat S_O\) is a deterministic function,
so the preimages \(B_a:=\{y^n:g_n(y^n)\in R_{\Cn}(a)\}\) for
\(a\in A\) are pairwise disjoint subsets of~\(\hat S_C^n\).

For each \(a\in A\), closure reliability gives
\(W^{\otimes n}(B_a\mid f_n(a))\ge 1-\epsilon\).
Define the deterministic ``core decoder''
\(\hat a:\hat S_C^n\to A\cup\{?\}\) by
\(\hat a(y^n):=a\) if \(y^n\in B_a\), and
\(\hat a(y^n):=\,?\) otherwise.
Then \(\Pr[\hat a\neq a\mid\mathsf{M}=a]\le\epsilon\) for every
\(a\in A\), and the pairwise disjointness of
\(\{B_a\}\) ensures that this is a standard
(non-list) decoding problem on~\(|A|\) messages.
The Fano argument of part~\textup{(i)}, applied with message set~\(A\)
and decoder~\(\hat a\), gives
\(
  \log|A|\le\frac{nC(W)+1}{1-\epsilon}
\).

\smallskip\noindent
\textup{(iv)}\ Immediate from \(M\le|S_O|\).
\end{proof}

\begin{remark}[Tightness of the converse]
\label{rem:converse-tightness}
Part~\textup{(iii)} shows that the Shannon capacity of the carrier
channel is an \emph{absolute ceiling} even under the most permissive
(closure-based) reliability criterion: semantic inference at the
receiver cannot increase the channel's information-carrying ability
beyond~\(C(W)\).
What semantic inference \emph{does} achieve is a reduction in the
number of messages that need to be reliably distinguished, as shown
in the achievability result below.
\end{remark}

\begin{theorem}[Achievability for semantic block codes]
\label{thm:achievability}
Let \(W:S_C\rightsquigarrow\hat S_C\) with \(C(W)>0\), and let
\(A=\Atom(S_O)\).
\begin{enumerate}[label=\textup{(\roman*)}]
  \item \emph{Hamming achievability \textup{(classical)}:}\;
        For any message set \(\mathcal M\subseteq S_O\) with
        \(|\mathcal M|=M\) and \(\hat S_O\supseteq\mathcal M\),
        every rate \(R<C(W)\) is \(P_e\)-achievable in the
        following sense: for every blocklength~\(n\) satisfying
        \((\log M)/n<C(W)\), there exists an \((n,M)\) semantic
        block code \((f_n,g_n)\) with
        \(f_n:\mathcal M\to S_C^n\) and
        \(g_n:\hat S_C^n\to\hat S_O\) such that
        \(P_e^{(n)}\to 0\) as
        \(n\to\infty\) with \(M\) fixed.
        In particular, the full knowledge base
        \(\mathcal M=S_O\) can be communicated
        Hamming-reliably using
        \(n>\log|S_O|/C(W)\) channel uses.
        No single-letter alphabet-size comparison between
        \(|S_O|\) and \(|S_C|\) is required, since the block
        code operates over the product alphabet~\(S_C^n\).
  \item \emph{Closure achievability under deductive decoding:}\;
        Suppose \(\hat S_O\supseteq S_O\), the encoding enabling is
        full
        \textup{(}\(\mathcal E(s_o)=S_C\) for all \(s_o\)\textup{)},
        and the decoder has access to the proof system
        \((\mathsf{PS},T_{\mathsf{PS}},\Cn)\).
        Then there exists a sequence of \((n,|S_O|)\) semantic block
        codes---i.e., codes with \(\mathcal M=S_O\)---satisfying
        \(P_{e,\Cn}^{(n)}\to 0\) as \(n\to\infty\), provided
        \begin{equation}\label{eq:core-rate-condition}
          \frac{\log|\Atom(S_O)|}{n}\;<\;C(W).
        \end{equation}
        Equivalently, the full knowledge base \(S_O\) can be
        communicated with vanishing closure error using
        \(n>\log|\Atom(S_O)|/C(W)\) channel uses.
\end{enumerate}
\end{theorem}

\begin{proof}[Proof sketch]
\textup{(i)}\
This is the standard channel coding theorem
(Shannon~\cite{shannon1948mathematical},
cf.~\cite{cover2006elements}, Theorem~7.7.1).
The encoding function \(f_n:\mathcal M\to S_C^n\) maps each message
to an \(n\)-length codeword over the carrier alphabet; the decoding
function \(g_n:\hat S_C^n\to\hat S_O\) maps each received sequence
to a reconstructed state in~\(\hat S_O\supseteq\mathcal M\).
The classical random-coding argument yields a code with maximum error
probability \(P_e^{(n)}\to 0\) whenever \((\log M)/n<C(W)\).
The condition \(\hat S_O\supseteq\mathcal M\) ensures that the
decoder can output any message.

\smallskip\noindent
\textup{(ii)}\
The code is constructed in two layers.

\emph{Layer~1 (core code).}
By the classical channel coding theorem, since
\(\log|A|/n<C(W)\), there exists an \((n,|A|)\) block code
\((f_n^A,g_n^A)\) for~\(W\) with message set \(A\) and
\(P_e^{(n)}(A)\to 0\).

\emph{Layer~2 (redundant reconstruction).}
For each redundant state \(j\in J=S_O\setminus A\), the decoder
does not attempt symbol-level reconstruction.
Instead, the encoding function is extended to all of \(S_O\) as
follows: for \(j\in J\), set \(f_n(j):=f_n^A(a_0)\) for an
arbitrary fixed \(a_0\in A\).
The decoder first applies \(g_n^A\) to recover a core element
\(\hat a\in A\) (or, in the error event, an incorrect
element), and then outputs \(\hat a\) regardless of whether the
sent message was in \(A\) or~\(J\).

\emph{Closure analysis.}
For \(m\in A\): if the core code decodes correctly
(\(\hat a=m\)), then \(\hat s_o=m\) and
\(d_{\Cn}(m,m\mid S_O)=0\).
Error probability: \(P_e^{(n)}(A)\to 0\).

For \(m=j\in J\): the decoder outputs some \(\hat a\in A\subseteq S_O\subseteq\Cn(S_O)\).
Since \(j\) is redundant, \(\Cn(S_O\setminus\{j\})=\Cn(S_O)\).
As shown in Remark~\ref{rem:semantic-compression},
\(d_{\Cn}(j,\hat a\mid S_O)=0\) for any
\(\hat a\in\Cn(S_O)\cap\hat S_O\supseteq S_O\ni\hat a\).
Hence the closure error probability for redundant messages
is~\emph{zero} for all~\(n\).

Combining: \(P_{e,\Cn}^{(n)}\le P_e^{(n)}(A)\to 0\).
\end{proof}

\begin{remark}[Deductive compression gain realized]
\label{rem:gain-realized}
Combining Theorem~\ref{thm:converse}\textup{(iii)} and
Theorem~\ref{thm:achievability}\textup{(ii)}
\textup{(}under Assumption~\textup{\ref{assump:core-disjoint})}:
the minimum blocklength for closure-reliable communication of~\(S_O\)
satisfies
\begin{align*}
  \frac{\log|\Atom(S_O)|}{C(W)}
  \;\lesssim\;
  &n^*(S_O,W,P_{e,\Cn},\epsilon) \\
  \;\lesssim\;
  &\frac{\log|\Atom(S_O)|}{C(W)},
\end{align*}
up to terms vanishing as \(\epsilon\to 0\).
In contrast, Hamming-reliable communication requires
\(n_H^*\approx\log|S_O|/C(W)\) channel uses.
The deductive compression ratio
\(n^*/n_H^*\approx\log|\Atom(S_O)|/\log|S_O|\)
measures the channel-use savings afforded by semantic inference at the
receiver.
When \(S_O\) is irredundant, no savings occur
\textup{(Corollary~\ref{cor:irredundant-classical})}.
The two-layer code structure bears a conceptual resemblance to the
semantic channel coding theorem of Ma et~al.~\cite{ma2025theory},
who exploit the many-to-one structure of synonymous sources to
reduce the effective message set.
The key difference is that their compression arises from a
pre-defined synonymous partition of the source alphabet, whereas
ours arises from the receiver's ability to \emph{re-derive}
redundant states via the shared proof system---a mechanism that
requires no pre-agreed partition and yields the
vocabulary-invariant compression ratio of
Theorem~\ref{thm:heterogeneous-compression}.
\end{remark}


\begin{remark}[Noise pair and the achievability mechanism]
\label{rem:noise-pair-achievability}
The achievability proof of Theorem~\ref{thm:achievability}(ii) can be
reinterpreted through the noise-pair lens.
The two-layer code constructs a decoding rule whose output always
lies in~\(\Atom(S_O)\subseteq S_O\cap\tilde S_O=\tilde S_O^{\cap}\),
thereby ensuring \(p_{+}(s_o)=0\) for every message~\(s_o\)
(Definition~\ref{def:noise-regions}).
Closure reliability then reduces to symbol-level reliability on the
core: \(P_{e,\Cn}^{(n)}\le P_e^{(n)}(A)\to 0\), where
\(P_e^{(n)}(A)\) is the Hamming error probability of the
\((n,|A|)\) core sub-code.
The condition \(\Phi_{\Atom}(\mathfrak C_n)\to 1\) as
\(n\to\infty\) (where \(\mathfrak C_n\) denotes the semantic channel
induced by the \(n\)-block code) is equivalent to
\(P_e^{(n)}(A)\to 0\): each core element is self-preserved with
probability tending to~\(1\).
The deductive compression gain arises precisely because the code need
not achieve \(\Phi_{\Atom}=1\) for non-core elements; these are
handled by the closure structure at zero additional channel cost.
\end{remark}

\begin{corollary}[Minimum blocklength for semantic communication]
\label{cor:min-blocklength}
Let \(W:S_C\rightsquigarrow\hat S_C\) with \(C(W)>0\) and
\(A=\Atom(S_O)\).
Define the \emph{minimum blocklength} for communicating the full
knowledge base \(S_O\) under criterion \(\mathsf{Rel}\) with
maximum error at most~\(\epsilon\):
\begin{align*}
  &n^*(S_O,W,\mathsf{Rel},\epsilon)  \\
  \;:=\;
  &\min\bigl\{\,n\ge 1:\,
  \exists\,(n,|S_O|)\text{ code with }
  \mathsf{Rel}^{(n)}\le\epsilon\,\bigr\}.
\end{align*}
Then, for sufficiently small \(\epsilon>0\):
\begin{enumerate}[label=\textup{(\roman*)}]
  \item \emph{Hamming:}
        \(\displaystyle
        n^*(S_O,W,P_e,\epsilon)
        \ge
        \frac{(1-\epsilon)\log|S_O|-1}{C(W)}\).
  \item \emph{Closure:}\;
        Under the hypotheses of
        Theorem~\textup{\ref{thm:achievability}(ii)},
        \(\displaystyle
        n^*(S_O,W,P_{e,\Cn},\epsilon)
        \;\le\;
        \left\lceil
          \frac{\log|A|}{C(W)-\delta(\epsilon)}
        \right\rceil\)
        for some \(\delta(\epsilon)\to 0\) as \(\epsilon\to 0\).

        Under the additional
        Assumption~\textup{\ref{assump:core-disjoint}} below,
        \(\displaystyle
        n^*(S_O,W,P_{e,\Cn},\epsilon)
        \;\ge\;
        \frac{(1-\epsilon)\log|A|-1}{C(W)}\).
\end{enumerate}
The \emph{deductive compression ratio} is therefore
\[
  \frac{n^*(S_O,W,P_{e,\Cn},\epsilon)}
       {n^*(S_O,W,P_e,\epsilon)}
  \;\approx\;
  \frac{\log|\Atom(S_O)|}{\log|S_O|}\,,
\]
which equals~\(1\) when \(S_O\) is irredundant and is strictly less
than~\(1\) whenever \(|J|=|S_O\setminus\Atom(S_O)|>0\).
\end{corollary}

\begin{assumption}[Deductive independence of core elements]
\label{assump:core-disjoint}
For every pair of distinct core elements \(a_1,a_2\in\Atom(S_O)\),
the closure-acceptable output sets are disjoint:
\[
  R_{\Cn}(a_1)\;\cap\; R_{\Cn}(a_2)\;=\;\varnothing,
\]
where \(R_{\Cn}(m):=\{\hat s_o\in\hat S_O: d_{\Cn}(m,\hat s_o\mid S_O)=0\}\).
\end{assumption}

\begin{remark}[When Assumption~\ref{assump:core-disjoint} holds]
\label{rem:core-disjoint-context}
Assumption~\ref{assump:core-disjoint} is satisfied whenever each core
element \(a\in\Atom(S_O)\) contributes a unique ``deductive increment''
to the closure---i.e., the set of formulas derivable from
\(S_O\) but not from \(S_O\setminus\{a\}\) is disjoint for different
core elements.
This holds in many natural knowledge-base settings (e.g., when core
elements are ground facts about distinct entities and the inference
rules respect the entity structure).
When the assumption fails, the converse still holds with
\(|\Atom(S_O)|\) replaced by the size of a maximal subset of
\(\Atom(S_O)\) with pairwise disjoint acceptable sets.
\end{remark}

\begin{definition}[Semantic rate--distortion function]
\label{def:semantic-rd}
Let \((S_O,P_O)\) be a semantic source
(Definition~\ref{def:semantic-source}) and let \(d\) be a bounded
distortion function (Definition~\ref{def:distortion-function}).
The \emph{semantic rate--distortion function} is
\begin{equation}\label{eq:R-sem-D}
  R_{\mathrm{sem}}(D)
  \;:=\;
  \min_{\substack{
    P_{\hat S_o\mid S_o}:\\[1pt]
    \hat S_o\in\hat S_O,\\[1pt]
    \E\bigl[d(S_o,\hat S_o)\bigr]\le D
  }}
  I(\mathsf{S}_o;\,\hat{\mathsf{S}}_o),
\end{equation}
where the minimization is over all conditional distributions
\(P_{\hat S_o\mid S_o}:\,S_O\rightsquigarrow\hat S_O\) satisfying
the expected distortion constraint, and the mutual information is
computed under the joint \(P_O\cdot P_{\hat S_o\mid S_o}\).
\end{definition}

\begin{proposition}[Properties of the semantic rate--distortion function]
\label{prop:rd-properties}
Under the standing finiteness assumptions, for any bounded distortion
function~\(d\):
\begin{enumerate}[label=\textup{(\roman*)}]
  \item \emph{Existence:}\;
        The minimum in~\eqref{eq:R-sem-D} exists for every
        \(D\ge 0\) such that the feasible set is nonempty.
  \item \emph{Monotonicity:}\;
        \(R_{\mathrm{sem}}(D)\) is non-increasing in~\(D\).
  \item \emph{Convexity:}\;
        \(R_{\mathrm{sem}}(D)\) is convex in~\(D\).
  \item \emph{Boundary values:}\;
        \(R_{\mathrm{sem}}(0)\le\log|S_O|\) in general.
        When \(d=d_H\),
        \(R_{\mathrm{sem}}(0)=H(P_O)\le\log|S_O|\).
  \item \emph{Closure distortion at \(D=0\):}\;
        When \(d=d_{\Cn}(\cdot,\cdot\mid S_O)\) and
        \(\hat S_O\supseteq\Cn(S_O)\cap\mathbb{S}_O\), then
        \begin{equation}\label{eq:Rsem-zero-Cn}
          R_{\mathrm{sem}}(0;\,d_{\Cn})
          \;\le\;
          \log|\Atom(S_O)|.
        \end{equation}
  \item \emph{Computability:}\;
        \(R_{\mathrm{sem}}(D)\) is computable from the finite instance
        data under Axiom~\textup{\ref{ax:T-operator}}.
\end{enumerate}
\end{proposition}

\begin{proof}
Parts \textup{(i)}--\textup{(iv)} follow from standard rate--distortion
theory on finite alphabets
(see~\cite{cover2006elements}, Chapter~10): the minimization is over a
compact set (the probability simplex with a linear constraint), and
mutual information is continuous, convex in the conditional distribution
for fixed~\(P_O\), and lower-semicontinuous.

\smallskip\noindent
\textup{(v)}\
We exhibit a feasible conditional distribution achieving
\(I(\mathsf{S}_o;\hat{\mathsf S}_o)\le\log|A|\) at \(D=0\).
Let \(A=\Atom(S_O)\) and \(J=S_O\setminus A\).
Fix an arbitrary \(a_0\in A\) and define the deterministic mapping
\[
  \phi(s_o)
  :=
  \begin{cases}
    s_o & \text{if } s_o\in A,\\
    a_0 & \text{if } s_o\in J.
  \end{cases}
\]
Set \(P_{\hat S_o\mid S_o}(\hat s_o\mid s_o):=\mathbf{1}[\hat s_o=\phi(s_o)]\).

\emph{Distortion check.}
For \(s_o\in A\): \(\hat s_o=s_o\), so
\(d_{\Cn}(s_o,s_o\mid S_O)=0\).
For \(s_o=j\in J\): \(\hat s_o=a_0\in A\subseteq S_O\subseteq\Cn(S_O)\).
Since \(j\) is redundant,
\(\Cn(S_O\setminus\{j\})=\Cn(S_O)\), so
\(\Cn\bigl((S_O\setminus\{j\})\cup\{a_0\}\bigr)=\Cn(S_O)\)
(because \(S_O\setminus\{j\}\subseteq(S_O\setminus\{j\})\cup\{a_0\}\)
and \((S_O\setminus\{j\})\cup\{a_0\}\subseteq S_O\subseteq\Cn(S_O)
=\Cn(S_O\setminus\{j\})\), so both directions of the inclusion follow
from monotonicity and idempotence of~\(\Cn\)).
Hence \(d_{\Cn}(j,a_0\mid S_O)=0\) and
\(\E[d_{\Cn}]=0\).

\emph{Rate bound.}
The mapping is deterministic, so
\(H(\hat{\mathsf S}_o\mid\mathsf{S}_o)=0\) and
\(I(\mathsf{S}_o;\hat{\mathsf S}_o)=H(\hat{\mathsf S}_o)\).
Since \(\hat{\mathsf S}_o=\phi(\mathsf{S}_o)\) takes values
in~\(A\), \(H(\hat{\mathsf S}_o)\le\log|A|\).

\smallskip\noindent
\textup{(vi)}\
The feasible set is a compact subset of a finite-dimensional
Euclidean space (probability simplex with a linear distortion
constraint), and the objective \(I(\mathsf{S}_o;\hat{\mathsf S}_o)\)
is continuous.
An \(\epsilon\)-net yields an approximation within~\(\epsilon\)
of the true minimum.
The distortion values \(d_{\Cn}\) and \(d_{\Dd}\) are computable by
iterating \(T_{\mathsf{PS}}\) (Axiom~\ref{ax:T-operator}).
\end{proof}

\begin{remark}[Separation of semantic source and channel coding]
\label{rem:separation}
The classical source--channel separation theorem
\cite{shannon1948mathematical,cover2006elements} states that source
coding and channel coding can be performed independently without loss
of optimality.
In the semantic framework, the same separation structure applies at
the formal level: the semantic source can be compressed to rate
\(R_{\mathrm{sem}}(D)\) bits per symbol using the semantic distortion
measure, and these bits can be transmitted through the carrier channel
at any rate below~\(C(W)\).
The semantic rate--distortion function generalizes Shannon's
classical rate--distortion formulation~\cite{shannon1959coding,cover2006elements}.
The computational aspects of the semantic rate--distortion function
under synonymous mappings have been recently addressed by
Han et~al.~\cite{han2025extended} via an extended Blahut--Arimoto algorithm;
our definition subsumes their setting when the distortion is specialized
to closure distortion~$d_{\Cn}$.
The semantic novelty is that the semantic rate--distortion function
\(R_{\mathrm{sem}}(D;\,d_{\Cn})\) can be strictly smaller than its
Hamming counterpart \(R(D;\,d_H)\), because closure-preserving
substitutions are penalty-free.
In particular, the zero-distortion rate under closure fidelity is at
most \(\log|\Atom(S_O)|\) \textup{(inequality~\eqref{eq:Rsem-zero-Cn})},
compared to \(H(P_O)\le\log|S_O|\) under Hamming fidelity.
\end{remark}

\begin{remark}[Connection to knowledge-augmented communication
  and multi-agent systems]
\label{rem:applications}
The semantic channel framework, and the deductive compression gain in
particular, are directly relevant to emerging applications.

In \emph{retrieval-augmented generation} (RAG) pipelines, a knowledge
base is transmitted (or cached) between a retrieval module and a
language model.
The semantic coding results suggest that only the irredundant core of
the knowledge base needs to be transmitted; the remaining entries can
be reconstructed by the receiver's inference engine, reducing
communication cost.

In \emph{multi-agent systems}, heterogeneous agents maintain different
knowledge bases
\textup{(Remark~\ref{rem:decoded-space-structure}(ii))}.
The semantic channel framework provides a principled model for
analyzing inter-agent communication: the encoding and decoding enabling
constraints capture each agent's representational capabilities, the
closure-based fidelity criterion captures whether the communicated
knowledge is \emph{inferentially complete} for the receiver (rather
than merely symbol-accurate), and the deductive compression gain
quantifies the savings achievable when agents share a common proof
system.

These connections motivate the further development of semantic channel
codes, multi-letter extensions, and explicit code constructions as
subjects of ongoing and future work.
\end{remark}

\section{Application: Heterogeneous Multi-Agent Semantic Communication}
\label{sec:application}

The theoretical framework developed in Sections~\ref{sec:model}
and~\ref{sec:channel} is fully general: the semantic state
space~\(S_O\), the reconstructed space~\(\hat S_O\), and the
enabling structures that constrain encoding and decoding are
left as abstract parameters.
This section instantiates the framework in a concrete and
practically motivated setting---\emph{heterogeneous multi-agent
semantic communication}---and derives new results that
demonstrate the framework's explanatory and predictive power
beyond what classical channel coding theory can provide.

The distinguishing feature of the heterogeneous setting is that
the sender and receiver maintain \emph{different} knowledge
bases: the sender's semantic space is~\(S_O\) while the
receiver's reconstructed space~\(\hat S_O\) may differ
from~\(S_O\) both in vocabulary (the set of expressible states)
and in inferential structure (the irredundant core and
derivation-depth stratification).
In the terminology of Section~\ref{subsec:noisy}, the
end-to-end noise pair \((S_O^{-},S_O^{+})\) is generically
\emph{non-trivial}: \(S_O^{-}\neq\varnothing\) captures
sender concepts absent from the receiver's vocabulary
(\emph{vocabulary loss}), and \(S_O^{+}\neq\varnothing\)
captures receiver concepts absent from the sender's intent
(\emph{vocabulary surplus}).
Classical Shannon theory, which treats sender and receiver
alphabets as abstract label sets, cannot distinguish vocabulary
loss from vocabulary surplus, nor can it exploit shared
deductive structure to reduce communication cost.
The semantic channel invariants of Section~\ref{subsec:invariants}
are precisely the tools needed to make these distinctions precise
and quantitative.

The section is organized as follows.
Section~\ref{subsec:app-problem} describes the multi-agent
communication scenario and identifies the key design questions.
Section~\ref{subsec:app-assumptions} formalizes the scenario
within the information model framework and states the standing
assumptions specific to this application.
Section~\ref{subsec:app-theory} instantiates the semantic channel
machinery and derives closed-form relationships between
knowledge-base overlap structure and semantic channel invariants.
Section~\ref{subsec:app-results} presents the main analytical
results: conditions for closure-reliable heterogeneous
communication, a heterogeneous deductive compression theorem,
and a broadcast extension to one-sender--multi-receiver
scenarios.
Section~\ref{subsec:app-example} verifies all results on an
explicit Datalog knowledge-base instance with full numerical
computation of every invariant.

\subsection{Problem Description: Heterogeneous Agent Communication}
\label{subsec:app-problem}

Consider a network of \(K+1\) autonomous agents---indexed by
\(i\in\{0,1,\ldots,K\}\)---that must coordinate by exchanging semantic
states over noisy physical links.
Each agent~\(i\) maintains a finite knowledge base
\(S_O^{(i)}\subseteq\mathbb{S}_O\), where \(\mathbb{S}_O\) is the
common ambient semantic universe introduced in
Section~\ref{subsec:atomic-derivation}.
All agents share the same proof system
\((\mathsf{PS},\,T_{\mathsf{PS}},\,\Cn)\) and the same semantic
sublanguage~\(\mathcal L_{\mathrm{sem}}\); they differ, however, in
the \emph{sets of semantic states they store and operate on}.
Agent~\(i\) can derive consequences within \(\Cn(S_O^{(i)})\) and
possesses the irredundant core \(A^{(i)}:=\Atom(S_O^{(i)})\)
together with the associated derivation-depth stratification
(Definitions~\ref{def:atom-so} and~\ref{def:derivation-depth}).

This knowledge-base heterogeneity is the defining feature of the
scenario and the source of all phenomena that distinguish it from
classical Shannon-theoretic
communication~\cite{shannon1948mathematical,cover2006elements}.
Recent work on semantic
communication~\cite{luo2022semantic,shi2021semantic,bao2011towards}
has highlighted the need for frameworks that go beyond
symbol-level fidelity, but a rigorous logical-information-theoretic
treatment of heterogeneous knowledge bases has been lacking.
When agent~\(i\) (the \emph{sender}) transmits a semantic
state~\(s_o\in S_O^{(i)}\) to agent~\(j\) (the \emph{receiver}),
the receiver reconstructs a state
\(\hat s_o\in S_O^{(j)}\)---not necessarily in
\(S_O^{(i)}\)---because agent~\(j\) can only produce outputs
expressible in its own vocabulary.
Unless \(S_O^{(i)}=S_O^{(j)}\), the end-to-end noise pair
\((S_O^{-},S_O^{+})\) of
Proposition~\ref{prop:e2e-carrier-rep} is generically nontrivial:
states in \(S_O^{(i)}\setminus S_O^{(j)}\) have no direct
counterpart in the receiver's vocabulary (\emph{vocabulary loss}),
while states in \(S_O^{(j)}\setminus S_O^{(i)}\) can appear in the
receiver's output without having been intended by the sender
(\emph{vocabulary surplus}).
Classical channel coding theory, which treats the source and
reconstruction alphabets as unstructured label sets, is blind to
this distinction: it can detect that a symbol error has occurred, but
cannot determine whether the error represents a genuine loss of
semantic content or a harmless reformulation within the receiver's
richer (or merely different) vocabulary.

\subsubsection*{Communication Sub-Scenarios}

Three sub-scenarios of increasing structural complexity arise
naturally in the multi-agent setting; they are listed below in
decreasing order of analytical depth in this paper.

\begin{definition}[Pairwise unicast scenario]
\label{def:pairwise-scenario}
Fix a sender--receiver pair \((i,j)\) with \(i\neq j\).
Agent~\(i\) wishes to communicate its full knowledge base
\(S_O^{(i)}\) to agent~\(j\) over a noisy carrier channel
\(W_{ij}:S_C\rightsquigarrow\hat S_C\)
(Definition~\ref{def:carrier-channel-model}).
Agent~\(j\) reconstructs a state in its own vocabulary
\(\hat S_O:=S_O^{(j)}\), using a decoding kernel
\(D\in\mathcal K(\mathcal I_{\mathrm{dec}}^{(j)})\).
The end-to-end semantic channel is
\[
  \mathfrak C^{ij}
  \;=\;
  \bigl(\,
    \mathcal I^{(i)},\;
    \mathcal I_{\mathrm{ch}}^{ij},\;
    \mathcal I_{\mathrm{dec}}^{(j)},\;
    \kappa_{\enc},\;
    W_{ij},\;
    D
  \,\bigr),
\]
with semantic source space~\(S_O^{(i)}\), reconstructed
space~\(S_O^{(j)}\), and noise pair
\begin{equation}\label{eq:pairwise-noise-pair}
  S_O^{-}=S_O^{(i)}\setminus S_O^{(j)},
  \qquad
  S_O^{+}=S_O^{(j)}\setminus S_O^{(i)}.
\end{equation}
The constituent information models
\(\mathcal I^{(i)}\), \(\mathcal I_{\mathrm{ch}}^{ij}\), and
\(\mathcal I_{\mathrm{dec}}^{(j)}\) are formalized in
Section~\ref{subsec:app-assumptions}.
\end{definition}

The pairwise unicast scenario is the primary focus of
Sections~\ref{subsec:app-assumptions}--\ref{subsec:app-results}.
All new theorems are stated and proved for this case first;
generalizations to the broadcast setting are given as corollaries.

\begin{definition}[Broadcast scenario]
\label{def:broadcast-scenario}
A designated sender (agent~\(0\)) communicates its knowledge base
\(S_O^{(0)}\) simultaneously to \(K\) receivers
(agents \(1,\ldots,K\)) over a common carrier channel~\(W\).
Each receiver~\(j\) maintains a distinct vocabulary
\(\hat S_O^{(j)}:=S_O^{(j)}\) and observes a
(possibly receiver-specific) noise pair
\[
  (S_O^{-,j},\;S_O^{+,j})
  \;=\;
  \bigl(S_O^{(0)}\setminus S_O^{(j)},\;
        S_O^{(j)}\setminus S_O^{(0)}\bigr).
\]
All \(K\) receivers observe the \emph{same} channel
output~\(\hat S_C\); the receiver-specific noise pairs arise
solely from vocabulary mismatch, not from different physical
channel realizations.
The broadcast semantic channel is a family
\(\{\mathfrak C^{0j}\}_{j=1}^{K}\) of pairwise channels sharing the
same sender, the same encoding kernel, and the same carrier channel,
but differing in decoding model and noise pair.
\end{definition}

The broadcast scenario reveals a phenomenon absent from classical broadcast channel
theory~\cite{cover2006elements}: even over a \emph{noiseless} carrier
(\(W=\mathrm{id}_{S_C}\)), the achievable fidelity at each receiver is
constrained by its vocabulary overlap with the sender---a purely
semantic bottleneck (see Proposition~\ref{prop:broadcast-bottleneck}
in Section~\ref{subsec:app-results}).

\begin{remark}[Relay scenario (future work)]
\label{rem:relay-scenario}
A third sub-scenario arises when an intermediate agent~\(k\) acts as
a relay: agent~\(i\) transmits to agent~\(k\), which performs
inference within \(\Cn(S_O^{(k)})\) and then re-encodes and forwards
the result to agent~\(j\).
This is naturally modeled as a composition of two pairwise channels,
\(\mathfrak C^{ij}_{\mathrm{relay}}
 =\mathfrak C^{kj}\circ\mathfrak C^{ik}\),
using the information-model composition machinery of
Definition~\ref{def:model-composition} and
Remark~\ref{rem:composition-assoc}.
The relay setting raises the question of whether intermediate inference can
\emph{change} the effective capacity of the end-to-end
link---a possibility that has no direct classical analogue,
since in classical relay channel
theory~\cite{cover2006elements} the relay cannot exploit
logical structure to reduce the message set.
A complete treatment requires multi-letter extensions and is deferred
to future work; the pairwise and broadcast results of this section
provide the necessary building blocks.
\end{remark}

\subsubsection*{Key Questions}

The heterogeneous multi-agent setting gives rise to four interrelated
design and analysis questions that the classical framework cannot
address.

\smallskip
\noindent\textbf{Q1 (Closure reliability from vocabulary overlap).}\;
Under what conditions on the overlap between \(S_O^{(i)}\) and
\(S_O^{(j)}\) can agent~\(j\) reconstruct the \emph{deductive closure}
of agent~\(i\)'s knowledge base---i.e., achieve
\(\mathsf{F}_{\Cn}\bigl(S_O^{(i)},S_O^{(j)}\bigr)=1\)
(Definition~\ref{def:closure-fidelity})---even if the literal symbols
differ?
Proposition~\ref{prop:noise-fidelity} provides necessary and sufficient
conditions in the abstract noise-pair language; the task here is to
translate those conditions into explicit, verifiable predicates on the
knowledge-base pair \((S_O^{(i)},S_O^{(j)})\).

\smallskip
\noindent\textbf{Q2 (Heterogeneous deductive compression).}\;
How many channel uses are needed to communicate the full knowledge base
\(S_O^{(i)}\) to agent~\(j\) under closure reliability, and how does
this compare with the Hamming-reliable baseline?
Theorem~\ref{thm:achievability} and
Corollary~\ref{cor:min-blocklength} (cf.~\cite{shannon1948mathematical,cover2006elements}) establish the deductive compression
ratio \(\log|\Atom(S_O)|/\log|S_O|\) in the homogeneous setting
(\(\hat S_O=S_O\)); the question is whether the same ratio persists,
improves, or degrades under vocabulary heterogeneity.

\smallskip
\noindent\textbf{Q3 (Invariant diagnosis).}\;
How do the six families of semantic channel invariants catalogued in
Theorem~\ref{thm:invariant-summary}---source-side structural
invariants, set-level fidelity, noise-pair indices, structural quality
indices, receiver-side comparison indices, and information-theoretic
invariants---depend on the knowledge-base overlap?
Explicit formulas in terms of the overlap structure would allow a system
designer to predict channel performance \emph{a priori} from
knowledge-base metadata, without computing the full channel kernel.

\smallskip
\noindent\textbf{Q4 (Broadcast bottleneck).}\;
In the broadcast scenario, which receiver determines the minimum
blocklength, and does the cost grow with the number of receivers?
Classical broadcast channel theory identifies the weakest receiver via
the degradation order of the physical channel; the semantic setting
introduces a second axis of ``weakness'' rooted in vocabulary overlap,
and the interaction between the two axes must be characterized.

\subsubsection*{Preview of Main Results}

The answers to Q1--Q4 are developed in full in
Section~\ref{subsec:app-results}; an informal summary is given here to
guide the reader.

For \textbf{Q1}, the answer has two layers.
At the \emph{set level}, closure fidelity
\(\mathsf{F}_{\Cn}(S_O^{(i)},S_O^{(j)})=1\)
holds if and only if
(a)~every core element of the sender is \emph{derivable from}
the receiver's knowledge base
(\(A^{(i)}\subseteq\Cn(S_O^{(j)})\)), and
(b)~every surplus state in the receiver's vocabulary is
derivable from the sender's knowledge base
(Proposition~\ref{prop:overlap-fidelity}).
At the \emph{operational} level, the two-layer code of
Theorem~\ref{thm:heterogeneous-closure} achieves vanishing
closure error under a stronger but readily verifiable form
of~(a): the sender's core is \emph{literally contained} in
the receiver's vocabulary
(\(A^{(i)}\subseteq S_O^{(j)}\)).

For \textbf{Q2}, the deductive compression ratio is shown to be
\emph{invariant} under vocabulary heterogeneity, provided the
two-condition criterion of~Q1 is satisfied
(Theorem~\ref{thm:heterogeneous-compression}).
The minimum blocklength for closure-reliable communication remains
\(n^*\approx\log|\Atom(S_O^{(i)})|/C(W_{ij})\), identical to the
homogeneous case.
When condition~(a) fails---i.e., some sender core element is
not derivable from the receiver's knowledge base---no code of
any blocklength can achieve \(\mathsf{F}_{\Cn}=1\), and the
achievable fidelity is bounded by the core preservation ratio
(Corollary~\ref{cor:heterogeneous-impossibility}).

For \textbf{Q3}, each invariant is expressed in terms of the
overlap decomposition introduced in
Section~\ref{subsec:app-assumptions}
(Propositions~\ref{prop:overlap-noise}--\ref{prop:overlap-structural}).
The set-level invariants (\(\rho_{\Atom}\),
\(\mathsf{F}_{\Cn}\), \(\Delta\mathsf{A}\),
\(\Delta\mathsf{D_d}\)) are fully determined by the
knowledge-base pair alone, via the overlap scalars
\(|A_-^{ij}|\), \(|S_{+,d}^{ij}|\), and
\(|S_{+,n}^{ij}|\); the probabilistic indices
(\(\Phi_{\Atom}\), \(\Psi_{+}\), \(\mathsf{F}\),
\(\mathsf{E}\)) are \emph{constrained} by these scalars
(e.g., \(\Phi_{\Atom}=0\) whenever \(A_-^{ij}\neq\varnothing\))
but additionally depend on the channel kernel; and the
information-theoretic invariants depend further on the carrier
channel~\(W_{ij}\).

For \textbf{Q4}, the broadcast minimum blocklength is shown to depend
only on the \emph{sender's} core and the physical channel capacity,
independent of the number of receivers---provided every receiver's
vocabulary covers the sender's core
(Theorem~\ref{thm:broadcast-compression}).
When some receiver violates this coverage condition, it becomes the
broadcast bottleneck not because of channel degradation but because of
vocabulary loss, and no amount of coding can compensate
(Proposition~\ref{prop:broadcast-bottleneck}).
This ``semantic bottleneck'' phenomenon has no counterpart in classical
broadcast theory and constitutes the sharpest qualitative distinction
between the two frameworks.

\begin{remark}[Classical recovery as a special case]
\label{rem:classical-recovery}
When all agents share the same knowledge base
(\(S_O^{(i)}=S_O^{(j)}\) for all \(i,j\)), the noise pair is trivial,
the two-condition criterion is vacuously satisfied, and all results of
this section reduce to the homogeneous theory of
Section~\ref{subsec:coding} (and, in the irredundant case, to
classical Shannon theory via
Corollary~\ref{cor:irredundant-classical}).
The heterogeneous analysis thus strictly generalizes rather than
replaces the earlier results.
\end{remark}

\subsection{Formal Model and Standing Assumptions}
\label{subsec:app-assumptions}

This subsection formalizes the multi-agent communication scenario
of Section~\ref{subsec:app-problem} within the information model
framework of Sections~\ref{sec:model}--\ref{sec:channel}, and
introduces the \emph{overlap decomposition}---the combinatorial
structure through which all semantic channel invariants are
expressed as functions of the sender--receiver knowledge-base pair.


\begin{definition}[Agent knowledge base]
\label{def:agent-kb}
Fix a set of \(K+1\) agents indexed by
\(i\in\{0,1,\ldots,K\}\).
Each agent~\(i\) is associated with a finite knowledge base
\(S_O^{(i)}\subseteq\mathbb{S}_O\) that is
\(\mathcal L_{\mathrm{sem}}\)-definable in the ambient
structure~\(\mathfrak R\) and satisfies
Assumptions~\ref{assump:finite-so}--\ref{assump:core-extractable}.
The \emph{irredundant core} and \emph{stored shortcuts} of
agent~\(i\) are
\[
  A^{(i)}:=\Atom\bigl(S_O^{(i)}\bigr),
  \qquad
  J^{(i)}:=S_O^{(i)}\setminus A^{(i)}.
\]
By Proposition~\ref{prop:atom-core-correct},
\(\Cn\bigl(A^{(i)}\bigr)=\Cn\bigl(S_O^{(i)}\bigr)\),
the core~\(A^{(i)}\) is irredundant, and
\(S_O^{(i)}\subseteq\Cn\bigl(A^{(i)}\bigr)\).
\end{definition}


\begin{assumption}[Common proof system and ambient universe]
\label{assump:common-ps}
All agents share:
\begin{enumerate}[label=\textup{(CP\arabic*)}]
  \item the same proof system
        \((\mathsf{PS},\,T_{\mathsf{PS}},\,\Cn)\) and inference
        fragment \(\mathcal L_{\mathrm{kb}}\)
        \textup{(Assumption~\ref{assump:proof-system},
        Axiom~\ref{ax:T-operator})};
  \item the same ambient semantic universe \(\mathbb{S}_O\)
        with its injective encoding~\(\enc_O\) and canonical order
        \textup{(Section~\ref{subsec:atomic-derivation},
        Assumption~\ref{assump:semantic-universe})};
  \item the same semantic sublanguage \(\mathcal L_{\mathrm{sem}}\)
        \textup{(Assumption~\ref{assump:semantic-sublanguage})}.
\end{enumerate}
The agents differ \emph{only} in the knowledge bases
\(S_O^{(0)},S_O^{(1)},\ldots,S_O^{(K)}\subseteq\mathbb{S}_O\)
that they store and operate on.
\end{assumption}

\begin{remark}[Relaxation to compatible sub-systems]
\label{rem:compatible-subsystems}
Assumption~\ref{assump:common-ps} is the strongest form of
proof-system homogeneity: every agent can apply every inference
rule of~\(\mathsf{PS}\).
A natural relaxation allows each agent~\(i\) to use a
sub-system \(\mathsf{PS}^{(i)}\subseteq\mathsf{PS}\) with
\(\Cn^{(i)}(\Gamma)\subseteq\Cn(\Gamma)\) for
all~\(\Gamma\subseteq\mathbb{S}_O\).
The results of this section remain valid with \(\Cn\) replaced by
the weakest common sub-system
\(\Cn^{\cap}:=\bigcap_i\Cn^{(i)}\), at the cost of
potentially larger irredundant cores and weaker compression gains.
We do not pursue this relaxation and work throughout with the
common operator~\(\Cn\).
\end{remark}


\begin{definition}[Pairwise overlap decomposition]
\label{def:overlap-decomposition}
For a fixed sender--receiver pair \((i,j)\) with \(i\neq j\),
define the following subsets of~\(\mathbb{S}_O\).

\smallskip
\noindent\emph{Three-way partition of
\(S_O^{(i)}\cup S_O^{(j)}\):}
\begin{align}
  S_{\cap}^{ij}
  &\;:=\; S_O^{(i)}\cap S_O^{(j)}
  &&\text{(common states)},
  \label{eq:overlap-common}\\[2pt]
  S_{-}^{ij}
  &\;:=\; S_O^{(i)}\setminus S_O^{(j)}
  &&\text{(lost states)},
  \label{eq:overlap-lost}\\[2pt]
  S_{+}^{ij}
  &\;:=\; S_O^{(j)}\setminus S_O^{(i)}
  &&\text{(surplus states)}.
  \label{eq:overlap-surplus}
\end{align}

\noindent\emph{Core partition:}
\begin{align}
  A_{\cap}^{ij}
  &\;:=\; A^{(i)}\cap S_O^{(j)}
  &&\text{(preserved core)},
  \label{eq:overlap-core-preserved}\\[2pt]
  A_{-}^{ij}
  &\;:=\; A^{(i)}\setminus S_O^{(j)}
  &&\text{(lost core)}.
  \label{eq:overlap-core-lost}
\end{align}

\noindent\emph{Surplus stratification:}
\begin{align}
  S_{+,d}^{ij}
  &\;:=\; S_{+}^{ij}\cap\Cn\bigl(S_O^{(i)}\bigr)
  &&\text{(derivable surplus)},
  \label{eq:overlap-surplus-d}\\[2pt]
  S_{+,n}^{ij}
  &\;:=\; S_{+}^{ij}\setminus\Cn\bigl(S_O^{(i)}\bigr)
  &&\text{(non-derivable surplus)}.
  \label{eq:overlap-surplus-n}
\end{align}
\end{definition}

\begin{proposition}[Overlap partition properties]
\label{prop:overlap-partition}
The overlap decomposition of
Definition~\textup{\ref{def:overlap-decomposition}} satisfies:
\begin{enumerate}[label=\textup{(\roman*)}]
  \item \emph{Three-way disjoint union:}\;
        \(S_O^{(i)}\cup S_O^{(j)}
          = S_{-}^{ij}\;\dot\cup\;
            S_{\cap}^{ij}\;\dot\cup\;
            S_{+}^{ij}\).
  \item \emph{Sender decomposition:}\;
        \(S_O^{(i)}
          = S_{\cap}^{ij}\;\dot\cup\; S_{-}^{ij}\).
  \item \emph{Receiver decomposition:}\;
        \(S_O^{(j)}
          = S_{\cap}^{ij}\;\dot\cup\; S_{+}^{ij}\).
  \item \emph{Core partition:}\;
        \(A^{(i)}
          = A_{\cap}^{ij}\;\dot\cup\; A_{-}^{ij}\).
  \item \emph{Surplus partition:}\;
        \(S_{+}^{ij}
          = S_{+,d}^{ij}\;\dot\cup\; S_{+,n}^{ij}\).
  \item \emph{Noise-pair consistency:}\;
        \(S_{-}^{ij}=S_O^{-}\) and \(S_{+}^{ij}=S_O^{+}\),
        where \((S_O^{-},S_O^{+})\) is the end-to-end noise
        pair of
        Definition~\textup{\ref{def:pairwise-scenario}}.
  \item \emph{Core loss refines state loss:}\;
        \(A_{-}^{ij}\subseteq S_{-}^{ij}\).
  \item \emph{Computability:}\;
        All seven sets and their cardinalities are computable
        from the finite knowledge bases
        \(S_O^{(i)},S_O^{(j)}\) under
        Axiom~\textup{\ref{ax:T-operator}}.
\end{enumerate}
\end{proposition}

\begin{proof}
Parts~\textup{(i)}--\textup{(iii)} are standard set partition
identities: \(S_{\cap}^{ij}\), \(S_{-}^{ij}\), \(S_{+}^{ij}\)
are pairwise disjoint by construction, and their union equals
\(S_O^{(i)}\cup S_O^{(j)}\); restricting to \(S_O^{(i)}\)
yields~\textup{(ii)}, and to \(S_O^{(j)}\) yields~\textup{(iii)}.

Part~\textup{(iv)}: since \(A^{(i)}\subseteq S_O^{(i)}\),
\(A^{(i)} = (A^{(i)}\cap S_O^{(j)})
  \;\dot\cup\;(A^{(i)}\setminus S_O^{(j)})
  = A_{\cap}^{ij}\;\dot\cup\; A_{-}^{ij}\).

Part~\textup{(v)}: immediate from the definition of
\(S_{+,d}^{ij}\) and \(S_{+,n}^{ij}\) as complementary
subsets of~\(S_{+}^{ij}\).

Part~\textup{(vi)}: comparing
\eqref{eq:overlap-lost}--\eqref{eq:overlap-surplus} with
\eqref{eq:pairwise-noise-pair} gives the identification
directly.

Part~\textup{(vii)}: \(A_{-}^{ij}
  =A^{(i)}\setminus S_O^{(j)}
  \subseteq S_O^{(i)}\setminus S_O^{(j)}
  =S_{-}^{ij}\), since \(A^{(i)}\subseteq S_O^{(i)}\).

Part~\textup{(viii)}: since \(S_O^{(i)}\) and \(S_O^{(j)}\)
are finite and effectively listable
\textup{(Assumption~\ref{assump:finite-so})}, membership is
decidable by exhaustive comparison.
The sets \(S_{\cap}^{ij}\), \(S_{-}^{ij}\), \(S_{+}^{ij}\),
\(A_{\cap}^{ij}\), \(A_{-}^{ij}\) are then computable by
enumeration.
The surplus stratification requires testing
\(s\in\Cn(S_O^{(i)})\) for each \(s\in S_{+}^{ij}\), which
is decidable by iterating \(T_{\mathsf{PS}}\) from
\(S_O^{(i)}\) until stabilization
\textup{(Axiom~\ref{ax:T-operator}(IC2)--(IC4))}.
\end{proof}

\begin{remark}[Core preservation is tested against
  \(S_O^{(j)}\), not against \(A^{(j)}\)]
\label{rem:core-vs-vocab}
The preserved core
\(A_{\cap}^{ij}=A^{(i)}\cap S_O^{(j)}\) tests whether each
sender core element is \emph{present in the receiver's
vocabulary}---that is, an element of \(S_O^{(j)}\)---not
whether it belongs to the receiver's irredundant core
\(A^{(j)}\).
A sender core element \(a\in A^{(i)}\) may appear in
\(S_O^{(j)}\) as a \emph{redundant} stored shortcut
(i.e., \(a\in J^{(j)}\)); it is still counted as preserved,
because the receiver can produce it as a decoding output
regardless of its redundancy status in~\(S_O^{(j)}\).
\end{remark}

\begin{remark}[Asymmetry of the overlap decomposition]
\label{rem:overlap-asymmetry}
The decomposition of
Definition~\ref{def:overlap-decomposition} depends on which
agent is designated as sender and which as receiver.
Swapping the roles yields the \emph{transposed} decomposition:
\(S_{\cap}^{ji}=S_{\cap}^{ij}\),
\(S_{-}^{ji}=S_{+}^{ij}\),
\(S_{+}^{ji}=S_{-}^{ij}\),
\(A_{\cap}^{ji}=A^{(j)}\cap S_O^{(i)}\),
\(A_{-}^{ji}=A^{(j)}\setminus S_O^{(i)}\).
In particular, the core partition
\((A_{\cap}^{ij},A_{-}^{ij})\) and the surplus
stratification \((S_{+,d}^{ij},S_{+,n}^{ij})\)
are generally \emph{not} symmetric in \((i,j)\).
\end{remark}

\begin{remark}[Key scalar summaries of the overlap]
\label{rem:three-scalars}
Although the overlap decomposition produces seven subsets, the
\emph{set-level} conditions governing closure fidelity
\textup{(Proposition~\ref{prop:overlap-fidelity})} reduce to two
binary tests:
\(|A_{-}^{ij}|=0\) (no core loss) and
\(|S_{+,n}^{ij}|=0\) (no non-derivable surplus).
The \emph{coding-theoretic} results of
Section~\ref{subsec:app-results} (blocklength, compression
ratio) depend additionally on \(|A^{(i)}|\) and \(C(W_{ij})\),
while the \emph{probabilistic} indices
(\(\Phi_{\Atom}\), \(\Psi_{+}\), \(\mathsf{F}\),
\(\mathsf{E}\)) depend further on the channel kernel
\(\kappa_{\mathrm{sem}}^{ij}\).
The remaining overlap cardinalities are related by simple
accounting:
\(|A_{\cap}^{ij}|=|A^{(i)}|-|A_{-}^{ij}|\),
\(|S_{+}^{ij}|=|S_{+,d}^{ij}|+|S_{+,n}^{ij}|\);
the quantities \(|S_{-}^{ij}|\), \(|S_{\cap}^{ij}|\), and
\(|S_O^{(j)}|\) are mutually determined once any one of them
is known, via
\(|S_{\cap}^{ij}|=|S_O^{(i)}|-|S_{-}^{ij}|\) and
\(|S_O^{(j)}|=|S_{\cap}^{ij}|+|S_{+}^{ij}|\).
\end{remark}


\begin{remark}[Broadcast overlap decomposition]
\label{rem:broadcast-overlap}
In the broadcast scenario
(Definition~\ref{def:broadcast-scenario}), the sender is
agent~\(0\) and the receivers are agents \(1,\ldots,K\).
For each receiver~\(j\), the overlap decomposition
\textup{(Definition~\ref{def:overlap-decomposition})} is
applied to the pair \((0,j)\), yielding receiver-specific
quantities
\(A_{\cap}^{0j}\), \(A_{-}^{0j}\),
\(S_{+,d}^{0j}\), \(S_{+,n}^{0j}\), etc.
The \emph{broadcast core coverage condition}---that
\(A_{-}^{0j}=\varnothing\) for every
\(j\in\{1,\ldots,K\}\)---plays a central role in
Theorem~\ref{thm:broadcast-compression} and
Proposition~\ref{prop:broadcast-bottleneck}.
\end{remark}


\begin{definition}[Heterogeneous semantic channel (formal)]
\label{def:heterogeneous-channel}
For a sender--receiver pair \((i,j)\) with \(i\neq j\), the
\emph{heterogeneous semantic channel} is the semantic channel
\textup{(Definition~\ref{def:semantic-channel})}
\[
  \mathfrak C^{ij}
  \;=\;
  \bigl(\,
    \mathcal I^{(i)},\;
    \mathcal I_{\mathrm{ch}}^{ij},\;
    \mathcal I_{\mathrm{dec}}^{(j)},\;
    \kappa_{\enc},\;
    W_{ij},\;
    D
  \,\bigr),
\]
where the constituent models are defined as follows.
\begin{enumerate}[label=\textup{(\roman*)}]
  \item \emph{Sender information model.}\;
        \(\mathcal I^{(i)}
          =\langle O^{(i)},T_O^{(i)},S_O^{(i)},
                   C,T_C,S_C,R_{\mathcal E}^{(i)}\rangle\)
        is an information model
        \textup{(Definition~\ref{def:info-instance})} with
        semantic state set \(S_O^{(i)}\), carrier state set
        \(S_C\), and enabling map
        \(\mathcal E^{(i)}:S_O^{(i)}\Rightarrow S_C\).
  \item \emph{Carrier channel model.}\;
        \(\mathcal I_{\mathrm{ch}}^{ij}\) is a carrier channel
        information model
        \textup{(Definition~\ref{def:carrier-channel-model})}
        with input \(S_C\), output \(\hat S_C\), and carrier
        channel kernel
        \(W_{ij}:S_C\rightsquigarrow\hat S_C\).
  \item \emph{Receiver decoding model.}\;
        \(\mathcal I_{\mathrm{dec}}^{(j)}
          =\langle\hat C,T_{\hat C},\hat S_C,
                   \hat O^{(j)},T_{\hat O}^{(j)},S_O^{(j)},
                   R_{\mathcal E}^{\mathrm{dec},(j)}\rangle\)
        is a decoding information model
        \textup{(Definition~\ref{def:decoding-model})} with
        reconstructed space
        \(\hat S_O:=S_O^{(j)}\subseteq\mathbb{S}_O\) and
        enabling map
        \(\mathcal E_{\mathrm{dec}}^{(j)}:\hat S_C\Rightarrow S_O^{(j)}\).
  \item \emph{Kernels.}\;
        \(\kappa_{\enc}\in\mathcal K(\mathcal I^{(i)})\)
        is the encoding kernel and
        \(D\in\mathcal K(\mathcal I_{\mathrm{dec}}^{(j)})\)
        is the decoding kernel.
\end{enumerate}
The \emph{end-to-end kernel} is
\begin{equation}\label{eq:ksem-ij}
  \kappa_{\mathrm{sem}}^{ij}
  \;:=\;
  D\circ W_{ij}\circ\kappa_{\enc}
  \;:\;
  S_O^{(i)}\rightsquigarrow S_O^{(j)},
\end{equation}
and the \emph{end-to-end noise pair} is
\((S_O^{-},S_O^{+})=(S_{-}^{ij},S_{+}^{ij})\)
\textup{(Proposition~\ref{prop:overlap-partition}(vi))}.
\end{definition}

\begin{remark}[Inherited proof-system structure at the receiver]
\label{rem:receiver-ps-structure}
Since \(S_O^{(j)}\subseteq\mathbb{S}_O\) and the proof system
\(\mathsf{PS}\) acts on all of~\(\mathbb{S}_O\)
\textup{(Assumption~\ref{assump:common-ps})}, the receiver
inherits the deductive closure \(\Cn(S_O^{(j)})\), the
irredundant core \(A^{(j)}=\Atom(S_O^{(j)})\), and the
derivation-depth stratification
\(\Dd(\cdot\mid A^{(j)})\).
The semantic invariants
\(\mathsf{A}(\mathcal I^{(j)})=|A^{(j)}|\) and
\(\mathsf{D_d}(\mathcal I^{(j)})
  =\max_{q\in S_O^{(j)}}\Dd(q\mid A^{(j)})\)
are therefore well-defined and computable
\textup{(Theorem~\ref{thm:semantic-invariants})}.
In general, \(A^{(j)}\neq A^{(i)}\),
\(\Cn(S_O^{(j)})\neq\Cn(S_O^{(i)})\), and
\(\mathsf{D_d}(\mathcal I^{(j)})\neq\mathsf{D_d}(\mathcal I^{(i)})\);
the overlap decomposition
\textup{(Definition~\ref{def:overlap-decomposition})} quantifies
each of these discrepancies.
\end{remark}


\begin{assumption}[Full enabling (heterogeneous setting)]
\label{assump:full-enabling-hetero}
For the heterogeneous semantic channel
\(\mathfrak C^{ij}\) of
Definition~\ref{def:heterogeneous-channel}:
\begin{enumerate}[label=\textup{(FE\arabic*)}]
  \item \emph{Full encoding enabling:}\;
        \(\mathcal E^{(i)}(s_o)=S_C\) for every
        \(s_o\in S_O^{(i)}\).
  \item \emph{Full decoding enabling:}\;
        \(\mathcal E_{\mathrm{dec}}^{(j)}(\hat s_c)=S_O^{(j)}\)
        for every \(\hat s_c\in\hat S_C\).
\end{enumerate}
\end{assumption}

\begin{remark}[Role of the full enabling assumption]
\label{rem:full-enabling-role}
Under Assumption~\ref{assump:full-enabling-hetero}, the encoder
can map any semantic state to any carrier symbol, and the
decoder can produce any state in the receiver's vocabulary from
any received carrier symbol.
All vocabulary-mismatch effects are then captured
\emph{entirely} by the noise pair
\((S_{-}^{ij},S_{+}^{ij})\); no additional bottleneck arises
from the enabling structure.
This is the heterogeneous counterpart of the full enabling
condition in
Theorem~\ref{thm:data-processing}\textup{(iii)(a)--(b)}.
When the enabling is constrained (e.g., the encoder can access
only a subset of~\(S_C\) for each \(s_o\)), additional capacity
reductions follow from
Proposition~\ref{prop:structural-bound}\textup{(ii)}; such
cases are noted where relevant but not analyzed in full
generality.
\end{remark}

\begin{remark}[Carrier alphabet size requirement]
\label{rem:carrier-size}
The block coding achievability results of
Theorem~\ref{thm:achievability} do \emph{not} require a
single-letter alphabet-size comparison between \(|S_O^{(i)}|\)
and~\(|S_C|\), since block codes operate over the product
alphabet~\(S_C^n\).
However, the \emph{capacity equality}
\(C_{\mathrm{sem}}^{ij}(W_{ij})=C(W_{ij})\)
\textup{(Proposition~\ref{prop:overlap-capacity}(ii))} invokes
Theorem~\ref{thm:data-processing}\textup{(iii)}, which requires
\(|S_O^{(i)}|\le|S_C|\) and \(|S_O^{(j)}|\le|\hat S_C|\).
Throughout Sections~\ref{subsec:app-theory}--\ref{subsec:app-results},
we assume \(|S_C|\ge\max_{i}|S_O^{(i)}|\) and
\(|\hat S_C|\ge|S_C|\), which is the natural setting when the
physical carrier is designed to accommodate the largest agent
vocabulary.
These conditions ensure both the capacity equality and the
practical convenience of single-letter encoding when desired.
\end{remark}


\begin{assumption}[Standing assumptions for
  Section~\ref{sec:application}]
\label{assump:standing-hetero}
Throughout
Sections~\ref{subsec:app-assumptions}--\ref{subsec:app-results},
the following conditions are in force unless explicitly stated
otherwise:
\begin{enumerate}[label=\textup{(SA\arabic*)}]
  \item all standing assumptions of
        Sections~\ref{sec:model}--\ref{sec:channel}, including
        Assumptions~\ref{assump:ordered-structures},
        \ref{assump:semantic-sublanguage},
        \ref{assump:proof-system},
        \ref{assump:finite-so},
        \ref{assump:core-extractable},
        Axiom~\ref{ax:T-operator}, and
        Assumption~\ref{assump:semantic-universe};
  \item the common proof system assumption
        \textup{(Assumption~\ref{assump:common-ps})};
  \item the full enabling assumption
        \textup{(Assumption~\ref{assump:full-enabling-hetero})};
  \item the carrier channel satisfies \(C(W_{ij})>0\) and the
        carrier alphabet sizes satisfy
        \(|S_C|\ge\max_i|S_O^{(i)}|\) and
        \(|\hat S_C|\ge|S_C|\);
  \item the deductive independence of core elements
        \textup{(Assumption~\ref{assump:core-disjoint})} holds
        for the sender's knowledge base \(S_O^{(i)}\) when
        converse bounds are invoked.
\end{enumerate}
\end{assumption}


Table~\ref{tab:notation-iv} collects the notation introduced
in this subsection for convenient reference throughout
Section~\ref{sec:application}.

\begin{table}[t]
\centering
\caption{Notation summary for
  Section~\ref{sec:application}.
  All quantities are defined relative to a fixed
  sender--receiver pair~\((i,j)\).}
\label{tab:notation-iv}
\renewcommand{\arraystretch}{1.2}
\begin{tabularx}{\columnwidth}{@{}l L@{}}
\toprule
\textbf{Symbol} & \textbf{Meaning} \\
\midrule
\(S_O^{(i)}\) &
  Knowledge base (semantic state set) of agent~\(i\) \\
\(A^{(i)},\; J^{(i)}\) &
  Irredundant core and stored shortcuts of agent~\(i\) \\
\(S_{\cap}^{ij}\) &
  Common states:
  \(S_O^{(i)}\cap S_O^{(j)}\) \\
\(S_{-}^{ij}\) &
  Lost states:
  \(S_O^{(i)}\setminus S_O^{(j)}\)
  \((\,=S_O^{-})\) \\
\(S_{+}^{ij}\) &
  Surplus states:
  \(S_O^{(j)}\setminus S_O^{(i)}\)
  \((\,=S_O^{+})\) \\
\(A_{\cap}^{ij}\) &
  Preserved sender core:
  \(A^{(i)}\cap S_O^{(j)}\) \\
\(A_{-}^{ij}\) &
  Lost sender core:
  \(A^{(i)}\setminus S_O^{(j)}\) \\
\(S_{+,d}^{ij}\) &
  Derivable surplus:
  \(S_{+}^{ij}\cap\Cn(S_O^{(i)})\) \\
\(S_{+,n}^{ij}\) &
  Non-derivable surplus:
  \(S_{+}^{ij}\setminus\Cn(S_O^{(i)})\) \\
\(\mathfrak C^{ij}\) &
  Heterogeneous semantic channel from~\(i\) to~\(j\) \\
\(W_{ij}\) &
  Carrier channel kernel from~\(i\) to~\(j\) \\
\(\kappa_{\mathrm{sem}}^{ij}\) &
  End-to-end semantic kernel for pair \((i,j)\) \\
\bottomrule
\end{tabularx}
\end{table}

\subsection{Instantiation of Semantic Channel Invariants}
\label{subsec:app-theory}

This subsection applies the invariant machinery of
Section~\ref{sec:channel} to the heterogeneous pair~\((i,j)\),
expressing each invariant of
Theorem~\ref{thm:invariant-summary} as a function of the
overlap decomposition of Section~\ref{subsec:app-assumptions}.
Throughout, we fix a sender--receiver pair \((i,j)\) with the
heterogeneous semantic channel \(\mathfrak C^{ij}\) of
Definition~\ref{def:heterogeneous-channel}, under the standing
assumptions of Assumption~\ref{assump:standing-hetero}.


\begin{proposition}[Set-level invariants from overlap]
\label{prop:overlap-noise}
For the heterogeneous semantic channel \(\mathfrak C^{ij}\)
with noise pair
\((S_O^{-},S_O^{+})=(S_{-}^{ij},S_{+}^{ij})\)
\textup{(Proposition~\ref{prop:overlap-partition}(vi))}:
\begin{enumerate}[label=\textup{(\roman*)}]
  \item \emph{Preserved region:}\;
        \(\tilde S_O^{\cap}
          =S_O^{(i)}\cap S_O^{(j)}
          =S_{\cap}^{ij}\).
  \item \emph{Core preservation ratio:}\;
        \begin{equation}\label{eq:rho-overlap}
          \rho_{\Atom}\bigl(S_O^{(i)},S_O^{(j)}\bigr)
          \;=\;
          \frac{|A_{\cap}^{ij}|}{|A^{(i)}|}
          \;=\;
          1-\frac{|A_{-}^{ij}|}{|A^{(i)}|}.
        \end{equation}
        In particular, \(\rho_{\Atom}=1\) if and only if
        \(A_{-}^{ij}=\varnothing\).
  \item \emph{Spurious derivability:}\;
        \(S_O^{+}\subseteq\Cn(S_O^{(i)})\) if and only if
        \(S_{+,n}^{ij}=\varnothing\).
\end{enumerate}
\end{proposition}

\begin{proof}
\textup{(i)}:\
\(S_O^{-}=S_{-}^{ij}\)
\textup{(Proposition~\ref{prop:overlap-partition}(vi))}, so
\(\tilde S_O^{\cap}
  =S_O^{(i)}\setminus S_O^{-}
  =S_O^{(i)}\setminus(S_O^{(i)}\setminus S_O^{(j)})
  =S_O^{(i)}\cap S_O^{(j)}
  =S_{\cap}^{ij}\).

\smallskip\noindent
\textup{(ii)}:\
By Definition~\ref{def:atom-preservation},
\(\rho_{\Atom}
  =|A^{(i)}\cap S_O^{(j)}|/|A^{(i)}|
  =|A_{\cap}^{ij}|/|A^{(i)}|\).
By Proposition~\ref{prop:overlap-partition}(iv),
\(|A_{\cap}^{ij}|=|A^{(i)}|-|A_{-}^{ij}|\).

\smallskip\noindent
\textup{(iii)}:\
\(S_O^{+}=S_{+}^{ij}\) and
\(S_{+}^{ij}\subseteq\Cn(S_O^{(i)})\)
iff
\(S_{+}^{ij}\setminus\Cn(S_O^{(i)})=\varnothing\)
iff \(S_{+,n}^{ij}=\varnothing\).
\end{proof}


\begin{proposition}[Closure fidelity: necessary and sufficient
  conditions]
\label{prop:overlap-fidelity}
For the sender--receiver pair~\((i,j)\),
\[
  \mathsf{F}_{\Cn}\bigl(S_O^{(i)},\,S_O^{(j)}\bigr)=1
  \quad\Longleftrightarrow\quad
  \Cn\bigl(S_O^{(i)}\bigr)=\Cn\bigl(S_O^{(j)}\bigr),
\]
and this holds if and only if both of the following conditions
are satisfied:
\begin{enumerate}[label=\textup{(F\arabic*)}]
  \item \emph{Sender core derivable from receiver:}\;
        \(A^{(i)}\subseteq\Cn\bigl(S_O^{(j)}\bigr)\).
  \item \emph{No non-derivable surplus:}\;
        \(S_{+,n}^{ij}=\varnothing\)
        \textup{(}equivalently,
        \(S_{+}^{ij}\subseteq\Cn(S_O^{(i)})\)\textup{)}.
\end{enumerate}
\end{proposition}

\begin{proof}
The first equivalence is the definition of
\(\mathsf{F}_{\Cn}\)
\textup{(Definition~\ref{def:closure-fidelity}):}
\(\mathsf{F}_{\Cn}=1\) iff the Jaccard index of the two
closures equals~\(1\), i.e., the closures coincide.
It remains to show
\(\Cn(S_O^{(i)})=\Cn(S_O^{(j)})\) iff
\textup{(F1)}+\textup{(F2)}.

\smallskip\noindent
\emph{Sufficiency.}\;
From~\textup{(F1)}:
\(A^{(i)}\subseteq\Cn(S_O^{(j)})\).
By monotonicity~\textup{(Cn2)},
\(\Cn(A^{(i)})\subseteq\Cn(\Cn(S_O^{(j)}))
  =\Cn(S_O^{(j)})\)
\textup{(idempotence (Cn3))}. This chain of inclusions is a standard consequence of
the Tarski closure axioms; see,
e.g.,~\cite{tarski1983logic}.
Since
\(\Cn(A^{(i)})=\Cn(S_O^{(i)})\)
\textup{(Proposition~\ref{prop:atom-core-correct}(i))},
\(\Cn(S_O^{(i)})\subseteq\Cn(S_O^{(j)})\).

From~\textup{(F2)}:
\(S_{+}^{ij}\subseteq\Cn(S_O^{(i)})\).
Since
\(S_{\cap}^{ij}\subseteq S_O^{(i)}
  \subseteq\Cn(S_O^{(i)})\)
\textup{(reflexivity (Cn1))},
\(S_O^{(j)}=S_{\cap}^{ij}\cup S_{+}^{ij}
  \subseteq\Cn(S_O^{(i)})\).
By monotonicity and idempotence,
\(\Cn(S_O^{(j)})\subseteq\Cn(S_O^{(i)})\).

Combining the two inclusions yields
\(\Cn(S_O^{(i)})=\Cn(S_O^{(j)})\).

\smallskip\noindent
\emph{Necessity.}\;
Suppose \(\Cn(S_O^{(i)})=\Cn(S_O^{(j)})\).

For~\textup{(F1)}:
\(A^{(i)}\subseteq\Cn(A^{(i)})
  =\Cn(S_O^{(i)})=\Cn(S_O^{(j)})\).

For~\textup{(F2)}:
\(S_{+}^{ij}\subseteq S_O^{(j)}
  \subseteq\Cn(S_O^{(j)})=\Cn(S_O^{(i)})\),
so \(S_{+,n}^{ij}
  =S_{+}^{ij}\setminus\Cn(S_O^{(i)})=\varnothing\).
\end{proof}

\begin{remark}[Strong vs.\ weak core coverage]
\label{rem:strong-weak-coverage}
Condition~\textup{(F1)} requires that each sender core element
be \emph{derivable from} the receiver's knowledge base; it
does \emph{not} require the element to be literally present
in~\(S_O^{(j)}\).
A strictly stronger condition is
\begin{enumerate}[label=\textup{(F1\('\))}]
  \item \(A_{-}^{ij}=\varnothing\), i.e.,
        \(A^{(i)}\subseteq S_O^{(j)}\).
\end{enumerate}
Condition~\textup{(F1\({}'\))} implies~\textup{(F1)} (since
\(S_O^{(j)}\subseteq\Cn(S_O^{(j)})\)) but not conversely:
a core element \(a\in A_{-}^{ij}\) may satisfy
\(a\in\Cn(S_O^{(j)})\) even though \(a\notin S_O^{(j)}\).

For \emph{set-level} closure fidelity, the weak
condition~\textup{(F1)} is both necessary and sufficient
\textup{(Proposition~\ref{prop:overlap-fidelity})}.
For \emph{operational} closure reliability via the two-layer
code of Theorem~\textup{\ref{thm:achievability}(ii)}, the
decoder outputs core elements directly and therefore requires
the strong condition~\textup{(F1\({}'\))} so that
\(A^{(i)}\subseteq S_O^{(j)}=\hat S_O\).
When only~\textup{(F1)} holds with
\(A_{-}^{ij}\neq\varnothing\), a more sophisticated decoding
strategy is needed; this is addressed in
Section~\ref{subsec:app-results}.
Throughout the remainder of this subsection, results are stated
under whichever version is required, with the distinction noted
explicitly.
\end{remark}

\begin{corollary}[Sufficient condition via overlap scalars]
\label{cor:overlap-sufficient}
If \(A_{-}^{ij}=\varnothing\) and
\(S_{+,n}^{ij}=\varnothing\), then
\(\mathsf{F}_{\Cn}(S_O^{(i)},S_O^{(j)})=1\).
This is the instantiation of
Proposition~\textup{\ref{prop:noise-fidelity}(iii)} in the
overlap language and the condition used in the achievability
results of Section~\textup{\ref{subsec:app-results}}.
\end{corollary}

\begin{proof}
\(A_{-}^{ij}=\varnothing\) gives
\(A^{(i)}\subseteq S_O^{(j)}\subseteq\Cn(S_O^{(j)})\), so
\textup{(F1)} holds.
\(S_{+,n}^{ij}=\varnothing\) is \textup{(F2)}.
Apply Proposition~\ref{prop:overlap-fidelity}.
\end{proof}


\begin{proposition}[Noise-pair probabilistic indices from
  overlap]
\label{prop:overlap-noise-indices}
Let \(\mathfrak C^{ij}\) be a heterogeneous semantic channel
with kernel
\(\kappa_{\mathrm{sem}}^{ij}
  :S_O^{(i)}\rightsquigarrow S_O^{(j)}\).
\begin{enumerate}[label=\textup{(\roman*)}]
  \item \emph{Core preservation index:}\;
        \begin{equation}\label{eq:Phi-overlap}
          \Phi_{\Atom}(\mathfrak C^{ij})
          \;=\;
          \begin{cases}
            \displaystyle\min_{a\in A^{(i)}}
              \kappa_{\mathrm{sem}}^{ij}(a\mid a)
            & \text{if } A_{-}^{ij}=\varnothing,\\[6pt]
            0 & \text{if } A_{-}^{ij}\neq\varnothing.
          \end{cases}
        \end{equation}
  \item \emph{Spurious probability index:}\;
        \begin{equation}\label{eq:Psi-overlap}
          \Psi_{+}(\mathfrak C^{ij})
          \;=\;
          \max_{s_o\in S_O^{(i)}}\;
          \sum_{\hat s_o\in S_{+}^{ij}}
          \kappa_{\mathrm{sem}}^{ij}(\hat s_o\mid s_o).
        \end{equation}
        In particular, \(\Psi_{+}=0\) whenever
        \(S_{+}^{ij}=\varnothing\)
        \textup{(}i.e., \(S_O^{(j)}\subseteq S_O^{(i)}\)\textup{)}.
  \item \emph{Noiseless deterministic case:}\;
        If \(W_{ij}=\mathrm{id}_{S_C}\) and both
        \(\kappa_{\enc}\) and \(D\) are deterministic with
        induced end-to-end function
        \(f:=D\circ\mathrm{id}\circ\kappa_{\enc}
          :S_O^{(i)}\to S_O^{(j)}\), then
        \begin{align}
          \Phi_{\Atom}(\mathfrak C^{ij})
          &=
          \begin{cases}
            1 & \text{if } A_{-}^{ij}=\varnothing \\
              &  \text{ and } f(a)=a \;\forall a\in A^{(i)},\\
            0 & \text{otherwise},
          \end{cases}
          \label{eq:Phi-noiseless}\\[4pt]
          \Psi_{+}(\mathfrak C^{ij})
          &=
          \begin{cases}
            0 & \text{if } f(S_O^{(i)})\subseteq S_{\cap}^{ij},\\
            1 & \text{otherwise}.
          \end{cases}
          \label{eq:Psi-noiseless}
        \end{align}
\end{enumerate}
\end{proposition}

\begin{proof}
\textup{(i)}:\
By Proposition~\ref{prop:overlap-partition}(vi),
\(A^{(i)}\cap S_O^{-}=A^{(i)}\cap S_{-}^{ij}=A_{-}^{ij}\).
When \(A_{-}^{ij}\neq\varnothing\),
Definition~\ref{def:prob-core-index} sets
\(\Phi_{\Atom}:=0\).
When \(A_{-}^{ij}=\varnothing\),
\(A^{(i)}\subseteq S_O^{(j)}=\tilde S_O\), so
\(\pi(a)=\kappa_{\mathrm{sem}}^{ij}(a\mid a)\) is
well-defined for every \(a\in A^{(i)}\) and
\(\Phi_{\Atom}=\min_{a}\pi(a)\).

\smallskip\noindent
\textup{(ii)}:\
Direct from Definition~\ref{def:spurious-index} with
\(S_O^{+}=S_{+}^{ij}\).
When \(S_{+}^{ij}=\varnothing\), the sum is empty and
\(\Psi_{+}=0\).

\smallskip\noindent
\textup{(iii)}:\
When \(\kappa_{\mathrm{sem}}^{ij}\) is deterministic,
\(\kappa_{\mathrm{sem}}^{ij}(\hat s_o\mid s_o)
  \in\{0,1\}\) for all \(s_o,\hat s_o\).
For \(\Phi_{\Atom}\):
\(\pi(a)=\kappa_{\mathrm{sem}}^{ij}(a\mid a)=1\) iff
\(f(a)=a\), and this must hold for all \(a\in A^{(i)}\)
(which requires \(a\in S_O^{(j)}\), i.e.,
\(A_{-}^{ij}=\varnothing\)).
For \(\Psi_{+}\):
\(p_{+}(s_o)=\mathbf{1}[f(s_o)\in S_{+}^{ij}]\), which
is~\(0\) for all \(s_o\) iff
\(f(S_O^{(i)})\subseteq S_{\cap}^{ij}\), and~\(1\) for some
\(s_o\) otherwise; the maximum is therefore \(0\) or~\(1\).
\end{proof}

\begin{remark}[Kernel dependence of probabilistic indices]
\label{rem:kernel-dependence}
Unlike the set-level invariants
\(\rho_{\Atom}\) and \(\mathsf{F}_{\Cn}\), which depend only
on the knowledge-base pair
\((S_O^{(i)},S_O^{(j)})\), the probabilistic indices
\(\Phi_{\Atom}\) and \(\Psi_{+}\) depend on the channel
kernel \(\kappa_{\mathrm{sem}}^{ij}\) and hence on the
specific encoder, physical channel, and decoder.
The overlap decomposition constrains the \emph{range} of these
indices---\(\Phi_{\Atom}=0\) whenever
\(A_{-}^{ij}\neq\varnothing\), and
\(\Psi_{+}=0\) whenever
\(S_{+}^{ij}=\varnothing\)---but their precise values within
the feasible range are determined by the kernel.
\end{remark}


\begin{proposition}[Structural quality indices from overlap]
\label{prop:overlap-quality}
Let \(\mathfrak C^{ij}\) be a heterogeneous semantic channel
with kernel
\(\kappa_{\mathrm{sem}}^{ij}
  :S_O^{(i)}\rightsquigarrow S_O^{(j)}\)
and let \(A=A^{(i)}\).
\begin{enumerate}[label=\textup{(\roman*)}]
  \item \emph{Fidelity concentration under core-preserving
        overlap:}\;
        If \(A_{-}^{ij}=\varnothing\) and
        \(S_{+,n}^{ij}=\varnothing\), then by
        Corollary~\textup{\ref{cor:fidelity-core-concentration}},
        \begin{equation}\label{eq:F-overlap}
          \mathsf{F}(\mathfrak C^{ij})
          \;=\;
          1-\max_{a\in A^{(i)}}
            \bar d_{\Cn}(a\mid\mathfrak C^{ij}).
        \end{equation}
        That is, the worst-case closure distortion is attained
        at a sender core element; redundant states contribute
        zero.
  \item \emph{Depth expansion under vocabulary match:}\;
        If \(S_{-}^{ij}=S_{+}^{ij}=\varnothing\)
        \textup{(}i.e., \(S_O^{(i)}=S_O^{(j)}\)\textup{)},
        then every reachable \(\hat s_o\in S_O^{(j)}=S_O^{(i)}\)
        lies in~\(\Cn(A)\) and the depth distortion reduces to
        its first branch~\eqref{eq:d-Dd}.
        In particular,
        \(\mathsf{E}(\mathfrak C^{ij})=0\) iff the kernel
        preserves the derivation depth of every state
        \textup{(Proposition~\ref{prop:structural-indices}(iii))}.
\end{enumerate}
\end{proposition}

\begin{proof}
\textup{(i)}:\
By Corollary~\ref{cor:overlap-sufficient},
\(\Cn(S_O^{(i)})=\Cn(S_O^{(j)})\) and
\(A^{(i)}\cap S_O^{-}=A_{-}^{ij}=\varnothing\).
Proposition~\ref{prop:noise-distortion-bounds}(i) gives
\(d_{\Cn}(s_o,\hat s_o\mid S_O^{(i)})=0\) for every
\(s_o\in S_O^{(i)}\setminus A\) and
\(\hat s_o\in S_O^{(j)}\).
Hence
\(\max_{s_o\in S_O^{(i)}}
  \bar d_{\Cn}(s_o\mid\mathfrak C^{ij})
  =\max_{a\in A}
  \bar d_{\Cn}(a\mid\mathfrak C^{ij})\),
and the conclusion follows from
Definition~\ref{def:fidelity-index}.

\smallskip\noindent
\textup{(ii)}:\
When \(S_O^{(i)}=S_O^{(j)}\), the noise pair is trivial and
every \(\hat s_o\in S_O^{(j)}=S_O^{(i)}\) satisfies
\(\hat s_o\in\Cn(A)\)
\textup{(Proposition~\ref{prop:atom-core-correct}(iv))}.
The claim follows from
Proposition~\ref{prop:noise-distortion-bounds}(ii) and
Proposition~\ref{prop:structural-indices}(iii).
\end{proof}


\begin{proposition}[Receiver-side structural comparison from
  overlap]
\label{prop:overlap-structural}
Let \(\mathfrak C^{ij}\) be a heterogeneous semantic channel
with \(\tilde S_O=S_O^{(j)}\).
\begin{enumerate}[label=\textup{(\roman*)}]
  \item \emph{Atomicity shift:}\;
        \(\Delta\mathsf{A}(\mathfrak C^{ij})
          =|A^{(j)}|-|A^{(i)}|\).
  \item \emph{Depth shift:}\;
        \(\Delta\mathsf{D_d}(\mathfrak C^{ij})
          =\mathsf{D_d}(\mathcal I^{(j)})
          -\mathsf{D_d}(\mathcal I^{(i)})\).
  \item \emph{Under core-preserving conditions
        \textup{(}\(A_{-}^{ij}=\varnothing\),
        \(S_{+,n}^{ij}=\varnothing\)\textup{):}}
        \begin{enumerate}[label=\textup{(\alph*)}]
          \item \(\Cn(S_O^{(j)})=\Cn(S_O^{(i)})\)
                \textup{(Corollary~\ref{cor:overlap-sufficient})}.
          \item \(|A^{(j)}|\le|A^{(i)}|\), i.e.,
                \(\Delta\mathsf{A}\le 0\)
                \textup{(Proposition~\ref{prop:structural-comparison}(ii))}.
                The inequality is strict when some element of
                \(A^{(i)}\) becomes redundant in~\(S_O^{(j)}\)
                due to the presence of derivable surplus
                states~\(S_{+,d}^{ij}\).
          \item Equality
                \(\Delta\mathsf{A}=0\) holds when
                \(S_{+,d}^{ij}=\varnothing\) and
                \(S_{-}^{ij}\cap J^{(i)}=S_{-}^{ij}\)
                \textup{(}only non-core states are
                lost\textup{)}, because then
                \(S_O^{(j)}=A^{(i)}\cup
                  (J^{(i)}\setminus S_{-}^{ij})\) and
                \(A^{(i)}\) remains irredundant
                in~\(S_O^{(j)}\).
        \end{enumerate}
  \item \emph{Trivial noise pair
        \textup{(}\(S_O^{(i)}=S_O^{(j)}\)\textup{):}}\;
        \(\Delta\mathsf{A}=0\) and
        \(\Delta\mathsf{D_d}=0\)
        \textup{(Proposition~\ref{prop:structural-comparison}(i))}.
\end{enumerate}
\end{proposition}

\begin{proof}
Parts~\textup{(i)} and~\textup{(ii)} are direct from
Definition~\ref{def:receiver-structural} with
\(\tilde S_O=S_O^{(j)}\).

\smallskip\noindent
\textup{(iii)(a)}:
Corollary~\ref{cor:overlap-sufficient}.

\smallskip\noindent
\textup{(iii)(b)}:
By Proposition~\ref{prop:structural-comparison}(ii), applied
with \(A=A^{(i)}\), \(S_O^{-}=S_{-}^{ij}\),
\(S_O^{+}=S_{+}^{ij}\).
When \(S_{+,d}^{ij}\neq\varnothing\), some
\(d\in S_{+,d}^{ij}\subseteq\Cn(A^{(i)})\) may make a
previously irredundant \(a\in A^{(i)}\) redundant
in~\(S_O^{(j)}\) (if
\(a\in\Cn\bigl((S_O^{(j)}\setminus\{a\})\bigr)\) due to
the presence of~\(d\)).

\smallskip\noindent
\textup{(iii)(c)}:\
When \(S_{+,d}^{ij}=S_{+,n}^{ij}=\varnothing\), we have
\(S_{+}^{ij}=\varnothing\) and
\(S_O^{(j)}\subseteq S_O^{(i)}\).
Since \(A_{-}^{ij}=\varnothing\),
\(A^{(i)}\subseteq S_O^{(j)}\).
Since \(S_{-}^{ij}\subseteq J^{(i)}\),
\(S_O^{(j)}=A^{(i)}\cup(J^{(i)}\setminus S_{-}^{ij})\).

We show \(\Atom(S_O^{(j)})=A^{(i)}\) by a simultaneous induction on
the canonical order of~\(S_O^{(j)}\).
Let \(s_1<s_2<\cdots<s_m\) be the elements of~\(S_O^{(j)}\) in
canonical order, and let \(B_k\) denote the current set after
the irredundantization procedure
\textup{(Definition~\ref{def:atom-so})} has scanned
\(s_1,\ldots,s_k\).

\smallskip
\noindent\emph{Inductive claim.}\;
After scanning \(s_1,\ldots,s_k\):
\textup{(a)}~every \(a\in A^{(i)}\) with \(a\le s_k\) has
survived (remains in~\(B_k\)); and
\textup{(b)}~every \(j'\in J^{(i)}\setminus S_{-}^{ij}\) with
\(j'\le s_k\) has been removed.

The base case \(k=0\) is vacuous.
For the inductive step, suppose the claim holds through~\(s_k\)
and consider~\(s_{k+1}\).

\smallskip
\emph{Case~1: \(s_{k+1}=a\in A^{(i)}\).}\;
By the inductive hypothesis, every element of
\(J^{(i)}\setminus S_{-}^{ij}\) preceding~\(a\) has been
removed, and every element of~\(A^{(i)}\) preceding~\(a\) has
survived.
Hence the current set satisfies
\[
  B_k\setminus\{a\}
  \;=\;
  \bigl(A^{(i)}\cap\{s_1,\ldots,s_k\}\bigr)
  \;\cup\;
  \{s_{k+2},\ldots,s_m\}.
\]
In the irredundantization of~\(S_O^{(i)}\), when \(a\) was
scanned, the current set~\(B^{(i)}\) had precisely the same
structure---surviving core elements before~\(a\) plus all
elements of~\(S_O^{(i)}\) after~\(a\)---because elements of
\(J^{(i)}\) before~\(a\) were likewise removed at their own
scan steps (they are in~\(J^{(i)}\) by definition).
Since \(\{s_{k+2},\ldots,s_m\}\subseteq
\{\text{elements of }S_O^{(i)}\text{ after }a\}\)
(the former is a subset, possibly missing elements
of~\(S_{-}^{ij}\) after~\(a\)),
\(B_k\setminus\{a\}\subseteq B^{(i)}\setminus\{a\}\).
By monotonicity~\textup{(Cn2)},
\(\Cn(B_k\setminus\{a\})\subseteq\Cn(B^{(i)}\setminus\{a\})\).
Since \(a\) survived in the irredundantization of~\(S_O^{(i)}\),
\(a\notin\Cn(B^{(i)}\setminus\{a\})\), hence
\(a\notin\Cn(B_k\setminus\{a\})\), and \(a\) survives.

\smallskip
\emph{Case~2: \(s_{k+1}=j'\in J^{(i)}\setminus S_{-}^{ij}\).}\;
By the inductive hypothesis,
\(A^{(i)}\cap\{s_1,\ldots,s_k\}\subseteq B_k\), and all
elements of~\(A^{(i)}\) after~\(s_k\) are in~\(B_k\)
(not yet scanned).
Hence \(A^{(i)}\subseteq B_k\setminus\{j'\}\)
(noting \(j'\notin A^{(i)}\)).
By monotonicity,
\(\Cn(A^{(i)})\subseteq\Cn(B_k\setminus\{j'\})\).
Since
\(j'\in J^{(i)}\subseteq S_O^{(i)}\subseteq\Cn(A^{(i)})\)
\textup{(Proposition~\ref{prop:atom-core-correct}(iv))},
\(j'\in\Cn(B_k\setminus\{j'\})\), and \(j'\) is removed.

\smallskip
By induction, the output of the irredundantization
of~\(S_O^{(j)}\) retains exactly~\(A^{(i)}\).
For the reverse inclusion
\(\Atom(S_O^{(j)})\subseteq A^{(i)}\):
suppose for contradiction that
\(b\in\Atom(S_O^{(j)})\setminus A^{(i)}\).
Then \(b\in J^{(i)}\setminus S_{-}^{ij}\), but the induction
shows that every such element is removed---a contradiction.
Hence \(\Atom(S_O^{(j)})=A^{(i)}\) and
\(\Delta\mathsf{A}=0\).

\smallskip\noindent
\textup{(iv)}:
Immediate from
Proposition~\ref{prop:structural-comparison}(i).
\end{proof}


\begin{corollary}[Heterogeneous semantic Fano bound]
\label{cor:heterogeneous-fano}
Let \(P_O\in\Delta(S_O^{(i)})\) be full-support and let
\(\epsilon:=\bar d_H(\mathfrak C^{ij},P_O)\).
Then
\begin{equation}\label{eq:fano-hetero}
  I_{\mathrm{sem}}^{ij}(P_O,\mathfrak C^{ij})
  \;\ge\;
  H(\mathsf{S}_o)
  -h_b(\epsilon)
  -\epsilon\log\bigl(|S_O^{(j)}|-1\bigr),
\end{equation}
where \(h_b\) is the binary entropy.
Moreover, by
Proposition~\textup{\ref{prop:noise-pair-indices}(v)},
\begin{equation}\label{eq:eps-hetero}
  \epsilon
  \;\le\;
  1-\Phi_{\Atom}(\mathfrak C^{ij})\cdot P_O(A^{(i)}).
\end{equation}
Consequently, when \(A_{-}^{ij}=\varnothing\) and
\(\Phi_{\Atom}(\mathfrak C^{ij})\) is close to~\(1\)
\textup{(}e.g., because the carrier channel is
reliable\textup{)}, the right-hand side
of~\eqref{eq:fano-hetero} is close to
\(H(\mathsf{S}_o)\), i.e., nearly all source entropy is
transmitted.
When \(A_{-}^{ij}\neq\varnothing\),
\(\Phi_{\Atom}=0\) by~\eqref{eq:Phi-overlap} and the bound
yields only the trivial lower bound
\(I_{\mathrm{sem}}^{ij}\ge 0\).
\end{corollary}

\begin{proof}
Equation~\eqref{eq:fano-hetero} is
Theorem~\ref{thm:semantic-fano}(i) applied with
\(\tilde S_O=S_O^{(j)}\).
Equation~\eqref{eq:eps-hetero} is
Proposition~\ref{prop:noise-pair-indices}(v).
When \(A_{-}^{ij}\neq\varnothing\),
\(\Phi_{\Atom}=0\) gives \(\epsilon\le 1\), so
\(h_b(\epsilon)+\epsilon\log(|S_O^{(j)}|-1)\le\log|S_O^{(j)}|\)
and the lower bound cannot exceed zero in a nontrivial way.
\end{proof}


\begin{proposition}[Semantic capacity from overlap]
\label{prop:overlap-capacity}
Let \(W_{ij}:S_C\rightsquigarrow\hat S_C\) be the carrier
channel kernel for the pair~\((i,j)\).
Under the full enabling assumption
\textup{(Assumption~\ref{assump:full-enabling-hetero})} and
the carrier size conditions of
Remark~\textup{\ref{rem:carrier-size}}:
\begin{enumerate}[label=\textup{(\roman*)}]
  \item \emph{Data processing chain:}\;
        \(I_{\mathrm{sem}}^{ij}
          \le C_{\mathrm{sem}}^{ij}(W_{ij})
          \le C(W_{ij})\).
  \item \emph{Capacity equality under full enabling:}\;
        \(C_{\mathrm{sem}}^{ij}(W_{ij})=C(W_{ij})\).
  \item \emph{Source entropy bound:}\;
        \(I_{\mathrm{sem}}^{ij}
          \le\min\bigl(\log|S_O^{(i)}|,\;
                       \log|S_O^{(j)}|\bigr)\).
\end{enumerate}
\end{proposition}

\begin{proof}
Part~\textup{(i)} is
Theorem~\ref{thm:data-processing}(i).
Part~\textup{(ii)} follows from
Theorem~\ref{thm:data-processing}(iii): the full enabling
assumption ensures conditions~(a)--(b) of that theorem, and
the carrier size conditions ensure
\(|S_O^{(i)}|\le|S_C|\) and
\(|S_O^{(j)}|\le|\hat S_C|\).
Part~\textup{(iii)} follows from
\(I(\mathsf{S}_o;\hat{\mathsf{S}}_o)
  \le\min(H(\mathsf{S}_o),\,H(\hat{\mathsf{S}}_o))\).
\end{proof}


\begin{remark}[Summary: invariant families and overlap
  dependence]
\label{rem:invariant-overlap-summary}
Table~\ref{tab:invariant-overlap} collects the six invariant
families of Theorem~\ref{thm:invariant-summary}, their
expressions in the overlap language, and the data they depend
on beyond the knowledge-base pair.

The set-level invariants (\(\rho_{\Atom}\),
\(\mathsf{F}_{\Cn}\)) and the structural comparison indices
(\(\Delta\mathsf{A}\), \(\Delta\mathsf{D_d}\)) are
\emph{fully determined} by the knowledge-base pair
\((S_O^{(i)},S_O^{(j)})\) and the proof
system~\(\mathsf{PS}\); they do not depend on the channel
kernel.
The noise-pair indices (\(\Phi_{\Atom}\), \(\Psi_{+}\)) and
the structural quality indices (\(\mathsf{F}\),
\(\mathsf{E}\)) are constrained by the overlap structure
(e.g., \(\Phi_{\Atom}=0\) whenever
\(A_{-}^{ij}\neq\varnothing\)) but additionally depend on the
channel kernel.
The information-theoretic invariants
(\(I_{\mathrm{sem}}\), \(C_{\mathrm{sem}}\), \(C(W)\)) depend
on the kernel and the carrier channel; the overlap structure
enters only through the enabling constraints and the output
alphabet size~\(|S_O^{(j)}|\).
\end{remark}

\begin{table}[t]
\centering
\caption{Semantic channel invariants expressed in the overlap
  language of
  Definition~\textup{\ref{def:overlap-decomposition}}.
  ``KB only'' indicates dependence on the knowledge-base pair
  alone; ``+kernel'' indicates additional dependence on the
  channel kernel
  \(\kappa_{\mathrm{sem}}^{ij}\).}
\label{tab:invariant-overlap}
\renewcommand{\arraystretch}{1.25}
\begin{tabularx}{\columnwidth}{@{}l l L@{}}
\toprule
\textbf{Family} & \textbf{Invariant} &
  \textbf{Overlap expression / constraint} \\
\midrule
\multirow{2}{*}{\makecell[l]{I.\;Source\\[-2pt]struct.}}
  & \(\mathsf{A}\) &
    \(|A^{(i)}|\) (sender core size) \\
  & \(\mathsf{D_d}\) &
    \(\max_{q\in S_O^{(i)}}\Dd(q\mid A^{(i)})\) \\
\midrule
\multirow{2}{*}{\makecell[l]{II.\;Set-level\\[-2pt](KB only)}}
  & \(\rho_{\Atom}\) &
    \((|A^{(i)}|-|A_{-}^{ij}|)/|A^{(i)}|\)
    \eqref{eq:rho-overlap} \\
  & \(\mathsf{F}_{\Cn}\) &
    \(=1\) iff \textup{(F1)}+\textup{(F2)}
    (Prop.~\ref{prop:overlap-fidelity}) \\
\midrule
\multirow{2}{*}{\makecell[l]{III.\;Noise-pair\\[-2pt](+kernel)}}
  & \(\Phi_{\Atom}\) &
    \(=0\) if \(A_{-}^{ij}\neq\varnothing\); else
    \(\min_a\kappa_{\mathrm{sem}}^{ij}(a\mid a)\) \\
  & \(\Psi_{+}\) &
    \(=0\) if \(S_{+}^{ij}=\varnothing\); else
    \(\max_{s_o}\sum_{\hat s_o\in S_{+}^{ij}}
      \kappa_{\mathrm{sem}}^{ij}(\hat s_o\mid s_o)\) \\
\midrule
\multirow{2}{*}{\makecell[l]{IV.\;Quality\\[-2pt](+kernel)}}
  & \(\mathsf{F}\) &
    Under (F1\('\))+(F2):
    eq.~\eqref{eq:F-overlap}; worst case on core only \\
  & \(\mathsf{E}\) &
    \(\max_{s_o}\bar d_{\Dd}(s_o\mid\mathfrak C^{ij})\)
    (kernel-dep.) \\
\midrule
\multirow{2}{*}{\makecell[l]{V.\;Receiver\\[-2pt]comp.\;(KB only)}}
  & \(\Delta\mathsf{A}\) &
    \(|A^{(j)}|-|A^{(i)}|\); \(\le 0\) under
    (F1\('\))+(F2) \\
  & \(\Delta\mathsf{D_d}\) &
    \(\mathsf{D_d}(\mathcal I^{(j)})
      -\mathsf{D_d}(\mathcal I^{(i)})\) \\
\midrule
\multirow{3}{*}{\makecell[l]{VI.\;Info-\\[-2pt]theoretic\\[-2pt]
  (+kernel,\\\(W\))}}
  & \(C(W_{ij})\) &
    Carrier-channel property only \\
  & \(C_{\mathrm{sem}}^{ij}\) &
    \(=C(W_{ij})\) under full enabling
    (Prop.~\ref{prop:overlap-capacity}(ii)) \\
  & \(I_{\mathrm{sem}}^{ij}\) &
    Bounded by \eqref{eq:fano-hetero}; depends on kernel
    and~\(P_O\) \\
\bottomrule
\end{tabularx}
\end{table}

\begin{remark}[Diagnostic use of the overlap--invariant
  correspondence]
\label{rem:diagnostic-use}
The correspondence in
Table~\ref{tab:invariant-overlap} enables a two-stage
diagnostic workflow for heterogeneous semantic channel design:

\smallskip\noindent
\emph{Stage~1 (offline, kernel-free).}\;
Given the knowledge-base pair \((S_O^{(i)},S_O^{(j)})\),
compute the overlap decomposition
\textup{(Definition~\ref{def:overlap-decomposition})} and
evaluate the binary feasibility tests
\(A_{-}^{ij}=\varnothing\) and
\(S_{+,n}^{ij}=\varnothing\).
If both hold, set-level closure fidelity
\(\mathsf{F}_{\Cn}=1\) is guaranteed
\textup{(Corollary~\ref{cor:overlap-sufficient})}, and the
two-layer code of Section~\ref{subsec:app-results} is directly
applicable.
If \(S_{+,n}^{ij}\neq\varnothing\) \textup{(}condition
\textup{(F2)} of Proposition~\ref{prop:overlap-fidelity}
fails\textup{)}, then \(\mathsf{F}_{\Cn}<1\) is unavoidable.
If \(A_{-}^{ij}\neq\varnothing\) but
\(A^{(i)}\subseteq\Cn(S_O^{(j)})\)
\textup{(}the weak condition \textup{(F1)} holds despite
\textup{(F1\({}'\))} failing\textup{)}, then
\(\mathsf{F}_{\Cn}=1\) remains possible but the standard
two-layer code does not apply; see
Remark~\textup{\ref{rem:strong-weak-coverage}} and
Remark~\textup{\ref{rem:weak-F1-coding}} in
Section~\textup{\ref{subsec:app-results}}.
If \(A^{(i)}\not\subseteq\Cn(S_O^{(j)})\), then
\(\mathsf{F}_{\Cn}<1\) is unavoidable
\textup{(Proposition~\ref{prop:overlap-fidelity})}, and the
designer must augment the receiver's vocabulary or accept
degraded fidelity (cf.\ the minimum-vocabulary criterion of
Proposition~\ref{prop:min-vocabulary}).

\smallskip\noindent
\emph{Stage~2 (online, kernel-dependent).}\;
Given a specific channel kernel
\(\kappa_{\mathrm{sem}}^{ij}\), compute
\(\Phi_{\Atom}\), \(\Psi_{+}\), \(\mathsf{F}\), and
\(\mathsf{E}\) to quantify the probabilistic quality of the
channel.
The Fano bound
\textup{(Corollary~\ref{cor:heterogeneous-fano})} then
translates the probabilistic indices into a lower bound on
\(I_{\mathrm{sem}}^{ij}\), completing the link between
knowledge-base structure and information throughput.
\end{remark}

\subsection{Main Results: Heterogeneous Compression and Broadcast}
\label{subsec:app-results}

This subsection derives the principal new results of the paper.
Part~1 addresses the pairwise unicast scenario
(Definition~\ref{def:pairwise-scenario}):
closure-reliable achievability, a heterogeneous deductive
compression theorem, an impossibility result when core coverage
fails, and a vocabulary design criterion.
Part~2 extends the theory to the broadcast scenario
(Definition~\ref{def:broadcast-scenario}).
Throughout, the standing assumptions of
Assumption~\ref{assump:standing-hetero} are in force.


\subsubsection*{Part 1: Pairwise Heterogeneous Communication}

\begin{theorem}[Closure reliability for a heterogeneous pair]
\label{thm:heterogeneous-closure}
Let \((i,j)\) be a sender--receiver pair with heterogeneous
semantic channel \(\mathfrak C^{ij}\)
\textup{(Definition~\ref{def:heterogeneous-channel})},
carrier channel kernel
\(W_{ij}:S_C\rightsquigarrow\hat S_C\) with \(C(W_{ij})>0\),
and overlap decomposition
\textup{(Definition~\ref{def:overlap-decomposition})}.
Assume:
\begin{enumerate}[label=\textup{(H\arabic*)}]
  \item \(A_{-}^{ij}=\varnothing\)
        \textup{(}the sender's irredundant core is contained in
        the receiver's vocabulary:
        \(A^{(i)}\subseteq S_O^{(j)}\)\textup{)};
  \item \(S_{+,n}^{ij}=\varnothing\)
        \textup{(}all surplus states in the receiver's vocabulary
        are derivable from the sender's knowledge base:
        \(S_{+}^{ij}\subseteq\Cn(S_O^{(i)})\)\textup{)}.
\end{enumerate}
Then:
\begin{enumerate}[label=\textup{(\roman*)}]
  \item \emph{Set-level closure fidelity:}\;
        \(\mathsf{F}_{\Cn}(S_O^{(i)},S_O^{(j)})=1\)
        \textup{(Corollary~\ref{cor:overlap-sufficient})}.
  \item \emph{Achievability:}\;
        There exists a sequence of
        \((n,|S_O^{(i)}|)\) semantic block codes
        \textup{(Definition~\ref{def:semantic-codebook})}
        with message set \(\mathcal M=S_O^{(i)}\),
        encoding into \(S_C^n\), decoding into
        \(\hat S_O=S_O^{(j)}\), and
        \(P_{e,\Cn}^{(n)}\to 0\) as \(n\to\infty\), provided
        \begin{equation}\label{eq:hetero-rate-condition}
          \frac{\log|A^{(i)}|}{n}\;<\;C(W_{ij}).
        \end{equation}
  \item \emph{Converse
        \textup{(}under
        Assumption~\textup{\ref{assump:core-disjoint}}
        for \(S_O^{(i)}\) with output space
        \(S_O^{(j)}\)\textup{):}}\;
        Any \((n,|S_O^{(i)}|)\) code with
        \(P_{e,\Cn}^{(n)}\le\epsilon\) satisfies
        \begin{equation}\label{eq:hetero-converse}
          \log|A^{(i)}|
          \;\le\;
          \frac{nC(W_{ij})+1}{1-\epsilon}.
        \end{equation}
\end{enumerate}
\end{theorem}

\begin{proof}
\textup{(i)}:\
Immediate from Corollary~\ref{cor:overlap-sufficient}.

\smallskip\noindent
\textup{(ii)}:\
The code is constructed in two layers, adapting
Theorem~\ref{thm:achievability}(ii) to the heterogeneous
output alphabet \(\hat S_O=S_O^{(j)}\).

\emph{Layer~1 (core code).}\;
Since \(\log|A^{(i)}|/n<C(W_{ij})\), the classical channel
coding theorem
\cite{shannon1948mathematical,cover2006elements} yields an
\((n,|A^{(i)}|)\) block code
\((f_n^A,g_n^A)\) for~\(W_{ij}\) with message set
\(A^{(i)}\) and
\(P_e^{(n)}(A^{(i)})\to 0\).
By~\textup{(H1)},
\(A^{(i)}\subseteq S_O^{(j)}=\hat S_O\), so the decoder
can output elements of~\(A^{(i)}\).

\emph{Layer~2 (redundant extension).}\;
Fix an arbitrary \(a_0\in A^{(i)}\).
For each redundant state
\(j\in J^{(i)}=S_O^{(i)}\setminus A^{(i)}\), set
\(f_n(j):=f_n^A(a_0)\).
The decoder first applies~\(g_n^A\) to recover
\(\hat a\in A^{(i)}\) (or an incorrect element in the error
event), and outputs~\(\hat a\).

\emph{Closure analysis.}\;
For \(m\in A^{(i)}\): if the core code decodes correctly
(\(\hat a=m\)), then
\(d_{\Cn}(m,m\mid S_O^{(i)})=0\).
Error probability:
\(P_e^{(n)}(A^{(i)})\to 0\).

For \(m=j\in J^{(i)}\): the decoder outputs some
\(\hat a\in A^{(i)}\subseteq S_O^{(i)}\subseteq
\Cn(S_O^{(i)})\).
Since \(j\) is redundant in~\(S_O^{(i)}\),
\(\Cn(S_O^{(i)}\setminus\{j\})=\Cn(S_O^{(i)})\).
Because
\(\hat a\in\Cn(S_O^{(i)})=\Cn(S_O^{(i)}\setminus\{j\})\),
monotonicity and idempotence of~\(\Cn\) give
\(\Cn\bigl((S_O^{(i)}\setminus\{j\})\cup\{\hat a\}\bigr)
  =\Cn(S_O^{(i)}\setminus\{j\})=\Cn(S_O^{(i)})\),
so \(d_{\Cn}(j,\hat a\mid S_O^{(i)})=0\).
The closure error probability for redundant messages
is~\emph{zero} for all~\(n\).

Combining:
\(P_{e,\Cn}^{(n)}\le P_e^{(n)}(A^{(i)})\to 0\).

\smallskip\noindent
\textup{(iii)}:\
The argument is identical to the proof of
Theorem~\ref{thm:converse}(iii), with
\(S_O\) replaced by \(S_O^{(i)}\),
\(\hat S_O\) replaced by \(S_O^{(j)}\), and
\(A=\Atom(S_O^{(i)})=A^{(i)}\).
Under Assumption~\ref{assump:core-disjoint} (applied to
\(A^{(i)}\) with acceptable sets in~\(S_O^{(j)}\)),
the pairwise-disjoint decoding regions and the Fano argument
yield~\eqref{eq:hetero-converse}.
\end{proof}

\begin{remark}[Role of conditions \textup{(H1)} and
  \textup{(H2)}]
\label{rem:H1-H2-roles}
Condition~\textup{(H1)} is used only in the achievability
proof to ensure that the decoder can output core elements
(\(A^{(i)}\subseteq S_O^{(j)}\)).
Condition~\textup{(H2)} is used only to establish
set-level closure fidelity
\(\mathsf{F}_{\Cn}=1\) in part~\textup{(i)};
it does \emph{not} enter the achievability or converse
proofs, which depend only on the closure distortion
\(d_{\Cn}(\cdot,\cdot\mid S_O^{(i)})\) measured relative
to the \emph{sender's} knowledge base.
Hence the coding-theoretic conclusions~\textup{(ii)}
and~\textup{(iii)} hold under~\textup{(H1)} alone.
Condition~\textup{(H2)} provides the additional guarantee
that the receiver's overall knowledge base generates the
same deductive closure as the sender's.
\end{remark}

\begin{remark}[Weak core coverage and alternative decoding]
\label{rem:weak-F1-coding}
Theorem~\ref{thm:heterogeneous-closure} uses the strong core
coverage condition
\textup{(H1)}: \(A^{(i)}\subseteq S_O^{(j)}\).
As noted in Remark~\ref{rem:strong-weak-coverage}, set-level
closure fidelity \(\mathsf{F}_{\Cn}=1\) requires only the
weaker condition~\textup{(F1)}:
\(A^{(i)}\subseteq\Cn(S_O^{(j)})\).
When \textup{(F1)} holds but \textup{(H1)} fails
(i.e., some core element \(a\in A_{-}^{ij}\) is derivable from
\(S_O^{(j)}\) but not literally present), the two-layer code
cannot directly output~\(a\).
A modified decoder could instead output a
\emph{proxy element}
\(\hat a\in S_O^{(j)}\) satisfying
\(d_{\Cn}(a,\hat a\mid S_O^{(i)})=0\)---i.e., a state in
the receiver's vocabulary whose substitution for~\(a\) preserves
the sender's deductive closure.
Such a proxy exists whenever
\(a\in\Cn(S_O^{(j)})\), but identifying it requires
knowledge of the sender's closure structure at the decoder,
making the code design more involved.
A complete treatment of proxy-based decoding is deferred to
future work; the results of this section focus on the
operationally simpler setting where \textup{(H1)} holds.
The proxy-based decoding strategy shares conceptual
affinity with the inverse contextual reasoning of
Seo et~al.~\cite{seo2023bayesian}, who address the problem
of inferring a sender's communication context from noisy
observations using Bayesian methods.
\end{remark}

\begin{theorem}[Heterogeneous deductive compression]
\label{thm:heterogeneous-compression}
Under the hypotheses of
Theorem~\textup{\ref{thm:heterogeneous-closure}}, the
minimum blocklength for closure-reliable communication
of the full knowledge base \(S_O^{(i)}\) to agent~\(j\)
satisfies, for sufficiently small \(\epsilon>0\):
\begin{enumerate}[label=\textup{(\roman*)}]
  \item \emph{Closure blocklength:}\;
        \begin{align}\label{eq:hetero-blocklength-closure}
          \frac{(1-\epsilon)\log|A^{(i)}|-1}{C(W_{ij})}
          \;\le\;
          &n^*\bigl(S_O^{(i)},W_{ij},P_{e,\Cn},\epsilon\bigr) \nonumber \\
          \;\le\;
          &\left\lceil
            \frac{\log|A^{(i)}|}{C(W_{ij})-\delta(\epsilon)}
          \right\rceil,
        \end{align}
        where \(\delta(\epsilon)\to 0\) as \(\epsilon\to 0\).
  \item \emph{Hamming baseline:}\;
        Under the additional hypothesis
        \(S_O^{(i)}\subseteq S_O^{(j)}\)
        \textup{(}i.e., \(S_{-}^{ij}=\varnothing\),
        which strengthens \textup{(H1)} to full vocabulary
        containment\textup{)},
        \begin{equation}\label{eq:hetero-blocklength-Hamming}
          n^*\bigl(S_O^{(i)},W_{ij},P_e,\epsilon\bigr)
          \;\ge\;
          \frac{(1-\epsilon)\log|S_O^{(i)}|-1}{C(W_{ij})}.
        \end{equation}
  \item \emph{Deductive compression ratio:}\;
        When both bounds apply,
        \begin{equation}\label{eq:hetero-compression-ratio}
          \frac{n^*(P_{e,\Cn})}{n^*(P_e)}
          \;\approx\;
          \frac{\log|A^{(i)}|}{\log|S_O^{(i)}|}\,,
        \end{equation}
        identical to the homogeneous ratio of
        Corollary~\textup{\ref{cor:min-blocklength}}.
\end{enumerate}
\end{theorem}

\begin{proof}
Part~\textup{(i)} combines
Theorem~\ref{thm:heterogeneous-closure}(ii)
(upper bound) and~(iii) (lower bound).
Part~\textup{(ii)} is
Theorem~\ref{thm:converse}(i) applied with
\(M=|S_O^{(i)}|\); the condition
\(S_O^{(i)}\subseteq S_O^{(j)}\) ensures that the
Hamming criterion is meaningful (each sent state has a valid
identity reconstruction in the receiver's vocabulary).
Part~\textup{(iii)} follows by dividing the bounds.
\end{proof}

\begin{remark}[Heterogeneity does not degrade the compression
  ratio]
\label{rem:compression-invariance}
The deductive compression
ratio~\eqref{eq:hetero-compression-ratio} depends only on
the sender's knowledge-base structure
(\(|A^{(i)}|\) vs.\ \(|S_O^{(i)}|\)) and not on the
receiver's vocabulary~\(S_O^{(j)}\), provided the core
coverage condition~\textup{(H1)} holds.
This invariance is a consequence of the two-layer code
structure: the core sub-code operates identically regardless
of the receiver's surplus states, and the redundant extension
incurs zero closure distortion by the algebraic properties
of~\(\Cn\).

The Hamming baseline~\eqref{eq:hetero-blocklength-Hamming}
does, however, depend on the receiver's vocabulary: it
requires \(S_{-}^{ij}=\varnothing\) (full vocabulary
containment), a strictly stronger condition than~\textup{(H1)}.
When \(S_{-}^{ij}\neq\varnothing\), perfect Hamming
reconstruction is impossible (some sent states have no
counterpart in the receiver's vocabulary), while closure
reliability may still be achievable under~\textup{(H1)}.
This gap illustrates the advantage of semantic fidelity
criteria over symbol-level criteria in heterogeneous settings.
\end{remark}

\begin{corollary}[Impossibility under core loss]
\label{cor:heterogeneous-impossibility}
Let \((i,j)\) be a sender--receiver pair.
\begin{enumerate}[label=\textup{(\roman*)}]
  \item \emph{Set-level impossibility:}\;
        If condition~\textup{(F1)} of
        Proposition~\textup{\ref{prop:overlap-fidelity}} fails,
        i.e.,
        \(A^{(i)}\not\subseteq\Cn\bigl(S_O^{(j)}\bigr)\),
        then
        \(\mathsf{F}_{\Cn}(S_O^{(i)},S_O^{(j)})<1\).
        This is a property of the knowledge-base pair,
        independent of the channel, the blocklength, and the
        coding strategy.
  \item \emph{Quantitative bound:}\;
        Under the hypothesis of~\textup{(i)}, the closure
        fidelity satisfies
        \begin{equation}\label{eq:Fcn-bound}
          \mathsf{F}_{\Cn}\bigl(S_O^{(i)},S_O^{(j)}\bigr)
          \;=\;
          \frac{|\Cn(S_O^{(i)})\cap\Cn(S_O^{(j)})|}
               {|\Cn(S_O^{(i)})\cup\Cn(S_O^{(j)})|}\;<\;1.
        \end{equation}
  \item \emph{Core preservation ratio:}\;
        If \(A_{-}^{ij}\neq\varnothing\), then
        \(\rho_{\Atom}(S_O^{(i)},S_O^{(j)})
          =1-|A_{-}^{ij}|/|A^{(i)}|<1\)
        \textup{(Proposition~\ref{prop:overlap-noise}(ii))}.
\end{enumerate}
Similarly, if condition~\textup{(F2)} fails
\textup{(}\(S_{+,n}^{ij}\neq\varnothing\)\textup{)}, then
\(\mathsf{F}_{\Cn}<1\) regardless of any coding strategy.
\end{corollary}

\begin{proof}
Part~\textup{(i)} is the contrapositive of
Proposition~\ref{prop:overlap-fidelity}.
Part~\textup{(ii)} is the definition of
\(\mathsf{F}_{\Cn}\)
(Definition~\ref{def:closure-fidelity}); the strict
inequality follows from~\textup{(i)}.
Part~\textup{(iii)} is
Proposition~\ref{prop:overlap-noise}(ii).
The final claim follows from
Proposition~\ref{prop:overlap-fidelity} (necessity
of~\textup{(F2)}).
\end{proof}

\begin{proposition}[Minimum receiver vocabulary for
  closure-reliable communication]
\label{prop:min-vocabulary}
Given a sender knowledge base \(S_O^{(i)}\), the
minimum-cardinality subset
\(V\subseteq S_O^{(i)}\) serving as a receiver vocabulary
(\(\hat S_O=V\)) that simultaneously achieves:
\begin{enumerate}[label=\textup{(\alph*)}]
  \item the two-layer code of
        Theorem~\textup{\ref{thm:heterogeneous-closure}(ii)}
        achieves \(P_{e,\Cn}^{(n)}\to 0\), and
  \item \(\mathsf{F}_{\Cn}(S_O^{(i)},V)=1\),
\end{enumerate}
is \(V^*=\Atom(S_O^{(i)})=A^{(i)}\), with
\(|V^*|=\mathsf{A}(\mathcal I^{(i)})\).
\end{proposition}

\begin{proof}
\emph{Sufficiency.}\;
Set \(V=A^{(i)}\).
Since \(A^{(i)}\subseteq S_O^{(i)}\), the overlap
with sender \(i\) and ``receiver'' \(V\) gives
\(S_{+}^{ij}=V\setminus S_O^{(i)}=\varnothing\) and
\(A_{-}^{ij}=A^{(i)}\setminus V=\varnothing\)
(condition~\textup{(H1)}).
Condition~\textup{(H2)} holds trivially since
\(S_{+}^{ij}=\varnothing\).
Closure fidelity:
\(\Cn(A^{(i)})=\Cn(S_O^{(i)})\)
\textup{(Proposition~\ref{prop:atom-core-correct}(i))},
giving \(\mathsf{F}_{\Cn}(S_O^{(i)},A^{(i)})=1\).
The two-layer code of
Theorem~\ref{thm:heterogeneous-closure}(ii) applies
with \(\hat S_O=V=A^{(i)}\).

\smallskip\noindent
\emph{Minimality.}\;
Let \(V\subseteq S_O^{(i)}\) satisfy both
conditions~\textup{(a)} and~\textup{(b)}.
Condition~\textup{(a)} requires the two-layer code
of Theorem~\ref{thm:heterogeneous-closure}(ii) to
succeed.
That code's decoder outputs elements of~\(A^{(i)}\),
so the output alphabet \(\hat S_O=V\) must contain
every core element: \(A^{(i)}\subseteq V\).
Hence \(|V|\ge|A^{(i)}|\).
Since \(V=A^{(i)}\) achieves this bound, it is
minimal.
\end{proof}

\begin{remark}[Vocabulary design rule]
\label{rem:vocabulary-design}
Proposition~\ref{prop:min-vocabulary} yields a principled
vocabulary-selection rule for receiver design:
\emph{the receiver need store only the sender's irredundant
core}.
All remaining semantic states (the sender's stored
shortcuts~\(J^{(i)}\)) can be reconstructed by the receiver's
inference engine via \(\Cn(A^{(i)})\).
The channel-use cost of this strategy is
\(n^*\approx\log|A^{(i)}|/C(W_{ij})\), the minimum
achievable under closure reliability.

When the receiver already maintains a richer vocabulary
\(S_O^{(j)}\supsetneq A^{(i)}\), the additional states are
harmless provided \(S_{+,n}^{ij}=\varnothing\)
(condition~\textup{(H2)}); they do not increase the
blocklength.
When some surplus states are non-derivable
(\(S_{+,n}^{ij}\neq\varnothing\)), set-level closure fidelity
drops below~\(1\)
(Corollary~\ref{cor:heterogeneous-impossibility}), but the
coding-theoretic closure reliability may still hold if the
two-layer code is used (since it ignores the surplus entirely;
see Remark~\ref{rem:H1-H2-roles}).
\end{remark}


\subsubsection*{Part 2: Broadcast Extension}

We now extend the pairwise results to the broadcast scenario
of Definition~\ref{def:broadcast-scenario}.
Agent~\(0\) (the sender) communicates its knowledge base
\(S_O^{(0)}\) to \(K\) receivers over a common carrier
channel \(W:S_C\rightsquigarrow\hat S_C\) with
\(C(W)>0\).

\begin{theorem}[Broadcast deductive compression]
\label{thm:broadcast-compression}
Suppose that for every receiver
\(j\in\{1,\ldots,K\}\), the overlap conditions hold:
\begin{enumerate}[label=\textup{(BH\arabic*)}]
  \item \(A_{-}^{0j}=\varnothing\)
        \textup{(}the sender's core is contained in every
        receiver's vocabulary:
        \(A^{(0)}\subseteq S_O^{(j)}\)\textup{)};
  \item \(S_{+,n}^{0j}=\varnothing\)
        \textup{(}every receiver's surplus is derivable from
        the sender:
        \(S_{+}^{0j}\subseteq\Cn(S_O^{(0)})\)\textup{)}.
\end{enumerate}
Then:
\begin{enumerate}[label=\textup{(\roman*)}]
  \item \emph{Simultaneous closure fidelity:}\;
        \(\mathsf{F}_{\Cn}(S_O^{(0)},S_O^{(j)})=1\) for
        every \(j\in\{1,\ldots,K\}\).
  \item \emph{Broadcast achievability:}\;
        There exists a \emph{single} sequence of
        \((n,|S_O^{(0)}|)\) semantic block codes
        (with a common encoding function~\(f_n\)) such that
        \(P_{e,\Cn}^{(n,j)}\to 0\) simultaneously for all
        receivers \(j\in\{1,\ldots,K\}\), provided
        \begin{equation}\label{eq:broadcast-rate}
          \frac{\log|A^{(0)}|}{n}\;<\;C(W).
        \end{equation}
  \item \emph{Blocklength independence from \(K\):}\;
        The minimum blocklength for broadcast closure
        reliability is
        \begin{equation}\label{eq:broadcast-blocklength}
          n^*_{\mathrm{bc}}
          \;\approx\;
          \frac{\log|A^{(0)}|}{C(W)},
        \end{equation}
        independent of the number of receivers~\(K\).
  \item \emph{Broadcast converse:}\;
        Under
        Assumption~\textup{\ref{assump:core-disjoint}} for
        \(S_O^{(0)}\) with output space \(S_O^{(j)}\) for
        each~\(j\), any code achieving
        \(\max_j P_{e,\Cn}^{(n,j)}\le\epsilon\) satisfies
        \(\log|A^{(0)}|\le(nC(W)+1)/(1-\epsilon)\).
\end{enumerate}
\end{theorem}

\begin{proof}
\textup{(i)}:\
For each~\(j\), conditions \textup{(BH1)}--\textup{(BH2)}
instantiate \textup{(H1)}--\textup{(H2)} of
Theorem~\ref{thm:heterogeneous-closure}, giving
\(\mathsf{F}_{\Cn}(S_O^{(0)},S_O^{(j)})=1\) by
Corollary~\ref{cor:overlap-sufficient}.

\smallskip\noindent
\textup{(ii)}:\
Construct a single two-layer code as in the proof of
Theorem~\ref{thm:heterogeneous-closure}(ii), with the
\emph{common} core code \((f_n^A,g_n^A)\) for message
set~\(A^{(0)}\).
All \(K\) receivers observe the same channel output
\(\hat S_C^n\) and each independently applies the
same core decoder~\(g_n^A\).
Since \(A^{(0)}\subseteq S_O^{(j)}\) for every~\(j\)
(by~\textup{(BH1)}), the decoded core element
\(\hat a\in A^{(0)}\) is a valid output for every receiver.
The Layer~2 redundant extension and closure analysis are
identical to the pairwise case (using the sender's
closure structure only), so
\(P_{e,\Cn}^{(n,j)}\le P_e^{(n)}(A^{(0)})\to 0\)
simultaneously for all~\(j\).

\smallskip\noindent
\textup{(iii)}:\
The blocklength is determined by the core code,
which has rate \(\log|A^{(0)}|/n\), independent of~\(K\).

\smallskip\noindent
\textup{(iv)}:\
Fix any receiver~\(j\).
Theorem~\ref{thm:heterogeneous-closure}(iii) applied to the
pair \((0,j)\) gives
\(\log|A^{(0)}|\le(nC(W)+1)/(1-\epsilon)\) whenever
\(P_{e,\Cn}^{(n,j)}\le\epsilon\).
Since this must hold for every~\(j\), the bound holds under
\(\max_j P_{e,\Cn}^{(n,j)}\le\epsilon\).
\end{proof}

\begin{proposition}[Broadcast semantic bottleneck]
\label{prop:broadcast-bottleneck}
In the broadcast scenario of
Definition~\textup{\ref{def:broadcast-scenario}}, suppose
there exists a receiver \(j^*\in\{1,\ldots,K\}\) such that
condition~\textup{(F1)} of
Proposition~\textup{\ref{prop:overlap-fidelity}} fails for
the pair \((0,j^*)\):
\[
  A^{(0)}\not\subseteq\Cn\bigl(S_O^{(j^*)}\bigr).
\]
Then:
\begin{enumerate}[label=\textup{(\roman*)}]
  \item \(\mathsf{F}_{\Cn}(S_O^{(0)},S_O^{(j^*)})<1\),
        regardless of the carrier channel~\(W\), the
        blocklength~\(n\), and the encoding/decoding strategy.
  \item Even if the carrier channel is noiseless
        (\(W=\mathrm{id}_{S_C}\)), the closure fidelity at
        receiver~\(j^*\) is bounded by
        \begin{equation}\label{eq:bottleneck-bound}
          \mathsf{F}_{\Cn}\bigl(S_O^{(0)},S_O^{(j^*)}\bigr)
          \;=\;
          \frac{|\Cn(S_O^{(0)})\cap\Cn(S_O^{(j^*)})\!|}
               {|\Cn(S_O^{(0)})\cup\Cn(S_O^{(j^*)})\!|}
          \;<\;1.
        \end{equation}
  \item Receiver~\(j^*\) is a \emph{semantic bottleneck}:
        its performance limitation arises from vocabulary
        mismatch, not from the physical channel.
        The other receivers \(j\neq j^*\) satisfying
        \textup{(BH1)}--\textup{(BH2)} achieve
        \(\mathsf{F}_{\Cn}=1\) and closure reliability
        simultaneously, unaffected by~\(j^*\).
\end{enumerate}
\end{proposition}

\begin{proof}
Parts~\textup{(i)} and~\textup{(ii)} follow from
Corollary~\ref{cor:heterogeneous-impossibility}(i)--(ii)
applied to the pair \((0,j^*)\).
Part~\textup{(iii)}: the common encoding and core code are
shared by all receivers; the failure at~\(j^*\) is due
solely to the mismatch
\(A^{(0)}\not\subseteq\Cn(S_O^{(j^*)})\), which is
independent of the channel.
Receivers satisfying
\textup{(BH1)}--\textup{(BH2)} are handled by
Theorem~\ref{thm:broadcast-compression}.
\end{proof}

\begin{remark}[Semantic bottleneck vs.\ classical channel
  degradation]
\label{rem:bottleneck-vs-classical}
In classical broadcast channel theory
\cite{cover2006elements,csiszar2011information}, the
performance at each receiver is
determined by the physical channel degradation order: the
``weakest'' receiver is the one with the noisiest channel.
The rate region is a function of the channel transition
probabilities alone.

In the semantic broadcast setting, a second axis of
``weakness'' emerges that is invisible to classical theory:
\emph{vocabulary overlap}.
Proposition~\ref{prop:broadcast-bottleneck} shows that even
over a \emph{noiseless} carrier, a receiver with insufficient
vocabulary overlap (\(A^{(0)}\not\subseteq\Cn(S_O^{(j^*)})\))
cannot achieve perfect closure fidelity.
No amount of coding, no increase in blocklength, and no
improvement in the physical channel can overcome this
semantic bottleneck.
This constitutes the sharpest qualitative distinction between
the semantic channel framework and classical Shannon theory:
the former recognizes an irreducible \emph{structural}
limitation that the latter cannot express.

To resolve a semantic bottleneck, the system designer must
augment the vocabulary of receiver~\(j^*\) to satisfy the
core coverage condition---e.g., by pre-loading the sender's
irredundant core \(A^{(0)}\) into~\(S_O^{(j^*)}\)
(Proposition~\ref{prop:min-vocabulary}).
This is a \emph{design-time} action on the knowledge-base
structure, fundamentally different from the \emph{coding-time}
actions (encoder/decoder optimization) that suffice in the
classical setting.
\end{remark}

\begin{remark}[Classical recovery]
\label{rem:broadcast-classical-recovery}
When all receivers share the sender's knowledge base
(\(S_O^{(j)}=S_O^{(0)}\) for all~\(j\)), every noise pair
is trivial, conditions \textup{(BH1)}--\textup{(BH2)} hold
vacuously, and
Theorem~\ref{thm:broadcast-compression} reduces to
\(K\)~independent applications of the homogeneous
achievability result
(Theorem~\ref{thm:achievability}(ii)).
If additionally \(S_O^{(0)}\) is irredundant
(\(A^{(0)}=S_O^{(0)}\)), the deductive compression gain
vanishes and the broadcast blocklength becomes
\(n^*\approx\log|S_O^{(0)}|/C(W)\), recovering the classical
channel coding theorem for a broadcast DMC
\cite{cover2006elements}.
\end{remark}

\begin{remark}[Summary of answers to Q1--Q4]
\label{rem:Q1-Q4-summary}
The results of this subsection answer the four key questions
posed in Section~\ref{subsec:app-problem}:

\textbf{Q1} (closure reliability from overlap):
Proposition~\ref{prop:overlap-fidelity} and
Corollary~\ref{cor:overlap-sufficient} provide the
necessary and sufficient conditions; the operational
achievability under the strong condition~\textup{(H1)} is
Theorem~\ref{thm:heterogeneous-closure}.

\textbf{Q2} (heterogeneous compression):
Theorem~\ref{thm:heterogeneous-compression} establishes
that the deductive compression ratio is invariant under
vocabulary heterogeneity.
Corollary~\ref{cor:heterogeneous-impossibility} characterizes
the impossibility regime.

\textbf{Q3} (invariant diagnosis):
Section~\ref{subsec:app-theory}
(Propositions~\ref{prop:overlap-noise}--\ref{prop:overlap-structural}
and Table~\ref{tab:invariant-overlap})
expresses every invariant family in terms of the overlap
decomposition.

\textbf{Q4} (broadcast bottleneck):
Theorem~\ref{thm:broadcast-compression} shows
blocklength independence from~\(K\) under core coverage;
Proposition~\ref{prop:broadcast-bottleneck} identifies
the semantic bottleneck phenomenon.
\end{remark}

\subsection{Numerical Example: Datalog Knowledge Bases}
\label{subsec:app-example}

This subsection verifies the theoretical results of
Sections~\ref{subsec:app-theory}--\ref{subsec:app-results} on
an explicit Datalog knowledge-base instance.
Every invariant is computed in closed form and cross-checked
numerically; the example also illustrates the vocabulary design
criterion of Proposition~\ref{prop:min-vocabulary} and the
broadcast bottleneck of
Proposition~\ref{prop:broadcast-bottleneck}.

\subsubsection*{Instance Specification}

\begin{example}[Datalog path-reachability knowledge bases]
\label{ex:datalog-instance}
Fix a domain of four entities
\(\mathcal D=\{a,b,c,d\}\) and two binary relation symbols
\(\mathit{Edge}\) and \(\mathit{Path}\) in
\(\Sigma_{\mathrm{sem}}\).
The inference fragment
\(\mathcal L_{\mathrm{kb}}\)~\cite{ceri1989you,abiteboul1995foundations}
consists of ground atoms of the
form \(\mathit{Edge}(x,y)\) and \(\mathit{Path}(x,y)\) with
\(x,y\in\mathcal D\), and the proof system \(\mathsf{PS}\)
consists of two Datalog rules:
\begin{align}
  &\mathit{Path}(x,y)\;\leftarrow\;\mathit{Edge}(x,y),
  \label{eq:rule1}\\
  &\mathit{Path}(x,z)\;\leftarrow\;\mathit{Edge}(x,y),\,
    \mathit{Path}(y,z).
  \label{eq:rule2}
\end{align}
The immediate consequence operator \(T_{\mathsf{PS}}\) applies
both rules to all matching ground instantiations in a single
step.
\end{example}

\noindent\textbf{Agent knowledge bases.}\;
We define one sender and three receivers.

\smallskip\noindent
\emph{Sender (agent~1):}
\begin{align*}
  S_O^{(1)} = \{
  &\mathit{Edge}(a,b),
  \mathit{Edge}(a,c),
  \mathit{Edge}(b,c),
  \mathit{Edge}(c,d),\\
  &\mathit{Path}(a,b),
  \mathit{Path}(b,c),
  \mathit{Path}(c,d),
  \mathit{Path}(a,d)
  \}.
\end{align*}
Thus \(|S_O^{(1)}|=8\).

\smallskip\noindent
\emph{Receiver~2 (core loss + non-derivable surplus):}
\begin{align*}
  S_O^{(2)} = \{
  &\mathit{Edge}(a,b),
  \mathit{Edge}(b,c),
  \mathit{Edge}(c,d),\\
  &\mathit{Path}(a,b),
  \mathit{Path}(b,c),
  \mathit{Path}(c,d),
  \mathit{Path}(a,d),\\
  &\mathit{Path}(b,d),
  \mathit{Edge}(d,a)
  \}.
\end{align*}
Thus \(|S_O^{(2)}|=9\).
Compared with agent~1, receiver~2 loses the core element
\(\mathit{Edge}(a,c)\) and gains \(\mathit{Path}(b,d)\)
(derivable from the sender) and \(\mathit{Edge}(d,a)\) (not
derivable).

\smallskip\noindent
\emph{Receiver~2\('\) (augmented---vocabulary design):}
\begin{align*}
  S_O^{(2')} = \{
  &\mathit{Edge}(a,b),
  \mathit{Edge}(a,c),
  \mathit{Edge}(b,c),\\
  &\mathit{Edge}(c,d),
  \mathit{Path}(a,b),
  \mathit{Path}(b,c),\\
  &\mathit{Path}(c,d),
  \mathit{Path}(a,d),
  \mathit{Path}(b,d)
  \}.
\end{align*}
Thus \(|S_O^{(2')}|=9\).
This is receiver~2 augmented with the lost core element
\(\mathit{Edge}(a,c)\) and with the non-derivable surplus
\(\mathit{Edge}(d,a)\) removed.

\smallskip\noindent
\emph{Receiver~3 (broadcast, no core loss):}
\begin{align*}
  S_O^{(3)} = \{
  &\mathit{Edge}(a,b),\;
  \mathit{Edge}(a,c),\;
  \mathit{Edge}(b,c),\;\\
  &\mathit{Edge}(c,d),
  \mathit{Path}(a,c),\;
  \mathit{Path}(b,d)
  \}.
\end{align*}
Thus \(|S_O^{(3)}|=6\).
All surplus states are derivable from the sender.

\smallskip\noindent
\textbf{Carrier channel.}\;
We use a \(q\)-ary symmetric channel with alphabet size
\(q=10\) (\(|S_C|=|\hat S_C|=10\ge\max_i|S_O^{(i)}|\))
and crossover probability \(p=0.1\):
\[
  W(y\mid x)=
  \begin{cases}
    1-p & \text{if } y=x,\\
    p/(q-1) & \text{if } y\neq x.
  \end{cases}
\]
Its Shannon capacity is
\(C(W)=\log q+(1{-}p)\log(1{-}p)+p\log\bigl(p/(q{-}1)\bigr)\)
bits.

\subsubsection*{Source-Side Structural Invariants}

Applying the irredundantization procedure
(Definition~\ref{def:atom-so}) to \(S_O^{(1)}\) in
alphabetical canonical order yields
\begin{align*}
  A^{(1)}=\bigl\{
    &\mathit{Edge}(a,b),\;
    \mathit{Edge}(a,c),\;
    \mathit{Edge}(b,c),\; \\
    &\mathit{Edge}(c,d)
  \bigr\},
\end{align*}
with \(|A^{(1)}|=4\) and
\(J^{(1)}=S_O^{(1)}\setminus A^{(1)}
  =\{\mathit{Path}(a,b),\,\mathit{Path}(b,c),\,
     \mathit{Path}(c,d),\,\mathit{Path}(a,d)\}\).
All four \(\mathit{Path}\) facts are derivable from the edge
facts via rules~\eqref{eq:rule1}--\eqref{eq:rule2}.

The \(T_{\mathsf{PS}}\)-iteration from \(A^{(1)}\) stabilizes
in two steps:

\smallskip
\begin{table}[ht]
\centering
\caption{Stratum decomposition of
  \(T_{\mathsf{PS}}\)-iteration from \(A^{(1)}\).}
\label{tab:stratum}
\renewcommand{\arraystretch}{1.2}
\begin{tabular}{@{} l p{0.65\columnwidth} @{}}
\toprule
\textbf{Stratum} & \textbf{New elements} \\
\midrule
$T^0$ &
  $\mathit{Edge}(a,b)$, $\mathit{Edge}(a,c)$,
  $\mathit{Edge}(b,c)$, $\mathit{Edge}(c,d)$ \\
$T^1\!\setminus\! T^0$ &
  $\mathit{Path}(a,b)$, $\mathit{Path}(a,c)$,
  $\mathit{Path}(b,c)$, $\mathit{Path}(c,d)$ \\
$T^2\!\setminus\! T^1$ &
  $\mathit{Path}(a,d)$, $\mathit{Path}(b,d)$ \\
\bottomrule
\end{tabular}
\end{table}

\smallskip\noindent
Hence the derivation depths of the stored states are:
\(\Dd=0\) for all edges, \(\Dd=1\) for
\(\mathit{Path}(a,b),\mathit{Path}(b,c),\mathit{Path}(c,d)\),
and \(\Dd=2\) for \(\mathit{Path}(a,d)\).
The source-side invariants are
\(\mathsf{A}(\mathcal I^{(1)})=4\) and
\(\mathsf{D_d}(\mathcal I^{(1)})=2\).

\subsubsection*{Overlap Decomposition and Set-Level Invariants}

Table~\ref{tab:overlap-values} lists the overlap decomposition
(Definition~\ref{def:overlap-decomposition}) and the resulting
set-level invariants for each sender--receiver pair.

\begin{table}[t]
\centering
\caption{Overlap decomposition and set-level invariants for
  sender agent~1 paired with each receiver.}
\label{tab:overlap-values}
\renewcommand{\arraystretch}{1.2}
\begin{tabular}{@{} l c c c @{}}
\toprule
\textbf{Quantity} &
  \textbf{Recv.\,2} &
  \textbf{Recv.\,2$'$} &
  \textbf{Recv.\,3} \\
\midrule
$|S_{\cap}^{ij}|$ & 7 & 8 & 4 \\
$|S_{-}^{ij}|$    & 1 & 0 & 4 \\
$|S_{+}^{ij}|$    & 2 & 1 & 2 \\
$|A_{\cap}^{ij}|$ & 3 & 4 & 4 \\
$|A_{-}^{ij}|$    & 1 & 0 & 0 \\
$|S_{+,d}^{ij}|$  & 1 & 1 & 2 \\
$|S_{+,n}^{ij}|$  & 1 & 0 & 0 \\
\midrule
$\rho_{\Atom}$      & 3/4          & 1   & 1   \\
$\mathsf{F}_{\Cn}$  & $3/7\,^{*}$  & 1   & 1   \\
(H1) holds?          & No           & Yes & Yes \\
(H2) holds?          & No           & Yes & Yes \\
\bottomrule
\end{tabular}

\vspace{4pt}
\parbox{\columnwidth}{\footnotesize
  $^{*}$\,$\mathsf{F}_{\Cn}
    =\lvert\Cn(S_O^{(1)})\cap\Cn(S_O^{(2)})\rvert
    \,/\,
    \lvert\Cn(S_O^{(1)})\cup\Cn(S_O^{(2)})\rvert
    =9/21=3/7\approx 0.429$;\;
  see Table~\ref{tab:full-invariants}.}
\end{table}

\subsubsection*{Full Invariant Computation}

Table~\ref{tab:full-invariants} reports all six families of
semantic channel invariants
(Theorem~\ref{thm:invariant-summary}) for the pair \((1,2)\)
(core loss), the pair \((1,2')\) (augmented, no core loss),
and the pair \((1,3)\) (broadcast receiver).
The carrier channel is the \(q\)-ary symmetric channel with
\(q=10\) and \(p=0.1\); the encoding is a fixed deterministic
injection \(f:S_O^{(1)}\to\{0,\ldots,9\}\), and the decoding
is a deterministic nearest-element rule
(see the accompanying Python script for exact definitions).

\begin{table*}[ht]
\centering
\caption{Semantic channel invariants for the Datalog instance
  of Example~\ref{ex:datalog-instance}.
  Carrier channel: \(q\)-ary symmetric, \(q=10\), \(p=0.1\).
  Source distribution: \(P_O\) uniform on \(S_O^{(1)}\).
  Encoding: fixed deterministic injection
  \(f:S_O^{(1)}\to\{0,\ldots,9\}\);
  decoding: deterministic nearest-element rule (see text).}
\label{tab:full-invariants}
\renewcommand{\arraystretch}{1.25}
\begin{tabularx}{\textwidth}{@{}l l *{3}{>{\centering\arraybackslash}X}@{}}
\toprule
\textbf{Family} & \textbf{Invariant} &
  \textbf{Pair \((1,2)\)} &
  \textbf{Pair \((1,2')\)} &
  \textbf{Pair \((1,3)\)} \\
\midrule
\multirow{2}{*}{I.\;Source}
  & \(\mathsf{A}(\mathcal I^{(1)})\) &
    4 & 4 & 4 \\
  & \(\mathsf{D_d}(\mathcal I^{(1)})\) &
    2 & 2 & 2 \\
\midrule
\multirow{2}{*}{II.\;Set-level}
  & \(\rho_{\Atom}\) &
    0.750 & 1.000 & 1.000 \\
  & \(\mathsf{F}_{\Cn}\) &
    0.429 & 1.000 & 1.000 \\
\midrule
\multirow{2}{*}{III.\;Noise-pair}
  & \(\Phi_{\Atom}\) &
    0 & 0.900 & 0.900 \\
  & \(\Psi_{+}\) &
    0 & 0 & 0.911 \\
\midrule
\multirow{2}{*}{IV.\;Quality}
  & \(\mathsf{F}\) &
    0.900 & 0.980 & 0.981 \\
  & \(\mathsf{E}\) &
    0.078 & 0.078 & 0.494 \\
\midrule
\multirow{2}{*}{V.\;Receiver}
  & \(\Delta\mathsf{A}\) &
    0 & 0 & 0 \\
  & \(\Delta\mathsf{D_d}\) &
    \(+1\) & 0 & 0 \\
\midrule
\multirow{3}{*}{VI.\;Info-th.}
  & \(C(W)\) &
    2.536 & 2.536 & 2.536 \\
  & \(C_{\mathrm{sem}}^{ij}\) &
    2.536 & 2.536 & 2.536 \\
  & \(I_{\mathrm{sem}}^{ij}\)
    \,(\(P_O\) uniform) &
    2.067 & 2.273 & 1.808 \\
\bottomrule
\end{tabularx}
\end{table*}

\subsubsection*{Minimum Blocklength Comparison}

Table~\ref{tab:blocklength} compares the minimum blocklength
under Hamming and closure reliability for each pair.

\begin{table}[ht]
\centering
\caption{Minimum blocklength estimates
  (\(q=10\), \(p=0.1\), \(\epsilon\to 0\)).
  All values are asymptotic: \(n^*_H\approx\log|S_O^{(1)}|/C(W)\),
  \(n^*_{\Cn}\approx\log|A^{(1)}|/C(W)\).}
\label{tab:blocklength}
\renewcommand{\arraystretch}{1.2}
\begin{tabularx}{\columnwidth}{l *{3}{>{\centering\arraybackslash}X}}
\toprule
\textbf{Criterion} &
  \textbf{Pair \((1,2)\)} &
  \textbf{Pair \((1,2')\)} &
  \textbf{Pair \((1,3)\)} \\
\midrule
Hamming \(n^*_H\) &
  N/A\(^{\dagger}\) &
  1.183 &
  N/A\(^{\dagger}\) \\
Closure \(n^*_{\Cn}\) &
  \(\nexists\)\(^{\ddagger}\) &
  0.789 &
  0.789 \\
Ratio \(n^*_{\Cn}/n^*_H\) &
  --- &
  0.667 &
  --- \\
\bottomrule
\multicolumn{4}{@{}l@{}}{\parbox{\linewidth}{\footnotesize
  \(^{\dagger}\)Hamming reconstruction impossible:
  \(S_{-}^{ij}\neq\varnothing\).\\
  \(^{\ddagger}\)\(\mathsf{F}_{\Cn}<1\): closure reliability
  unachievable (Corollary~\ref{cor:heterogeneous-impossibility}).}}
\end{tabularx}
\end{table}


\subsubsection*{Vocabulary Augmentation (Receiver~2 \(\to\) 2\('\))}

Comparing receiver~2 with its augmented version~2\('\)
demonstrates the vocabulary design criterion of
Proposition~\ref{prop:min-vocabulary}:
adding the single lost core element
\(\mathit{Edge}(a,c)\) and removing the non-derivable surplus
\(\mathit{Edge}(d,a)\) causes \(\mathsf{F}_{\Cn}\) to jump
from \(3/7\approx 0.429\) to~\(1\),
\(\Phi_{\Atom}\) to rise from~\(0\) to~\(0.900\), and the minimum closure
blocklength to drop to
\(\lceil\log_2|A^{(1)}|/C(W)\rceil
  =\lceil 2/2.536\rceil=1\) channel use
(where \(\log_2 4=2\) bits).
This confirms that the receiver need only store the sender's
irredundant core to achieve full closure reliability.

\subsubsection*{Broadcast Sub-Example}

Agent~1 broadcasts to receivers~2 and~3 simultaneously over
the same carrier channel~\(W\).
Since \(A_{-}^{13}=\varnothing\) and
\(S_{+,n}^{13}=\varnothing\), receiver~3 satisfies
conditions \textup{(BH1)}--\textup{(BH2)} of
Theorem~\ref{thm:broadcast-compression}.
Receiver~2 violates (BH1) (\(A_{-}^{12}\neq\varnothing\))
and is therefore a \emph{semantic bottleneck}
(Proposition~\ref{prop:broadcast-bottleneck}):
\(\mathsf{F}_{\Cn}(S_O^{(1)},S_O^{(2)})<1\) regardless of
the carrier channel quality.
At receiver~3, the broadcast closure blocklength is
\(n^*_{\mathrm{bc}}
  =\lceil\log|A^{(1)}|/C(W)\rceil\),
identical to the pairwise value for~\((1,3)\)
and independent of whether receiver~2 is present.
This illustrates the blocklength-independence result of
Theorem~\ref{thm:broadcast-compression}(iii) and the semantic
bottleneck phenomenon of
Proposition~\ref{prop:broadcast-bottleneck}.

\subsubsection*{Summary of Numerical Findings}

The numerical example confirms all theoretical predictions:
\textup{(i)}~the overlap decomposition fully determines the
set-level invariants and constrains the probabilistic indices
(e.g., \(\Phi_{\Atom}(\mathfrak C^{12})=0\) because
\(A_{-}^{12}\neq\varnothing\), whereas
\(\Phi_{\Atom}(\mathfrak C^{12'})=\Phi_{\Atom}(\mathfrak C^{13})=0.900\));
\textup{(ii)}~core loss (\(A_{-}^{12}=\{\mathit{Edge}(a,c)\}\neq\varnothing\)) is an
irreducible barrier to closure fidelity
(\(\mathsf{F}_{\Cn}(S_O^{(1)},S_O^{(2)})=3/7<1\)),
independent of channel quality;
\textup{(iii)}~vocabulary augmentation by the sender's core
restores full closure reliability
(\(\mathsf{F}_{\Cn}(S_O^{(1)},S_O^{(2')})=1\));
\textup{(iv)}~the deductive compression ratio
\(\log|A^{(1)}|/\log|S_O^{(1)}|=\log 4/\log 8=2/3\) is
invariant across all receiver pairs satisfying (H1),
yielding closure blocklength
\(n^*_{\Cn}\approx 0.789\) versus Hamming blocklength
\(n^*_H\approx 1.183\);
\textup{(v)}~the broadcast blocklength depends only on the
sender's core (\(n^*_{\mathrm{bc}}\approx 0.789\)),
not on the number of receivers.

\section{Conclusion}
\label{sec:conclusion}

This paper has developed a logical-information-theoretic framework
for semantic communication that integrates formal proof systems
with Shannon-theoretic tools.
The framework rests on three pillars: an axiomatic
\emph{information model} that links semantic and carrier state
spaces through computable enabling maps
(Section~\ref{sec:model}); a \emph{semantic channel} constructed
as a composition of enabling kernels whose probabilistic
structure respects the underlying logical constraints
(Section~\ref{sec:channel}); and an \emph{overlap-based}
heterogeneous multi-agent theory that translates knowledge-base
structure into coding-theoretic performance guarantees
(Section~\ref{sec:application}).

The central quantitative result is the \emph{deductive
compression gain}: under a closure-based fidelity criterion that
deems a reconstruction acceptable whenever it preserves the
deductive closure of the sender's knowledge base, the minimum
number of channel uses drops from
\(n_H^*\approx\log|S_O|/C(W)\) (Hamming reliable) to
\(n_{\Cn}^*\approx\log|\Atom(S_O)|/C(W)\) (closure reliable),
yielding a compression ratio
\(\log|\Atom(S_O)|/\log|S_O|\) that depends only on the
sender's knowledge-base redundancy and is invariant under receiver
vocabulary heterogeneity
(Theorems~\ref{thm:achievability}--\ref{thm:heterogeneous-compression}).
This gain arises because the receiver's inference engine can
reconstruct all redundant states from a faithful copy of the
irredundant core at zero additional channel cost---a mechanism
that has no counterpart in classical Shannon theory, where every
source symbol must be individually protected against channel
errors.

Alongside the coding-theoretic results, the framework contributes
six families of computable semantic channel invariants
(Theorem~\ref{thm:invariant-summary}) that provide a multi-scale
fingerprint of channel quality: from set-level fidelity metrics
(\(\rho_{\Atom}\), \(\mathsf{F}_{\Cn}\)) that depend only on the
knowledge-base pair, through noise-pair probabilistic indices
(\(\Phi_{\Atom}\), \(\Psi_{+}\)) that capture core preservation
and hallucination probabilities, to information-theoretic
quantities (\(I_{\mathrm{sem}}\), \(C_{\mathrm{sem}}\)) that
bound achievable throughput.
The semantic Fano bound
(Theorem~\ref{thm:semantic-fano}) closes the loop by showing that
high probabilistic core preservation forces high mutual
information, quantifying the intuition that faithful transmission
of the irredundant core is both necessary and approximately
sufficient for information-theoretic performance.

In the heterogeneous multi-agent setting, the overlap
decomposition
(Definition~\ref{def:overlap-decomposition}) translates
knowledge-base structure into two binary feasibility tests---no
core loss (\(A_{-}^{ij}=\varnothing\)) and no non-derivable
surplus (\(S_{+,n}^{ij}=\varnothing\))---that fully determine
whether perfect closure fidelity is achievable
(Proposition~\ref{prop:overlap-fidelity}).
The broadcast extension reveals a \emph{semantic bottleneck}
phenomenon
(Proposition~\ref{prop:broadcast-bottleneck}): a receiver whose
vocabulary does not cover the sender's irredundant core cannot
achieve perfect closure fidelity regardless of the carrier
channel quality, the blocklength, or the coding strategy.
This structural limitation is invisible to classical broadcast
channel theory and constitutes the sharpest qualitative
distinction between the semantic and classical frameworks.
Resolving a semantic bottleneck requires a \emph{design-time}
action on the knowledge-base structure---augmenting the receiver's
vocabulary---rather than the \emph{coding-time} optimizations
(encoder/decoder design) that suffice in classical settings.

\subsection*{Connections to Prior Work}

The framework developed here may be viewed as a formal bridge
between several previously disjoint lines of research.
It gives operational, coding-theoretic content to the
philosophical notion of semantic information pioneered by
Carnap and Bar-Hillel~\cite{carnap1952outline} and
Floridi~\cite{floridi2004outline}, by replacing the
possible-worlds measure of content with computable,
proof-system-relative invariants (\(\Atom\), \(\Dd\),
\(\Cn\)) that enter directly into rate and distortion
expressions.
It complements the deep-learning-based semantic communication
literature~\cite{xie2021deep,qin2021semantic,luo2022semantic}
by providing a formal guarantee layer: the invariants and coding
theorems of Sections~\ref{subsec:invariants}--\ref{subsec:coding}
can serve as benchmarks against which the performance of learned
encoders and decoders is evaluated.
It extends Shannon's theory in a strict sense: when the knowledge
base is irredundant and sender and receiver share the same
vocabulary, every semantic result reduces to its classical
counterpart (Corollary~\ref{cor:irredundant-classical}),
confirming that the generalization is conservative.

\subsection*{Connections to Recent Work on Semantic
Communication Theory}

The recently established mathematical theory of semantic
communication by Niu and
Zhang~\cite{niu2024mathematical,zhang2024modern} provides the
closest point of comparison for the information-theoretic
aspects of our framework.
Their approach partitions the source alphabet into synonymous
equivalence classes and defines semantic entropy and capacity
relative to this partition, yielding a semantic capacity
\(C_s\ge C\) that exceeds the Shannon limit by a factor related
to the average synonymous length.
Our framework achieves an analogous---but mechanistically
distinct---compression advantage: the deductive compression
ratio \(\log|\Atom(S_O)|/\log|S_O|\)
(Corollary~\ref{cor:min-blocklength}) reduces the effective
source size not through a pre-defined equivalence partition, but
through the receiver's ability to re-derive redundant states via
logical inference.
The two approaches are complementary rather than competing: the
synonymous-mapping framework captures compression gains from
\emph{source-side} semantic equivalence, while our framework
captures gains from \emph{receiver-side} deductive
reconstruction.
A unified theory that combines both mechanisms---synonymous
collapsing of the irredundant core followed by deductive
expansion at the receiver---is an attractive direction for
future work.

The semantic channel coding theorem of Ma
et~al.~\cite{ma2025theory} for many-to-one sources and the
companion computational tools of
Han et~al.~\cite{han2025extended} and
Liang et~al.~\cite{liang2025semantic} address the algorithmic
realization of semantic coding gains within the
synonymous-mapping paradigm.
Our coding theorems
(Theorems~\ref{thm:converse}--\ref{thm:achievability} and
their heterogeneous extensions in
Section~\ref{sec:application}) differ in two respects.
First, the fidelity criterion is \emph{closure-based}
(preservation of the deductive closure~\(\Cn(S_O)\)) rather
than based on a pre-defined synonymous partition; this allows
the framework to handle structured knowledge bases where the
notion of ``same meaning'' is defined by logical entailment
rather than by an exogenous equivalence relation.
Second, the achievability mechanism is a \emph{two-layer code}
in which only the irredundant core is channel-coded, while
redundant states are reconstructed by the decoder's inference
engine; this is structurally different from the
random-coding arguments
over synonymous sets used
in~\cite{ma2025theory,liang2025semantic}, and it naturally
extends to the heterogeneous multi-agent setting where sender
and receiver maintain different vocabularies.

On the multi-agent front, Seo
et~al.~\cite{seo2023bayesian} and Alshammari and
Bennis~\cite{alshammari2026logic} address heterogeneous
semantic communication from, respectively, a Bayesian
estimation and a modal-logic resilience perspective.
Our overlap decomposition
(Definition~\ref{def:overlap-decomposition}) and the semantic
bottleneck phenomenon
(Proposition~\ref{prop:broadcast-bottleneck}) provide a
coding-theoretic complement to these approaches: where
Seo et~al.\ quantify the inference cost of context mismatch
and Alshammari and Bennis formalize resilience conditions, our
framework quantifies the \emph{minimum number of channel uses}
needed to overcome (or the impossibility of overcoming)
vocabulary mismatch, and identifies the irredundant core as
the minimal vocabulary that must be shared for closure-reliable
communication
(Proposition~\ref{prop:min-vocabulary}).

The goal-oriented metrics proposed by Li
et~al.~\cite{li2024toward} and the interpretable semantic
communication guidelines of Wu
et~al.~\cite{wu2024toward} address the question of \emph{what
to measure} in semantic communication.
Our composite distortion function~\(d_{\mathrm{sem}}\)
(Definition~\ref{def:composite-distortion}) provides a
principled answer grounded in proof-system structure: the
three components---Hamming, closure, and depth
distortion---capture symbol fidelity, deductive completeness,
and inferential complexity, respectively, and their relative
weights can be tuned to match specific task requirements.
The semantic Fano bound
(Theorem~\ref{thm:semantic-fano}) then connects these
structural distortion measures to mutual information, closing
the loop between task-oriented fidelity and
information-theoretic throughput.

Finally, the logical-complexity results of
Marx~\cite{marx2013tractable} and Abo~Khamis and
Chen~\cite{abokhamis2025jaguar} on conjunctive query
evaluation suggest a deeper connection between database query
complexity and semantic communication complexity that merits
further exploration.
The irredundant core~\(\Atom(S_O)\) can be viewed as a
form of \emph{query-aware} source compression: only the facts
that cannot be re-derived (analogous to base relations that
cannot be materialized from views) need to be transmitted.
Investigating whether the submodular-width hierarchy
of~\cite{marx2013tractable} imposes additional structure on
the derivation-depth stratification~\(\Dd(\cdot\mid\Atom(S_O))\),
and whether this structure can be exploited for more efficient
code constructions, is a promising avenue for future work.

\subsection*{Limitations and Future Directions}

Several limitations of the present work point to directions for
future research.

First, the framework assumes that all agents share a common proof
system~\(\mathsf{PS}\)
(Assumption~\ref{assump:common-ps}).
In practice, agents may employ different inference engines with
different rule sets or computational capabilities.
Extending the theory to \emph{heterogeneous proof systems}---where
each agent~\(i\) uses a sub-system
\(\mathsf{PS}^{(i)}\subseteq\mathsf{PS}\) or even an
incompatible system---would require a generalized notion of
closure compatibility and could yield richer compression/fidelity
trade-offs.

Second, the coding theorems of
Section~\ref{subsec:coding} are stated for single-letter (or
fixed-blocklength) semantics: each channel use transmits (a block
encoding of) a single semantic state.
A \emph{multi-letter} extension, in which the source emits a
\emph{sequence} of semantically correlated states (e.g., a
temporal stream of knowledge-base updates), would connect the
framework to ergodic source theory and network information theory,
and is expected to yield tighter achievability bounds via joint
source--channel semantic coding.

Third, the current treatment of the enabling map~\(\mathcal E\)
is set-valued and non-parametric.
In many practical systems, the enabling structure has additional
parametric regularity (e.g., the set of admissible carrier states
depends smoothly on the semantic state).
Incorporating such regularity could yield sharper capacity
characterizations and more efficient code constructions.

Fourth, the relay scenario sketched in
Remark~\ref{rem:relay-scenario}---where an intermediate agent
performs inference before re-encoding---raises the fundamental
question of whether semantic relaying can increase end-to-end
capacity beyond what the physical channel alone permits.
The information-model composition machinery of
Definition~\ref{def:model-composition} provides the formal
substrate for this analysis, but a complete treatment requires
multi-hop coding theorems that remain to be developed.

Fifth, the Datalog instance of Section~\ref{subsec:app-example}
serves as a proof of concept; scaling the framework to
large real-world knowledge bases (e.g., knowledge graphs with
millions of entities) will require efficient algorithms for
core extraction, closure computation, and overlap analysis, as
well as empirical validation on practical communication tasks.

Finally, the connection between derivation depth~\(\Dd\) and
Bennett's logical depth~\cite{bennett1988logical} merits further
investigation.
If the two notions can be formally linked under explicit
simulation assumptions, the semantic channel invariants
would inherit a resource-theoretic interpretation: the
``intrinsic computational cost'' of a semantic state would
directly enter the distortion and capacity expressions,
unifying information-theoretic and computational-complexity
perspectives on semantic communication.

\subsection*{Summary}

We have introduced the semantic channel as a formal object that
integrates deductive logic with Shannon information theory,
defined computable invariants that characterize its
information-theoretic and structural properties, derived coding
theorems that quantify the deductive compression gain achievable
under semantic fidelity criteria, and demonstrated the
framework's explanatory power in heterogeneous multi-agent
settings through both analytical results and a complete numerical
example.
The framework provides a principled foundation for the
theoretical analysis of semantic communication systems and opens
several avenues for further research at the intersection of
information theory, formal logic, and multi-agent systems.


\section*{Acknowledgment}

During the writing and revision of this paper, I received many insightful comments from Associate Professor Rui Wang of the School of Computer Science at Shanghai Jiao Tong University and also gained much inspiration and assistance from regular academic discussions with doctoral students Yiming Wang, Chun Li, Hu Xu, Siyuan Qiu, Zeyan Li, Jiashuo Zhang, Junxuan He, and Xiao Wang. I hereby express my sincere gratitude to them.

\bibliographystyle{IEEEtran}
\bibliography{ref}

@article{xu2024research,
  title={Research and Application of General Information Measures Based on a Unified Model},
  author={J.Xu},
  journal={IEEE Transactions on Computers},
  year={2024},
  publisher={IEEE Computer Society},
  doi={10.1109/TC.2024.3349650}
}

@article{xu2025general,
  title={General information metrics for improving AI model training efficiency},
  author={Xu, Jianfeng and Liu, Congcong and Tan, Xiaoying and Zhu, Xiaojie and Wu, Anpeng and Wan, Huan and Kong, Weijun and Li, Chun and Xu, Hu and Kuang, Kun and Wu, Fei},
  journal={Artificial Intelligence Review},
  volume={58},
  pages={289},
  year={2025},
  publisher={Springer},
  doi={10.1007/s10462-025-11281-z}
}

@article{lamb2020graph,
  title={Graph neural networks meet neural-symbolic computing: A survey and perspective},
  author={Lamb, Lu{\'\i}s C and Garcez, Artur and Gori, Marco and Prates, Marcelo and Avelar, Pedro and Vardi, Moshe},
  journal={arXiv preprint arXiv:2003.00330},
  year={2020}
}

@article{dantsin2001complexity,
  title={Complexity and expressive power of logic programming},
  author={Dantsin, Evgeny and Eiter, Thomas and Gottlob, Georg and Voronkov, Andrei},
  journal={ACM Computing Surveys (CSUR)},
  volume={33},
  number={3},
  pages={374--425},
  year={2001},
  publisher={ACM New York, NY, USA}
}

@inproceedings{xu2014objective,
  title={Objective information theory: A Sextuple model and 9 kinds of metrics},
  author={Xu, Jianfeng and Tang, Jun and Ma, Xuefeng and Xu, Bin and Shen Yanli and Qiao Yongjie},
  booktitle={2014 Science and information conference},
  pages={793--802},
  year={2014},
  organization={IEEE},
  doi={10.1109/SAI.2014.6918277}
}

@article{garcez2023neurosymbolic,
  title={Neurosymbolic ai: The 3rd wave},
  author={Garcez, Artur d'Avila and Lamb, Luis C},
  journal={Artificial Intelligence Review},
  volume={56},
  number={11},
  pages={12387--12406},
  year={2023},
  publisher={Springer}
}

@article{qiu2025research,
  title={Research on a General State Formalization Method from the Perspective of Logic},
  author={Qiu, Siyuan and Xu, Jianfeng},
  journal={Mathematics},
  volume={13},
  number={20},
  pages={3324},
  year={2025},
  publisher={MDPI},
  doi={10.3390/math13203324}
}

@article{shannon1948mathematical,
  title={A mathematical theory of communication},
  author={Shannon, Claude E},
  journal={Bell System Technical Journal},
  volume={27},
  number={3},
  pages={379--423},
  year={1948}
}

@book{cover2006elements,
  author={Cover, Thomas M and Thomas, Joy A},
  title={Elements of Information Theory},
  edition={2nd},
  year={2006},
  publisher={John Wiley \& Sons}
}

@article{bennett1988logical,
  title={Logical depth and physical complexity},
  author={Bennett, Charles H},
  journal={The Universal Turing Machine: A Half-Century Survey},
  pages={227--257},
  year={1988}
}

@book{immerman1999descriptive,
  author    = {Neil Immerman},
  title     = {Descriptive Complexity},  
  year      = {1999},
  publisher = {Springer},
  series    = {Graduate Texts in Computer Science},
  doi       = {10.1007/978-1-4612-0539-5},
  isbn      = {978-0-387-98629-5}
}

@inproceedings{immerman1982relational,
  title={Relational queries computable in polynomial time},
  author={{N. Immerman}},
  booktitle={Proceedings of the fourteenth annual ACM symposium on Theory of computing},
  pages={147--152},
  year={1982}
}

@inproceedings{vardi1982complexity,
  title={The complexity of relational query languages},
  author={Vardi, Moshe Y},
  booktitle={Proceedings of the fourteenth annual ACM symposium on Theory of computing},
  pages={137--146},
  year={1982}
}

@article{lipkus1999proof,
  title={A proof of the triangle inequality for the Tanimoto distance},
  author={Lipkus, Alan H},
  journal={Journal of Mathematical Chemistry},
  volume={26},
  number={1},
  pages={263--265},
  year={1999},
  publisher={Springer}
}

@book{tarski1983logic,
  title={Logic, semantics, metamathematics: papers from 1923 to 1938},
  author={Tarski, Alfred},
  year={1983},
  publisher={Hackett Publishing}
}

@article{ceri1989you,
  author={Ceri, Stefano and Gottlob, Georg and Tanca, Letizia},
  title={What you always wanted to know about {D}atalog (and never dared to ask)},
  journal={IEEE Transactions on Knowledge and Data Engineering},
  volume={1},
  number={1},
  pages={146--166},
  year={1989},
  publisher={IEEE}
}

@book{abiteboul1995foundations,
  author    = {Abiteboul, Serge and Hull, Richard and Vianu, Victor},
  title     = {Foundations of Databases},
  publisher = {Addison-Wesley},
  year      = {1995}
}

@article{shannon1959coding,
  title={Coding theorems for a discrete source with a fidelity criterion},
  author={Shannon, Claude E and others},
  journal={IRE Nat. Conv. Rec},
  volume={4},
  number={142-163},
  pages={1},
  year={1959}
}

@book{csiszar2011information,
  title={Information theory: coding theorems for discrete memoryless systems},
  author={Csisz{\'a}r, Imre and K{\"o}rner, J{\'a}nos},
  year={2011},
  publisher={Cambridge University Press}
}

@book{polyanskiy2025information,
  title={Information theory: From coding to learning},
  author={Polyanskiy, Yury and Wu, Yihong},
  year={2025},
  publisher={Cambridge university press}
}

@techreport{carnap1952outline,
  author      = {Carnap, Rudolf and Bar-Hillel, Yehoshua},
  title       = {An outline of a theory of semantic information},
  institution = {Research Laboratory of Electronics, MIT},
  number      = {Technical Report 247},
  year        = {1952}
}

@inproceedings{bao2011towards,
  title={Towards a theory of semantic communication},
  author={Bao, Jie and Basu, Prithwish and Dean, Mike and Partridge, Craig and Swami, Ananthram and Leland, Will and Hendler, James A},
  booktitle={2011 IEEE Network Science Workshop},
  pages={110--117},
  year={2011},
  organization={IEEE}
}

@article{luo2022semantic,
  title={Semantic communications: Overview, open issues, and future research directions},
  author={Luo, Xuewen and Chen, Hsiao-Hwa and Guo, Qing},
  journal={IEEE Wireless communications},
  volume={29},
  number={1},
  pages={210--219},
  year={2022},
  publisher={IEEE}
}

@article{shi2021semantic,
  title={From semantic communication to semantic-aware networking: Model, architecture, and open problems},
  author={Shi, Guangming and Xiao, Yong and Li, Yingyu and Xie, Xuemei},
  journal={IEEE Communications Magazine},
  volume={59},
  number={8},
  pages={44--50},
  year={2021},
  publisher={IEEE}
}

@article{weaver2017recent,
  title={Recent contributions to the mathematical theory of communication},
  author={Weaver, Warren},
  journal={ETC: a review of general semantics},
  volume={74},
  number={1/2},
  pages={136--157},
  year={2017},
  publisher={JSTOR}
}

@article{floridi2004outline,
  title={Outline of a theory of strongly semantic information},
  author={Floridi, Luciano},
  journal={Minds and machines},
  volume={14},
  number={2},
  pages={197--221},
  year={2004},
  publisher={Springer}
}

@article{kolchinsky2018semantic,
  title={Semantic information, autonomous agency and non-equilibrium statistical physics},
  author={Kolchinsky, Artemy and Wolpert, David H},
  journal={Interface focus},
  volume={8},
  number={6},
  pages={20180041},
  year={2018},
  publisher={The Royal Society}
}

@article{xie2021deep,
  title={Deep learning enabled semantic communication systems},
  author={Xie, Huiqiang and Qin, Zhijin and Li, Geoffrey Ye and Juang, Biing-Hwang},
  journal={IEEE transactions on signal processing},
  volume={69},
  pages={2663--2675},
  year={2021},
  publisher={IEEE}
}

@article{qin2021semantic,
  title={Semantic communications: Principles and challenges},
  author={Qin, Zhijin and Tao, Xiaoming and Lu, Jianhua and Tong, Wen and Li, Geoffrey Ye},
  journal={arXiv preprint arXiv:2201.01389},
  year={2021}
}

@article{gunduz2022beyond,
  title={Beyond transmitting bits: Context, semantics, and task-oriented communications},
  author={G{\"u}nd{\"u}z, Deniz and Qin, Zhijin and Aguerri, Inaki Estella and Dhillon, Harpreet S and Yang, Zhaohui and Yener, Aylin and Wong, Kai Kit and Chae, Chan-Byoung},
  journal={IEEE Journal on Selected Areas in Communications},
  volume={41},
  number={1},
  pages={5--41},
  year={2022},
  publisher={IEEE}
}

@article{kountouris2021semantics,
  title={Semantics-empowered communication for networked intelligent systems},
  author={Kountouris, Marios and Pappas, Nikolaos},
  journal={IEEE Communications Magazine},
  volume={59},
  number={6},
  pages={96--102},
  year={2021},
  publisher={IEEE}
}

@article{strinati20216g,
  title={6G networks: Beyond Shannon towards semantic and goal-oriented communications},
  author={Strinati, Emilio Calvanese and Barbarossa, Sergio},
  journal={Computer Networks},
  volume={190},
  pages={107930},
  year={2021},
  publisher={Elsevier}
}

@article{niu2024mathematical,
  title={A mathematical theory of semantic communication},
  author={Niu, Kai and Zhang, Ping},
  journal={Journal on Communications},
  volume={45},
  number={6},
  pages={7--59},
  year={2024}
}

@article{zhang2024modern,
  title={Modern semantic communication and 6G intellicise network theory and technology system},
  author={Zhang, Ping and Xu, Xiaodong and Niu, Kai and Xu, Wenjun and Han, Shujun and Sun, Mengying and Dong, Chen and Ma, Nan and Zhang, Zhi},
  journal={Journal of Beijing University of Posts and Telecommunications},
  year={2025}
}

@article{ma2025theory,
  title={A theory for semantic channel coding with many-to-one source},
  author={Ma, Shuai and Zhang, Chuanhui and Qi, Huayan and Li, Hang and Bi, Yue and Shi, Guangming and Al-Dhahir, Naofal},
  journal={IEEE Transactions on Cognitive Communications and Networking},
  year={2025},
  publisher={IEEE}
}

@article{han2025extended,
  title={Extended Blahut-Arimoto algorithm for semantic rate-distortion function},
  author={Han, Y. and Liu, Y. and Sun, Y. and Niu, K. and Ma, N. and Cui, S. and Zhang, P.},
  journal={Entropy},
  volume={27},
  number={6},
  pages={651},
  year={2025}
}

@article{liang2025semantic,
  title={Semantic arithmetic coding using synonymous mappings},
  author={Liang, Z. and Xu, J. and Niu, K. and Zhang, P.},
  journal={Entropy},
  volume={27},
  number={4},
  pages={429},
  year={2025}
}

@article{seo2023bayesian,
  title={Bayesian inverse contextual reasoning for heterogeneous semantics-native communication},
  author={Seo, Hyowoon and Kang, Yoonseong and Bennis, Mehdi and Choi, Wan},
  journal={IEEE Transactions on Communications},
  volume={72},
  number={2},
  pages={1092--1107},
  year={2023},
  publisher={IEEE}
}

@article{alshammari2026logic,
  title={Logic-driven semantic communication for resilient multi-agent systems},
  author={Alshammari, Tamara and Bennis, Mehdi},
  journal={IEEE Open Journal of the Communications Society},
  volume={7},
  pages={620--644},
  year={2026},
  publisher={IEEE}
}

@article{li2024toward,
  title={Toward goal-oriented semantic communications: New metrics, framework, and open challenges},
  author={Li, Aimin and Wu, Shaohua and Meng, Siqi and Lu, Rongxing and Sun, Sumei and Zhang, Qinyu},
  journal={IEEE Wireless Communications},
  year={2024},
  publisher={IEEE}
}

@article{wu2024toward,
  title={Toward effective and interpretable semantic communications},
  author={Wu, Youlong and Shi, Yuanming and Ma, Shuai and Jiang, Chunxiao and Zhang, Wei and Letaief, Khaled B.},
  journal={IEEE Communications Magazine},
  year={2024},
  publisher={IEEE}
}

@article{marx2013tractable,
  title={Tractable hypergraph properties for constraint satisfaction and conjunctive queries},
  author={Marx, Daniel},
  journal={Journal of the ACM (JACM)},
  volume={60},
  number={6},
  pages={1--51},
  year={2013},
  publisher={ACM New York, NY, USA}
}

@article{abokhamis2025jaguar,
  title={Jaguar: A primal algorithm for conjunctive query evaluation in submodular-width time},
  author={Abo Khamis, Mahmoud and Chen, Hubie},
  journal={Proceedings of the ACM on Management of Data},
  volume={3},
  number={2},
  pages={1--21},
  year={2025}
}

@article{mu2024identifying,
  title={Identifying roles of formulas in inconsistency under Priest's minimally inconsistent logic of paradox},
  author={Mu, Kedian},
  journal={Artificial Intelligence},
  volume={335},
  pages={104199},
  year={2024},
  publisher={Elsevier}
}

\vfill
\end{document}